\documentclass[journal]{IEEEtran}
\usepackage[utf8]{inputenc}
\usepackage{amsmath,amssymb,amsfonts}
\usepackage[nospace]{cite}

\usepackage{algorithmic}
\usepackage{multirow}

\usepackage{graphicx}
\usepackage{textcomp}
\usepackage{xcolor}
\usepackage{pifont}
\usepackage{amsthm}
\usepackage{hyperref}
\usepackage{mathrsfs}
\usepackage{booktabs}
\usepackage{hyperref}
\usepackage{cases}
\usepackage{comment}
\usepackage{empheq} 
\usepackage{caption}
\usepackage{subcaption}
\usepackage{fancyhdr}

\allowdisplaybreaks

\usepackage{url}  
\usepackage{graphicx}  
\usepackage[ruled,vlined,linesnumbered]{algorithm2e}
\usepackage{setspace}
\usepackage{floatpag} 
\usepackage{float}
\floatpagestyle{empty}

\newtheorem{defn}{Definition}

\newtheorem{lem}{Lemma}

\newtheorem{rem}{Remark}

\newcommand{\bp}[1]{{\mathbb{P}}\left[{#1}\right]}
\newcommand{\bfu}[1]{{\mathbb{F}}\left[{#1}\right]}
\newcommand{\qeda}{\hfill\ensuremath{\blacksquare}}
\newcommand{\bE}[1]{{\mathbb{E}}\left[{#1}\right]}

\usepackage{listings}

\title{Local Differential Privacy based Federated Learning for the Internet of Things}

\author{Yang~Zhao,~\IEEEmembership{Graduate Student Member,~IEEE,}
        Jun~Zhao,~\IEEEmembership{Member,~IEEE,}
		Mengmeng~Yang,~\IEEEmembership{Member,~IEEE,}
        Teng~Wang,~\IEEEmembership{Member,~IEEE,}
        Ning~Wang,~\IEEEmembership{Member,~IEEE,}
        Lingjuan~Lyu,~\IEEEmembership{Member,~IEEE,}
        Dusit~Niyato,~\IEEEmembership{Fellow,~IEEE,}
        and Kwok-Yan~Lam,~\IEEEmembership{Senior Member,~IEEE}
\thanks{
Yang Zhao and Jun Zhao are supported by 1) Nanyang Technological University (NTU) Startup Grant, 2) Alibaba-NTU Singapore Joint Research Institute (JRI), 3) Singapore Ministry of Education Academic Research Fund Tier 1 RG128/18, Tier 1 RG115/19, Tier 1 RT07/19, Tier 1 RT01/19, and Tier 2 MOE2019-T2-1-176, 4) NTU-WASP Joint Project, 5) Singapore National Research Foundation (NRF) under its Strategic Capability Research Centres Funding Initiative: Strategic Centre for Research in Privacy-Preserving Technologies \& Systems (SCRIPTS), 6)  Energy Research Institute @NTU (ERIAN), 7) Singapore NRF National Satellite of Excellence, Design Science and Technology for Secure Critical Infrastructure NSoE DeST-SCI2019-0012, 8) AI Singapore (AISG) 100 Experiments (100E) programme, and 9) NTU Project for Large Vertical Take-Off \& Landing (VTOL) Research Platform.
Mengmeng~Yang and Kwok-Yan~Lam are supported by the National Research Foundation, Prime Minister’s Office, Singapore under its Strategic Capability Research Centres Funding Initiative.
Ning~Wang is supported by the National Natural Science Foundation of China (61902365) and China Postdoctoral Science Foundation Grant (2019M652473). 
Dusit~Niyato is supported by the National Research Foundation (NRF), Singapore, under Singapore Energy Market Authority (EMA), Energy Resilience, NRF2017EWT-EP003-041, Singapore NRF2015-NRF-ISF001-2277, Singapore NRF National Satellite of Excellence, Design Science and Technology for Secure Critical Infrastructure NSoE DeST-SCI2019-0007, A*STAR-NTU-SUTD Joint Research Grant on Artificial Intelligence for the Future of Manufacturing RGANS1906, Wallenberg AI, Autonomous Systems and Software Program and Nanyang Technological University (WASP/NTU) under grant M4082187 (4080), Singapore Ministry of Education (MOE) Tier 1 (RG16/20), and NTU-WeBank JRI (NWJ-2020-004), Alibaba Group through Alibaba Innovative Research (AIR) Program and Alibaba-NTU Singapore Joint Research Institute (JRI).
(Corresponding author: Jun Zhao).

Yang Zhao, Jun Zhao, Mengmeng~Yang, Dusit Niyato and Kwok-Yan~Lam are with School of Computer Science and Engineering, Nanyang Technological University, Singapore, 639798. (Emails: s180049@e.ntu.edu.sg, junzhao@ntu.edu.sg, melody.yang@ntu.edu.sg, dniyato@ntu.edu.sg, kwokyan.lam@ntu.edu.sg).}
\thanks{Teng~Wang is with School of Cyberspace Security, Xi'an University of Posts $\&$ Telecommunications. (Email: wangteng@xupt.edu.cn).
}
\thanks{Ning~Wang is with College of Information Science and Engineering, Ocean University of China, 266102. (Email: wangning8687@ouc.edu.cn).
}

\thanks{Lingjuan~Lyu is with Department of Computer Science, National University of Singapore, 117417. (Email: lingjuanlvsmile@gmail.com).
}
}
\begin{document}

\maketitle

\thispagestyle{fancy}
\pagestyle{fancy}
\lhead{This paper appears in IEEE Internet of Things Journal (IoT-J). Please feel free to contact us for questions or remarks.}
\cfoot{\thepage}
\renewcommand{\headrulewidth}{0.4pt}
\renewcommand{\footrulewidth}{0pt}

\begin{abstract}
The Internet of Vehicles (IoV) is a promising branch of the Internet of Things.
IoV simulates a large variety of crowdsourcing applications such as Waze, Uber, and Amazon Mechanical Turk, etc. Users of these applications report the real-time traffic information to the cloud server which trains a machine learning model based on traffic information reported by users for intelligent traffic management.  However, crowdsourcing application owners can easily infer users' location information, traffic information, motor vehicle information, environmental information, etc., which raises severe sensitive personal information privacy concerns of the users. In addition, as the number of vehicles increases, the frequent communication  between vehicles and the cloud server incurs unexpected amount of communication cost. To avoid the privacy threat and reduce the communication cost, in this paper, we propose to integrate federated learning and local differential privacy (LDP) to facilitate the crowdsourcing applications to achieve the machine learning model.
Specifically, we propose four LDP  mechanisms to perturb gradients generated by vehicles.  The proposed \texttt{Three-Outputs} mechanism 
introduces three different output possibilities 
to deliver a high accuracy when the privacy budget is small. The output possibilities of \texttt{Three-Outputs} can be encoded with two bits to reduce the communication cost. Besides, to maximize the performance when the privacy budget is large, an optimal piecewise mechanism (\texttt{PM-OPT}) is proposed. We 
further propose a suboptimal mechanism (\texttt{PM-SUB}) with a simple formula and comparable utility to \texttt{PM-OPT}. Then, we build a novel hybrid mechanism by combining \texttt{Three-Outputs} and \texttt{PM-SUB}.
Finally, an \texttt{LDP-FedSGD} algorithm is proposed to coordinate the cloud server and vehicles to train  the model collaboratively. Extensive experimental results on real-world datasets 
validate that our proposed algorithms are capable of protecting privacy while guaranteeing utility.
\end{abstract}

The development of sensors and communication technologies for Internet of Things (IoT)  have enabled a fast and large-scale collection of user data, which has bred new services and applications such as the Waze application that provides the intelligent transportation routing service. This kind of service benefits users' daily life, but it may raise privacy concerns of sensitive data such as users' location information. To address these concerns, we propose a hybrid approach that integrates federated learning (FL)~\cite{mcmahan2016communication} with local differential privacy (LDP)~\cite{bindschaedler2017plausible} techniques. FL can facilitate the collaborative learning with uploaded gradients from users instead of sharing users' raw data. A honest-but-curious aggregator may be able to leverage users' uploaded gradients to infer the original data~\cite{hitaj2017deep, xu2019verifynet}.  Thus, we deploy LDP noises to gradients to ensure privacy while not compromising the utility of gradients.

\textbf{Federated learning with LDP.} 
In addition to LDP mechanisms, FL also provides privacy protection to the data by enabling users to maintain data locally. \texttt{FedSGD} algorithm~\cite{chen2016revisiting} allows users to submit gradients instead of true data. However, attackers may reverse the gradients to infer original data. By adding LDP noises to the gradients before uploading,  we obtain the \texttt{LDP-FedSGD} algorithm, which prevents attackers from deducing original data even though they obtain perturbed gradients. As a result, the FL server gathers and averages users' submitted perturbed gradients to obtain the averaged result to update the global model's parameters.

\begin{table}[]
\centering\caption{A comparison of the worst-case variances of existing $\epsilon$-LDP mechanisms on a single numeric attribute with a domain $[-1,1]$: \texttt{Duchi} of~\cite{duchi2013local} generating a binary output, \texttt{Laplace} of~\cite{dwork2006calibrating} with the addition of Laplace noise, and the Piecewise Mechanism (\texttt{PM}) of~\cite{wang2019collecting}. For an LDP mechanism $\mathcal{A}$, its worst-case variance is denoted by $V_{\mathcal{A}}$. We obtain this table based on results of~\cite{wang2019collecting}. }
\label{LDPcomparison1}
\begin{tabular}{ll}
\hline
\multicolumn{1}{|l|}{Range of $\epsilon$} & \multicolumn{1}{|l|}{Comparison of  mechanisms} \\ \hline
\multicolumn{1}{|l|}{$0<\epsilon<1.29$} & \multicolumn{1}{|l|}{$V_{\texttt{Duchi}} < V_{\texttt{PM}} < V_{\texttt{Laplace}} $} \\ \hline
\multicolumn{1}{|l|}{$1.29<\epsilon<2.32$} & \multicolumn{1}{|l|}{$V_{\texttt{PM}} < V_{\texttt{Duchi}} < V_{\texttt{Laplace}} $}  \\ \hline
\multicolumn{1}{|l|}{$\epsilon>2.32$} & \multicolumn{1}{|l|}{$V_{\texttt{PM}}  < V_{\texttt{Laplace}}< V_{\texttt{Duchi}} $}  \\ \hline
                       &                      
\end{tabular}
\end{table}

\begin{table}[]
\centering\caption{A comparison of the worst-case variances of our and existing $\epsilon$-LDP mechanisms on a single numeric attribute with a domain $[-1,1]$. \texttt{Three-Outputs} and \texttt{PM-SUB} are our main LDP mechanisms proposed in this paper. The results in this table show the advantages of our mechanisms over existing  mechanisms for a wide range of privacy parameter $\epsilon$.}
\label{LDPcomparison2}

\begin{tabular}{|l|l|}
\hline
Range of $\epsilon$ & Comparison of  mechanisms \\ \hline 
$0<\epsilon<\ln 2 \approx 0.69$ & $V_{\texttt{Duchi}} = V_{\texttt{Three-Outputs}} < V_{\texttt{PM-SUB}} < V_{\texttt{PM}} $ \\ \hline
$\ln 2<\epsilon<1.19$ & $V_{\texttt{Three-Outputs}} < V_{\texttt{Duchi}}   < V_{\texttt{PM-SUB}} < V_{\texttt{PM}} $  \\ \hline
$1.19<\epsilon<1.29$ & $V_{\texttt{Three-Outputs}}  < V_{\texttt{PM-SUB}} < V_{\texttt{Duchi}}   < V_{\texttt{PM}} $ \\ \hline
$1.29<\epsilon <2.56 $ & $V_{\texttt{Three-Outputs}}  < V_{\texttt{PM-SUB}} < V_{\texttt{PM}} < V_{\texttt{Duchi}}   $  \\ \hline
$2.56<\epsilon <3.27 $ & $ V_{\texttt{PM-SUB}} < V_{\texttt{Three-Outputs}}  <V_{\texttt{PM}} < V_{\texttt{Duchi}}   $ \\ \hline
$ \epsilon>3.27 $ & $ V_{\texttt{PM-SUB}} <V_{\texttt{PM}}  < V_{\texttt{Three-Outputs}} < V_{\texttt{Duchi}}   $ \\ \hline
\end{tabular}
\end{table}

\textbf{Existing LDP mechanisms.} Since the proposal of LDP in~\cite{dwork2006calibrating}, various LDP mechanisms have been proposed in the literature. Mechanisms for categorical data are presented in~\cite{bassily2015local, wang2017locally, erlingsson2014rappor}. For numerical data, for which new LDP mechanisms are developed in this paper, prior mechanisms of~\cite{dwork2006calibrating, duchi2013local, wang2019collecting} are discussed as follows. For simplicity, we consider data with a single numeric attribute which has a domain $[-1,1]$ (after normalization if the original domain is not $[-1,1]$). Extensions to the case of multiple numeric attributes will be discussed later in the paper.
\begin{itemize}
\item \textbf{\texttt{Laplace} of \cite{dwork2006calibrating}.} LDP can be understood as a variant of DP, where the difference is the definition of ``neighboring data(sets)''. In DP, two datasets are neighboring if they differ in just one record; in LDP, any two instances of the user's data are neighboring. Due to this connection between DP and LDP, the classical Laplace mechanism (referred to \texttt{Laplace} hereinafter) for DP can thus also be used to achieve LDP. Yet, \texttt{Laplace} may not achieve a high utility for some $\epsilon$ since {\ding{172} \texttt{Laplace} does not consider the difference between DP and LDP} .
\item \textbf{\texttt{Duchi} of \cite{duchi2013local}.} In view of the above drawback~\ding{172} of  \texttt{Laplace}, Duchi~\textit{et al.}~\cite{duchi2013local} introduce alternative mechanisms for LDP. For a single numeric attribute, one mechanism, hereinafter referred to \texttt{Duchi} of \cite{duchi2013local}, flips a coin with two possibilities to generate an output, where the probability of each possibility depends on the input. A disadvantage of~\texttt{Duchi} is as follows: {\ding{173}~Since the output of \texttt{Duchi} has only two possibilities, the utility may not be high for large $\epsilon$} (intuitively, for large $\epsilon$, the privacy protection is weak so the output should be close to the input which means the output should have many possibilities since the input can take any value in $[-1,1]$)\footnote{Note that any algorithm satisfying DP or LDP has the following property: the set of possible values for the output does not depend on the input (though the output distribution depends on the input). This can be easily seen by contradiction. Suppose an output $y$ is possible for input $x$ but not for $x'$ ($x$ and $x'$ satisfy  the neighboring relation in DP or LDP). Then $\bp{y \mid x}>0$ and $\bp{y \mid x'}=0$, resulting in $\bp{y \mid x} > e^{\epsilon}\bp{y \mid x'} $ and hence violating the privacy requirement ($\mathbb{P}[\cdot | \cdot]$ denotes conditional probability).}. 
\item \textbf{Piecewise Mechanism (\texttt{PM}) of~\cite{wang2019collecting}.} Due to the drawback~\ding{172} of  \texttt{Laplace}, and the drawback~\ding{173} of~\texttt{Duchi} above, Wang~\textit{et al.}~\cite{wang2019collecting} propose the Piecewise Mechanism (\texttt{PM}) which achieves higher utility than~\texttt{Duchi} for large $\epsilon$, since the output range of \texttt{PM} is continuous and has infinite possibilities, instead of just 2 possibilities as in~\texttt{Duchi}. In \texttt{PM}, the plot of the output's probability density function with respect to the output value  consists of three ``pieces'', among which the center piece has a higher probability than the other two. As the input increases, the length of the center piece
remains unchanged, but the length of the leftmost (resp., rightmost) piece increases (resp., decreases). Since \texttt{PM} is tailored for LDP, unlike \texttt{Laplace} for LDP, \texttt{PM} has a strictly lower  worst-case variance (i.e., the maximum variance with respect to the input given $\epsilon$) than  \texttt{Laplace} for \textit{any} $\epsilon$. 
\end{itemize}  

\textbf{Proposing new LDP mechanisms.} Based on results of~\cite{wang2019collecting}, Table~\ref{LDPcomparison1} on Page~\pageref{LDPcomparison1} shows that among the three mechanisms \texttt{Duchi}, \texttt{Laplace}, and \texttt{PM}, in terms of the worst-case variance, \texttt{Duchi} is the best for $0<\epsilon<1.29$, while \texttt{PM} is the best for $\epsilon>1.29$. Then a natural research question is that can we propose better or even optimal LDP mechanisms? The optimal LDP mechanism for numeric data is still open in the literature, but the optimal LDP mechanism for categorical data has been discussed by Kairouz~\textit{et~al.}~\cite{Kairouz2014ExtremalMF}. In particular,~\cite{Kairouz2014ExtremalMF} shows that for a categorical attribute (with a limited number of discrete values), the binary and randomized response mechanisms, are universally optimal for very small and large $\epsilon$, respectively. Although~\cite{Kairouz2014ExtremalMF} handles categorical data, its results can provide the following insight even for continuous numeric data: for very small $\epsilon$, a mechanism generating a binary output should be optimal; for very large $\epsilon$, the optimal mechanism's output should have infinite possibilities. This is also in consistent with the results of Table~\ref{LDPcomparison1} on Page~\pageref{LDPcomparison1}. Based on the above insight, intuitively, there may exist a range of medium $\epsilon$ where a mechanism with three, or four, or five, ... output possibilities can be optimal. To this end, \textbf{our first motivation of proposing new LDP mechanisms} is to develop a mechanism with three output possibilities such that the mechanism outperform existing mechanisms of Table~\ref{LDPcomparison1} for some $\epsilon$. The outcome of the above motivation is our mechanism called \textbf{\texttt{Three-Outputs}}. Since the analysis of \texttt{Three-Outputs} is already very complex, we do not consider a mechanism with four, or five, ... output possibilities.

In addition, \textbf{our second motivation of proposing new LDP mechanisms} is that despite the elegance of the Piecewise Mechanism (\texttt{PM}) of~\cite{wang2019collecting}, we can derive the optimal mechanism \textbf{\texttt{PM-OPT}}  under the ``piecewise framework'' of~\cite{wang2019collecting}, in order to improve \texttt{PM}. Note that \texttt{PM-OPT} is just optimal under the above framework and may not be the optimal LDP mechanism. Since the expressions for \texttt{PM-OPT} are quite complex, we present \textbf{\texttt{PM-SUB}} which compared with \texttt{PM-OPT} is suboptimal, but has simpler expressions and achieves a comparable utility.

Table~\ref{LDPcomparison2} on Page~\pageref{LDPcomparison2} gives a comparison of our and existing LDP mechanisms. For simplicity, we do not include \texttt{Laplace} in the comparison since \texttt{Laplace} is worse than \texttt{PM} for any $\epsilon$ according to  Table~\ref{LDPcomparison1}. As shown in Table~\ref{LDPcomparison2}, in terms of the worst-case variance for a single numeric attribute with a domain $[-1,1]$, \texttt{Three-Outputs} outperforms \texttt{Duchi} for $\epsilon>\ln 2 \approx 0.69$ (and is the same as \texttt{Duchi} for $\epsilon \leq \ln 2$), while \texttt{PM-SUB} beats \texttt{PM} for any $\epsilon$; moreover, \texttt{Three-Outputs} outperforms both \texttt{Duchi} and \texttt{PM} for $\ln 2<\epsilon<3.27$, while \texttt{PM-SUB} beats both \texttt{Duchi} and \texttt{PM} for $\epsilon>1.19$. 
 
We also follow the practice of~\cite{wang2019collecting}, which combines different mechanisms \texttt{Duchi} and \texttt{PM} to obtain a \underline{h}ybrid \underline{m}echanism \texttt{HM}. \texttt{HM} has a lower worst-case variance than those of \texttt{Duchi} and \texttt{PM}. In this paper, we combine \texttt{\underline{T}hree-Outputs} and \texttt{\underline{P}M-SUB} to propose \textbf{\texttt{HM-TP}}. The intuition is as follows. Given $\epsilon$, the variances of \texttt{Three-Outputs} and \texttt{PM-SUB} (denoted by $T_{\epsilon}(x)$ and $P_{\epsilon}(x)$) depend on the input $x$, and their maximal values (i.e., the worst-case variances) may be taken at different $x$. Hence, for a hybrid mechanism which probabilistically uses  \texttt{Three-Outputs} with probability $q$ or \texttt{PM-SUB} with probability $1-q$, given $\epsilon$, the worst-case variance $ \max_{x} \left( q \cdot T_{\epsilon}(x)+(1-q)\cdot P_{\epsilon}(x) \right)$ may be strictly smaller than the minimum of $ \max_{x} T_{\epsilon}(x) $ and $ \max_{x} P_{\epsilon}(x) $. We optimize $q$ for each $\epsilon$ to obtain \texttt{HM-TP}. Due to the already complex expression of \texttt{PM-OPT}, we do not consider the combination of \texttt{PM-OPT} and  \texttt{Three-Outputs}. 

\textbf{Contributions.} Our contributions can be summarized as follows:
\begin{itemize}
     \item Using the \texttt{LDP-FedSGD} algorithm for federated learning in IoV as a motivating context, we present novel LDP mechanisms for numeric data with a continuous domain. Among our proposed mechanisms, \texttt{Three-Outputs} and \texttt{PM-SUB} outperform existing mechanisms for a wide range of $\epsilon$, as shown in the theoretical results in Table~\ref{LDPcomparison2} and confirmed by experiments. In terms of comparing our \texttt{Three-Outputs} and \texttt{PM-SUB}, we have: \texttt{Three-Outputs}, whose output has three possibilities,  is better for small $\epsilon$, while \texttt{PM-SUB}, whose output can take infinite possibilities of an interval, has higher utility for large $\epsilon$.  Our \texttt{PM-SUB} is a slightly suboptimal version of our \texttt{PM-OPT} to simplify the expressions. We further combine \texttt{Three-Outputs} and \texttt{PM-SUB} to obtain a hybrid mechanism \texttt{HM-TP}, which achieves even higher utility.
    \item We discretize the continuous output ranges of our proposed mechanisms \texttt{PM-SUB} and \texttt{PM-OPT}. Through the discretization post-processing, we enable vehicles to use our proposed mechanisms. In Section~\ref{sec:experiment}, we confirm that the discretization post-processing algorithm maintains utility with our experiments, while reducing the communication cost.
    \item Experimental evaluation of our proposed mechanisms on real-world datasets and synthetic datasets demonstrates that our proposed mechanisms achieve higher accuracy in estimating the mean frequency of the data and performing empirical risk minimization tasks than existing approaches.
\end{itemize}

\textbf{Organization.} In the following, Section~\ref{sec:preliminaries} introduces the preliminaries. Then, we introduce related works in Section~\ref{sec:related_work}. Then, we illustrate the system model and the local differential privacy based FedSGD algorithm in Section~\ref{sec:system-model}. Section~\ref{sec:problem} presents the problem formation. Section~\ref{sec:single} proposes novel solutions for the single numerical data estimation. Section~\ref{sec:multiple} illustrates proposed  mechanisms used for multidimensional  numerical data estimation. Section~\ref{sec:experiment} demonstrates our experimental results. Section~\ref{sec:conclusion} concludes the paper.

\section{Preliminaries}\label{sec:preliminaries}

In local differential privacy, users complete the perturbation by themselves. To protect users' privacy, each user runs a random perturbation algorithm $\mathcal{M}$, and then he sends perturbed results to the aggregator. The privacy budget $\epsilon$ controls the privacy-utility trade-off, and a higher privacy budget means a lower privacy protection. As a result of this, we define local differential privacy as follows:

\begin{defn} (Local Differential Privacy.) Let $\mathcal{M}$ be a randomized function with domain $\mathbb{X}$ and range $\mathbb{Y}$; i.e., $\mathcal{M}$ maps each element in $\mathbb{X}$ to a probability distribution with sample space $\mathbb{Y}$. For a non-negative $\epsilon$, the randomized mechanism $\mathcal{M}$ satisfies $\epsilon$-local differential privacy if
\begin{align}
    \bigg|\ln   \frac{\mathbb{P}_{\mathcal{M}}[Y \in S | x]}{\mathbb{P}_{\mathcal{M}}[Y \in S  | x']}   \bigg| \leq \epsilon,~\forall x, x' \in \mathbb{X},~\forall S \subseteq \mathbb{Y}.\label{local-differential-privacy-fomular}
\end{align}
\end{defn}

$\mathbb{P}_{\mathcal{M}}[\cdot | \cdot]$ means conditional probability distribution depending on $\mathcal{M}$. In local differential privacy, the random perturbation is performed by users instead of a centralized aggregator. Centralized aggregator only receives perturbed results which make sure that the aggregator is unable to distinguish whether the true tuple is $x$ or $x'$ with high confidence~(controlled by the privacy budget $\epsilon$).

\begin{figure}[h]
    \centering
    \includegraphics[scale=0.5]{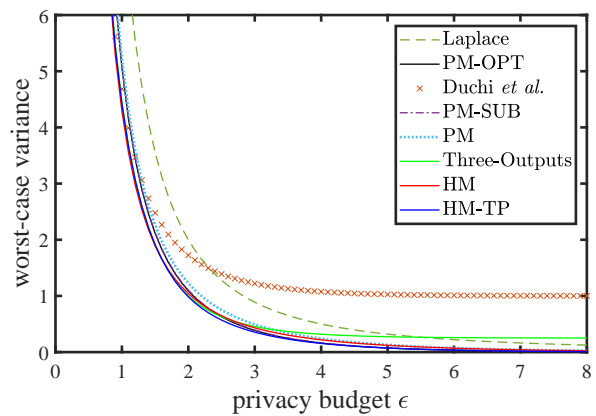}
    \caption{Different mechanisms' worst-case noise variance for one-dimensional numeric data versus the privacy budget $\epsilon$.}
    \label{fig:worst_var_comp}
\end{figure}

\section{Related Work}\label{sec:related_work}

Recently, local differential privacy has attracted much attention~\cite{ou2020singular, tang2019secure, sun2019relationship, zhao2019survey,sun2019distributed, gursoy2019secure, ghane2020preserving, lyu2020towards}. Several mechanisms for numeric data estimation have been proposed~\cite{duchi2013local, dwork2006calibrating, wang2019collecting, sun2020bisample}. (i) Dwork~\textit{et al.}~\cite{dwork2006calibrating} propose the Laplace mechanism which adds the Laplace noise to real one-dimensional data directly. The Laplace mechanism is originally used in the centralized differential privacy mechanism, and it can be applied to local differential privacy directly. (ii) For a single numeric attribute with a domain $[-1,1]$, Duchi~\textit{et al.}~\cite{duchi2013local} propose an LDP framework that provides output from $\{-C, C\}$, where $C>1$. (iii) Wang~\textit{et al.}~\cite{wang2019collecting} propose the piecewise mechanism (\texttt{PM}) which offers an output that contains infinite possibilities in the range of $[-A, A]$, where $A>1$. In addition, they apply LDP mechanism to preserve the privacy of gradients generated during machine learning tasks. Both approaches by Duchi~\textit{et al.}~\cite{duchi2013local} and Wang~\textit{et al.}~\cite{wang2019collecting} can be extended to the case of multidimensional numerical data.

\textit{Deficiencies of existing solutions.} 
Fig.~\ref{fig:worst_var_comp} illustrates that when $\epsilon \leq 2.3$, Laplace mechanism's worst-case noise variance is larger than that of Duchi~\textit{et al.}'s~\cite{duchi2018minimax} solution;
however, the Laplace mechanism outperforms Duchi~\textit{et al.}'s~\cite{duchi2018minimax} solution if $\epsilon$ is larger. The worst-case noise variance in \texttt{PM} is smaller than that of Laplace and Duchi~\textit{et al.}'s~\cite{duchi2018minimax} solution when $\epsilon$ is large. The $\texttt{HM}$ mechanism outperforms other existing solutions by taking advantage of Duchi~\textit{et al.}'s~\cite{duchi2018minimax} solution when $\epsilon$ is small and \texttt{PM} when $\epsilon$ is large. However, \texttt{PM} and \texttt{HM}'s outputs have infinite possibilities that are hard to encode. We would like to find a mechanism that can improve the utility of existing mechanisms.  In addition, we believe there is a mechanism that retains a high utility and is easy to encode its outputs. Based on the above intuition, we propose four novel mechanisms that can be used by vehicles in Section~\ref{sec:single}.

In addition, LDP has been widely used in the research of IoT~\cite{xu2019edgesanitizer,choi2018guaranteeing,he2017privacy,li2018privacy,arachchige2019local}. For example, Xu~\textit{et al.}~\cite{xu2019edgesanitizer} integrate deep learning with local differential privacy techniques and apply them to protect users' privacy in edge computing. They develop an EdgeSanitizer framework that forms a new protection layer against sensitive inference by leveraging a deep learning model to mask the learned features with noise and minimize data. Choi~\textit{et al.}~\cite{choi2018guaranteeing} explore the feasibility of applying LDP on ultra-low-power (ULP) systems. They use resampling, thresholding, and a privacy budget control algorithm to overcome the low resolution and fixed point nature of ULPs. He~\textit{et al.}~\cite{he2017privacy} address the location privacy and usage pattern privacy induced by the mobile edge computing's wireless task offloading feature by proposing a privacy-aware task offloading scheduling algorithm based on constrained Markov decision process. Li~\textit{et al.}~\cite{li2018privacy}~propose a scheme for privacy-preserving data aggregation  in the mobile edge computing to assist IoT applications with three participants, i.e., a public cloud center (PCC), an edge server (ES), and a terminal device (TD). TDs generate and encrypt data and send them to the ES, and then the ES submits the aggregated data to the PCC. The PCC uses its private key to recover the aggregated plaintext data. Their scheme provides source authentication and integrity and guarantees the data privacy of the TDs. In addition, their scheme can save half of the communication cost. To protect the privacy of massive data generated from IoT platforms, Arachchige~\textit{et al.}~\cite{arachchige2019local} design an LDP mechanism named as LATENT for deep learning. A randomization layer between the convolutional module and the fully connected module  is added to the LATENT to perturb data before data leave data owners for machine learning services. Pihur~\emph{et~al.}~\cite{pihur2019podium} propose the Podium Mechanism which is similar to our \texttt{PM-SUB}, but his mechanism is applicable to DP instead of LDP.

Moreover, federated learning or collaborative learning is an emerging distributed machine learning paradigm, and it is widely used to address data privacy problem in machine learning~\cite{mcmahan2016communication, ZhaoHWJSLH20}. 
Recently, federated learning is explored extensively in the Internet of Things recently~\cite{lim2019federated, hao2019efficient, lu2019collaborative, fantacci2020federated, saputra2019energy}. Lim~\textit{et al.}~\cite{lim2019federated} survey federated learning applications in mobile edge network comprehensively, including algorithms, applications and potential research problems, etc. Besides, Lu~\textit{et al.}~\cite{lu2019collaborative} propose CLONE which is a collaborative learning framework on the edges for connected vehicles, and it reduces the training  time while guaranteeing the prediction accuracy. Different from CLONE, our proposed approach utilizes local differential privacy noises to protect the privacy of the uploaded data. Furthermore, Fantacci~\textit{et al.}~\cite{fantacci2020federated} leverage  FL to protect the privacy of  mobile edge computing, while Saputra~\textit{et al.}~\cite{saputra2019energy} apply FL to predict the energy demand for electrical vehicle networks.

Furthermore, there have been many papers on federated learning and differential privacy such as~\cite{mcmahan2017learning, truex2019hybrid, hu2020personalized, bagdasaryan2019differential, hao2019efficient, triastcyn2019federated, wang2020federated, lyu2019fog,li2019privacy, zhao2019mobile}.  For example, Truex~\textit{et al.}~\cite{truex2019hybrid} utilize both secure multiparty computation and centralized differential privacy to prevent inference over both the messages exchanged in the process of training the model. However, they do not analyze the impact of the privacy budget on performance of FL. Hu~\textit{et al.}~\cite{hu2020personalized} propose a privacy-preserving FL approach for learning effective personalized models. They use Gaussian mechanism, a centralized DP mechanism, to protect the privacy of the model. Compared with them, our proposed LDP mechanisms provide a stronger privacy protection using the local differential privacy mechanism. Hao~\textit{et al.}~\cite{hao2019efficient} propose a differential enhanced federated learning  scheme for industrial artificial industry. Triastcyn~\textit{et al.}~\cite{triastcyn2019federated} employ Bayesian differential privacy on federated learning.  They make use of the centralized differential privacy mechanism to protect the privacy of gradients, but we leverage a stronger privacy-preserving mechanism (LDP) to protect each vehicle's privacy.

Additionally, DP can be applied to various FL algorithms such as \texttt{FedSGD}~\cite{chen2016revisiting} and \texttt{FedAvg}~\cite{mcmahan2016communication}. \texttt{FedAvg} requires users to upload model parameters instead of gradients in \texttt{FedSGD}. The advantage of \texttt{FedAvg} is that it allows users to train the model for multiple rounds locally before submitting gradients. McMahan~\textit{et al.}~\cite{mcmahan2017learning} propose to apply centralized DP to \texttt{FedAvg} and  \texttt{FedSGD} algorithm. In our paper, we deploy LDP mechanisms to gradients in \texttt{FedSGD} algorithm. Our future work is to develop LDP mechanisms for  state-of-the-art FL algorithms.  Zhao~\emph{et al.}~\cite{ZhaoWZZC19} propose a SecProbe mechanism to protect privacy and quality of participants' data by leveraging exponential mechanism and functional mechanism of differential privacy. SecProbe guarantees the high accuracy as well as privacy protection. In addition, it prevents unreliable participants in the collaborative learning.

\section{System Model and erential Privacy based FedSGD Algorithm} \label{sec:system-model}

\subsection{System Model}

In this work, we consider a scenario where a number of vehicles are connected with a cloud server as Fig.~\ref{fig:system_design}. Each vehicle is responsible for continuously performing training and inference locally based on data that it collects and the model initiated by the cloud server. Local training dataset is never uploaded to the  cloud server. After finishing predefined epochs locally, the cloud server calculates the average of uploaded gradients from vehicles and updates the global model with the average. The FL aggregator is honest-but-curious or semi-honest, which follows the FL protocol but it will try to learn additional information using received data~\cite{truex2019hybrid, lyu2020threats}. With the injected LDP noise, servers or attackers cannot retrieve users' information by reversing their uploaded gradients~\cite{hitaj2017deep, xu2019verifynet}. Thus, there is a need to deploy LDP mechanisms to the FL to develop a communication-efficient LDP-based FL algorithm.

\begin{figure}[!h]
    \centering
    \includegraphics[scale=0.65]{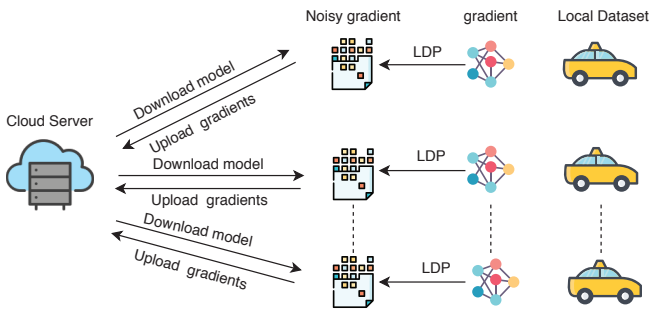}
    \caption{System Design.}
    \label{fig:system_design}
\end{figure}

\subsection{Federated learning with LDP: LDP-FedSGD}

In addition, we propose a local differential privacy based federated stochastic gradient descent algorithm (\texttt{LDP-FedSGD}) for our proposed system. Details of \texttt{LDP-FedSGD} are given in  Algorithm~\ref{algorithm-ldp-fedsgd}. Unlike the FedAvg algorithm, in the \texttt{FedSGD} algorithm, clients (i.e., vehicles) upload updated gradients instead of model parameters to the central aggregator (i.e., cloud server)~\cite{chen2016revisiting}. However, compared with the standard \texttt{FedSGD}~\cite{chen2016revisiting}, we add our proposed LDP mechanism proposed in Section~\ref{sec:multiple} to prevent the privacy leakage of gradients. Each vehicle locally takes one step of gradient descent on the current model using its local data, and then it perturbs the true gradient with Algorithm~\ref{PM-d-dimension-Algo}. The server aggregates and averages the updated gradients  from vehicles and then updates the model. To reduce the communication rounds, we separate vehicles into groups, so that the cloud server updates the model  after gathering gradient updates from vehicles in a group. In the following sections, we will introduce how we obtain the LDP algorithm in detail.

\begin{algorithm}[h]
\caption{Local Differential Privacy based FedSGD (\texttt{LDP-FedSGD}) Algorithm.}
 \label{algorithm-ldp-fedsgd}
 {\textbf{Server executes:}} \\
 {Server initializes the parameter as $\theta_0$;} \\
 \For{$t$ from $1$ to  maximal iteration number}
 {Server sends $\theta_{t-1}$ to vehicles in group $G_t$;\\
 \For{each vehicle $i$ in Group $G_t$}
 {VehicleUpdate($i$, $\Delta L$):}
 {Server computes the  average of the noisy gradient of group $G_t$ and updates the parameter from $\theta_{t-1}$ to $\theta_t$: $\theta_t \leftarrow \theta_{t-1} - \eta \cdot \frac{1}{|G_t|}\sum_{i\in G_t}\mathcal{M}(\Delta L(\theta_{t-1};x_i))$, where  $\eta_t$ is  the  learning rate;}\\
 \If{$\theta_t$ and $\theta_{t-1}$ are close enough or these remains no vehicle which has not participated in the computation}
    {break;}
  $t \rightarrow t+1;$}
  {\textbf{VehicleUpdate ($i$, $\Delta L$):}}\\
  {Compute the (true) gradient $\Delta L(\theta_{t-1};x_i)$, where $x_i$ is vehicle $i$'s data;\\
 Use local differential privacy-compliant algorithm $\mathcal{M}$ to compute the noisy gradient $\mathcal{M}(\Delta(\theta_{t-1};x_i))$;}
\end{algorithm}

\subsection{Comparing  LDP-FedSGD with other privacy-preserving federated learning paradigms}

The LDP-FedSGD algorithm incorporates LDP into federated learning. In addition to LDP-FedSGD, one may be interested in other ways of using DP in federated learning. To explain them, we categorize combinations of DP and federated learning (or distributed computations in general) by considering the place of perturbation (distributed/centralized perturbation) and privacy granularity (user-level/record-level privacy protection)~\cite{wang2019privacy}: 
\begin{itemize}
\item \textbf{Distributed/Centralized perturbation.} Note that differential privacy is achieved by introducing perturbation.
Distributed perturbation considers an honest-but-curious aggregator, while centralized perturbation needs a trusted aggregator. Both perturbation methods defend against  external inference attacks after model publishing.
\item \textbf{User-level/Record-level privacy protection.} In general, a differentially private algorithm ensures that the probability distributions of the outputs on two neighboring datasets do not differ much. The distinction between user-level and record-level privacy protection lies in how neighboring datasets are defined. We define that two datasets are user-neighboring if one dataset can be formed from the other dataset by adding or removing all one user's records arbitrarily. We define that two datasets are record-neighboring if one dataset can be obtained from the other dataset by changing a \textit{single} record of one user.
\end{itemize}

Based on the above, we further obtain four paradigms as follows: 1) \textbf{ULDP} (\underline{u}ser-\underline{l}evel privacy protection with \underline{d}istributed \underline{p}erturbation), 2) \textbf{RLDP} (\underline{r}ecord-\underline{l}evel privacy protection with \underline{d}istributed \underline{p}erturbation), 3) \textbf{RLCP} (\underline{r}ecord-\underline{l}evel privacy protection with \underline{c}entralized \underline{p}erturbation), and 4) \textbf{ULCP} (\underline{u}ser-\underline{l}evel privacy protection with \underline{c}entralized \underline{p}erturbation). The details are as follows and can also be found in the second author's prior work~\cite{wang2019privacy}.
\begin{enumerate}
\item In \textbf{ULDP}, each user $i$ selects a privacy parameter $\epsilon_i$ and applies a randomization algorithm $Y_i$ such that given any two instances $x_i$ and $x_i'$  of user $i$'s data  (which are \textit{user}-neighboring), and for any possible subset of outputs\footnote{For simplicity, we slightly abuse the notation and denote the output of algorithm $Y_i$ (resp., algorithm $Y$) by $Y_i$ (resp., $Y$).\label{footnoteabusenotation}} $\mathcal{Y}_i$ of $Y_i$, we obtain  
 $\bp{ Y_i \in \mathcal{Y}_i \mid x_i } \le e^{\epsilon_i} \times \bp{Y_i \in \mathcal{Y}_i \mid x_i '}$.
 Clearly, ULDP is achieved by each user implementing local differential privacy studied in this paper.
\item In \textbf{RLDP}, each user $i$ selects a privacy parameter $\epsilon_i$ and applies a randomization algorithm $Y_i$ such that for any two \textit{record}-neighboring instances $x_i$ and $x_i'$  of user $i$'s data  (i.e, $x_i$ and $x_i'$ differ in only one record), and for any possible subset of outputs $\mathcal{Y}_i$ of $Y_i$, we have 
$\bp{ Y_i \in \mathcal{Y}_i \mid x_i } \le e^{\epsilon_i} \times \bp{Y_i \in \mathcal{Y}_i \mid x_i '}$. Clearly, in \textbf{RLDP}, what each user does is just to apply standard differential privacy. In contrast, in \textbf{ULDP} above, each user applies local differential privacy.
\item In $\epsilon$-\textbf{RLCP}, the aggregator sets a privacy parameter $\epsilon$ and applies a randomization algorithm $Y$ such that for any user $i$, for any two \textit{record}-neighboring instances $x_i$ and $x_i'$   of user $i$'s data, and for any possible subset of outputs$^\text{\ref{footnoteabusenotation}}$ $\mathcal{Y}$ of $Y$, we obtain $\bp{ Y \in \mathcal{Y} \mid x_i } \le e^{\epsilon} \times \bp{Y \in \mathcal{Y} \mid x_i '}$. In other words, in \textbf{RLCP}, the aggregator applies standard differential privacy. For the aggregator to implement RLCP well, typically the aggregator should be able to bound the impact of each record on the information sent from a user to the  aggregator. Further discussions on this can be interesting, but we do not present more details since RLCP is not our paper's focus.
\item In $\epsilon$-\textbf{ULCP},  the aggregator sets a privacy parameter $\epsilon$ and applies a randomization algorithm $Y$ so that for any user $i$, for any two instances $x_i$ and $x_i'$ of user $i$'s data (which are \textit{user}-neighboring), and for any possible subset of outputs $\mathcal{Y}$ of $Y$, we have $\bp{ Y \in \mathcal{Y} \mid x_i } \le e^{\epsilon} \times \bp{Y \in \mathcal{Y} \mid x_i '}$. The difference RLCP and ULCP is that RLCP achieves record-level privacy protection while ULCP ensures the stronger user-level privacy protection.
\end{enumerate}

In the case of distributed perturbation, when all users set the same privacy parameter $\epsilon$, we refer to ULDP and RLDP above as $\epsilon$-ULDP and $\epsilon$-RLDP, respectively. Table~\ref{tableprivacy} presents a comparison of $\epsilon$-ULDP, $\epsilon$-RLDP, $\epsilon$-RLCP, and
$\epsilon$-ULCP. In this paper's focus, each user applies $\epsilon$-local differential privacy, so our framework is under ULDP. The reasons that we consider ULDP instead of RLDP, RLCP, and ULCP are as follows.
\begin{itemize}
\item We do not consider RLDP which implements perturbation at each user via standard differential privacy, since we aim to achieve user-level privacy protection instead of the weaker record-level privacy protection (a vehicle  is a user in our IoV applications and may have multiple records). The motivation is that often much data from a vehicle may be about the vehicle's regular driver, and it often makes more sense to protect all data about the regular driver instead of just protecting each single record. A similar argument has been recently stated in~\cite{mcmahan2017learning}, which incorporates user-level differential privacy into the training process of federated learning for  language modeling. Specifically,~\cite{mcmahan2017learning} considers user-level privacy to protect the privacy of all typed words of a user, and explains that such privacy protection is more reasonable than protecting individual words as in the case of record-level privacy. 
In addition, although we can compute the level of \underline{u}ser-level privacy from \underline{r}ecord-level privacy via the group privacy property of differential privacy (see Theorem 2.2 of~\cite{dwork2014algorithmic}), but this may significantly increase the privacy parameter and hence weaken the privacy protection if a user has many records (note that a larger privacy parameter $\epsilon$ in $\epsilon$-DP means weaker privacy protection). More specifically, for a user with $m$ records, according to the group privacy property~\cite{dwork2014algorithmic}, the privacy protection strength for the user under $\epsilon$-record-level privacy is just as that under $m\epsilon$-user-level privacy (i.e., for a user with $m$ records, $\epsilon$-RLDP ensures $m\epsilon$-ULDP; $\epsilon$-RLCP ensures $m\epsilon$-ULCP).
\item We also do not investigate RLCP and ULCP since this paper considers a honest-but-curious aggregator instead of a trusted aggregator. The aggregator is not completely trusted, so the perturbation is implemented at each user (i.e., vehicle in IoV).
\end{itemize}

\begin{table}[]
\caption{We compare different privacy notions in this table. In this paper, we focus on $\epsilon$-local differential privacy which achieves \underline{u}ser-\underline{l}evel privacy protection with \underline{d}istributed \underline{p}erturbation (ULDP). We do not consider \underline{r}ecord-\underline{l}evel privacy protection with \underline{d}istributed \underline{p}erturbation (RLDP) which implements perturbation at each user via standard differential privacy, since we aim to achieve user-level privacy protection instead of the weaker record-level privacy protection (a vehicle  is a user in our IoV applications and may have multiple records). We also do not investigate \underline{r}ecord/\underline{u}ser-\underline{l}evel privacy protection with \underline{c}entralized \underline{p}erturbation (RLCP/ULCP) since this paper considers a honest-but-curious aggregator instead of a trusted aggregator.}
\label{tableprivacy}
\setlength\tabcolsep{1pt}
\begin{tabular}{|l|l|l|}
\hline
\begin{tabular}[c]{@{}l@{}}privacy granularity and \\  place of perturbation\end{tabular} & \begin{tabular}[c]{@{}l@{}}privacy    \\ property\end{tabular}  & adversary model                                                                                                  \\ \hline
\begin{tabular}[c]{@{}l@{}}$\epsilon$-LDP (defined for\\\textbf{distributed} perturbation)\end{tabular}                                                   & $\epsilon$-ULDP & \multirow{2}{*}{\begin{tabular}[c]{@{}l@{}} defend against a  honest-but-curious \\ aggregator  \&external attacks  after model \\ publishing\end{tabular}}  \\ \cline{1-2}
                                           \begin{tabular}[c]{@{}l@{}}$\epsilon$-DP with\\ \textbf{distributed} perturbation\end{tabular}                                                   & $\epsilon$-RLDP &                               \\ \hline
\begin{tabular}[c]{@{}l@{}}$\epsilon$-DP with\\ \textbf{centralized} perturbation\end{tabular}                                                 & $\epsilon$-RLCP & \multirow{2}{*}{\begin{tabular}[c]{@{}l@{}}trusted aggregator; defend against  external \\ attacks after model publishing\end{tabular}} \\ \cline{1-2}
                                           \begin{tabular}[c]{@{}l@{}}user-level privacy with\\ \textbf{centralized} perturbation\end{tabular}                                                   & $\epsilon$-ULCP &     \\ \hline 
\end{tabular}
\end{table}

\section{Problem Formation}\label{sec:problem}

Let $x$ be a user’s true value, and $Y$ be the perturbed value. Under the perturbation mechanism $\mathcal{M}$, we use $\mathbb{E}_{\mathcal{M}}[Y|x]$ to denote the expectation of the randomized output $Y$ given input $x$. $\textup{Var}_{\mathcal{M}}[Y|x]$ is the variance of output $Y$ given input $x$. $\textup{MaxVar}(\mathcal{M})$ denotes the worst-case $\textup{Var}_{\mathcal{M}}[Y|x]$. We are interested in finding a privatization mechanism $\mathcal{M}$ that minimizes $\textup{MaxVar}(\mathcal{M})$ by solving the following constraint minimization problem:
\begin{align} 
  \nonumber &\min_{\mathcal{M}} \textup{MaxVar}(\mathcal{M}),\\
  \nonumber &~\textup{s.t.}~\textup{Eq}.~(\ref{local-differential-privacy-fomular}), \\
  \nonumber &~~~~~\mathbb{E}_{\mathcal{M}}[Y|x] = x,~\textup{and} \\
  \nonumber &~~~~~\mathbb{P}_{\mathcal{M}}[Y \in \mathbb{Y} | x] = 1.
\end{align}
The second constraint illustrates that our estimator is unbiased, and the third constraint shows the proper distribution where  $\mathbb{Y}$ is the range of randomized function $\mathcal{M}$. 
In the following sections, if $\mathcal{M}$ is clear from the context, we omit the subscript $\mathcal{M}$ for simplicity.

\section{Mechanisms for Estimation of A Single Numeric Attribute}\label{sec:single}
To solve the problem in Section~\ref{sec:problem}, we propose four local differential privacy mechanisms: \texttt{Three-Outputs}, \texttt{PM-OPT}, \texttt{PM-SUB}, and \texttt{HM-TP}. Fig.~\ref{fig:worst_var_comp} compares the worst-case noise variances of existing mechanisms and our proposed mechanisms. \texttt{Three-Outputs} has three discrete output possibilities, which incurs little communication cost because two bits are enough to encode three different outputs. Moreover, it achieves a small worst-case noise variance  in the high privacy regime (small privacy budget $\epsilon$). However, to maintain a low worst-case noise variance in the low privacy regime (large privacy budget $\epsilon$), we propose \texttt{PM-OPT} and \texttt{PM-SUB}. Both of them achieve higher accuracies than \texttt{Three-Outputs} and other existing solutions when the privacy budget $\epsilon$ is large. Additionally, we discretize their continuous ranges of output for vehicles to encode using a post-processing discretization algorithm. In the following sections, we will explain our proposed four mechanisms and the post-processing discretization algorithm in detail respectively.

\subsection{\texttt{Three-Outputs} Mechanism}
Now, we propose a mechanism with three output possibilities named as \texttt{Three-Outputs} which is illustrated in Algorithm~\ref{algo:three-outputs}. \texttt{Three-Outputs} ensures low communication cost while achieving a smaller worst-case noise variance than existing solutions in the high privacy regime (small privacy budget $\epsilon$).  Duchi~\textit{et~al.}'s~\cite{duchi2018minimax} solution contains two output possibilities, and it outperforms other approaches when the privacy budget is small. However, Kairouz~\textit{et~al.}~\cite{Kairouz2014ExtremalMF} prove that two outputs are not always optimal as $\epsilon$ increases. By outputting three values instead of two, \texttt{Three-Outputs} improves the performance as the privacy budget increases, which is shown in Fig.~\ref{fig:worst_var_comp}. When the privacy budget is small, \texttt{Three-Outputs} is equivalent to Duchi~\textit{et~al.}'s~\cite{duchi2018minimax} solution.

For notional simplicity, given a mechanism $\mathcal{M}$, we often write $\mathbb{P}_{\mathcal{M}}[Y=y \mid X = x]$ as $P_{y \leftarrow x}(\mathcal{M})$ below. We also sometimes omit $\mathcal{M}$ to obtain $\mathbb{P}[Y=y \mid X = x]$ and $P_{y \leftarrow x}$.

\begin{algorithm}[h]
 \caption{\texttt{Three-Outputs} Mechanism for One-Dimensional Numeric Data.}
 \label{algo:three-outputs}
  \KwIn{ tuple $x \in [-1,1]$ and privacy parameter $\epsilon$.} 
  \KwOut{tuple $Y \in \{-C, 0, C \}$.}
  {Sampling a random variable $u$ with the probability distribution as follows:
    \begin{align}
        &\bp{u=-1}= P_{-C \leftarrow x},  \nonumber \\
        &\bp{u=0}=P_{0 \leftarrow x},~\textup{and} \nonumber \\
        &\bp{u=1}= P_{C \leftarrow x},  \nonumber 
    \end{align}
    \nonumber where $P_{-C \leftarrow x}$, $P_{0 \leftarrow x}$ and  $P_{C \leftarrow x}$  are given in Eq.~(\ref{eq:probability--C}), Eq.~(\ref{eq:probability-C}) and Eq.~(\ref{eq:probability-0}).}\\
     \uIf{$u= -1$}{
        $Y= -C$;
      }
      \uElseIf{$u = 0$}{
        $Y = 0$;
      }
      \Else{
        $Y= C$;
      }
   return $Y$;
\end{algorithm}

Given a tuple $x \in [-1, 1]$, \texttt{Three-Outputs} returns a perturbed value $Y$ that equals $-C$, $0$ or $C$ with probabilities defined by
\begin{align}\label{eq:probability--C}
    P_{-C \leftarrow x}\hspace{-2pt} =\hspace{-2pt} \begin{cases}
        &\hspace{-12pt}\frac{1 - P_{0\leftarrow 0}}{2}\hspace{-2pt} + \hspace{-2pt}\left ( \frac{1-P_{0\leftarrow 0} }{2} \hspace{-2pt}- \hspace{-2pt}\frac{e^\epsilon -P_{0\leftarrow 0} }{e^\epsilon(e^\epsilon+1)}  \right )x,\text{if}~0\leq x\leq 1, \\
        &\hspace{-12pt}\frac{1 - P_{0\leftarrow 0}}{2} \hspace{-2pt}+\hspace{-2pt} \left ( \frac{e^\epsilon \hspace{-2pt}-\hspace{-2pt}P_{0\leftarrow 0} }{e^\epsilon+1}\hspace{-2pt}-\hspace{-2pt} \frac{1-P_{0\leftarrow 0} }{2} \right )x,\text{if }-1 \leq x\leq 0,
    \end{cases}
\end{align}
\begin{align}\label{eq:probability-C}
    P_{C \leftarrow x}\hspace{-2pt} =\hspace{-2pt} \begin{cases}
        &\hspace{-12pt}\frac{1 - P_{0\leftarrow 0}}{2} \hspace{-2pt}+\hspace{-2pt} \left ( \frac{e^\epsilon -P_{0\leftarrow 0} }{e^\epsilon+1}\hspace{-2pt}-\hspace{-2pt} \frac{1-P_{0\leftarrow 0} }{2} \right )x,\text{if }0 \leq x\leq 1, \\
        &\hspace{-12pt}\frac{1 - P_{0\leftarrow 0}}{2}\hspace{-2pt} +\hspace{-2pt} \left ( \frac{1-P_{0\leftarrow 0} }{2}\hspace{-2pt} -\hspace{-2pt} \frac{e^\epsilon -P_{0\leftarrow 0} }{e^\epsilon(e^\epsilon+1)}  \right ) x,\text{if}-1\leq x\leq 0,
    \end{cases}
\end{align}
\begin{align}\label{eq:probability-0}
  \textup{and}~  P_{0 \leftarrow x} = P_{0 \leftarrow 0} + (\frac{P_{0 \leftarrow 0}}{e^\epsilon} - P_{0 \leftarrow 0}) x,~\textup{if} -1\leq x\leq 1,
\end{align}
where $P_{0\leftarrow 0}$ is defined by
\begin{align}
       &\hspace{-10pt}P_{0 \leftarrow 0} := \nonumber \\\label{a_opt_1}
       &\hspace{-10pt}\begin{cases}
          & \hspace{-10pt}0, ~~\text{if}~\epsilon < \ln2, \\
          & \hspace{-10pt}  -\frac{1}{6}(-e^{2\epsilon}-4e^\epsilon-5 \\& +2\sqrt{\Delta_0} \cos(\frac{\pi}{3} + \frac{1}{3} \arccos(-\frac{\Delta_1}{2\Delta_0^\frac{3}{2}}))),   \text{if}~ \ln 2 \leq \epsilon \leq  \epsilon', \\
         &\hspace{-10pt} \frac{e^\epsilon}{e^\epsilon+2}, ~\text{if}~\epsilon >  \epsilon',
        \end{cases}
    \end{align} in which
    \begin{align}
        \Delta_0 &:= e^{4\epsilon}+14e^{3\epsilon}+50e^{2\epsilon}-2e^\epsilon+25,\label{eq:delta-0}\\
       \Delta_1 &:= -2e^{6\epsilon} -42e^{5\epsilon}-270e^{4\epsilon}-404e^{3\epsilon}-918e^{2\epsilon}\nonumber\\&+30e^\epsilon -250, \label{eq:delta-1} \\
       \textup{and}~\epsilon' &:= \ln \left(\frac{3 + \sqrt{65}}{2}\right) \approx \ln5.53.\label{eq:epsilon-hash}
    \end{align}

\textbf{Next, we will show how we derive the above probabilities.} 
For a mechanism which uses  $x \in [-1,1]$ as the input and only three possibilities $-C, 0, C$ for the output value, it satisfies
\begin{subnumcases}{}
\textup{$\epsilon$-LDP} : \frac{P_{C \leftarrow x}}{P_{C \leftarrow x'}}, \frac{P_{0 \leftarrow x}}{P_{0 \leftarrow x'}}, \frac{P_{-C \leftarrow x}}{P_{-C \leftarrow x'}} \in [e^{-\epsilon}, e^{\epsilon}], \label{LDP} \\ \textup{unbiased estimation:} \nonumber \\ \quad C \cdot P_{C \leftarrow x} + 0 \cdot P_{0 \leftarrow x}+(-C) \cdot P_{-C \leftarrow x}= x ,\label{unbiased}  \\ 
\textup{proper distribution:} \nonumber \\ \quad P_{y \leftarrow x} \geq 0~\textup{and}~ P_{C \leftarrow x} + P_{0 \leftarrow x} + P_{-C \leftarrow x} = 1.\label{proper}  
\end{subnumcases}

To calculate values of $P_{C \leftarrow x}$, $P_{0 \leftarrow x}$ and $P_{-C \leftarrow x}$, we use Lemma~\ref{lem-symmetrization} below to convert a mechanism $\mathcal{M}_1$ satisfying the requirements in (\ref{LDP})~(\ref{unbiased})~(\ref{proper}) to a symmetric mechanism $\mathcal{M}_2$. Then, we use Lemma~\ref{lem:lemma-2} below to transform the symmetric mechanism further to $\mathcal{M}_3$ whose worst-case noise variance is smaller than $\mathcal{M}_2$'s. Next, we use $P_{0 \leftarrow 1}$ to represent other probabilities, and then we prove that we get the minimum variance when $P_{0 \leftarrow 0} = e^{\epsilon} P_{0 \leftarrow 1}$ using Lemma~\ref{lem:lem-3}. Finally, Lemma~\ref{lem:optimal_a} and Lemma~\ref{lem:worst-case-variance-three-outputs} are used to obtain values for $P_{0 \leftarrow 0}$ and the worst-case noise variance of \texttt{Three-Outputs}, respectively. Thus, we can obtain values of $P_{C \leftarrow x}$, $P_{0 \leftarrow x}$ and $P_{-C \leftarrow x}$ using $P_{0 \leftarrow 0}$. In the following, we will illustrate above processes in detail.

By symmetry, for any $x \in [-1, 1]$, we enforce
\begin{subnumcases}{}
P_{C \leftarrow x}  =  P_{-C \leftarrow -x},\label{lem-symmetrization-conditions-a} \\ P_{0 \leftarrow x}  =  P_{0 \leftarrow -x},\label{lem-symmetrization-conditions-b}
\end{subnumcases}
where Eq.~(\ref{lem-symmetrization-conditions-b}) can be derived from Eq.~(\ref{lem-symmetrization-conditions-a}). The formal justification of Eq.~(\ref{lem-symmetrization-conditions-a})~(\ref{lem-symmetrization-conditions-b}) is given by Lemma~\ref{lem-symmetrization} below. Since the input domain $[-1,1]$ is symmetric, we can transform any mechanism satisfying requirements in (\ref{LDP})~(\ref{unbiased})~(\ref{proper}) to a symmetric mechanism while guaranteeing the worst-case noise variance will not increase in Lemma~\ref{lem-symmetrization}. Thus, we can derive probabilities when $x \in [-1, 0]$ using probabilities when $x \in [0,1]$ based on the symmetry.

\begin{lem} \label{lem-symmetrization}
For a mechanism $\mathcal{M}_1$ satisfying the requirements in (\ref{LDP})~(\ref{unbiased})~(\ref{proper}), the following symmetrization process to obtain a mechanism $\mathcal{M}_2$ will not increase (i.e., will reduce or not change) the worst-case noise variance, while mechanism $\mathcal{M}_2$ still satisfies the requirements in (\ref{LDP})~(\ref{unbiased})~(\ref{proper}).
Symmetrization: For $x \in [-1, 1]$,
\begin{align}
&P_{C \leftarrow x}(\mathcal{M}_2) =  P_{-C \leftarrow -x}(\mathcal{M}_2) = \frac{P_{C \leftarrow x}(\mathcal{M}_1) + P_{-C \leftarrow -x}(\mathcal{M}_1)}{2},\label{eq:symmetrization-1} \\ &P_{0 \leftarrow x}(\mathcal{M}_2) =  P_{0 \leftarrow -x}(\mathcal{M}_2) = \frac{P_{0 \leftarrow x}(\mathcal{M}_1) + P_{0 \leftarrow -x}(\mathcal{M}_1)}{2}.\label{eq:symmetrization-2} 
\end{align}
\end{lem}
\begin{proof}
    The proof details are given in Appendix~\ref{app:proof-of-lem-1} of the submitted supplementary file.
\end{proof}
Based on Lemma~\ref{lem-symmetrization}, we define a symmetric mechanism as follows.

\textbf{Symmetric Mechanism.} A mechanism under (\ref{LDP})~(\ref{unbiased})~(\ref{proper}) is called a symmetric mechanism if it satisfies Eq.~(\ref{lem-symmetrization-conditions-a})~(\ref{lem-symmetrization-conditions-b}). In the following, we only consider the symmetric mechanism $\mathcal{M}_2$.

Now, we design probabilities for the symmetric mechanism $\mathcal{M}_2$.  As $\mathcal{M}_2$ satisfies the unbiased estimation which is a linear relationship, we set probabilities as piecewise linear functions of $x$ as follows:

\textbf{Case 1:} For $x \in [0, 1]$,
\begin{align}
P_{C \leftarrow x} & = P_{C \leftarrow 0} + (P_{C \leftarrow 1}-P_{C \leftarrow 0})x,\label{eq:lem-1-eq-1} \\ P_{-C \leftarrow x} & = P_{-C \leftarrow 0} - (P_{-C \leftarrow 0}-P_{-C \leftarrow 1})x, \label{eq:lem-1-eq-2} \\    P_{0 \leftarrow x} &= 1-P_{-C \leftarrow 0}-P_{C \leftarrow 0} \nonumber \\& + (P_{-C \leftarrow 0}-P_{C \leftarrow 0} + P_{-C \leftarrow 1} -P_{C \leftarrow 1})x.
\end{align}

\textbf{Case 2:} For $x \in [-1, 0]$,
\begin{align}
P_{C \leftarrow x} & = P_{C \leftarrow 0} + (P_{C \leftarrow 0}-P_{C \leftarrow -1})x, \\ P_{-C \leftarrow x} & = P_{-C \leftarrow 0} - (P_{-C \leftarrow -1}-P_{-C \leftarrow 0})x, \\ P_{0 \leftarrow x} &= 1-P_{-C \leftarrow 0}-P_{C \leftarrow 0} \nonumber \\& + (P_{-C \leftarrow 0}-P_{C \leftarrow 0} + P_{-C \leftarrow 1} -P_{C \leftarrow 1})x.
\end{align}

Then, we may assign values to our designed probabilities above. We find that if a symmetric mechanism satisfies Eq.~(\ref{lemma-2-enforce-1}) and Eq.~(\ref{lemma-2-enforce-2}), it obtains a smaller worst-case noise variance.
From Lemma~\ref{lem:lemma-2} below, we enforce
\begin{subnumcases}{}
    P_{C \leftarrow 1} = e^{\epsilon} P_{C \leftarrow -1}, \label{lemma-2-enforce-1} \\
    P_{-C \leftarrow -1} = e^{\epsilon} P_{-C \leftarrow 1}.\label{lemma-2-enforce-2}
\end{subnumcases}
Hence, given a symmetric mechanism $\mathcal{M}_2$ satisfying Inequality~(\ref{lem:lem-2-assumption}), we can transform it to a new symmetric mechanism $\mathcal{M}_3$ which satisfies Eq.~(\ref{lemma-2-enforce-1}) and Eq.~(\ref{lemma-2-enforce-2}) through processes of Eq.~(\ref{eq:lem-2-1})~(\ref{eq:lem-2-2})~(\ref{eq:lem-2-3}) until $P_{C \leftarrow -1} = e^\epsilon P_{-C \leftarrow 1}$.  After transformation, the new mechanism $\mathcal{M}_3$ achieves a smaller worst-case noise variance than mechanism $\mathcal{M}_2$. Therefore, we use the new symmetric mechanism  $\mathcal{M}_3$ to replace $\mathcal{M}_2$ in the future's discussion. Details of transformation are in the Lemma~\ref{lem:lemma-2}.

\begin{lem}\label{lem:lemma-2}
For a symmetric mechanism $\mathcal{M}_2$,
if \begin{align}
    P_{C \leftarrow 1}(\mathcal{M}_2) <  e^{\epsilon} P_{C \leftarrow -1}(\mathcal{M}_2), \label{lem:lem-2-assumption}
\end{align}we set a symmetric mechanism $\mathcal{M}_3$ as follows: 
For $x \in [-1, 1]$,
\begin{align}
&P_{C \leftarrow x}(\mathcal{M}_3)  = P_{-C \leftarrow -x}(\mathcal{M}_3) \nonumber \\& = P_{C \leftarrow x}(\mathcal{M}_2) - \frac{e^{\epsilon} P_{C \leftarrow -1}(\mathcal{M}_2)-P_{C \leftarrow 1}(\mathcal{M}_2)}{e^{\epsilon}-1},\label{eq:lem-2-1}
\end{align}
\begin{align}
&P_{-C \leftarrow x}(\mathcal{M}_3)  = P_{C \leftarrow -x}(\mathcal{M}_3) \nonumber \\& = P_{-C \leftarrow x}(\mathcal{M}_2) - \frac{e^{\epsilon} P_{-C \leftarrow 1}(\mathcal{M}_2)-P_{-C \leftarrow -1}(\mathcal{M}_2)}{e^{\epsilon}-1},\label{eq:lem-2-2}
\end{align}
\begin{align}
&P_{0 \leftarrow x}(\mathcal{M}_3) = 1- P_{C \leftarrow x}(\mathcal{M}_3)-P_{-C \leftarrow x}(\mathcal{M}_3) \nonumber \\ & = P_{0 \leftarrow x}(\mathcal{M}_2) + \frac{2 (e^\epsilon P_{C \leftarrow -1}(\mathcal{M}_2) - P_{C \leftarrow 1}(\mathcal{M}_2))}{e^\epsilon - 1}.\label{eq:lem-2-3}
\end{align}
Moreover, the mechanism $\mathcal{M}_3$ has a worst-case noise variance smaller than that of $\mathcal{M}_2$, while $\mathcal{M}_3$ still satisfies the requirements in (\ref{LDP})~(\ref{unbiased})~(\ref{proper}).
\end{lem}

\begin{proof}
    The proof details are given in Appendix~\ref{app:proof-of-lem-2} of the submitted supplementary file.
\end{proof}

We have proved that the symmetric mechanism $\mathcal{M}_3$ has a smaller worst-case noise variance  than that of mechanism $\mathcal{M}_2$ in Lemma~\ref{lem:lemma-2}, and then we use mechanism $\mathcal{M}_3$ to obtain the relation between $P_{0 \leftarrow 1}$ and $P_{0 \leftarrow 0}$ to find the minimum variance. 
From Lemma \ref{lem:lem-3} below, we enforce
\begin{align}
P_{0 \leftarrow 0} = e^{\epsilon} P_{0 \leftarrow 1}.
\end{align}
Then, we use the following Lemma \ref{lem:lem-3} to obtain the relation between $P_{0 \leftarrow 1}$ and $P_{0 \leftarrow 0}$, so that we can obtain $P_{C \leftarrow x}$, $P_{0 \leftarrow x}$ and $P_{-C \leftarrow x}$ using $P_{0 \leftarrow 0}$.

\begin{lem}\label{lem:lem-3}
Given $P_{0 \leftarrow 0}$, the variance of the output given input $x$ is a strictly decreasing function of $P_{0 \leftarrow 1}$ and hence is minimized when $P_{0 \leftarrow 1}  = \frac{P_{0 \leftarrow 0}}{e^{\epsilon}}$.

\end{lem}

\begin{proof}
    The proof details are given in Appendix~\ref{app:proof-of-lem-3} of the submitted supplementary file.
\end{proof}

Lemma~\ref{lem:lem-3} shows that we get the minimum variance when $P_{0 \leftarrow 1}  = \frac{P_{0 \leftarrow 0}}{e^{\epsilon}}$. Hence, we replace $e^{\epsilon} P_{0 \leftarrow 1}$ with $P_{0 \leftarrow 0}$. Then, the variance is equivalent to
\begin{align}
    &\textup{Var}[Y | X =x] \nonumber \\&= \left( \frac{e^\epsilon + 1}{(e^\epsilon -1)(1-\frac{P_{0 \leftarrow 0}}{e^\epsilon})} \right)^2 \left( 1 - P_{0 \leftarrow 0} + (P_{0 \leftarrow 0} - \frac{P_{0 \leftarrow 0}}{e^\epsilon})|x|\right)\nonumber \\& \quad - x^2.~\label{eq:th-var}
\end{align}
Complete details for obtaining Eq.~(\ref{eq:th-var}) are in Appendix~\ref{app:proof-of-lem-3} of the submitted supplementary file.

Next, we use Lemma~\ref{lem:optimal_a} to obtain the optimal $P_{0 \leftarrow 0}$ in \texttt{Three-Outputs} to achieve the minimum worst-case variance as follows:

\begin{lem}\label{lem:optimal_a}
The optimal $P_{0 \leftarrow 0}$ to minimize the $\max_{x \in [-1,1]}\textup{Var}[Y|x]$ is defined by Eq.~(\ref{a_opt_1}).
\end{lem}

\begin{proof}
 The proof details are given in Appendix~\ref{appendix:proof_optimal_a} of the submitted supplementary file.
\end{proof}

\begin{rem}
    Fig.~\ref{fig:optimal_a} displays how $P_{0 \leftarrow 0}$ changes with $\epsilon$ in Eq.~(\ref{a_opt_1}). When the privacy budget $\epsilon$ is small, $P_{0 \leftarrow 0} = 0$. Thus, \texttt{Three-Outputs} is equivalent to Duchi~\textit{et~al.}'s~\cite{duchi2018minimax} solution when $P_{0 \leftarrow 0} = 0$. However, as the privacy budget $\epsilon$ increases, $P_{0 \leftarrow 0}$ increases, which means that the probability of outputting true value increases.
\end{rem}

\begin{figure}[!h]
    \centering
    \includegraphics[scale=0.4]{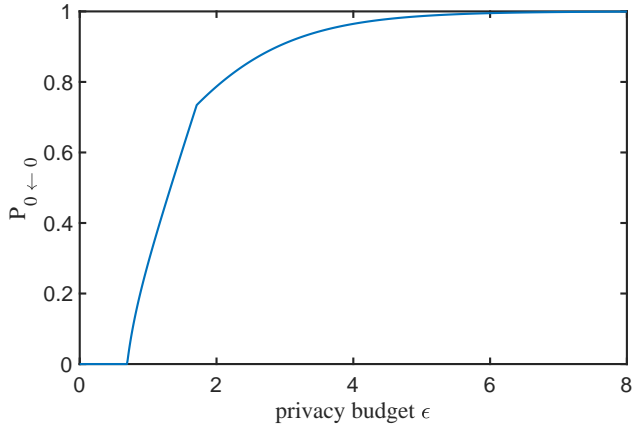}
    \caption{Optimal $P_{0 \leftarrow 0}$ if the privacy budget $\epsilon \in [0, 8]$.}
    \label{fig:optimal_a}
\end{figure}

By summarizing above, we obtain $P_{-C \leftarrow x}$, $P_{C \leftarrow x}$ and $P_{0 \leftarrow x}$ from Eq.~(\ref{eq:probability--C}), Eq.~(\ref{eq:probability-C}) and Eq.~(\ref{eq:probability-0}) using $P_{0 \leftarrow 0}$.

Then, we can calculate the optimal $P_{0 \leftarrow 0}$ to obtain the minimum worst-case noise variance of \texttt{Three-Outputs} as follows:

\begin{lem}\label{lem:worst-case-variance-three-outputs}
The minimum worst-case noise variance of \texttt{Three-Outputs} is obtained when $P_{0 \leftarrow 0}$ satisfies Eq.~(\ref{a_opt_1}).

\end{lem}
\begin{proof}
    The proof details are given in Appendix~\ref{appendix:lem-three-outputs-variance} of the submitted supplementary file.
\end{proof}

\textbf{A clarification about~\texttt{Three-Outputs}~versus~\texttt{Four-Outputs}.}
One may wonder why we consider a perturbation mechanism with three outputs (i.e., our \texttt{Three-Outputs}) instead of a perturbation mechanism with four outputs (referred to as \texttt{Four-Outputs}), since using two bits to encode the output of a perturbation mechanism can represent four outputs. The reason is as follows. The approach to design \texttt{Four-Outputs} is similar to that for \texttt{Three-Outputs}, but the detailed analysis for \texttt{Four-Outputs} will be even more tedious than that for \texttt{Three-Outputs} (which is already quite complex). Given above reasons, we elaborate \texttt{Three-Outputs} but not \texttt{Four-Outputs} in this paper.

\subsection{\texttt{PM-OPT} Mechanism}
Now, we advocate an optimal piecewise mechanism (\texttt{PM-OPT}) as shown in Algorithm~\ref{algo:PM-OPT} to get a small worst-case variance when the privacy budget is large. As shown in Fig.~\ref{fig:worst_var_comp},  \texttt{Three-Outputs}'s worst-case noise variance is smaller than \texttt{PM}'s when the privacy budget $\epsilon < 3.2$. But it loses the advantage when the privacy budget $\epsilon \geq 3.2$.  As the privacy budget increases, Kairouz~\textit{et~al.}~\cite{Kairouz2014ExtremalMF} suggested to send more information using more output possibilities. Besides, we observe that it is possible to improve Wang~\textit{et~al.}'s~\cite{wang2019collecting} \texttt{PM} to achieve a smaller worst-case noise variance. Thus, inspired by them, we propose an optimal piecewise mechanism named as \texttt{PM-OPT} with a smaller worst-case noise variance than \texttt{PM}.
\begin{algorithm}[!h]
 \caption{\texttt{PM-OPT} Mechanism for One-Dimensional Numeric Data under Local Differential Privacy.}
 \label{algo:PM-OPT}
  \KwIn{ tuple $x \in [-1,1]$ and privacy parameter $\epsilon$.}
  \KwOut{tuple $Y \in [-A,  A]$.}
  {Value $t$ is calculated in the Eq.~(\ref{eq:t-value});} \\
  {Sample $u$ uniformly at random from $[0,1]$;} \\
   \eIf{$u< \frac{e^\epsilon}{t + e^\epsilon}$}
   {Sample $Y$ uniformly at random from $[L(\epsilon,x,t),  R(\epsilon,x,t)];$}
   {Sample $Y$ uniformly at random from $[-A, L(\epsilon,x,t) ) \cup  ( R(\epsilon,x,t),A ];$}
   return $Y$;
\end{algorithm}

\begin{figure}[!h]
\centering 
  \includegraphics[width=0.45\textwidth]{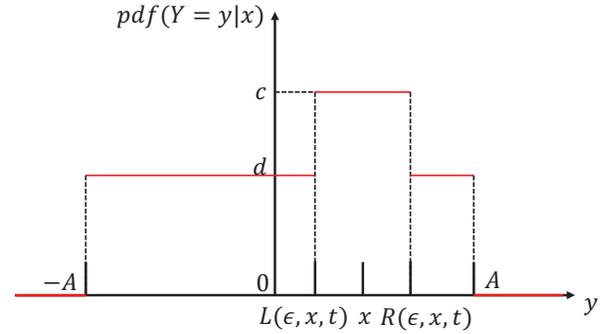} \vspace{-2pt}\caption{The probability density function $\bfu{Y=y|x}$ of the randomized output $Y$ after applying $\epsilon$-local differential privacy.}
 \label{PDFeps}
\end{figure}

For a true input $x \in [-1, 1]$, the probability density function of the randomized output $Y \in [-A, A]$ after applying local differential privacy is given by
\begin{subnumcases}{\hspace{-28pt}\bfu{Y=y|x}\hspace{-2pt}\label{eq:pdf-ldp}=\hspace{-3pt}}\label{eq:pdf-ldp-a}
  \hspace{-5pt}c, & \hspace{-18pt}for $ y\hspace{-3pt} \in\hspace{-3pt} [L(\epsilon,x,t),  R(\epsilon,x,t)]$,  \\\label{eq:pdf-ldp-b}
  \hspace{-5pt}d, & \hspace{-18pt}for $ y \hspace{-3pt}\in\hspace{-3pt} [-A, L(\epsilon,x,t))  \cup  (R(\epsilon,x,t), A]$, 
\end{subnumcases}
where
\begin{align}
    c &=  \frac{e^{\epsilon} t (e^{\epsilon}-1)}{2(t+e^\epsilon)^2},\label{eq:opt-pdf-c} \\
    d &=  \frac{ t (e^{\epsilon}-1)}{2(t+e^\epsilon)^2},\label{eq:opt-pdf-d} \\
    A &= \frac{(e^\epsilon+t)(t+1)}{t(e^\epsilon - 1)},\label{eq:opt-pdf-A} \\
    L(\epsilon,x,t) &= \frac{(e^{\epsilon}+t)(xt-1)}{t (e^{\epsilon}-1)}, \nonumber\\
    R(\epsilon,x,t) &= \frac{(e^{\epsilon}+t)(xt+1)}{t (e^{\epsilon}-1)},~\textup{and} \nonumber
\end{align}
\begin{subequations}
\begin{empheq}[left= {\hspace{-17pt}t\hspace{-1.5pt}=\hspace{-2.5pt}\empheqlbrace}]{align}
&\nonumber \hspace{-2.5pt}\frac{1}{2} \sqrt{e^{2\epsilon} \hspace{-1.5pt}+\hspace{-1.5pt} 2^{2/3} \sqrt[3]{e^{2\epsilon} \hspace{-1.5pt}-\hspace{-1.5pt} e^{4\epsilon}}} +\\&\nonumber \hspace{-1.5pt} \frac{1}{2} \sqrt{2 e^{2\epsilon} \hspace{-1.5pt}-\hspace{-1.5pt} 2^{2/3} \sqrt[3]{e^{2\epsilon} \hspace{-1.5pt}-\hspace{-1.5pt} e^{4\epsilon}} \hspace{1pt} +  \hspace{1pt} \frac{4 e^{\epsilon} \hspace{-1.5pt}-\hspace{-1.5pt} 2 e^{3\epsilon}}{\sqrt{e^{2\epsilon} \hspace{-1.5pt}+\hspace{-1.5pt} 2^{2/3} \sqrt[3]{e^{2\epsilon} \hspace{-1.5pt}-\hspace{-1.5pt} e^{4\epsilon}}}}} \\&\hspace{-1.5pt}-\hspace{-1.5pt}   \frac{e^{\epsilon}}{2},\hspace{10pt} \text{if}\hspace{5pt} \epsilon < \ln \sqrt{2}, \\& \nonumber \hspace{-2.5pt}-\frac{1}{2} \sqrt{e^{2\epsilon} \hspace{-1.5pt}+\hspace{-1.5pt} 2^{2/3} \sqrt[3]{e^{2\epsilon} \hspace{-1.5pt}-\hspace{-1.5pt} e^{4\epsilon}}} \hspace{-1.5pt}+\hspace{-1.5pt} \\& \nonumber \frac{1}{2} \sqrt{2 e^{2\epsilon} \hspace{-1.5pt}-\hspace{-1.5pt} 2^{2/3} \sqrt[3]{e^{2\epsilon} \hspace{-1.5pt}-\hspace{-1.5pt} e^{4\epsilon}} \hspace{1pt} - \hspace{1pt} \frac{4 e^{\epsilon} \hspace{-1.5pt}-\hspace{-1.5pt} 2 e^{3\epsilon}}{\sqrt{e^{2\epsilon} \hspace{-1.5pt}+\hspace{-1.5pt} 2^{2/3} \sqrt[3]{e^{2\epsilon} \hspace{-1.5pt}-\hspace{-1.5pt} e^{4\epsilon}}}}} \\& \hspace{-1.5pt}-\hspace{-1.5pt}   \frac{e^{\epsilon}}{2}, \hspace{10pt}~\text{if}~ \epsilon   > \ln \sqrt{2},
\\& \frac{\sqrt{3+2\sqrt{3}} - 1}{\sqrt{2}}  , \hspace{10pt}~\text{if}~\epsilon  = \ln \sqrt{2} .
\end{empheq}\label{eq:t-value}
\end{subequations}

The meaning of $t$ can be seen from $\frac{t-1}{t+1} = \frac{L(\epsilon,1,t)}{R(\epsilon,1,t)}$. When the input is $x = 1$, the length of the higher probability density function $\bfu{Y=y|x}=\frac{e^{\epsilon} t (e^{\epsilon}-1)}{2(t+e^\epsilon)^2}$ is $R(\epsilon,1,t) - L(\epsilon,1,t)$. $R(\epsilon,1,t)$ is the right boundary, and $L(\epsilon,1,t)$ is the left boundary. If $0 < t < \infty$, we can derive $\lim_{t \to 0} \frac{t-1}{t+1} = -1$, meaning the right boundary is opposite to the left boundary if $t$ is close to $0$. Since $\lim_{t \to \infty} \frac{t-1}{t+1} = 1$, it means that the right boundary is equal to the left boundary when $t$ is close to $\infty$. 

Moreover, Fig.~\ref{PDFeps} illustrates that the probability density function of Eq.~(\ref{eq:pdf-ldp}) contains three pieces. If $y \in [L(\epsilon,x,t),  R(\epsilon,x,t)]$, the probability density function is equal to $c$ which is higher than other two pieces $y \in [-A, L(\epsilon,x,t))$ and $y \in (R(\epsilon,x,t), A]$. We calculate the probability of a variable $Y$ falling in the interval $[L(\epsilon,x,t), R(\epsilon,x,t)]$ as $\bp{L(\epsilon,x,t) \leq Y \leq R(\epsilon,x,t)} = \int_{L(\epsilon,x,t)}^{R(\epsilon,x,t)} c~dY = \frac{e^\epsilon}{t+e^\epsilon}$.

Furthermore, we use the following lemmas to establish how we get the value $t$ in Eq.~(\ref{eq:pdf-ldp}).

\begin{lem}\label{PM-lemma-1}
Algorithm~\ref{algo:PM-OPT} achieves $\epsilon$-local differential privacy. Given an input value $x$, it returns a noisy value $Y$ with $\mathbb{E}[Y|x] = x$ and 
\begin{align}
\textup{Var}[Y|x] = \frac{t+1}{e^\epsilon-1}x^2 + \frac{(t+e^\epsilon)\big ((t+1)^3 + e^\epsilon -1 \big)}{3t^2(e^\epsilon-1)^2}.
\end{align}
\end{lem}

\begin{proof}
The proof details are given in Appendix~\ref{appendix:lem-1} of the submitted supplementary file.
\end{proof}
Thus, when $x=1$, we obtain the worst-case noise variance as follows:
\begin{align}
 \max_{x \in [-1, 1]} \textup{Var}[Y|x] = \frac{t+1}{e^\epsilon-1} + \frac{(t+e^\epsilon)\big ((t+1)^3 + e^\epsilon -1 \big)}{3t^2(e^\epsilon-1)^2}.\label{eq:max-var-y-to-x}
\end{align}

Then, we obtain the optimal $t$ in Lemma~\ref{eq:optimal-t} to minimize Eq.~(\ref{eq:max-var-y-to-x}).

\begin{lem}\label{eq:optimal-t}
The  optimal t for $\min_{t} \max_{x \in [-1, 1]} \textup{Var}[Y|x]$ is Eq.~(\ref{eq:t-value}).
\end{lem}

\begin{proof}
By computing the first-order derivative and second-order derivative of $\min_{t} \max_{x \in [-1, 1]} \textup{Var}[Y|x]$, we get the optimal $t$. The proof details are given in Appendix~\ref{appendix:solve_quartic} of the submitted supplementary file.
\end{proof}

\subsection{\texttt{PM-SUB} Mechanism}\label{sec:pm-sub-mechanism}
We propose a suboptimal piecewise mechanism (\texttt{PM-SUB}) to simplify the sophisticated computation of $t$ in Eq.~(\ref{PDFeps}) of \texttt{PM-OPT}, and details of \texttt{PM-SUB} are shown in Algorithm~\ref{alg:pm_sub}. 

\begin{algorithm}[h]
 \caption{\texttt{PM-SUB} Mechanism for One-Dimensional Numeric Data under Local Differential Privacy.}
  \label{alg:pm_sub}
  \KwIn{ tuple $x \in [-1,1]$ and privacy parameter $\epsilon$.} 
  \KwOut{tuple $Y \in [-A,  A]$.}
  {Sample $u$ uniformly at random from $[0,1]$;} \\
   \eIf{$u< \frac{e^\epsilon}{e^{\epsilon/3} + e^\epsilon}$}
   {Sample $Y$ uniformly at random from $[\frac{(e^{\epsilon}+e^{\epsilon/3})(xe^{\epsilon/3}-1)}{e^{\epsilon/3} (e^{\epsilon}-1)},  \frac{(e^{\epsilon}+e^{\epsilon/3})(xe^{\epsilon/3}+1)}{e^{\epsilon/3} (e^{\epsilon}-1)}];$}
   {Sample $Y$ uniformly at random from $[-A, \frac{(e^{\epsilon}+e^{\epsilon/3})(xe^{\epsilon/3}-1)}{e^{\epsilon/3} (e^{\epsilon}-1)} ) \cup  ( \frac{(e^{\epsilon}+e^{\epsilon/3})(xe^{\epsilon/3}+1)}{e^{\epsilon/3} (e^{\epsilon}-1)},A ];$}
   return $Y$;
\end{algorithm}

Fig.~\ref{fig:worst_var_comp} illustrates that \texttt{PM-OPT} achieves a smaller worst-case noise variance compared with \texttt{PM}, but the parameter $t$ for \texttt{PM-OPT} in Eq.~(\ref{eq:t-value}) is complicated to compute. Some vehicles are unable to process the complicated computation. To make $t$ simple for vehicles to implement, we need to find a simple expression for it  while ensuring the mechanism's performance. Then, we find that Wang~\textit{et al.}'s~\cite{wang2019collecting} \texttt{PM} is the case when $t = e^{\epsilon/2}$. Inspired by \texttt{PM}, $\ln t$ and $\epsilon$ can be linearly related. Then, we find that $\frac{\ln t}{\epsilon}$ is close to $\frac{1}{3}$ ($t$ for \texttt{PM-OPT} in Eq.~(\ref{eq:t-value})), so we can set $e^{\epsilon/3}$ as $t$ in Eq.~(\ref{eq:pdf-ldp}) for a new mechanism named as \texttt{PM-SUB}. The probability of a variable $Y$ falling in the interval $[L(\epsilon,x,e^{\epsilon/3}), R(\epsilon,x,e^{\epsilon/3})]$ is $\frac{e^\epsilon}{e^{\epsilon/3}+e^\epsilon}$, and we give the detail of proof in Appendix~\ref{appendix:y-pdf}.

\begin{figure}[h]
    \centering
    \includegraphics[scale=0.4]{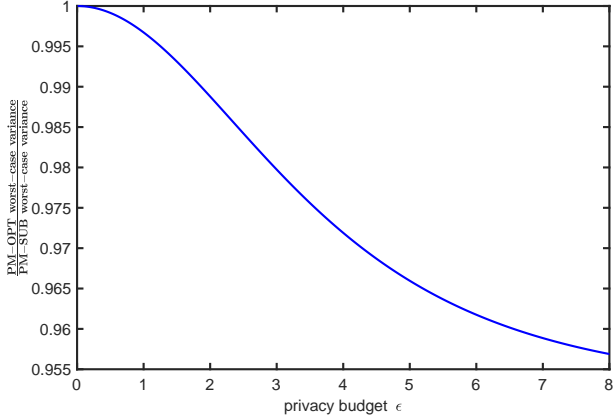}
    \caption{\texttt{PM-OPT}'s worst-case noise variance versus \texttt{PM-SUB}'s worst-case noise variance.}
    \label{fig:OPT-SUB-worst_var_comp}
\end{figure}

Similar to \texttt{PM-OPT}, we derive the worst-case noise variance of \texttt{PM-SUB} from Lemma~\ref{PM-lemma-1} with $t = e^{\epsilon/3}$ as follows:
\begin{align}
    \hspace{-5pt}\max_{x \in [-1, 1]} \textup{Var}[Y|x] = \frac{5e^{4\epsilon/3}}{3 (e^{\epsilon}-1)^2}     +\frac{5 e^{2\epsilon/3}}{3 (e^{\epsilon}-1)^2} +\frac{2 e^{\epsilon}}{(e^{\epsilon}-1)^2} .
\end{align}
As shown in Fig.~\ref{fig:OPT-SUB-worst_var_comp}, \texttt{PM-SUB}'s worst-case noise variance is close to \texttt{PM-OPT}'s, but it is smaller than \texttt{PM}'s, which can be observed in Fig.~\ref{fig:worst_var_comp}.

\subsection{Discretization Post-Processing}\label{algo:discretization}
Both \texttt{PM-OPT} and \texttt{PM-SUB}'s output ranges is $[-1,1]$ which is continuous, so that there are infinite output possibilities given an input $x$. Thus, it is difficult to encode their outputs for vehicles. Hence, we consider to apply a post-processing process to discretize the continuous output range into finite output possibilities. Algorithm~\ref{algo:discretization} shows our discretization post-processing steps.

\begin{algorithm}[h]
 \caption{Discretization Post-Processing.}
 \label{algo:discretization}
  \KwIn{Perturbed data $ y \in [-C, C]$, and domain $[-C, C]$ is separated into $2m$ pieces, where $m$ is a positive integer.} 
  \KwOut{Discrete data $Z$.}
  {Sample a Bernoulli variable $u$ such that
  \begin{align}
      \hspace{-30pt}\bp{u= 1}=
        \left (\frac{C\cdot (\left\lfloor \frac{m\cdot y}{C}\right\rfloor+1)}{m}-y \right) \cdot \frac{m}{C}; \nonumber
   \end{align} } \\
   \eIf{$u = 1$}
    {$Z = \frac{C\cdot \left\lfloor \frac{m\cdot y}{C}\right\rfloor}{m}$;}
    {$Z = \frac{C\cdot (\left\lfloor \frac{m\cdot y}{C}\right\rfloor+1)}{m}$;}
   return $Z$;
\end{algorithm}

The idea of Algorithm~\ref{algo:discretization} is as follows. We discretize the range of output into $2m$ parts due to the symmetric range $[-C,C]$, and then we obtain $2m+1$ output possibilities. After we get a perturbed data $y$, it will fall into one of $2m$ segments. Then, we categorize it to the left boundary or the right boundary of the segment, which resembles sampling a Bernoulli variable.

Next, we explain how we derive probabilities for the Bernoulli variable. Let the original input be $x$. A random variable $Y$ represents the intermediate output after the perturbation and a random variable $Z$ represents the output after the discretization. The range of $Y$ is $[-C, C]$. Because the range of output is symmetric with respect to $0$, we discretize both $[-C, 0]$ and $[0, C]$ into $m$ parts, where the value of $m$ depends on the user' requirement. Thus, we discretize $Y$ to $Z$ to take only the following $(2m+1)$ values:
\begin{align}
 & \left\{ i\times\frac{C}{m} : \text{integer }i \in \{-m,-m+1, \ldots, m\} \right\}.
\end{align}
When $Y$ is instantiated as $y \in [-C, C]$, we have the following  two cases:
\begin{itemize}
\item[\ding{172}] If $y$ is one of the above $(2m+1)$ values, we set $Z$ as $y$.
\item[\ding{173}] If $y$ is not one of the above $(2m+1)$ values, and then there exist some integer $k \in \{-m,-m+1, \ldots, m-1\} $ such that $\frac{kC}{m} < y < \frac{(k+1)C}{m} $. In fact, this gives $k < \frac{ym}{C} < k+1 $, so we can set $k:= \lfloor \frac{ym}{C} \rfloor$. Then conditioning on that $Y$ is instantiated as $y$, we set $Z$ as  $\frac{kC}{m}$ with probability $k+1- \frac{ym}{C}$ and as  $\frac{(k+1)C}{m}$ with probability $\frac{ym}{C}-k$, so that the expectation of $Z$ given $Y=y$ equals $y$ (as we will show in Eq.~(\ref{eq:expection-y-to-x}), this ensures that the expectation of $Z$ given the original input as $x$ equals $x$).
\end{itemize}

The following Lemma~\ref{lem:proof-discretize-probability} shows the probability distribution of assigning $y$ with a boundary value in the second case above when the intermediate output $y$ is not one of discrete $(2m+1)$ values.

\begin{lem}\label{lem:proof-discretize-probability}
After we obtain the intermediate output $y$ after perturbation, we discretize it to a random variable $Z$ equal to  $\frac{kC}{m}$ or $\frac{(k+1)C}{m}$ with the following probabilities:
    \begin{align}
         \bp{Z=z~|~Y=y}  = \begin{cases}
        k+1- \frac{ym}{C},  &\text{if} ~z=\frac{kC}{m}, \\[8pt]   \frac{ym}{C}-k, &\text{if} ~z=\frac{(k+1)C}{m}.
            \end{cases} \label{eq-Z-given-Y}
    \end{align}
\end{lem}

\begin{proof}
  The proof details are given in Appendix~\ref{appendix:proof-of-discreization-probability} of the submitted supplementary file.
\end{proof}

After discretization, the worst-case noise variance does not change or get worse proved by Lemma~\ref{lem:discretization-worst-case-var} as follows:

\begin{lem}\label{lem:discretization-worst-case-var}
   Let local differential privacy mechanism be Mechanism $\mathcal{M}_1$, and discretization algorithm be Mechanism $\mathcal{M}_2$. Let all of output possibilities of Mechanism $\mathcal{M}_1$ be $S_1$, and output possibilities of Mechanism $\mathcal{M}_2$ be $S_2$. $S_2 \subset S_1$. When given input $x$, $\mathcal{M}_1$ and $\mathcal{M}_2$ are unbiased. The worst-case noise variance of Mechanism $\mathcal{M}_2$ is greater than or equal to  the worst-case noise variance of Mechanism $\mathcal{M}_1$.
\end{lem}

\begin{proof}
 The proof details are given in Appendix~\ref{app:proof-of-lem-8} of the submitted supplementary file.
\end{proof}

\subsection{\texttt{HM-TP} Mechanism}

Fig.~\ref{fig:worst_var_comp} shows that \texttt{Three-Outputs} outperforms \texttt{PM-SUB} when the privacy budget $\epsilon$ is small, whereas \texttt{PM-SUB}  achieves a smaller variance if the privacy budget $\epsilon$ is large. To fully take advantage of two mechanisms, we combine \texttt{Three-Outputs} and \texttt{PM-SUB} to create a new hybrid mechanism named as \texttt{HM-TP}. Fig.~\ref{fig:worst_var_comp} illustrates that \texttt{HM-TP} obtains a lower worst-case noise variance than other solutions.

Hence, \texttt{HM-TP} invokes \texttt{PM-SUB} with probability $\beta$. Otherwise, it invokes \texttt{Three-Outputs}. We define the noisy variance of \texttt{HM-TP}  as $\textup{Var}_{\mathcal{H}}[Y|x]$ given inputs $x$ as follows:
\begin{align}
    \nonumber \textup{Var}_{\mathcal{H}}[Y|x] &= \beta \cdot \textup{Var}_{\mathcal{P}}[Y|x] + (1 - \beta) \cdot \textup{Var}_{\mathcal{T}}[Y|x],
\end{align} where $\textup{Var}_{\mathcal{P}}[Y|x]$ and $\textup{Var}_{\mathcal{T}}[Y|x]$ denote noisy outputs' variances incurred by \texttt{PM-SUB} and \texttt{Three-Outputs}, respectively. The following lemma presents the value of $\beta$:

\begin{lem}\label{lem:beta}
  The~$\max_{x \in [-1,1]} \textup{Var}_{\mathcal{H}}[Y|x]$~is minimized when $\beta$ is Eq.~(\ref{eq:beta-value}). Due to the complicated equation of $\beta$, we put it in the appendix.
\end{lem}

\begin{proof}
 The proof details are given in Appendix~\ref{proof_lemma_varH} of the submitted supplementary file.
\end{proof}

Since we have obtained the probability $\beta$, we can calculate the exact expression for the worst-case noise variance in Lemma~\ref{lem:hm-var} as follows:

\begin{lem}\label{lem:hm-var} If $\beta$ satisfies Lemma \ref{lem:beta}, we obtain the worst-case noise variance of \texttt{HM-TP} as
\begin{align}
    &\max_{x \in [-1,1]} \textup{Var}_{\mathcal{H}}[Y|x]= \nonumber\\
    & \begin{cases}
        \textup{Var}_{\mathcal{H}}[Y|x^*],\text{if}~ 0 < \beta < \frac{2(e^\epsilon-a)^2(e^\epsilon-1)-ae^\epsilon(e^\epsilon+1)^2}{2(e^\epsilon-a)^2(e^\epsilon+t)-ae^\epsilon(e^\epsilon+1)^2},\\
        \max \{\textup{Var}_{\mathcal{H}}[Y|0], \textup{Var}_{\mathcal{H}}[Y|1] \},~~\text{otherwise,} \nonumber
    \end{cases}
\end{align}
where $x^* := \frac{(\beta-1)ae^\epsilon(e^\epsilon+1)^2}{2(e^\epsilon-a)^2(\beta (e^\epsilon+ t)-e^\epsilon+1)}$ and $a = P_{0 \leftarrow 0}$ which is defined in Eq.~(\ref{a_opt_1}).
\end{lem}

\begin{proof}
The proof details are given in Appendix~\ref{sec:proof-hm-var} of the submitted supplementary file.
\end{proof}

\section{Mechanisms for Estimation of Multiple Numeric Attributes }\label{sec:multiple}

Now, we consider a case in which the user's data record contains $d>1$ attributes. There are three existing solutions to collect multiple attributes: (i) The straightforward approach which collects each attribute with privacy budget $\epsilon/d$. Based on the composition theorem~\cite{dwork2014algorithmic}, it satisfies $\epsilon$-LDP after collecting of all attributes. But the added noise can be excessive if $d$ is large~\cite{wang2019collecting}. (ii) Duchi~\textit{et~al.}'s~\cite{duchi2018minimax} solution, which is rather complicated, handles numeric attributes only.  (iii) Wang~\textit{et~al.}'s~\cite{wang2019collecting} solution is the advanced approach that deals with a data tuple containing both numeric and categorical attributes. Their algorithm requires to calculate an optimal $k < d$ based on the single dimensional attribute's $\epsilon$-LDP mechanism, and a user submits selected $k$ dimensional attributes instead of $d$ dimensions.

\begin{algorithm}[h]
 \caption{Mechanism for Multiple-Dimensional Numeric Attributes.}
 \label{PM-d-dimension-Algo}
  \KwIn{ tuple $x \in [-1,1]^d$ and privacy parameter $\epsilon$.} 
  \KwOut{tuple $Y \in [-A,  A]^d$.}
  {Let $Y = \langle 0, 0, \dots, 0 \rangle$;} \\
  {Let $k = \max \{1, \min \{d, \Bigl \lfloor \frac{\epsilon}{2.5} \Bigr \rfloor \} \}$;} \\
  {Sample $k$ values uniformly without replacement from $\{1,2, \dots ,d\}$;}\\
  \For{each sampled value $j$ }{
        Feed $x[t_j]$ and $\frac{\epsilon}{k}$ as input to \texttt{PM-SUB}, \texttt{Three-Outputs} or \texttt{HM-TP}, and obtain a noisy value $y_{j}$;\\
        $Y[t_j] = \frac{d}{k}y_{j}$;
   }
   \textbf{return}  $Y$;
\end{algorithm}

Thus, we follow Wang~\textit{et~al.}'s~\cite{wang2019collecting} idea to extend  Section~\ref{sec:single} to the case of multidimensional attributes. Algorithm~\ref{PM-d-dimension-Algo} shows the pseudo-code of our extension for our \texttt{PM-SUB}, \texttt{Three-Outputs}, and \texttt{HM-TP}. Given a tuple $x \in [-1, 1]^d$, the algorithm returns a perturbed tuple $Y$ that has non-zero value on $k$ attributes, where
\begin{align}
    k = \max \{1, \min \{ d, \Bigl \lfloor \frac{\epsilon}{2.5} \Bigr \rfloor \} \},\label{eq:optimal-k}
\end{align}
and Appendix~\ref{proof:k} of the submitted supplementary file proves our selected $k$ is optimal after extending \texttt{PM-SUB}, \texttt{Three-Outputs}, and \texttt{HM-TP} to support $d$ dimensional attributes.

Overall, our algorithm for collecting multiple attributes outperforms existing solutions, which is confirmed by our experiments in the Section~\ref{sec:experiment}. But \texttt{Three-Outputs}  uses only one more bit compared with Duchi~\textit{et al.}'s~\cite{duchi2018minimax} solution to encode outputs. Moreover, our \texttt{Three-Outputs} obtains a higher accuracy in the high privacy regime (where the privacy budget is small) and saves many bits for encoding since \texttt{PM} and \texttt{HM}'s continuous output range requires infinite bits to encode, whereas \texttt{PM-SUB} and \texttt{HM-TP}'s advantages are obvious at a large privacy budget. Furthermore, because vehicles can not encode continuous range, we discretize the continuous range of outputs to discrete outputs. Our experiments in Section~\ref{:discreization-results} confirm that we can achieve similar results to algorithms before discretizing by carefully designing the number of discrete parts. Hence, our proposed algorithms are obviously more suitable for vehicles than existing solutions.

Intuitively, Algorithm~\ref{PM-d-dimension-Algo} requires every user to submit $k$ attributes instead of $d$ attributes, such that the privacy budget for each attribute increases from $\epsilon/d$ to $\epsilon/k$, which helps to minimize the noisy variance. In addition, by setting $k$ as Eq.~(\ref{eq:optimal-k}), algorithm~\ref{PM-d-dimension-Algo} achieves an asymptotically optimal performance while preserving privacy, which we will prove using Lemma~\ref{lem-d-1} and~\ref{lem-d-2}. Lemma~\ref{lem-d-1} and~\ref{lem-d-2} are proved in the same way as that of Lemma $4$ and $5$ in~\cite{wang2019collecting}.

\begin{lem}\label{lem-d-1}
  Algorithm~\ref{PM-d-dimension-Algo} satisfies $\epsilon$-local differential privacy. In addition, given an input tuple $x$, it outputs a noisy tuple $Y$, such that for any $j \in [1, d]$, and each $t_j$ of those $k$ attributes is selected uniformly at random (without replacement) from all $d$ attributes of $x$, and then $\mathbb{E}[Y[t_j]] = x[t_j]$.
\end{lem}

\begin{proof} Algorithm~\ref{PM-d-dimension-Algo} composes $k$ numbers of $\epsilon$-LDP perturbation algorithms; thus, based on composition theorem of differential mechanism~\cite{mcsherry2007mechanism}, Algorithm~\ref{PM-d-dimension-Algo} satisfies $\epsilon$-LDP. As we can see from Algorithm~\ref{PM-d-dimension-Algo}, each perturbed output $Y$ equals to $\frac{d}{k}y_{j}$ with probability $\frac{k}{d}$ or equals to 0 with probability $1-\frac{k}{d}$. Thus, $\mathbb{E}[Y[t_j]] = \frac{k}{d} \cdot \mathbb{E}[\frac{d}{k} \cdot y_{j}] = \mathbb{E}[y_{j}] = x[t_j] $ holds.
\end{proof}

\begin{lem}\label{lem-d-2}
   For any $j \in [1,d]$, let $Z[t_j] = \frac{1}{n} \sum_{i=1}^{n} Y[t_j]$ and $X[t_j] = \frac{1}{n} \sum_{i=1}^{n} x [t_j]$. With at least $1-\beta$ probability, $$\max_{j \in [1,d]} |Z[t_j] - X[t_j]| = O \bigg ( \frac{\sqrt{d \ln(d/\beta))}}{\epsilon \sqrt{n}} \bigg ).$$
\end{lem}

\begin{proof}
The proof details are given in Appendix~\ref{proof:lem-d-2} of the submitted supplementary file.
\end{proof}

\begin{figure*}[h]
     \centering
     \includegraphics[width=\textwidth]{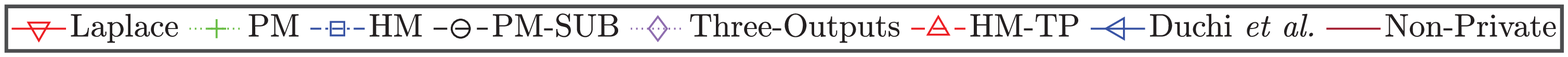}
\end{figure*}

\begin{figure*}[h]
     \centering
     \begin{subfigure}[b]{0.24\textwidth}
         \centering
         \includegraphics[width=\textwidth]{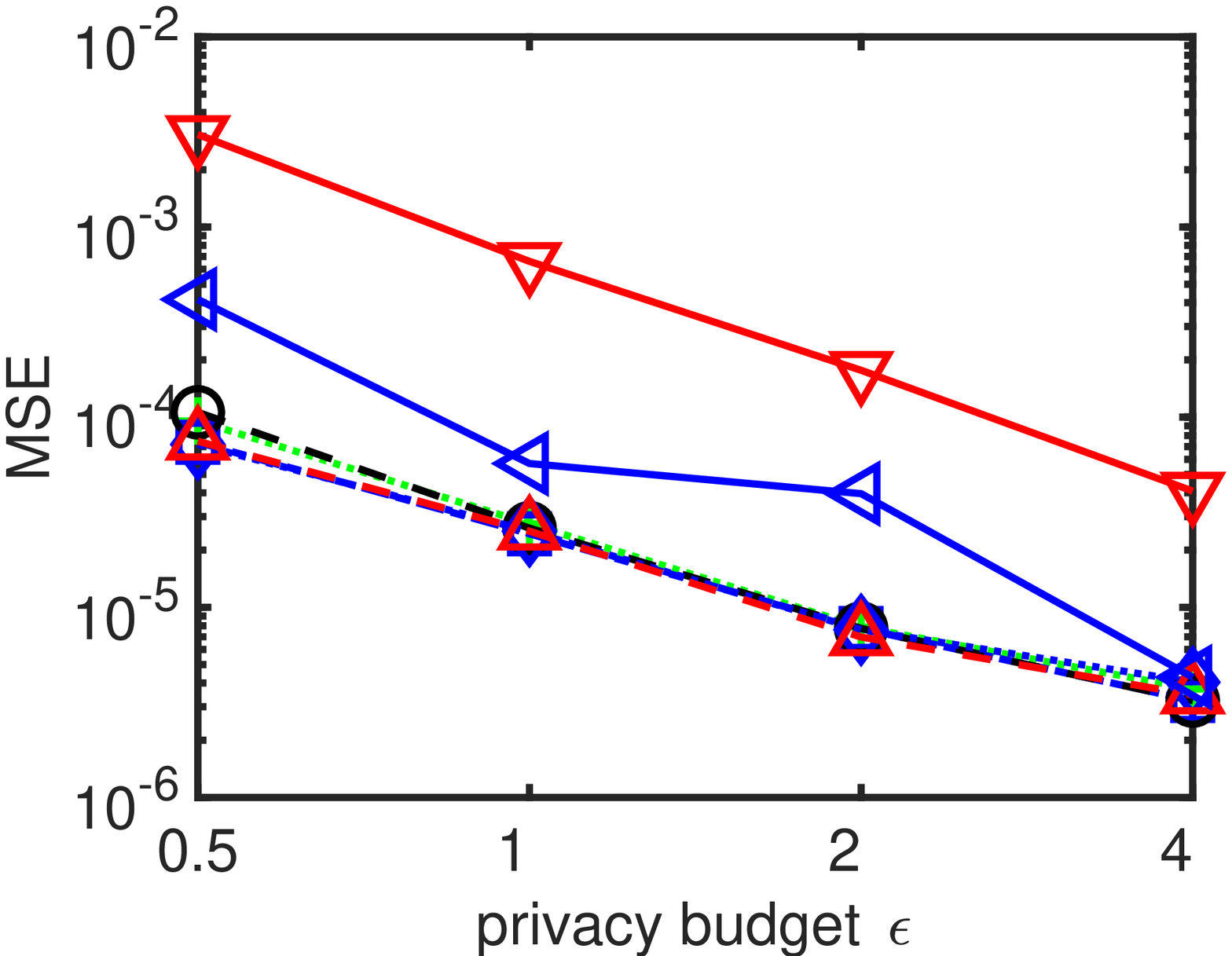}
         \caption{MX-Numeric}
         \label{fig:MSE_MX}
     \end{subfigure}
     \begin{subfigure}[b]{0.24\textwidth}
         \centering
         \includegraphics[width=\textwidth]{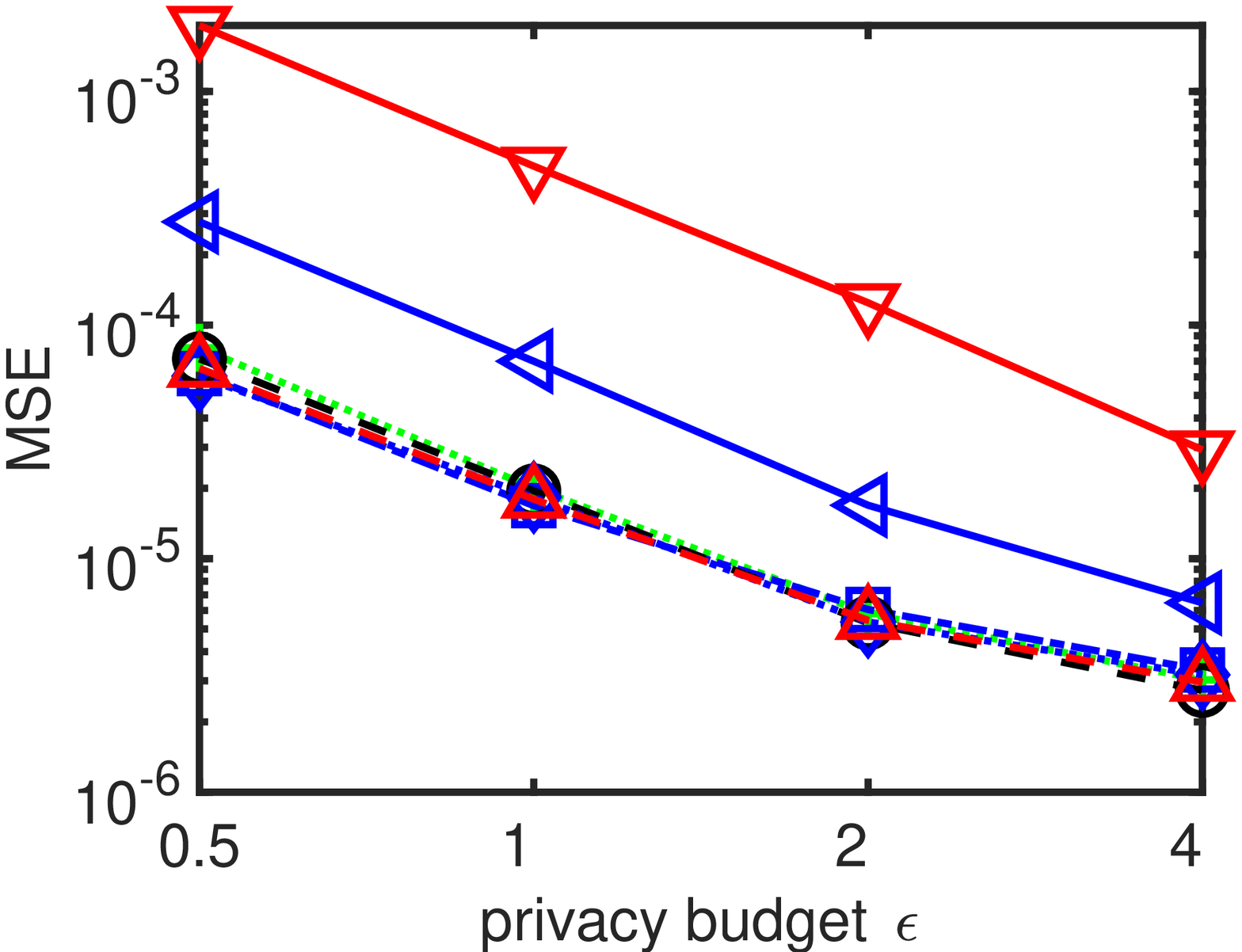}
         \caption{BR-Numeric}
         \label{fig:MSE_BR}
     \end{subfigure}
    \begin{subfigure}[b]{0.24\textwidth}
         \centering
         \includegraphics[width=\textwidth]{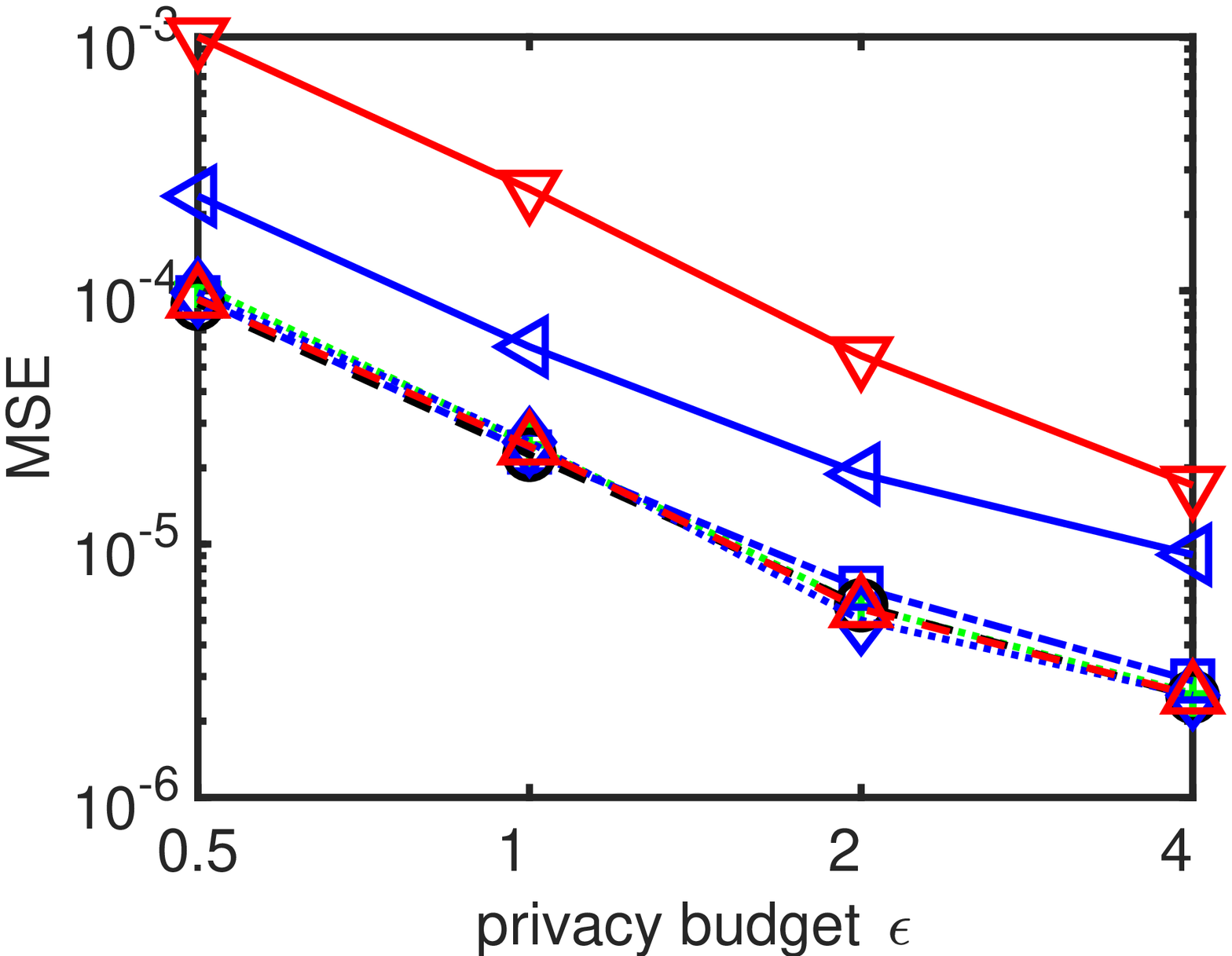}
         \caption{WISDM-Numeric}
         \label{fig:MSE_IOT}
     \end{subfigure}
    \begin{subfigure}[b]{0.24\textwidth}
         \centering
         \includegraphics[width=\textwidth]{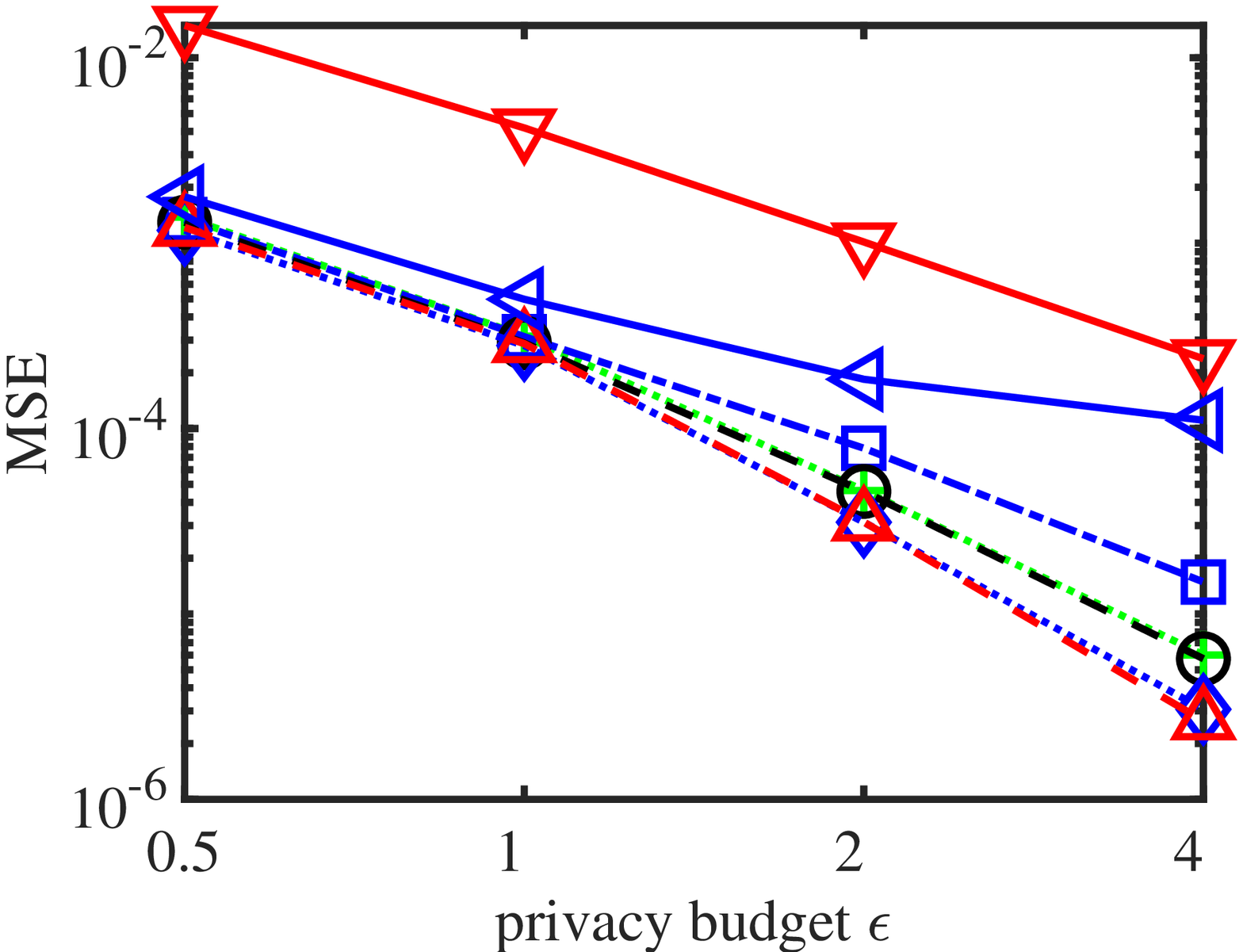}
         \caption{Vehicle-Numeric}
         \label{fig:MSE_IOT}
     \end{subfigure}
        \caption{Result accuracy for mean estimation (on numeric attributes).}
        \label{fig:MSE}
\end{figure*}

\begin{figure*}[h]
     \centering
     \begin{subfigure}[b]{0.24\textwidth}
         \centering
         \includegraphics[width=\textwidth]{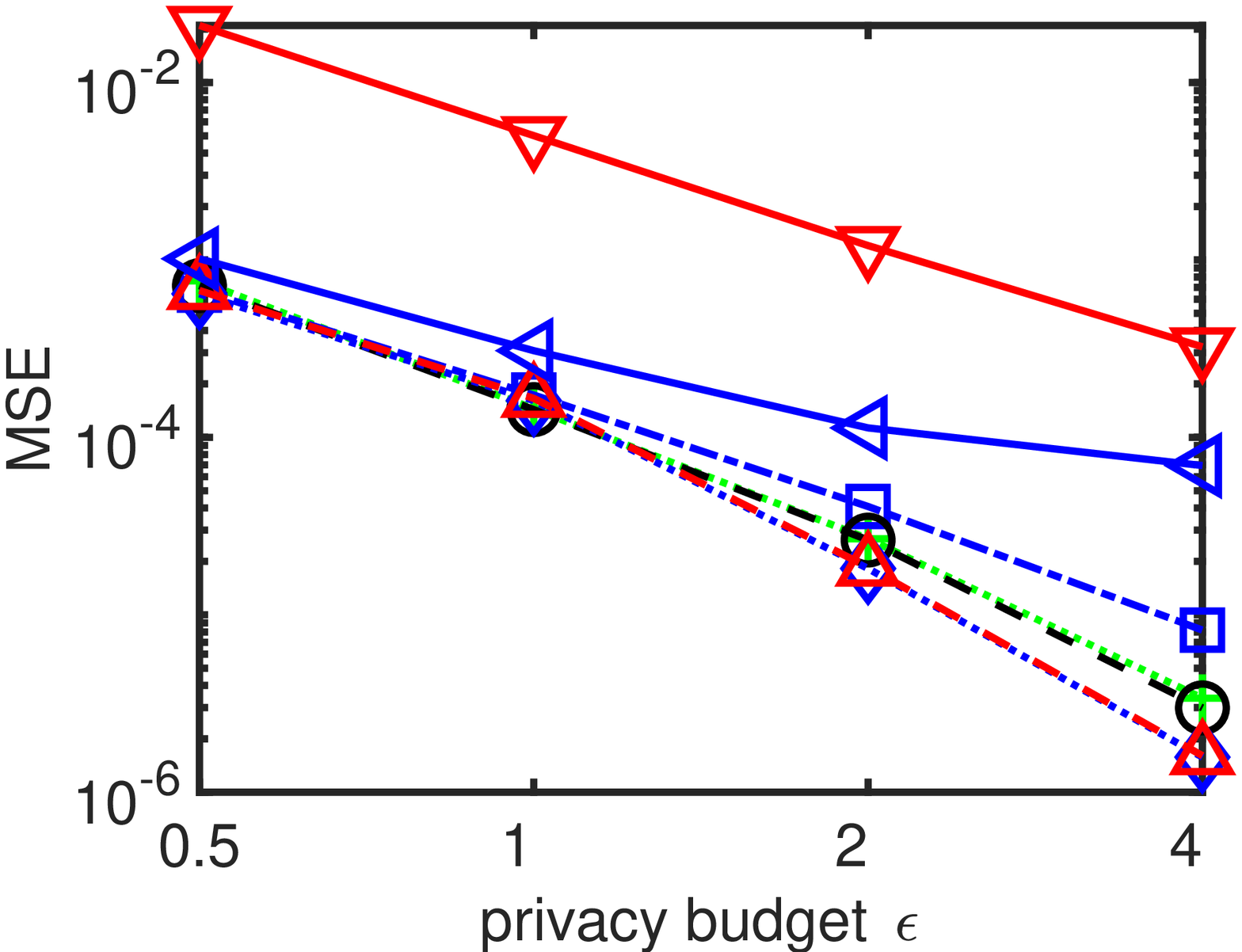}
         \caption{$\mu = 0$}
         \label{fig:synthetic_0}
     \end{subfigure}
     \begin{subfigure}[b]{0.24\textwidth}
         \centering
         \includegraphics[width=\textwidth]{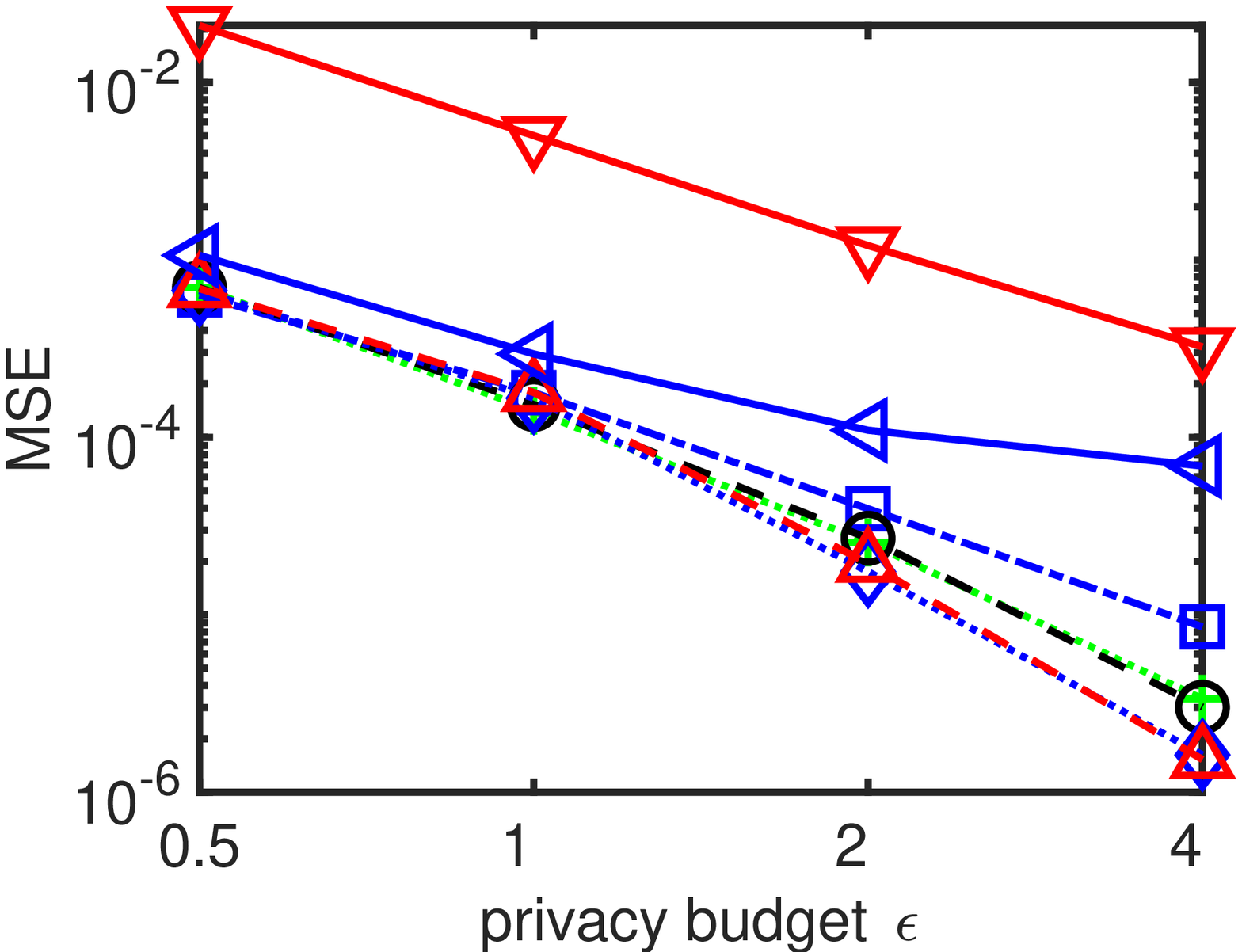}
         \caption{$\mu = 1/3$}
         \label{fig:synthetic_1}
     \end{subfigure}
     \begin{subfigure}[b]{0.24\textwidth}
         \centering
         \includegraphics[width=\textwidth]{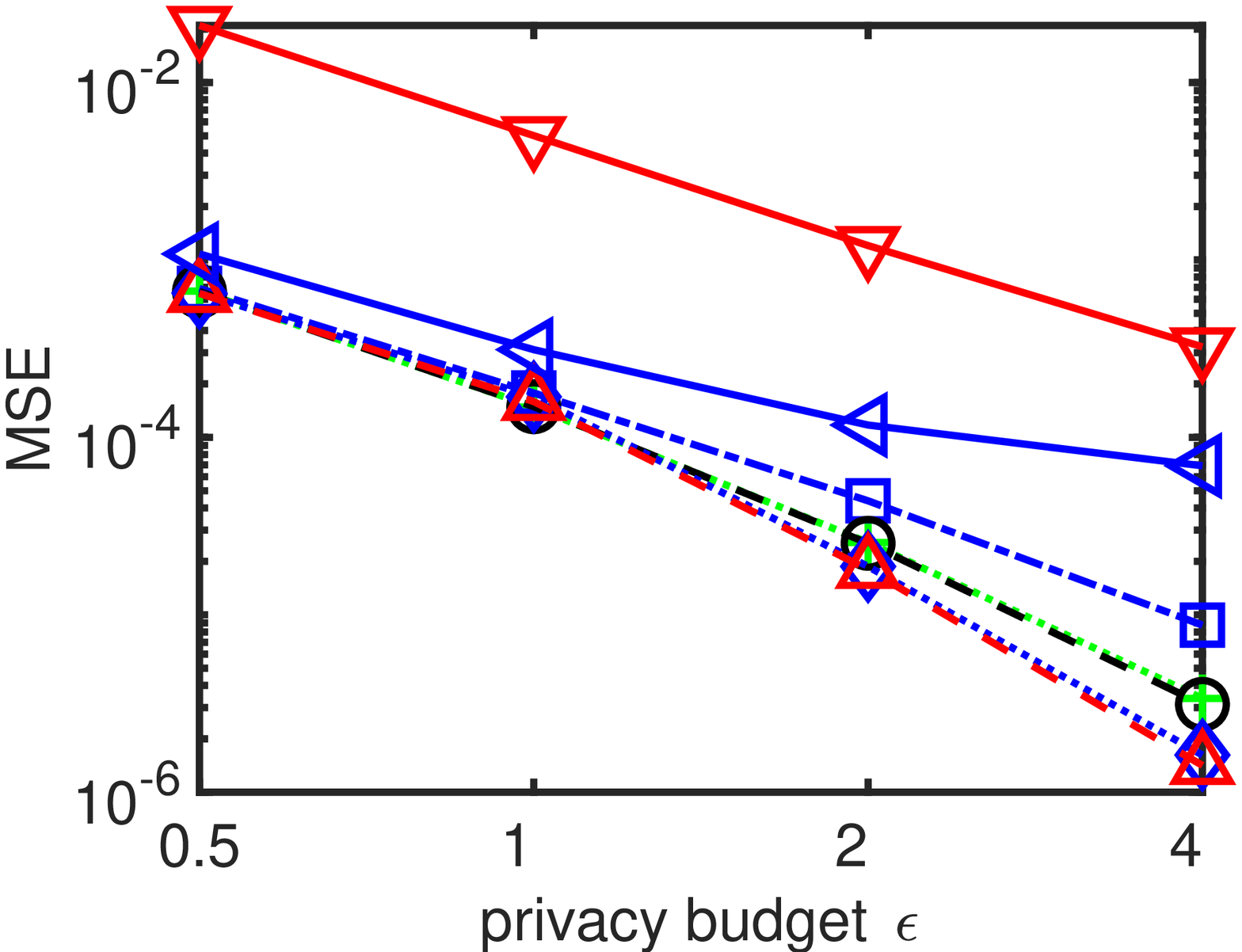}
         \caption{$\mu = 2/3$}
         \label{fig:synthetic_2}
     \end{subfigure}
     \begin{subfigure}[b]{0.24\textwidth}
         \centering
         \includegraphics[width=\textwidth]{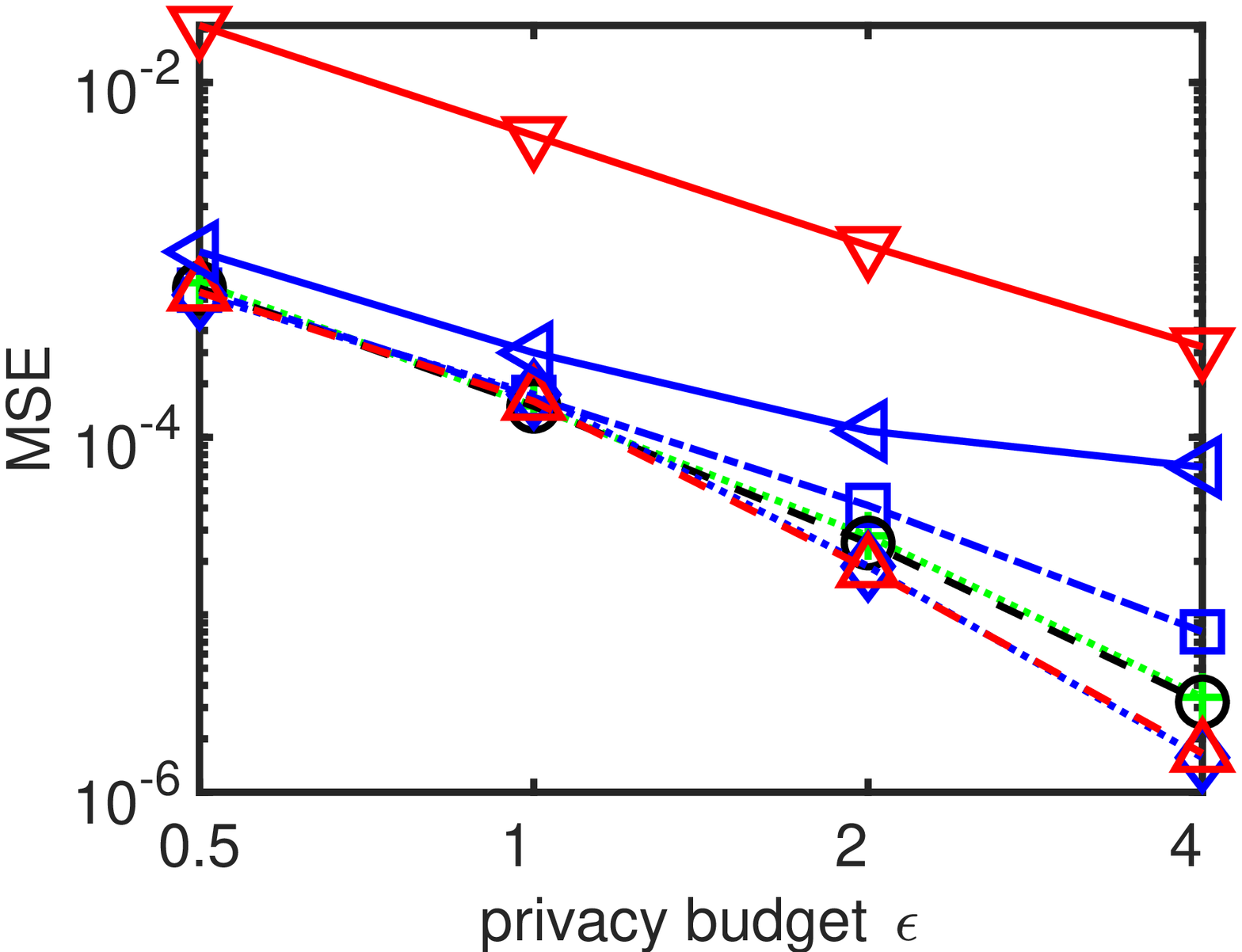}
         \caption{$\mu = 1$}
         \label{fig:synthetic_2}
     \end{subfigure}
        \caption{Result accuracy on synthetic datasets with 16 dimensions, each of which follows a Gaussian distribution $N(\mu, 1/16)$ truncated to $[-1, 1]$.}
        \label{fig:MSE_synthetic}
\end{figure*}

\begin{figure*}[h]
     \centering
     \begin{subfigure}[b]{0.24\textwidth}
         \includegraphics[width=\textwidth]{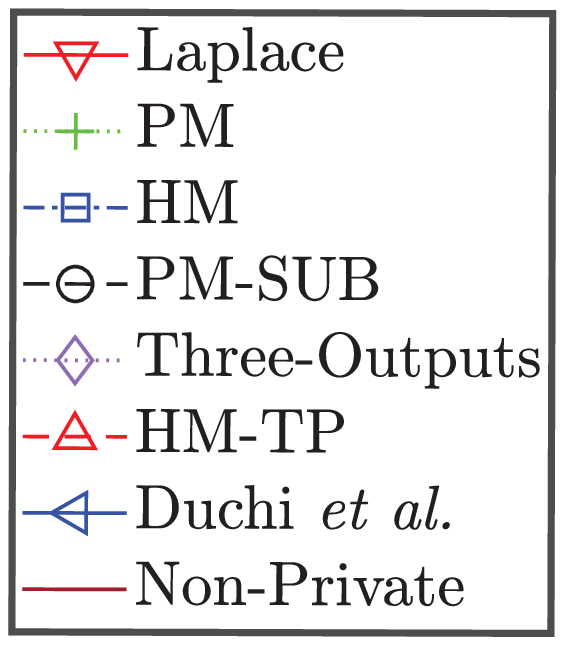}
     \end{subfigure}
     \begin{subfigure}[b]{0.24\textwidth}
         \centering
         \includegraphics[width=\textwidth]{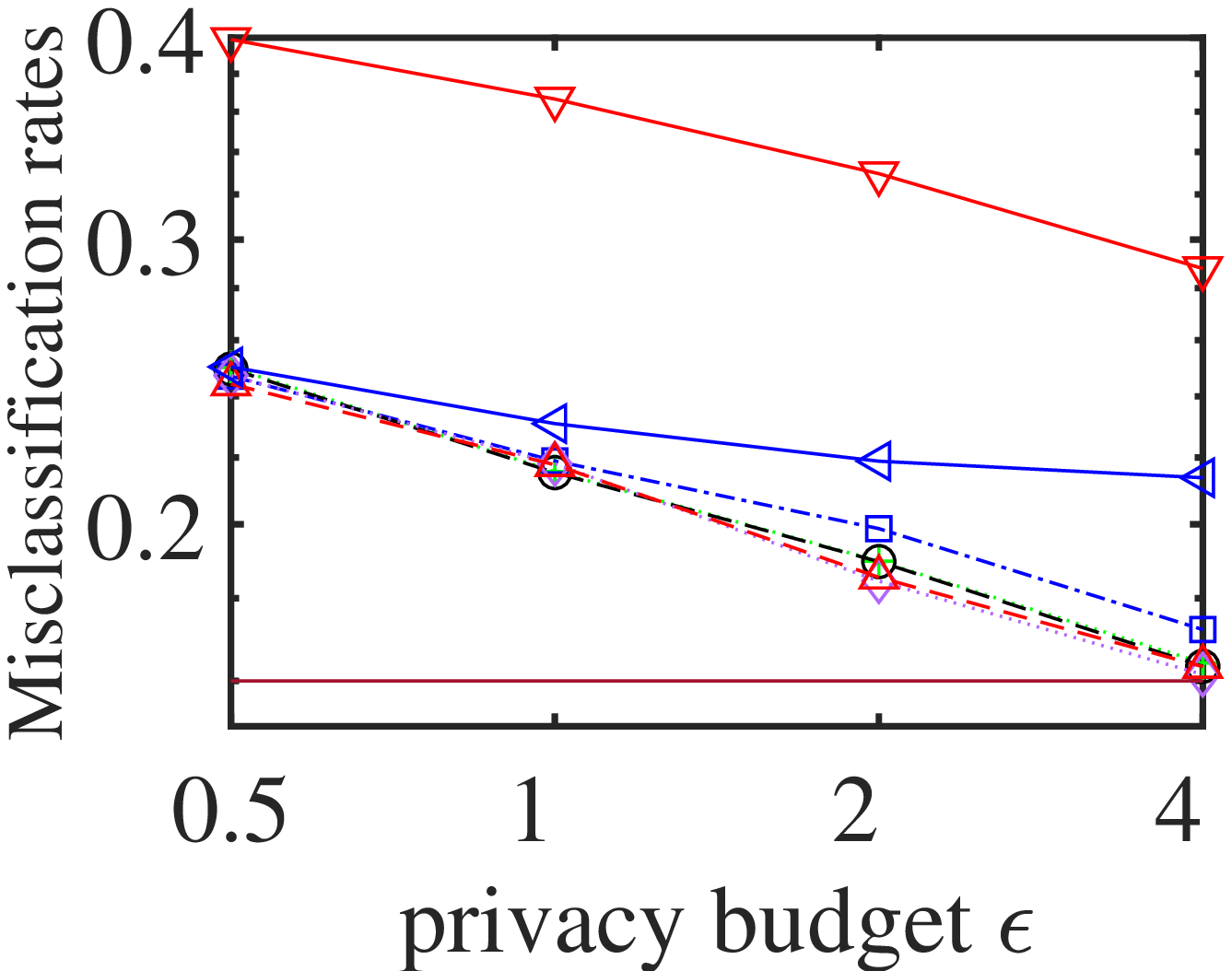}
         \caption{MX}
         \label{fig:log_MX}
     \end{subfigure}
     \begin{subfigure}[b]{0.24\textwidth}
         \centering
         \includegraphics[width=\textwidth]{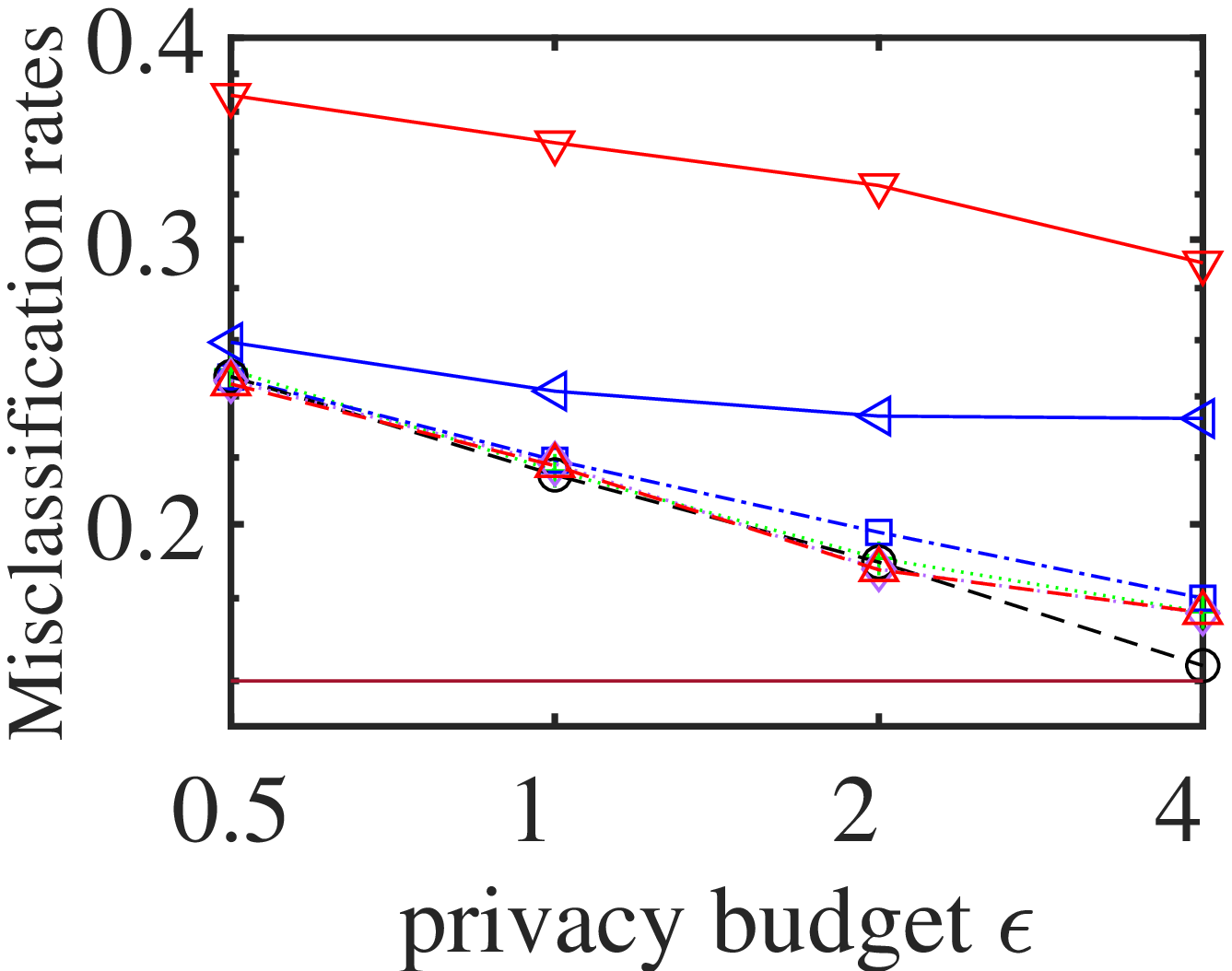}
         \caption{BR}
         \label{fig:log_BR}
     \end{subfigure}
     \begin{subfigure}[b]{0.24\textwidth}
         \centering
         \includegraphics[width=\textwidth]{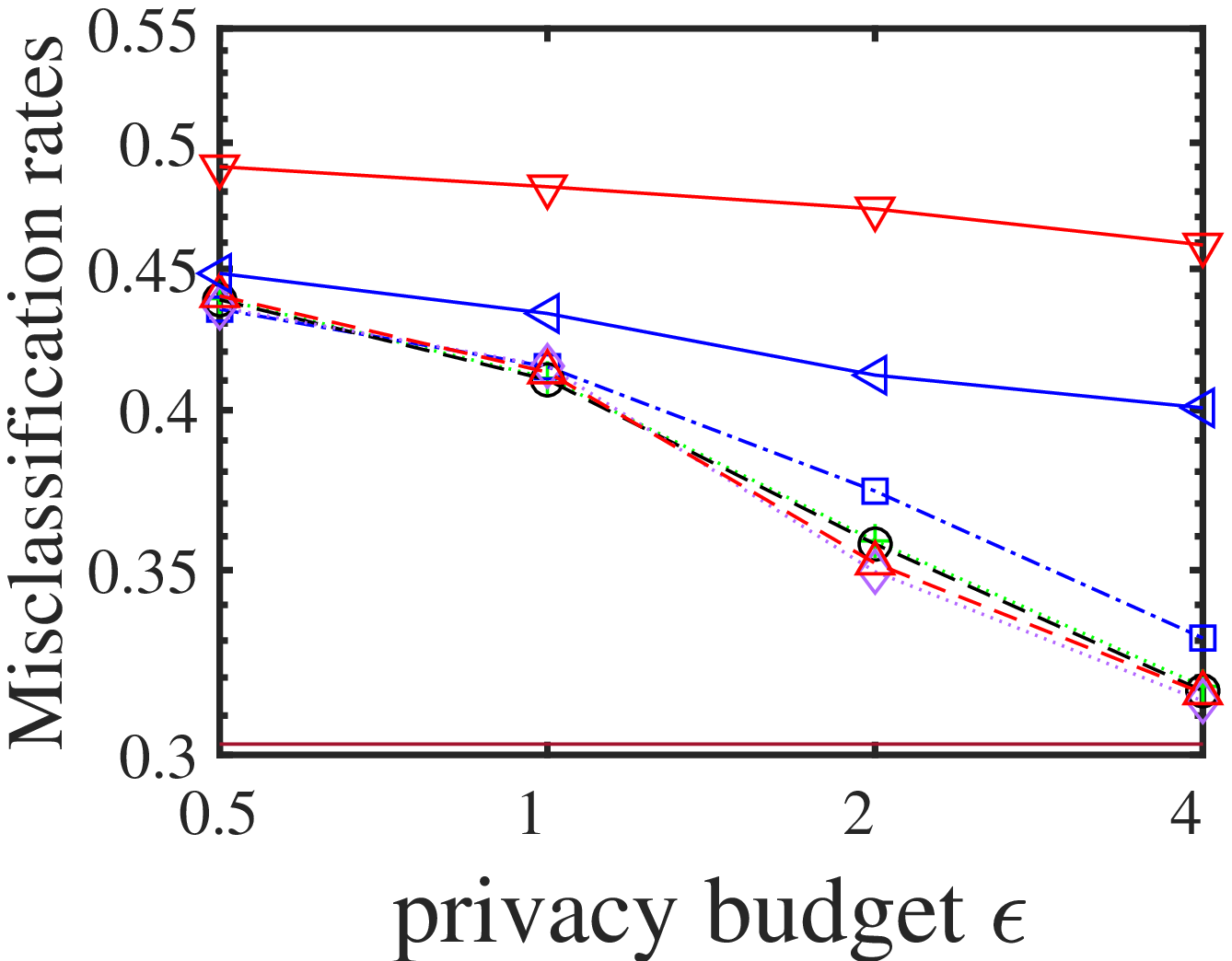}
         \caption{WISDN}
         \label{fig:log_iot}
     \end{subfigure}
         \caption{Logistic Regression.}
         \label{fig:log}
\end{figure*}

\begin{figure*}[h]
     \centering
     \begin{subfigure}[b]{0.24\textwidth}
         \includegraphics[width=\textwidth]{images/legend_v.eps}
     \end{subfigure}
     \begin{subfigure}[b]{0.24\textwidth}
         \includegraphics[width=\textwidth]{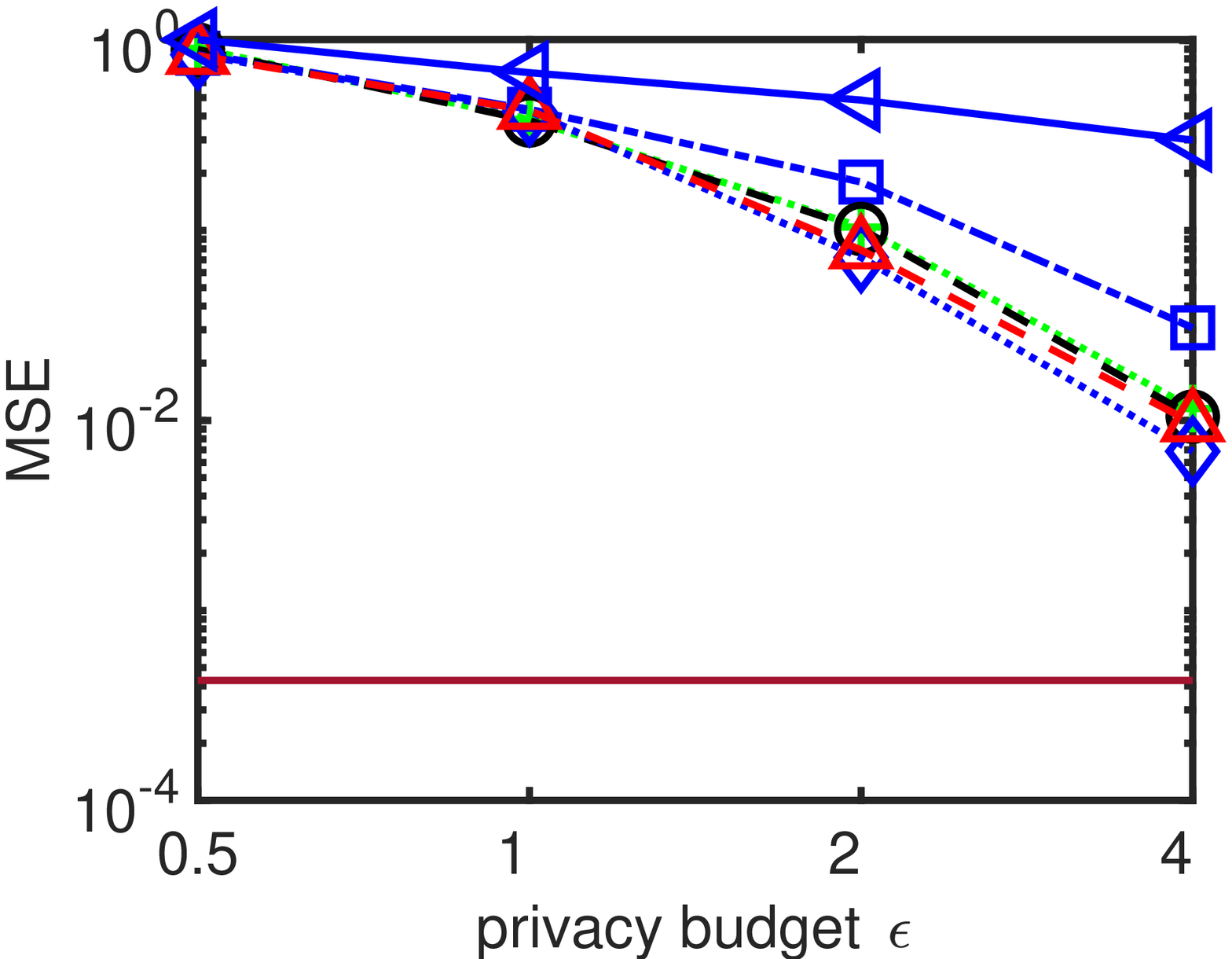}
         \caption{MX}
         \label{fig:linear_MX}
     \end{subfigure}
     \begin{subfigure}[b]{0.24\textwidth}
         \includegraphics[width=\textwidth]{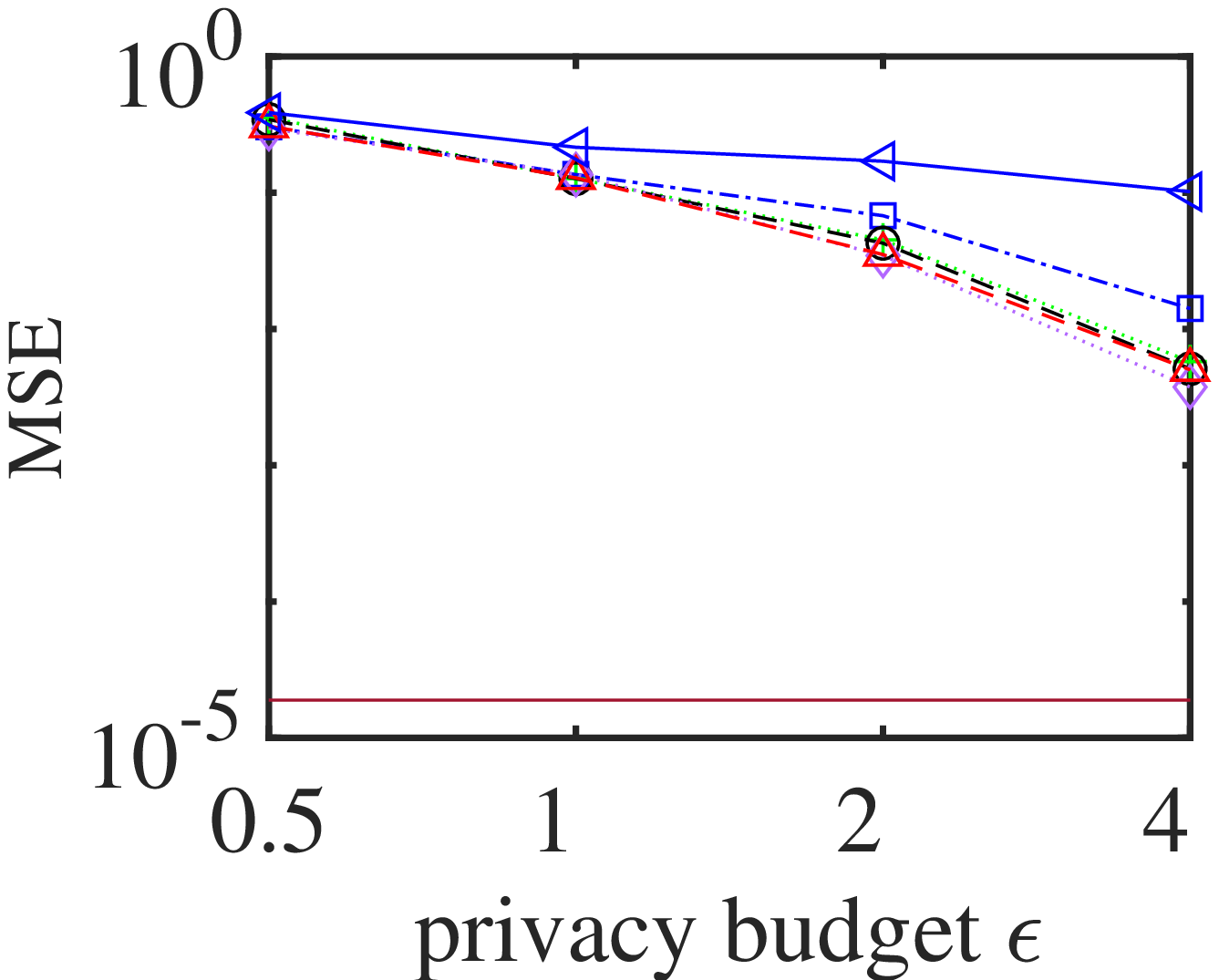}
         \caption{BR}
         \label{fig:linear_BR}
     \end{subfigure}
     \begin{subfigure}[b]{0.24\textwidth}
         \includegraphics[width=\textwidth]{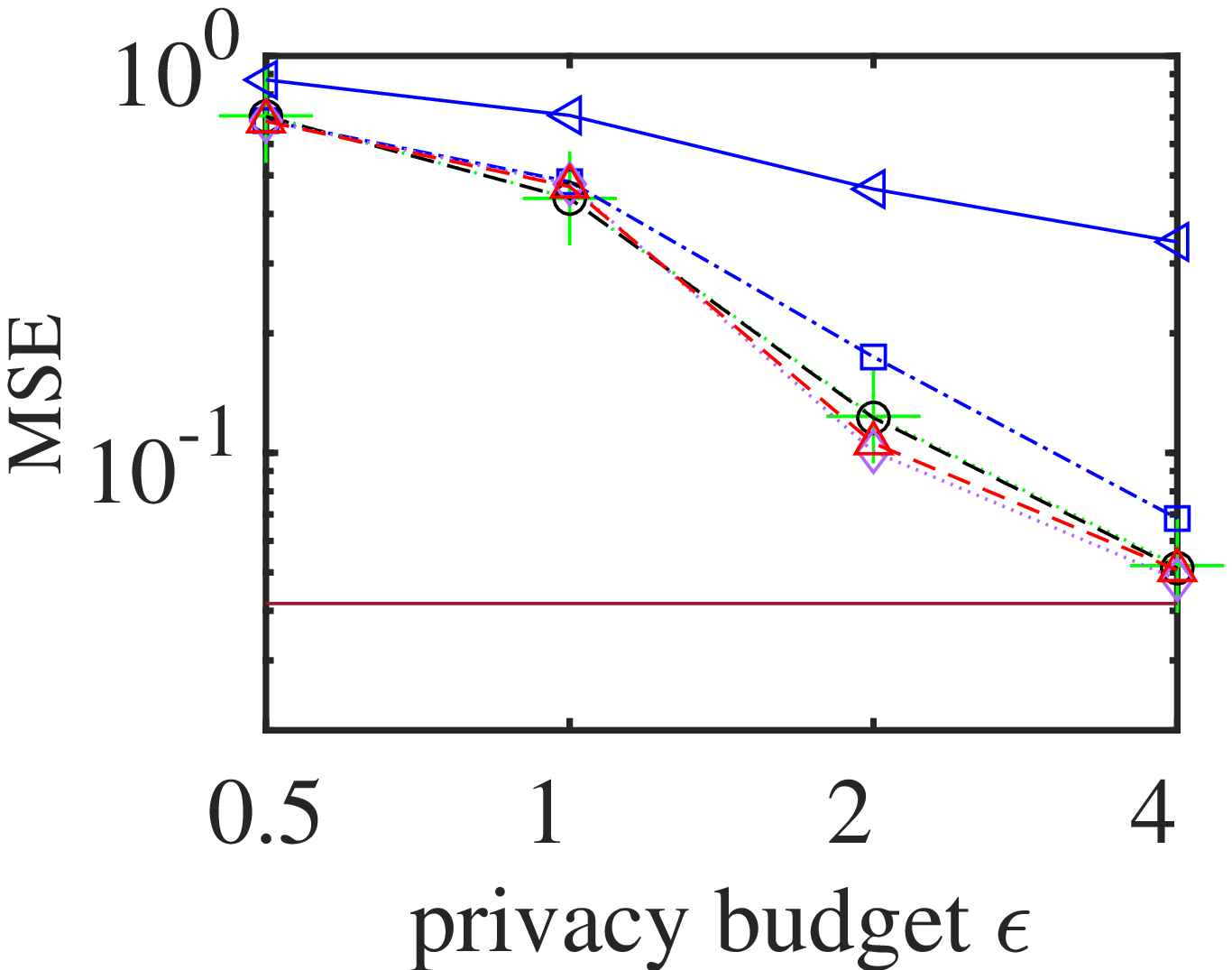}
         \caption{WISDN}
         \label{fig:linear_iot}
     \end{subfigure}
        \caption{Linear Regression.}
        \label{fig:linear}
\end{figure*}

\begin{figure*}[!h]
     \centering
     \begin{subfigure}[b]{0.24\textwidth}
         \includegraphics[width=\textwidth]{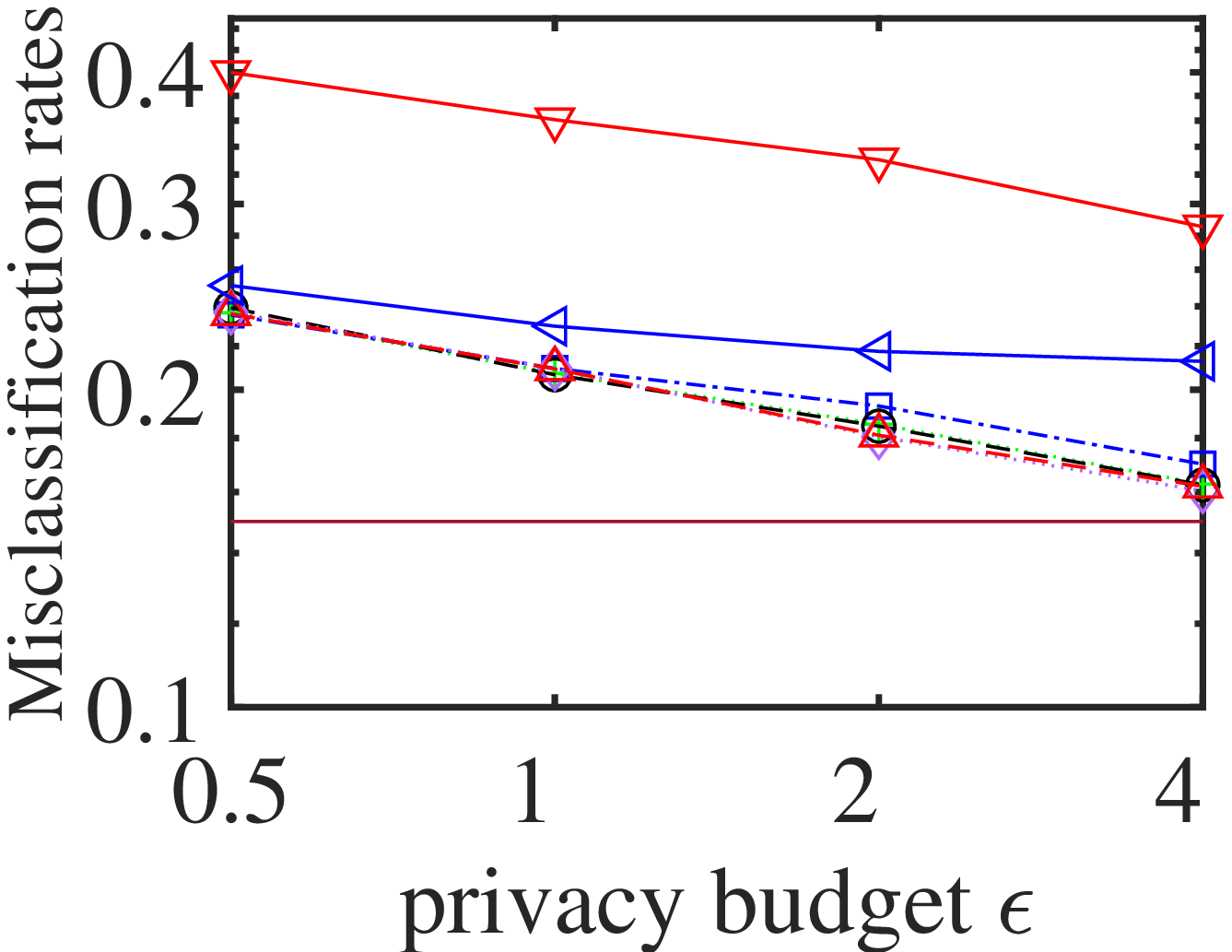}
         \caption{MX}
         \label{fig:svm_MX}
     \end{subfigure}
     \begin{subfigure}[b]{0.24\textwidth}
         \includegraphics[width=\textwidth]{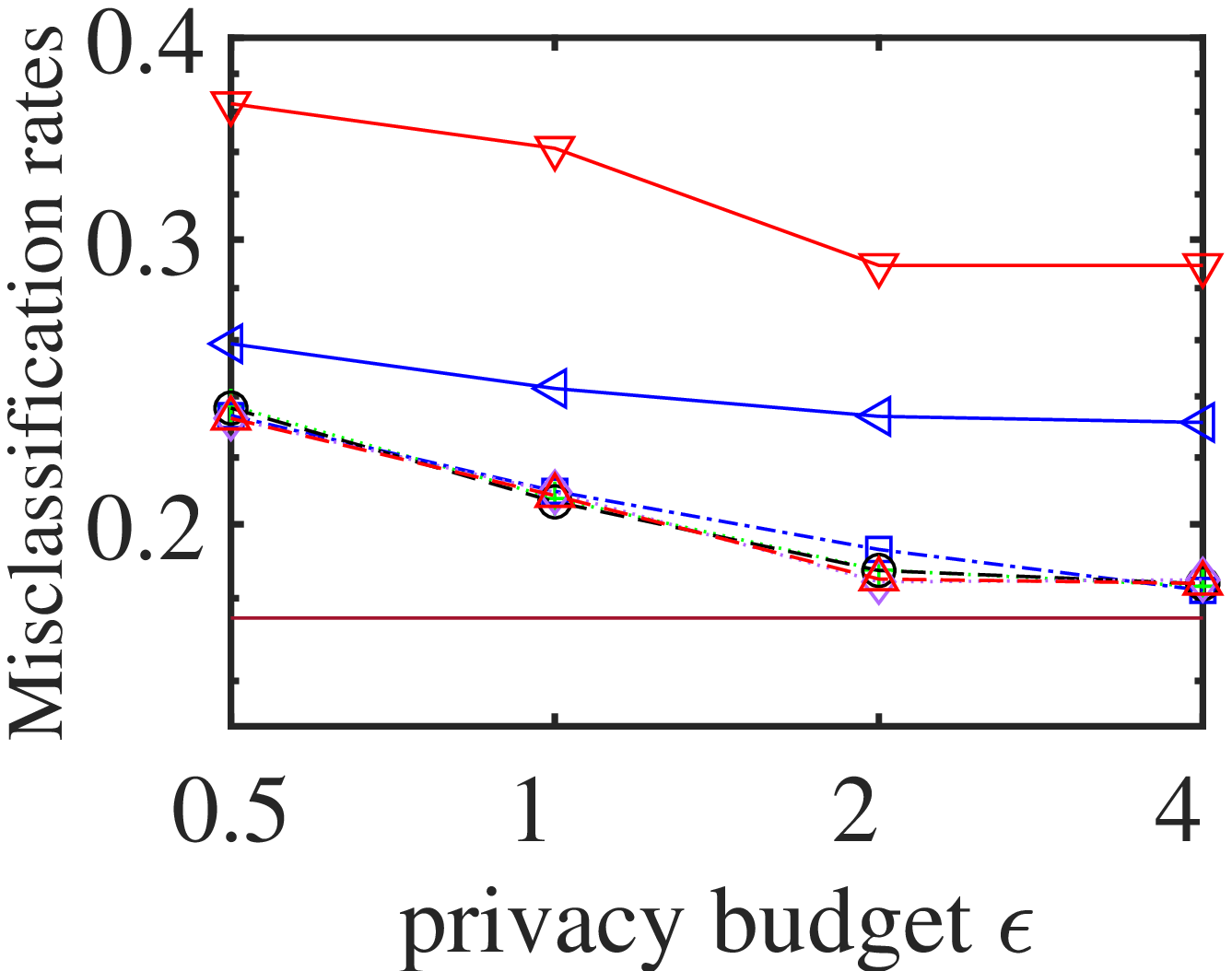}
         \caption{BR}
         \label{fig:svm_BR}
     \end{subfigure}
     \begin{subfigure}[b]{0.24\textwidth}
         \includegraphics[width=\textwidth]{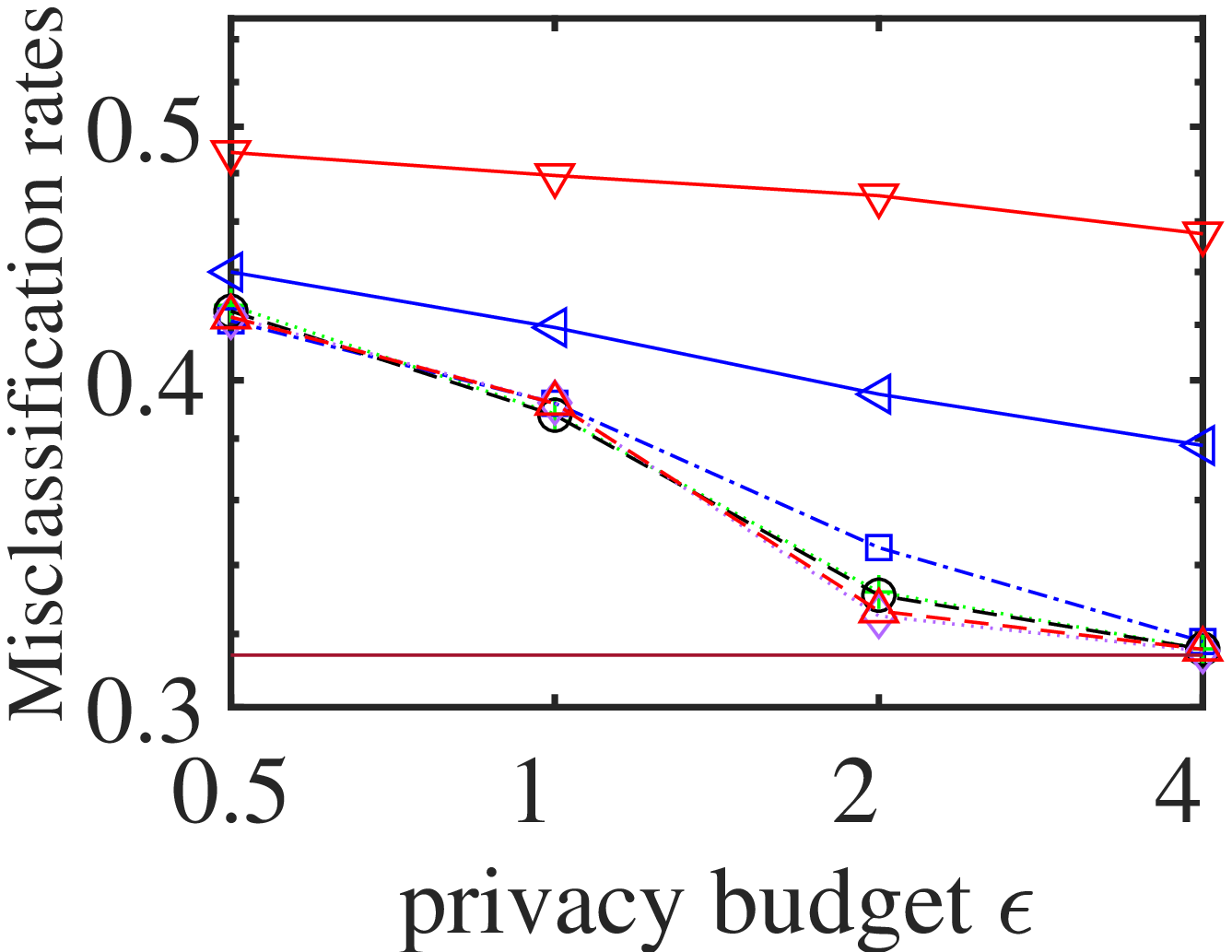}
         \caption{WISDN}
         \label{fig:svm_iot}
     \end{subfigure}
     \begin{subfigure}[b]{0.24\textwidth}
         \includegraphics[width=\textwidth]{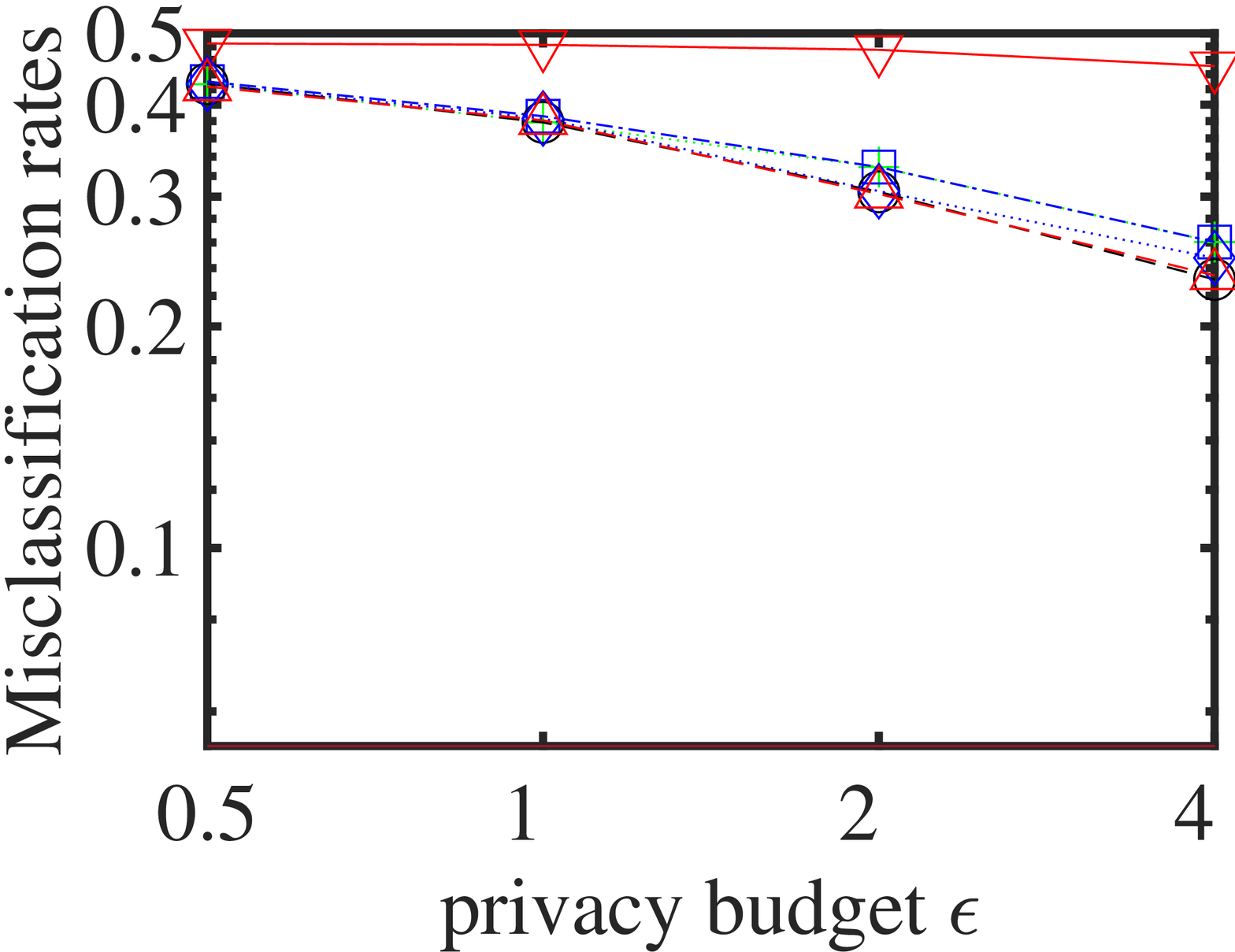}
         \caption{Vehicle}
         \label{fig:svm_vehicle}
     \end{subfigure}
          \caption{Support Vector Machines.}
          \label{fig:svm}
\end{figure*}

\begin{figure*}[h]
     \centering
     \begin{subfigure}[b]{0.24\textwidth}
         \includegraphics[width=\textwidth]{images/legend_v.eps}
     \end{subfigure}
     \begin{subfigure}[b]{0.24\textwidth}
         \centering
         \includegraphics[width=\textwidth]{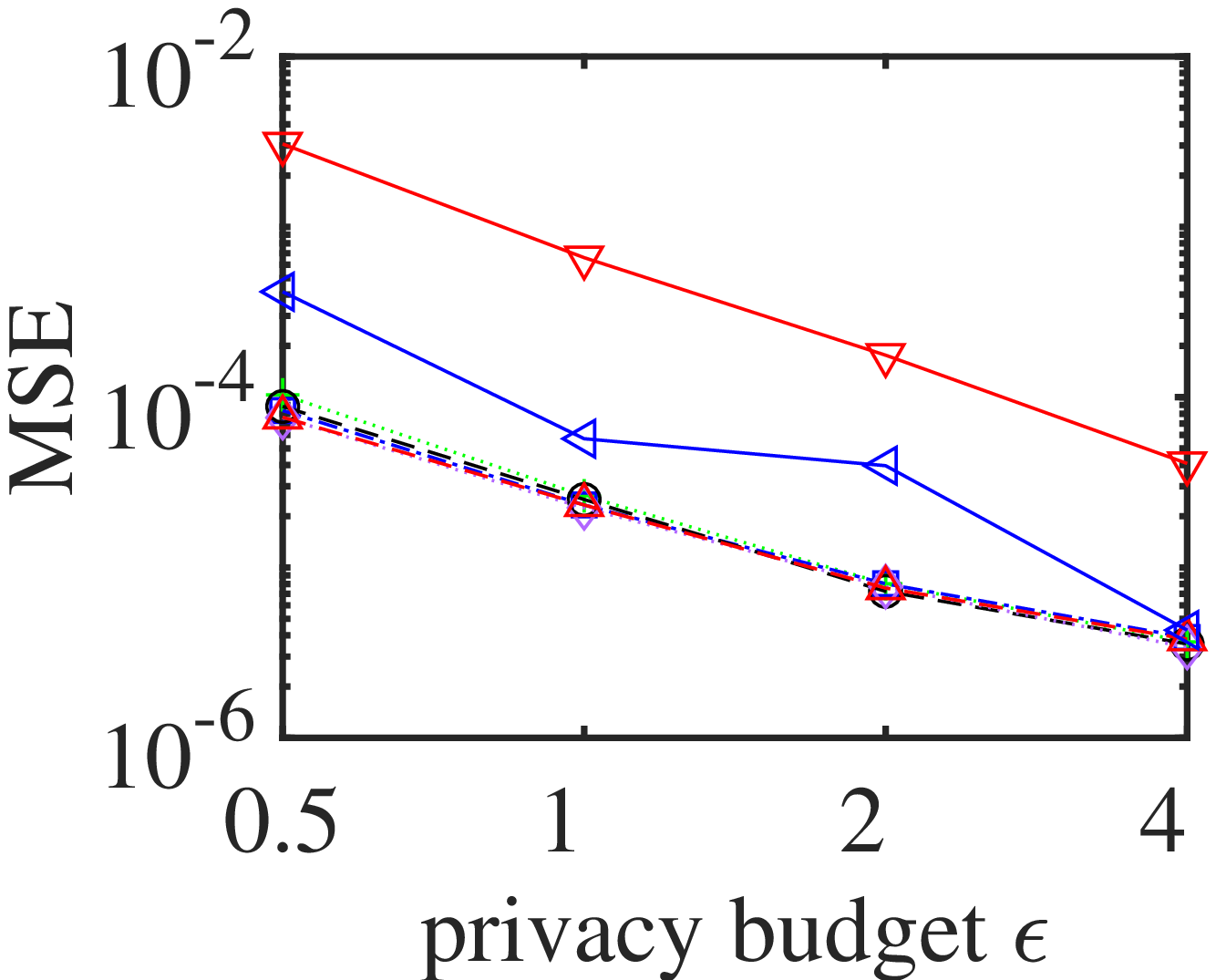}
         \caption{MX}
         \label{fig:discrete_mse_mx}
     \end{subfigure}
     \begin{subfigure}[b]{0.24\textwidth}
         \centering
         \includegraphics[width=\textwidth]{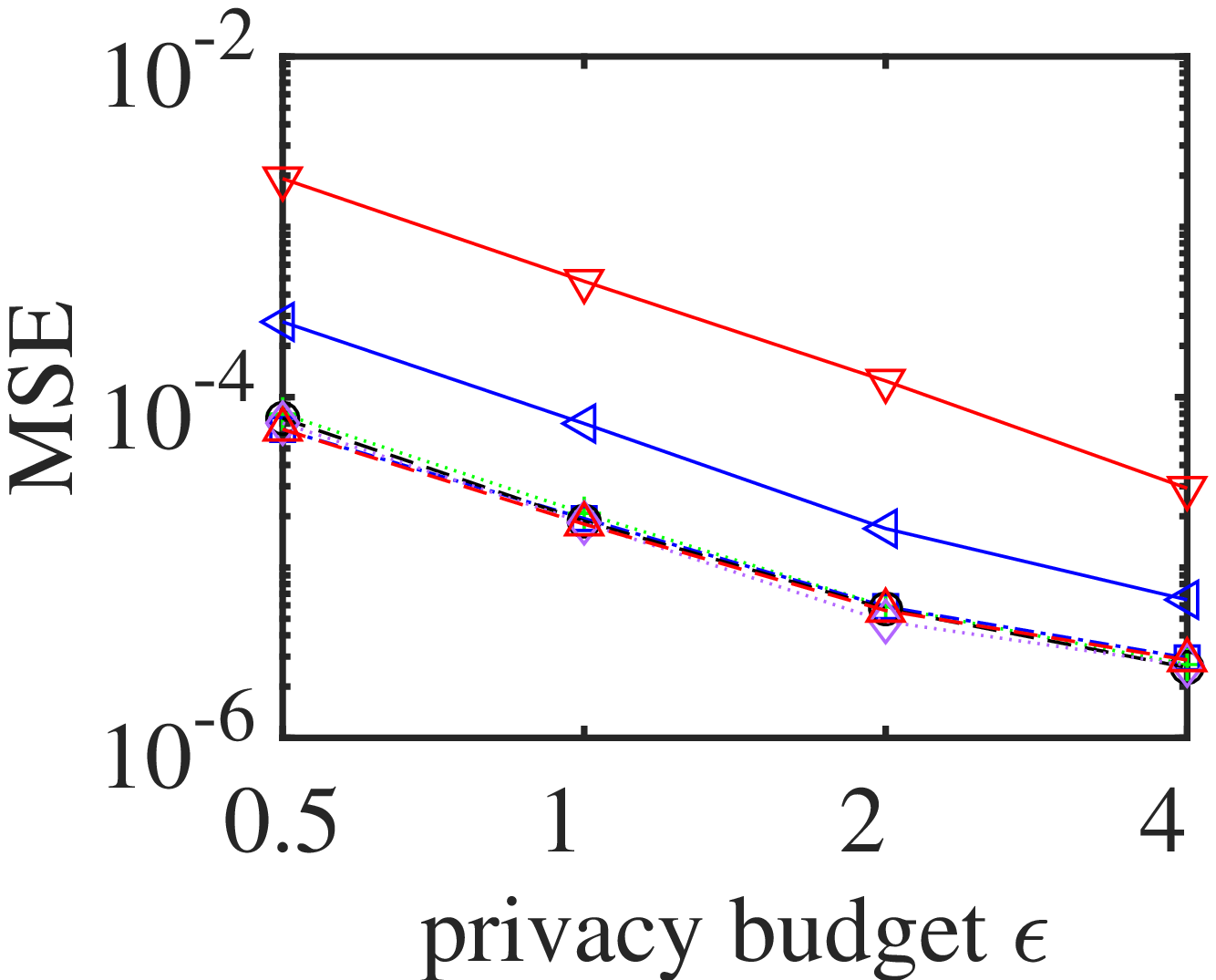}
         \caption{BR}
         \label{fig:discrete_mse_br}
     \end{subfigure}
     \begin{subfigure}[b]{0.24\textwidth}
         \centering
         \includegraphics[width=\textwidth]{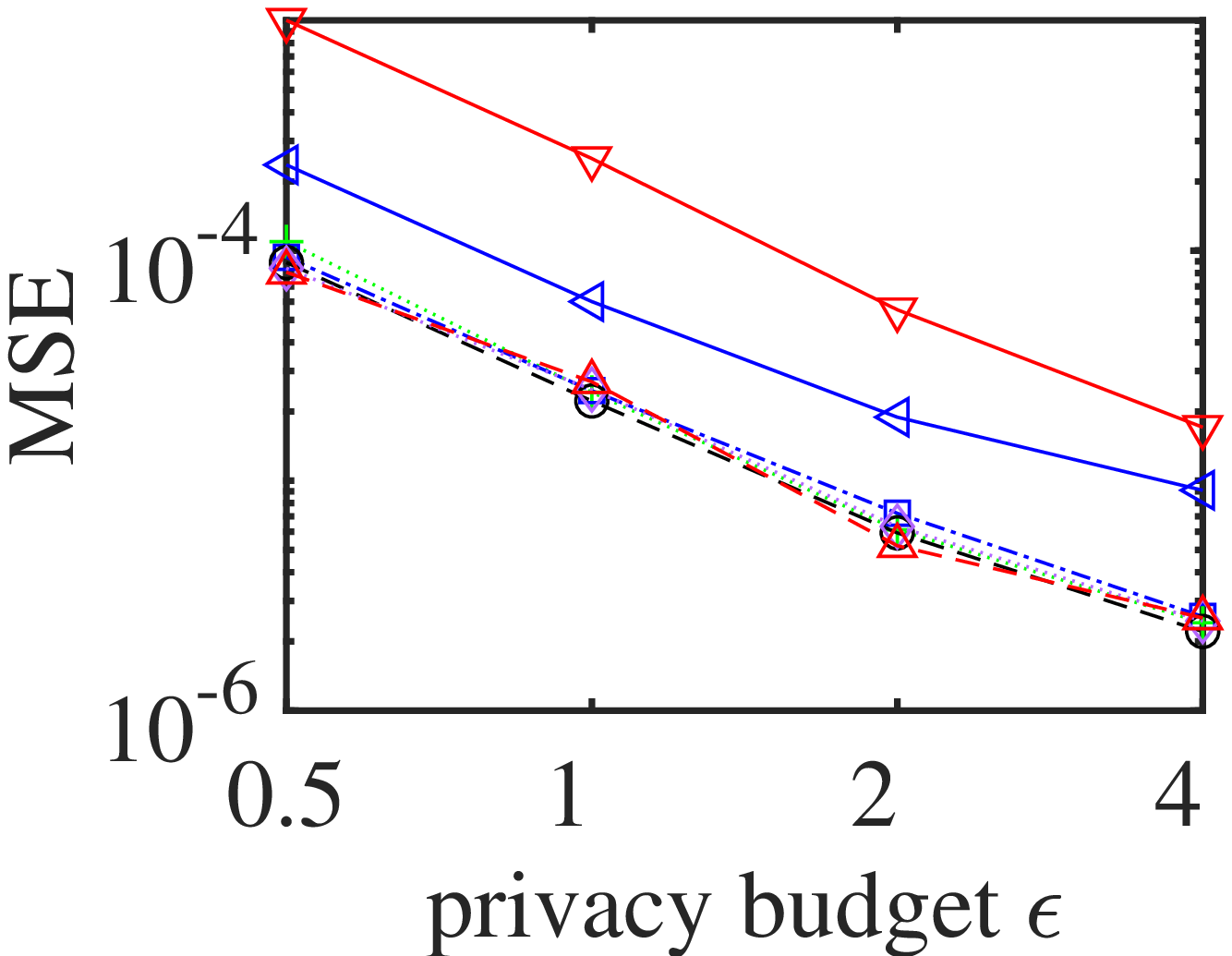}
         \caption{WISDN}
         \label{fig:discrete_mse_iot}
     \end{subfigure}
        \caption{Result accuracy for mean estimation with discretization post processing on \texttt{PM}, \texttt{HM}, and \texttt{HM-TP}. }
        \label{fig:MSE_discrete}
\end{figure*}

\begin{figure*}[h]
     \centering
     \begin{subfigure}[b]{0.24\textwidth}
         \includegraphics[width=\textwidth]{images/legend_v.eps}
     \end{subfigure}
     \begin{subfigure}[b]{0.24\textwidth}
         \centering
         \includegraphics[width=\textwidth]{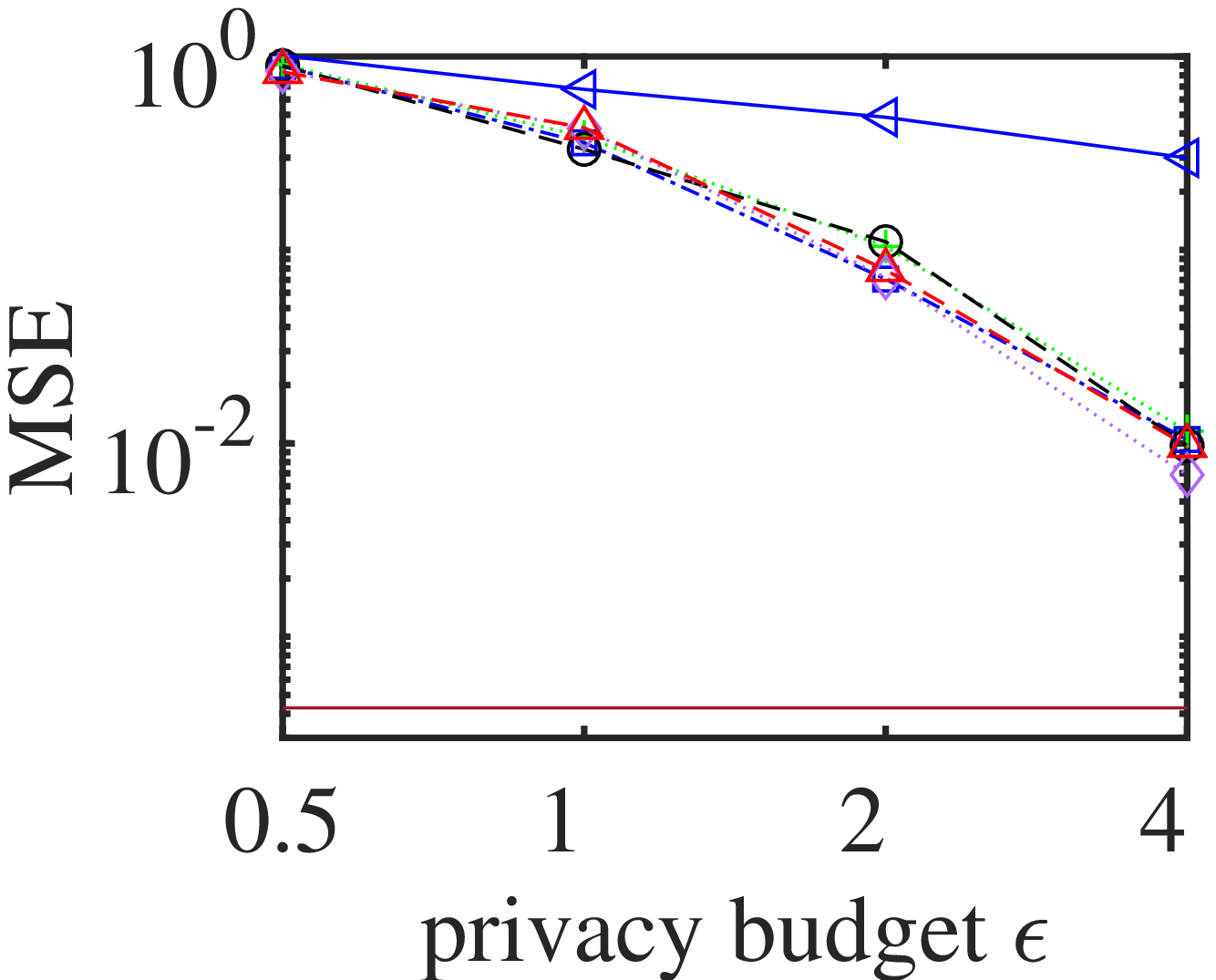}
         \caption{MX}
         \label{fig:discrete_linear_mx}
     \end{subfigure}
     \begin{subfigure}[b]{0.24\textwidth}
         \centering
         \includegraphics[width=\textwidth]{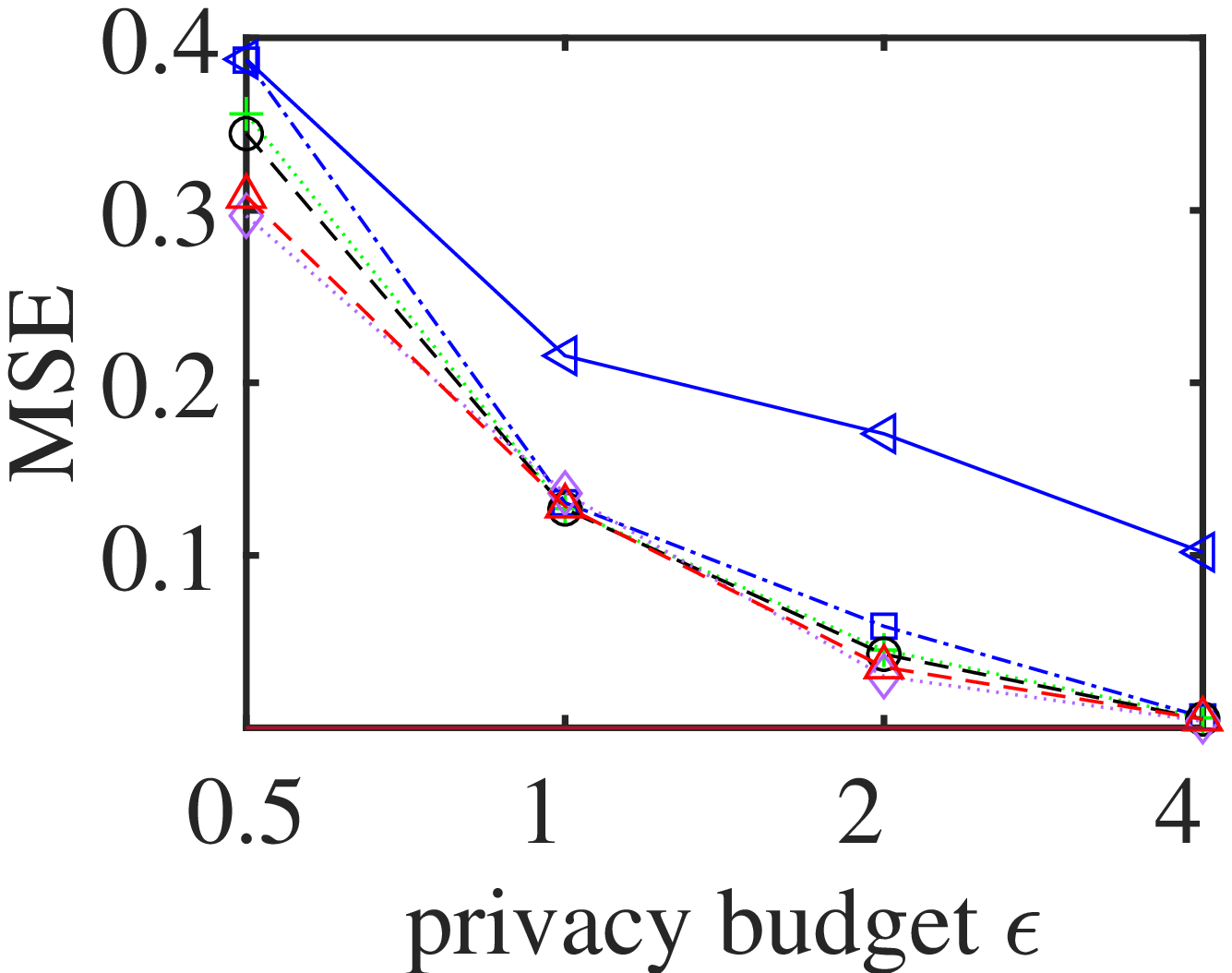}
         \caption{BR}
         \label{fig:discrete_linear_br}
     \end{subfigure}
     \begin{subfigure}[b]{0.24\textwidth}
         \centering
         \includegraphics[width=\textwidth]{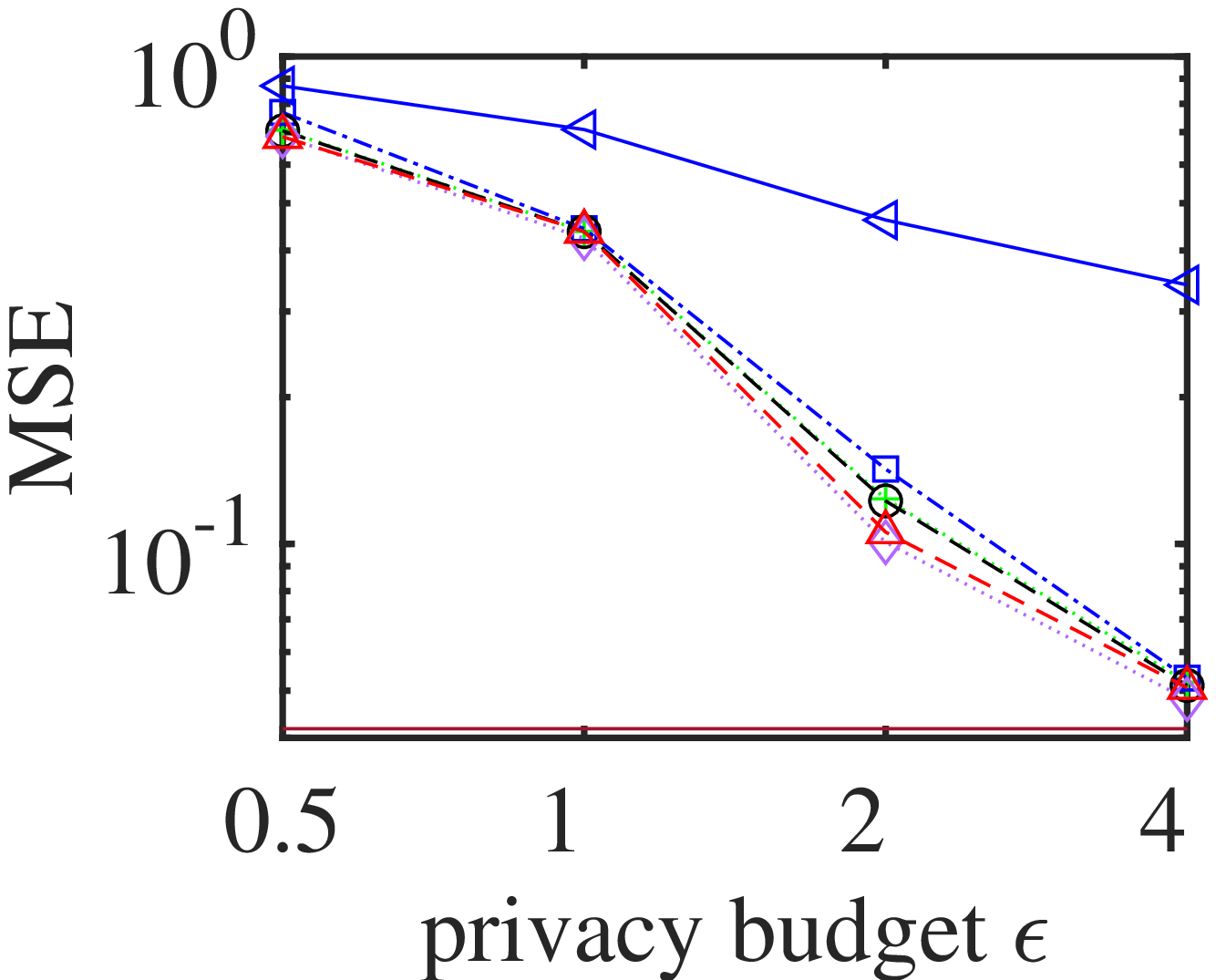}
         \caption{WISDN}
         \label{fig:discrete_linear_iot}
     \end{subfigure}
        \caption{Linear Regression with discretization post processing on \texttt{PM}, \texttt{HM}, and \texttt{HM-TP} (privacy parameter $\epsilon = 4$).}
        \label{fig:linear_discrete}
\end{figure*}

\begin{figure*}[h]
     \centering
     \begin{subfigure}[b]{0.24\textwidth}
         \includegraphics[width=\textwidth]{images/legend_v.eps}
     \end{subfigure}
     \begin{subfigure}[b]{0.24\textwidth}
         \centering
         \includegraphics[width=\textwidth]{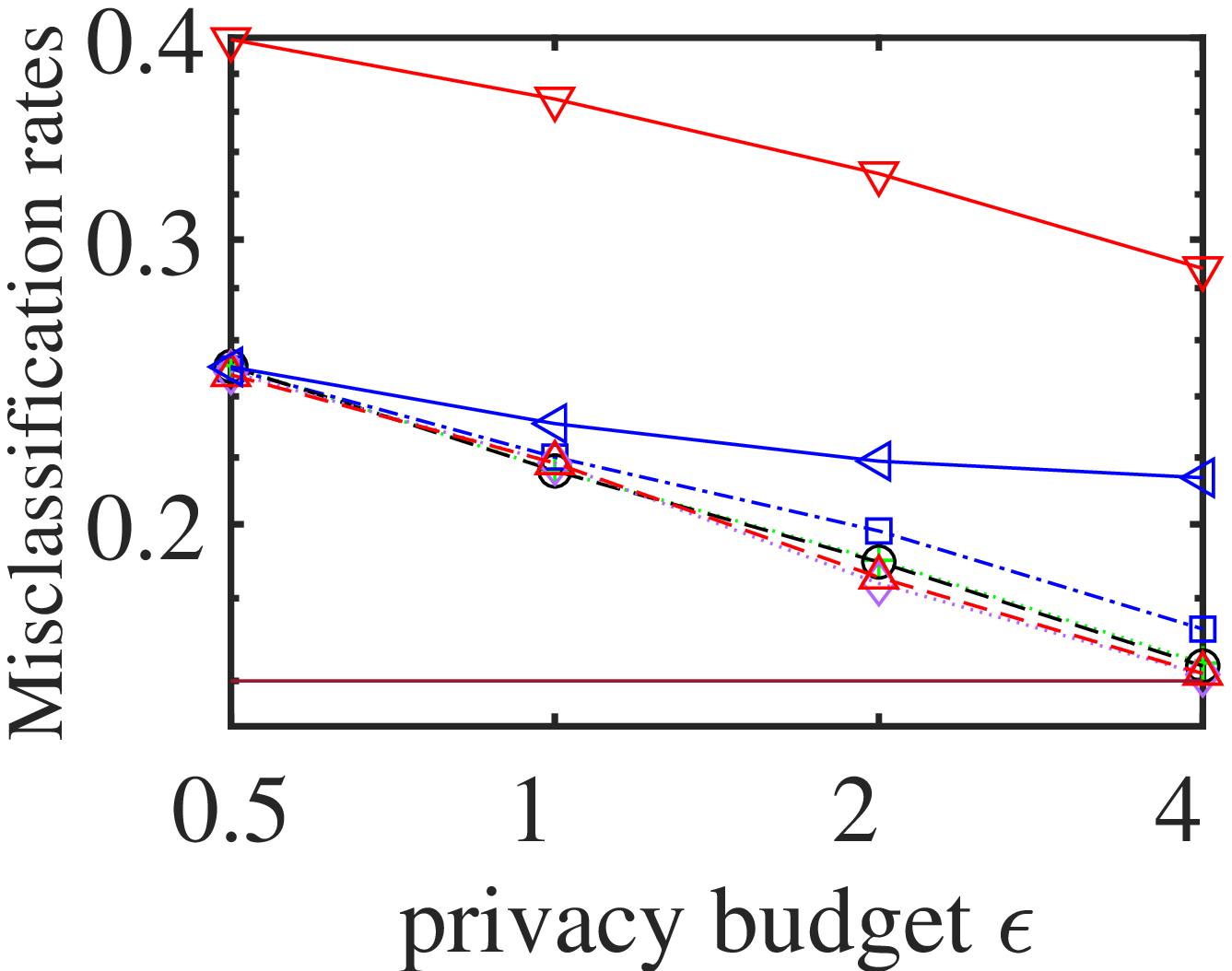}
         \caption{MX}
         \label{fig:discrete_log_mx}
     \end{subfigure}
     \begin{subfigure}[b]{0.24\textwidth}
         \centering
         \includegraphics[width=\textwidth]{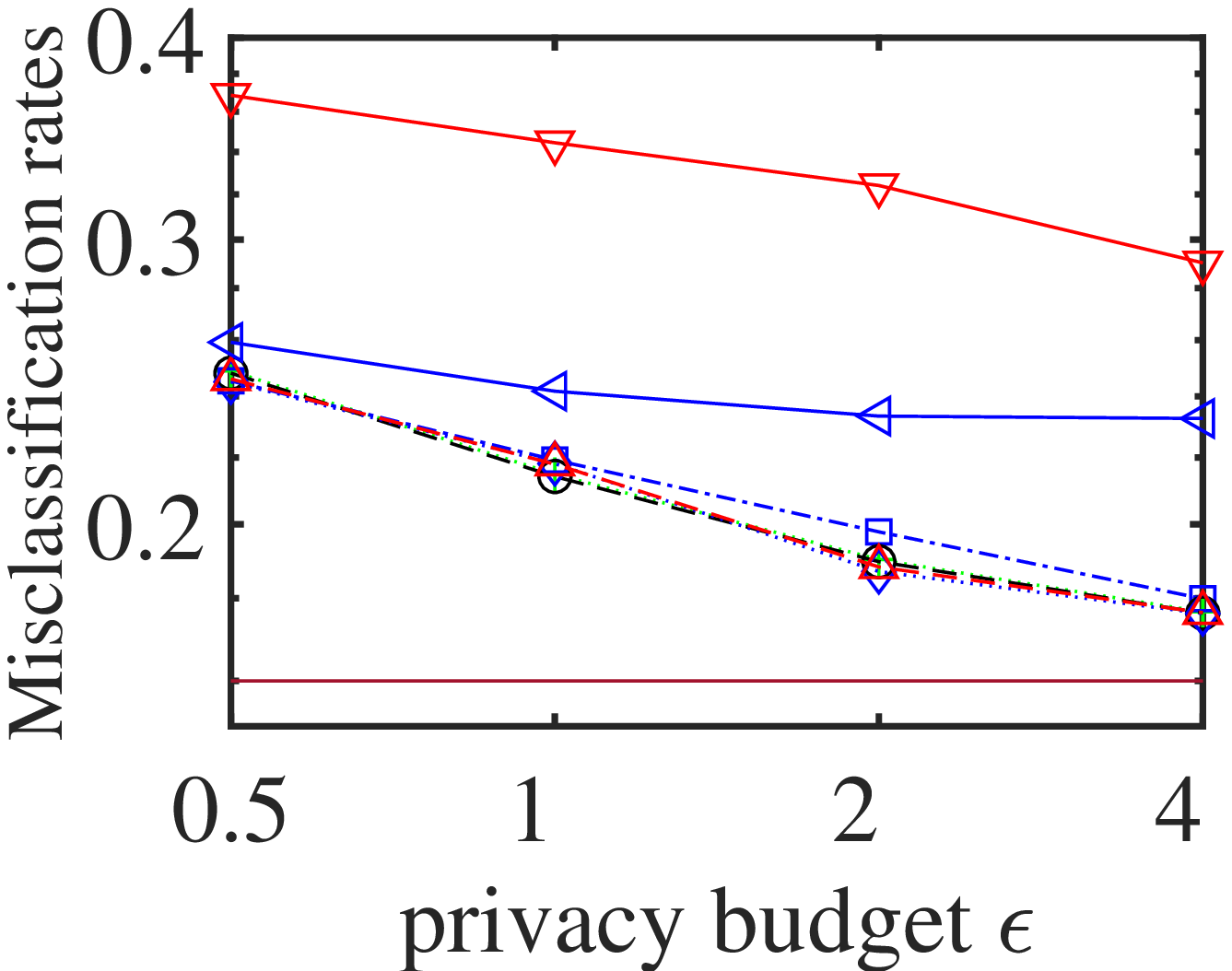}
         \caption{BR}
         \label{fig:discrete_log_br}
     \end{subfigure}
     \begin{subfigure}[b]{0.24\textwidth}
         \centering
         \includegraphics[width=\textwidth]{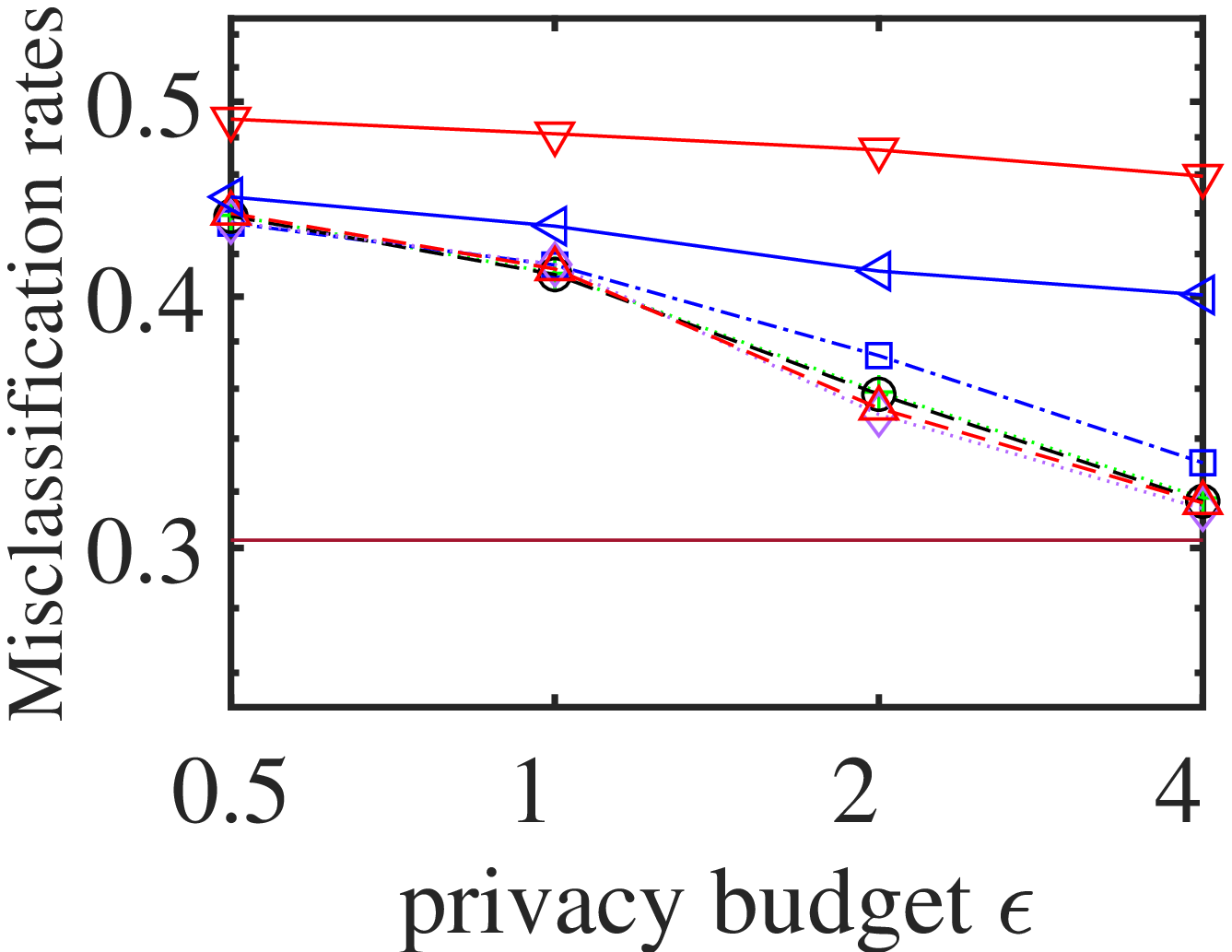}
         \caption{WISDN}
         \label{fig:discrete_log_iot}
     \end{subfigure}
        \caption{Logistic Regression with discretization post processing on \texttt{PM}, \texttt{HM}, and \texttt{HM-TP} (privacy budget $\epsilon = 4$).}
        \label{fig:log_discrete}
\end{figure*}

\begin{figure*}[h]
     \centering
     \begin{subfigure}[b]{0.24\textwidth}
         \centering
         \includegraphics[width=\textwidth]{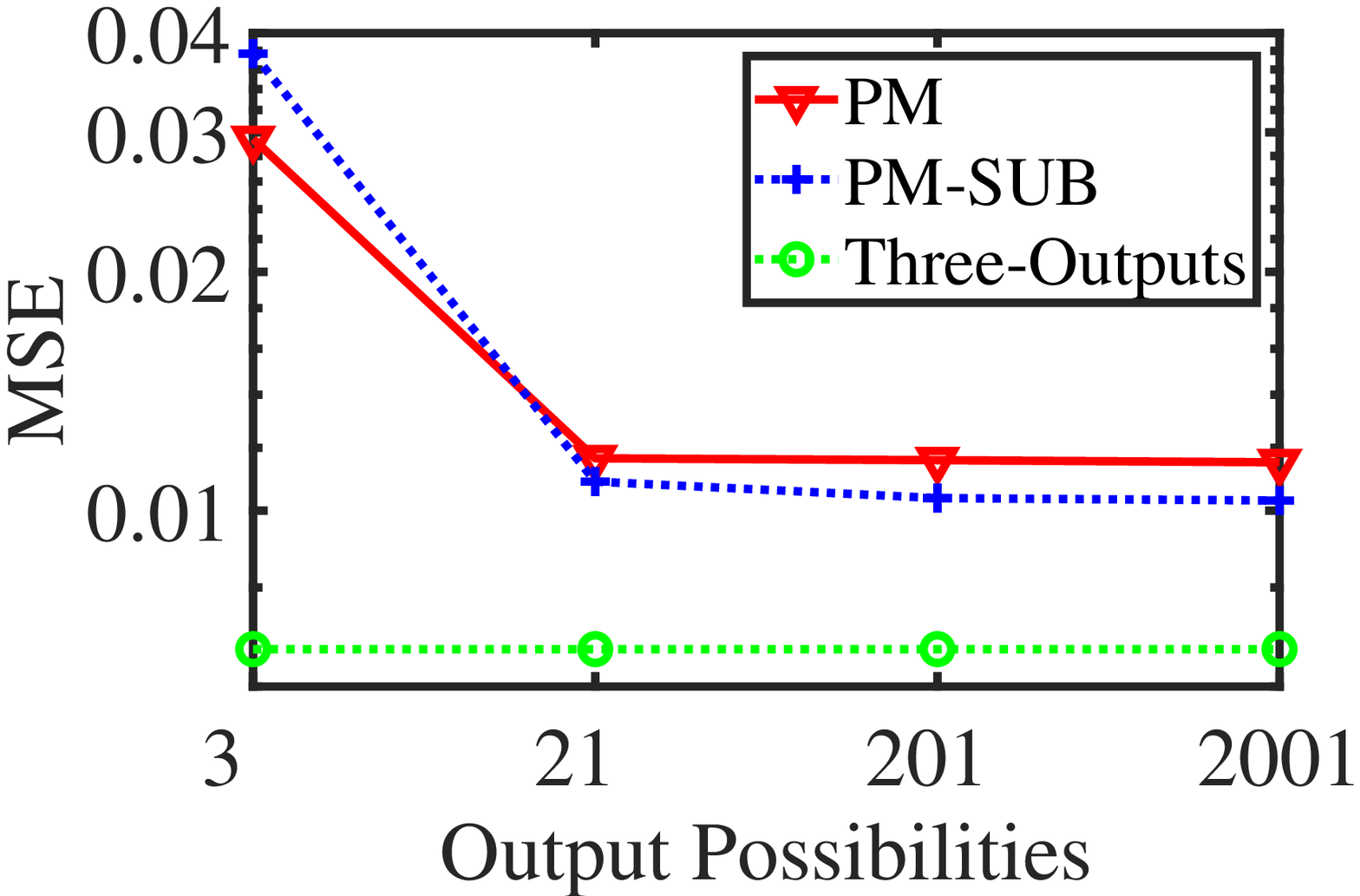}
         \caption{MX}
         \label{fig:discrete_output_linear_mx}
     \end{subfigure}
     \begin{subfigure}[b]{0.24\textwidth}
         \centering
         \includegraphics[width=\textwidth]{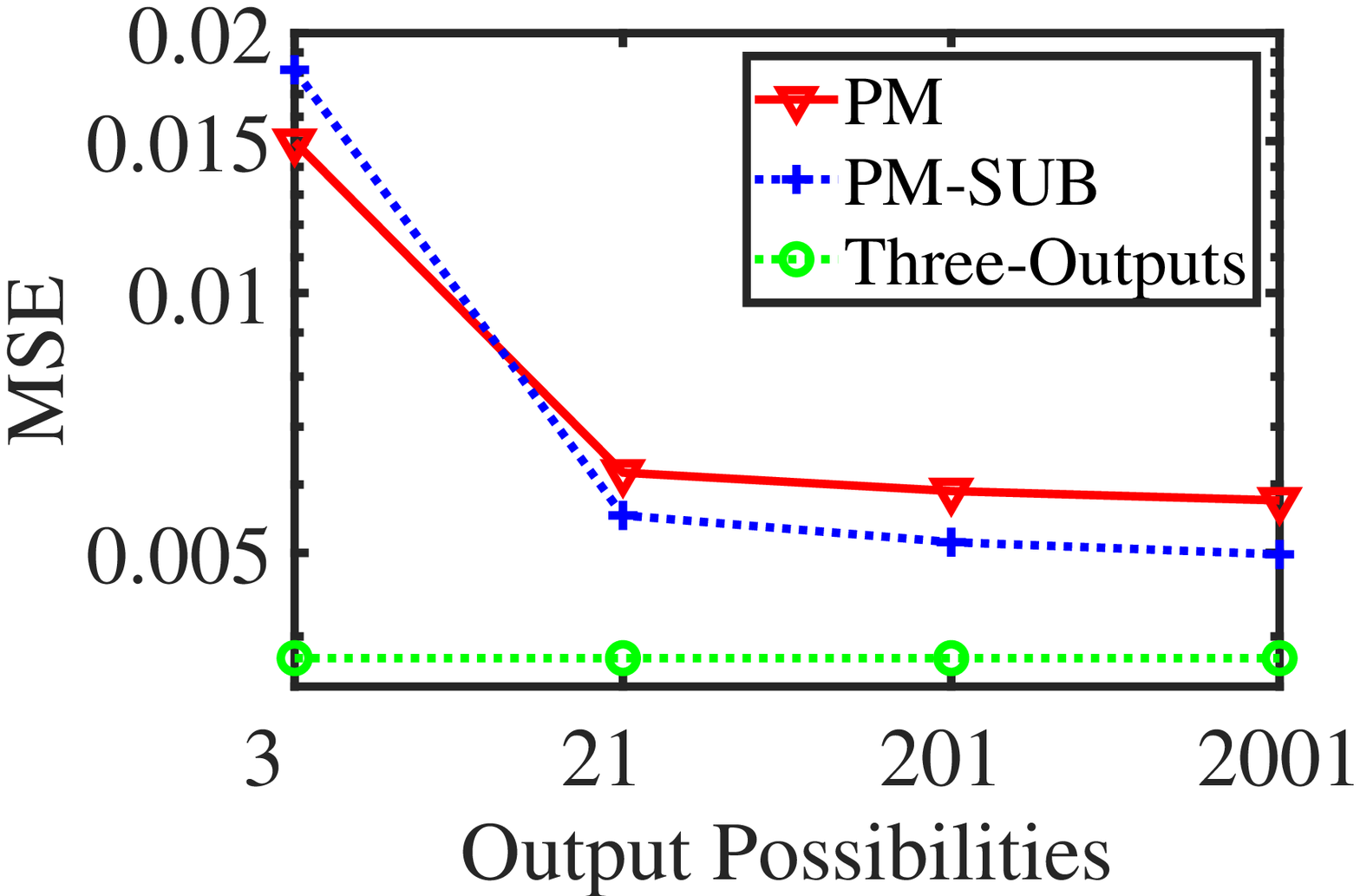}
         \caption{BR}
         \label{fig:discrete_output_linear_br}
     \end{subfigure}
     \begin{subfigure}[b]{0.24\textwidth}
         \centering
         \includegraphics[width=\textwidth]{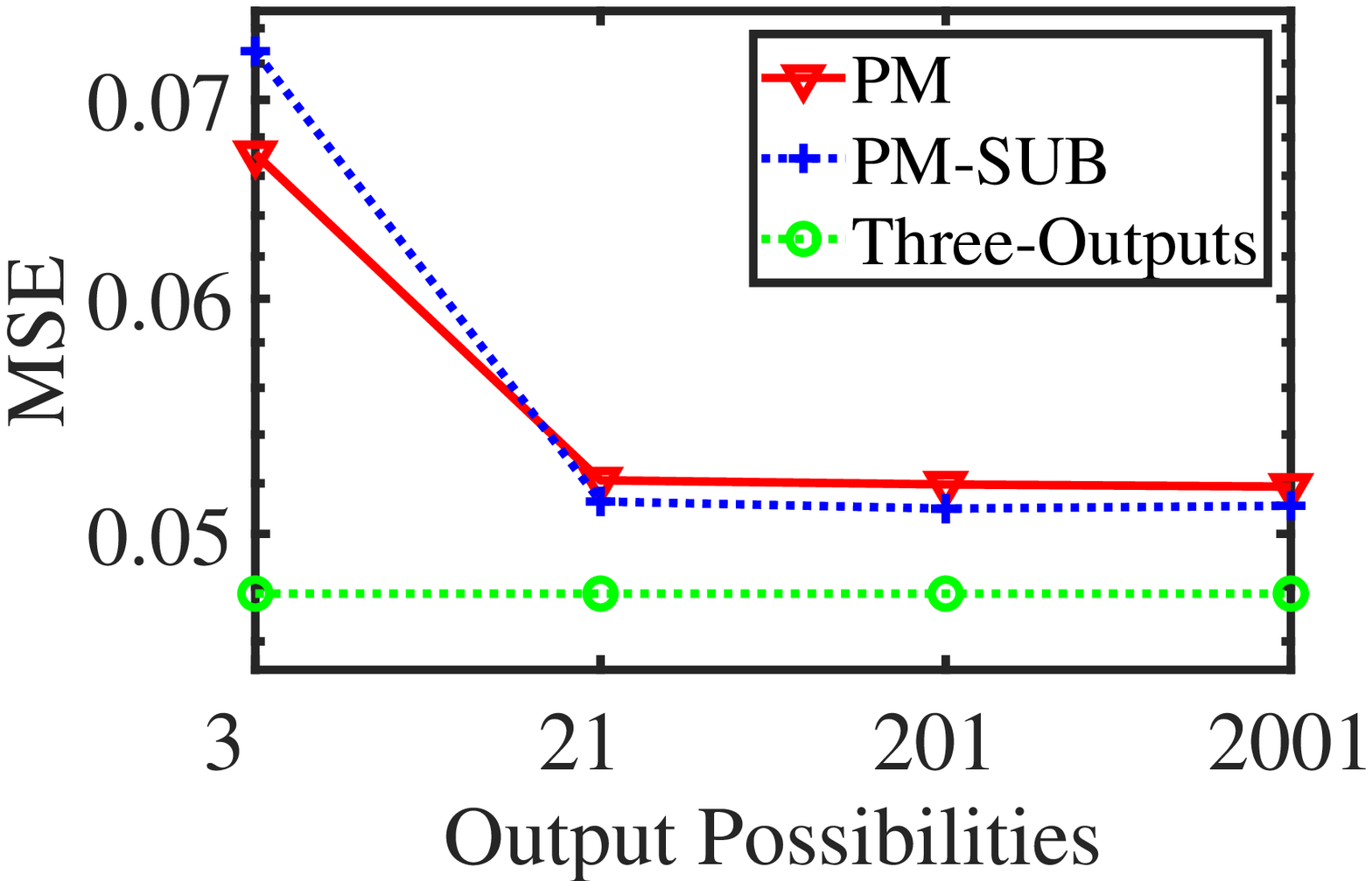}
         \caption{WISDN}
         \label{fig:discrete_output_linear_iot}
     \end{subfigure}
        \caption{Linear Regression with discretization post processing on \texttt{PM}, \texttt{HM}, and \texttt{HM-TP} (privacy budget $\epsilon = 4$).}
        \label{fig:linear_output_discrete}
\end{figure*}

\begin{figure}[!h]
    \centering
    \begin{subfigure}[b]{0.25\textwidth}
        \centering
        \includegraphics[scale=0.24]{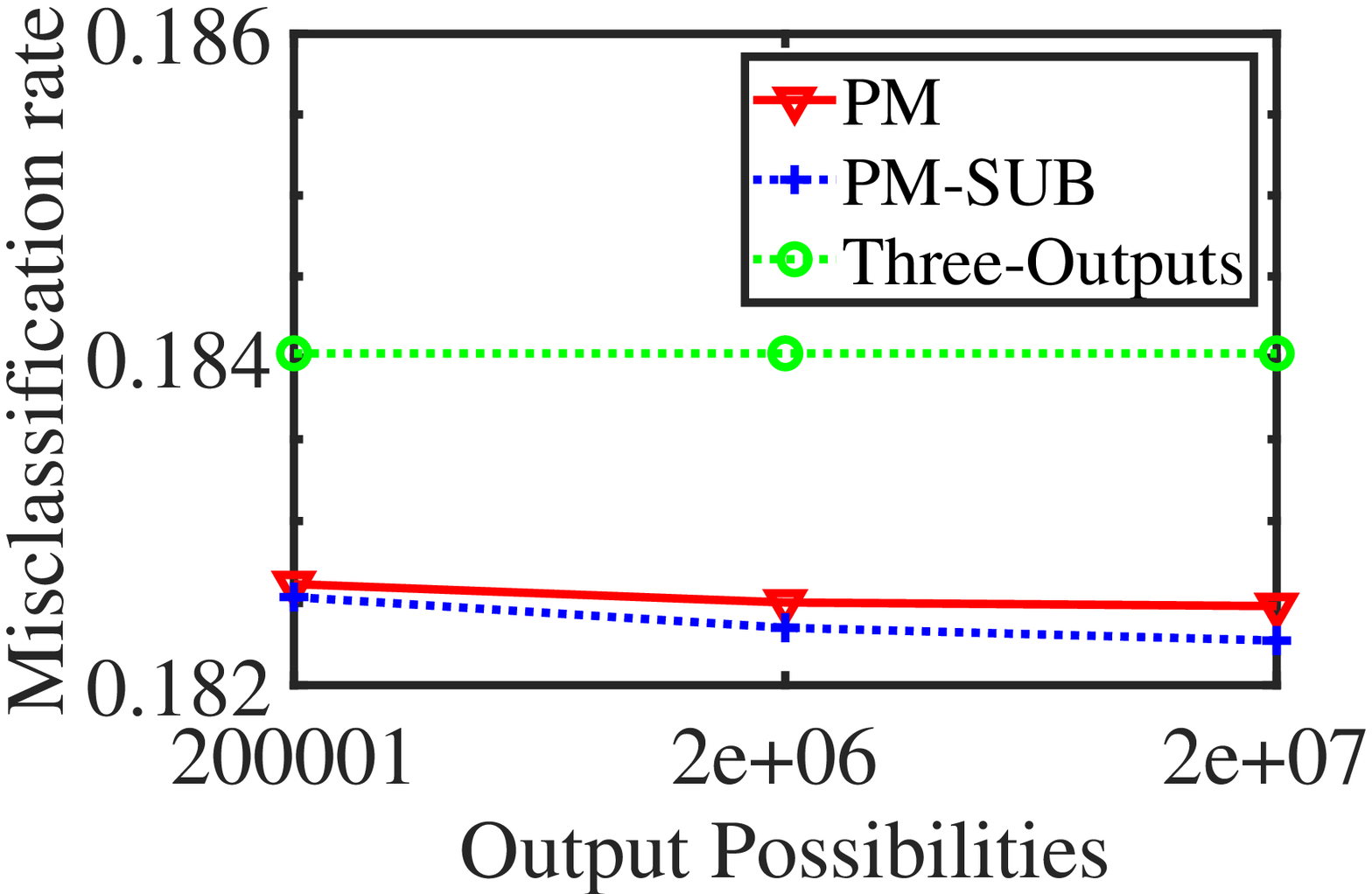}
        \caption{BR}
    \end{subfigure}
    \caption{Support Vector Machine with discretization post processing on \texttt{PM}, \texttt{HM}, and \texttt{HM-TP} (privacy budget $\epsilon = 5$).}
    \label{fig:svm_output_discrete}
\end{figure}

\section{Experiments}\label{sec:experiment}

We implemented both existing solutions and our proposed solutions, including \texttt{PM-SUB}, \texttt{Three-Outputs}, \texttt{HM-TP} proposed by us, \texttt{PM} and \texttt{HM} proposed by Wang~\textit{et~al.}~\cite{wang2019collecting}, Duchi~\textit{et al.}'s~\cite{duchi2018minimax} solution and the traditional Laplace mechanism. Our datasets include (i) the WISDM Human Activity Recognition dataset~\cite{kwapisz2011activity}  is a set of accelerometer data collecting on Android phones from $35$ subjects performing $6$ activities, where the domain of the timestamps of the phone's uptime is removed from the dataset, and the remaining $3$ numeric attributes are accelerations in $x, y,$ and $z$ directions measured by the Android phone's accelerometer and $2$ categorical attributes; (ii) two public datasets extracted from Integrated Public Use Microdata Series~\cite{mpcdrupl_2019} contain census records from Brazil (BR) and Mexico (MX). BR includes $4$M tuples and $16$ attributes, of which $6$ are numerical and $10$ are categorical. MX contains $4M$ records and $19$ attributes, of which $5$ are numerical and $14$ are categorical; (iii) a Vehicle dataset obtained by collecting from a distributed sensor network, including acoustic (microphone), seismic (geophone), and infrared (polarized IR sensor)~\cite{duarte2004vehicle}. The dataset contains $98528$ tuples and $101$ attributes, where $100$ attributes are numerical representing information such as the raw time series data observed at each sensor and acoustic feature vectors extracted from each sensor’s microphone. One attribute is categorical denoting different types of vehicles, which are labeled manually by a human operator to ensure high accuracy. Besides, information about vehicles is gathered to find out the type or brand of the vehicle. The Vehicle dataset is also used as the federated learning benchmark by~\cite{li2019fair}. We normalize the domain of each numeric attribute to $[-1, 1]$. In our experiments, we report average results over $100$ runs.

\subsection{Results on the Mean Values of Numeric Attributes}\label{sec:continous-mse}

We estimate the mean of every numeric attribute by collecting a noisy multidimensional tuple from each user. To compare with Wang~\textit{et~al.}'s~\cite{wang2019collecting} mechanisms, we follow their experiments and then divide the total privacy budget $\epsilon$ into two parts. Assume a tuple contains $d$ attributes which include $d_n$ numeric attributes and $d_c$ categorical attributes. Then, we allocate $d_n \epsilon /d$ budget to numeric attributes, and $d_c \epsilon /d$ to categorical ones, respectively. Our approach of using LDP for categorical data is same as that of Wang~\emph{et~al.}~\cite{wang2019collecting}. We estimate the mean value for each of the numeric attributes using existing methods: (i) Duchi~\textit{et~al.}'s~\cite{duchi2018minimax} solution handles multiple numeric attributes directly; (ii) when using the Laplace mechanism, it applies $\epsilon/d$ budget to each numeric attribute individually; (iii) \texttt{PM} and \texttt{HM} are from Wang~\textit{et~al.}~\cite{wang2019collecting}. In Section~\ref{sec:multiple}, we evaluate the mean square error (MSE) of the estimated mean values for numeric attributes using our proposed approaches. Fig.~\ref{fig:MSE} presents MSE results as a function of the total budget of $\epsilon$ in the datasets (WISDM, MX, BR and Vehicle). To simplify the complexity, we use last $6$ numerical attributes of the  Vehicle dataset to calculate MSE. Overall, our experimental evaluation shows that our proposed approaches outperform existing solutions. \texttt{HM-TP} outperforms existing solutions in all settings, whereas \texttt{PM-SUB}'s MSE is smaller than \texttt{PM}'s when privacy budget $\epsilon$ is large such as $4$, and \texttt{Three-Outputs}' performance is better at a small privacy budget. Hence, experimental results are in accordance with our theories.

We also run a set of experiments on synthetic datasets that contain numeric attributes only. We create four synthetic datasets, including $16$ numeric attributes where each attribute value is obtained by sampling from a Gaussian distribution with mean value $u \in \{0, \frac{1}{3}, \frac{2}{3}, 1\}$ and standard deviation of $\frac{1}{4}$. By evaluating the MSE in estimating mean values of numeric attributes with our proposed mechanisms, we present our experimental results in Fig.~\ref{fig:MSE_synthetic}. Hereby, we confirm that \texttt{PM-SUB}, \texttt{Three-Outputs} and \texttt{HM-TP} outperform existing solutions.

\subsection{Results on Empirical Risk Minimization}\label{sec:continous-empirical-risk-minimization}

In the following experiments, we evaluate the proposed algorithms' performance using linear regression, logistic regression, and SVM classification tasks. We change each categorical attribute $t_j$ with $k$ values into $k-1$ binary attributes with a domain $\{ -1, 1\}$, for example, given $t_j$, (i) $1$ represents the $l$-th ($l < k$) value on the $l$-th binary attribute and $-1$ on each of the rest of $k-2$ attributes; (ii) $-1$ represents the $k$-th value on all binary attributes. After the transformation, the dimension of WISDN is $43$, BR (resp. MX) is $90$ (resp. $94$) and Vehicle is $101$. 
Since both the BR and MX datasets contain the ``total income" attribute, we use it as the dependent variable and consider other attributes as independent variables. The Vehicle dataset is used for SVM~\cite{duarte2004vehicle, li2019fair}. Each tuple in the Vehicle dataset contains 100-dimensional feature and a binary label.

Consider each tuple of data as the dataset of a vehicle, so vehicles calculate gradients and run different LDP mechanisms to generate noisy gradients. Each mini-batch is a group of vehicles. Thus, the centralized aggregator i.e. cloud server updates the model after each group of vehicles send noisy gradients. 
The experiment involves $8$ competitors: \texttt{PM-SUB}, \texttt{Three-Outputs}, \texttt{HM-TP}, \texttt{PM}, \texttt{HM}, Duchi~\textit{et~al.}'s solution, Laplace and a non-private setting. We set the regularization factor $\lambda = 10^{-4}$ in all approaches. We use $10$-fold cross-validation $5$ times to evaluate the performance of each method in each dataset. Fig.~\ref{fig:log} and  Fig.~\ref{fig:svm} show that the proposed mechanisms (\texttt{PM-SUB}, \texttt{Three-Outputs}, and \texttt{HM-TP}) have lower misclassification rates than other mechanisms. Fig.~\ref{fig:linear} shows the MSE of the linear regression model. We ignore Laplace's result because its MSE clearly exceeds those of other mechanisms. In the selected privacy budgets, our proposed mechanisms (\texttt{PM-SUB}, \texttt{Three-Outputs}, and \texttt{HM-TP}) outperform existing approaches, including Laplace mechanism, Duchi~\textit{et~al.}'s solution, \texttt{PM}, and \texttt{HM}.

\subsection{Results after Discretization}\label{experiment:discreization-results}

In this section, we add a discretization post processing step in Algorithm~\ref{algo:discretization} to the implementation of mechanisms with continuous range of outputs, including \texttt{PM}, \texttt{PM-SUB}, \texttt{HM} and \texttt{HM-TP}. To confirm that the discretization is effective, we perform the following experiments. We separate the output domain $[-C, C]$ into $2000$ segments, and then we have $2001$ possible outputs given an initial input $x$. We add a discretization step to the experiments in Section~\ref{sec:continous-mse}. Fig.~\ref{fig:MSE_discrete} displays our experimental results. We confirm that our proposed approaches outperform existing solutions in estimating the mean value using  three real-world datasets: WISDM, MX, and BR after discretizing. 

In addition, we use log regression and linear regression to evaluate the performance after discretization. We repeat the experiments in Section~\ref{sec:continous-empirical-risk-minimization} with an additional discretization post processing step. Fig.~\ref{fig:linear_discrete} and Fig.~\ref{fig:log_discrete} present our experimental results. Compared with other approaches, the performance is similar to that before discretizing. Furthermore, Fig.~\ref{fig:linear_output_discrete} illustrates how the accuracy changes as output possibilities increase. It shows that the misclassification rate of the logistic regression task and the MSE of the linear regression task are related to the size of output possibilities. Although incurring with randomness, we find that the misclassification rate and MSE decrease as the number of output possibilities increases. When there are three output possibilities, it incurs randomness. Moreover, Fig.~\ref{fig:svm_output_discrete} shows that \texttt{PM-SUB} outperforms \texttt{Three-Outputs}, when the number of output possibilities is large. However, when we discretize the range of outputs into $2000$ segments, the performance is satisfactory and similar to the performance with a continuous range of outputs. Hence, our proposed approaches combined with the discretization step help retain the performance while enabling the usage in vehicles.

\section{Conclusion}\label{sec:conclusion}
In this paper, we propose \texttt{PM-OPT}, \texttt{PM-SUB}, \texttt{Three-Outputs}, and \texttt{HM-TP} local differential privacy mechanisms. These mechanisms effectively preserve  the privacy when collecting data records and computing accurate statistics in various data analysis tasks, including estimating the mean frequency and machine learning tasks such as SVM classification, logistic regression,  and linear regression. Moreover, we integrate our proposed local differential privacy mechanisms with \texttt{FedSGD} algorithm to create an \texttt{LDP-FedSGD} algorithm. The \texttt{LDP-FedSGD} algorithm enables the vehicular crowdsourcing applications to train a machine learning model to predict the traffic status while avoiding the privacy threat and reducing the communication cost. More specifically, by leveraging LDP mechanisms, adversaries are unable to deduce the exact location information of vehicles from uploaded gradients. Then, FL enables vehicles to train their local machine learning models using collected data and then send noisy gradients instead of data to the cloud server to obtain a global model. Extensive experiments demonstrate that our proposed approaches are effective and able to perform better than existing solutions. Further, we intend to apply  our  proposed LDP mechanisms to deep neural network to deal with more complex data analysis tasks.

\bibliographystyle{IEEEtran}
\bibliography{related}

 \newpage

\clearpage

\appendix
\subsection{\textbf{Proof of Lemma ~\ref{lem-symmetrization}}}\label{app:proof-of-lem-1}

The mechanism $\mathcal{M}_2$ satisfies the proper distribution in Eq.~(\ref{proper}) because of
\begin{align}
   & P_{C \leftarrow x}(\mathcal{M}_2) + P_{-C \leftarrow x}(\mathcal{M}_2) + P_{0 \leftarrow x}(\mathcal{M}_2) \nonumber \\ &  =  \frac{P_{C \leftarrow x}(\mathcal{M}_1) + P_{-C \leftarrow -x}(\mathcal{M}_1)}{2} \nonumber\\ &  + \frac{P_{-C \leftarrow x}(\mathcal{M}_1) + P_{C \leftarrow -x}(\mathcal{M}_1)}{2} \nonumber \\&  + \frac{P_{0 \leftarrow x}(\mathcal{M}_1) + P_{0 \leftarrow -x}(\mathcal{M}_1)}{2}  = 1.\nonumber
\end{align}
Besides, the mechanism $\mathcal{M}_2$ satisfies unbiased estimation in Eq.~(\ref{unbiased}) because
\begin{align}
    &C \cdot P_{C \leftarrow x}(\mathcal{M}_2) + (-C) \cdot P_{-C \leftarrow x}(\mathcal{M}_2) + 0 \cdot P_{0 \leftarrow x}(\mathcal{M}_2) \nonumber \\ & \quad =
    C \cdot \frac{P_{C \leftarrow x}(\mathcal{M}_1) + P_{-C \leftarrow -x}(\mathcal{M}_1)}{2} \nonumber\\ & \quad + (-C) \cdot \frac{P_{-C \leftarrow x}(\mathcal{M}_1) + P_{C \leftarrow -x}(\mathcal{M}_1)}{2} \nonumber\\ & \quad + 0 \cdot \frac{P_{0 \leftarrow x}(\mathcal{M}_1) + P_{0 \leftarrow -x}(\mathcal{M}_1)}{2} = x. \nonumber
\end{align}
In addition,
\begin{align}
   \frac{P_{C \leftarrow x}(\mathcal{M}_2)}{P_{C \leftarrow x'}(\mathcal{M}_2)} = \frac{P_{C \leftarrow x}(\mathcal{M}_1) + P_{-C \leftarrow -x}(\mathcal{M}_1)}{P_{C \leftarrow x'}(\mathcal{M}_1) + P_{-C \leftarrow -x'}(\mathcal{M}_1)}\label{eq:LDP-M2-1}
\end{align}
and
\begin{align}
   \frac{P_{-C \leftarrow x}(\mathcal{M}_2)}{P_{-C \leftarrow x'}(\mathcal{M}_2)}= \frac{P_{-C \leftarrow x}(\mathcal{M}_1) + P_{C \leftarrow -x}(\mathcal{M}_1)}{P_{-C \leftarrow x'}(\mathcal{M}_1) + P_{C \leftarrow -x'}(\mathcal{M}_1)}\label{eq:LDP-M2-2}
\end{align}
and
\begin{align}
  \frac{P_{0 \leftarrow x}(\mathcal{M}_2)}{P_{0 \leftarrow x'}(\mathcal{M}_2)} = \frac{P_{0 \leftarrow x}(\mathcal{M}_1) + P_{0 \leftarrow -x}(\mathcal{M}_1)}{P_{0 \leftarrow x'}(\mathcal{M}_1) + P_{0 \leftarrow -x'}(\mathcal{M}_1)}.\label{eq:LDP-M2-3}
\end{align}
According to (\ref{LDP}), we obtain
\begin{align}
  & \frac{e^{-\epsilon}(P_{C \leftarrow x'}(\mathcal{M}_1) + P_{-C \leftarrow -x'}(\mathcal{M}_1))}{P_{C \leftarrow x'}(\mathcal{M}_1) + P_{-C \leftarrow -x'}(\mathcal{M}_1)} \nonumber   \\& \quad \leq Eq.~(\ref{eq:LDP-M2-1}) \leq\frac{e^{\epsilon}(P_{C \leftarrow x'}(\mathcal{M}_1) + P_{-C \leftarrow -x'}(\mathcal{M}_1))}{P_{C \leftarrow x'}(\mathcal{M}_1) + P_{-C \leftarrow -x'}(\mathcal{M}_1)}, \nonumber
\end{align}
which is equavalent to
\begin{align}
  e^{-\epsilon} \leq  \textup{Eq}.~(\ref{eq:LDP-M2-1}) \leq e^{\epsilon}.
\end{align}
Similarly, we prove that Eq.~(\ref{eq:LDP-M2-2})~(\ref{eq:LDP-M2-3}) satisfy (\ref{LDP}). Hence, we conclude that $\mathcal{M}_2$ satisfies $\epsilon$-LDP requirements.

Then, we prove that the symmetrization process does not increase the worst-case noise variance as follows:

Since $\mathcal{M}_2$ satisfies the unbiased estimation of Eq.~(\ref{unbiased}), $\mathbb{E}[Y | X = x] = x$. Hence, the variance of mechanism $\mathcal{M}_2$ given $x$ is
\begin{align}
    &\textup{Var}_{\mathcal{M}_2}[Y | X = x] = \mathbb{E}[Y^2 | X = x] - (\mathbb{E}[Y | X = x])^2 \nonumber \\
    &= C^2\cdot P_{C \leftarrow x}(\mathcal{M}_2) + 0 \cdot P_{C \leftarrow x}(\mathcal{M}_2) \nonumber\\& \quad + (-C)^2 \cdot P_{-C \leftarrow x}(\mathcal{M}_2) - x^2 \nonumber \\ & = C^2 (1 - P_{0 \leftarrow x}(\mathcal{M}_2)) - x^2 \label{eq:var-m-2-0} \\& = C^2 \left (1 -  \frac{P_{0 \leftarrow x}(\mathcal{M}_1) + P_{0 \leftarrow -x}(\mathcal{M}_1)}{2}\right ) - x^2,\label{eq:var-m-2}
\end{align} or it changes to
\begin{align}
    &\textup{Var}_{\mathcal{M}_2}[Y | X = -x] = \mathbb{E}[Y^2 | X = -x] - (\mathbb{E}[Y | X = -x])^2 \nonumber \\
    &= C^2\cdot P_{C \leftarrow -x}(\mathcal{M}_2) + 0 \cdot P_{C \leftarrow -x}(\mathcal{M}_2) \nonumber\\& \quad + (-C)^2 \cdot P_{-C \leftarrow -x}(\mathcal{M}_2) - x^2 \nonumber \\ & = C^2 (1 - P_{0 \leftarrow -x}(\mathcal{M}_2)) - x^2 \label{eq:var-m-2-1} \\& = C^2 \left (1 -  \frac{P_{0 \leftarrow x}(\mathcal{M}_1) + P_{0 \leftarrow -x}(\mathcal{M}_1)}{2}\right ) - x^2,\label{eq:var-m-2}
\end{align}
when given $-x$.

The variance of mechanism $\mathcal{M}_1$ is
\begin{align}
    &\textup{Var}_{\mathcal{M}_1}[Y | X = x] = \mathbb{E}[Y^2 | X = x] - (\mathbb{E}[Y | X = x])^2 \nonumber \\
    &= C^2\cdot P_{C \leftarrow x}(\mathcal{M}_1) + 0 \cdot P_{C \leftarrow x}(\mathcal{M}_1) \nonumber\\& \quad + (-C)^2 \cdot P_{-C \leftarrow x}(\mathcal{M}_1) - x^2 \nonumber \\ & = C^2 (1 - P_{0 \leftarrow x}(\mathcal{M}_1)) - x^2,\label{eq:var-m-1}
\end{align}
or
\begin{align}
    &\textup{Var}_{\mathcal{M}_1}[Y | X = -x] = \mathbb{E}[Y^2 | X = -x] - (\mathbb{E}[Y | X = -x])^2 \nonumber \\
    &= C^2\cdot P_{C \leftarrow -x}(\mathcal{M}_1) + 0 \cdot P_{C \leftarrow -x}(\mathcal{M}_1) \nonumber\\& \quad + (-C)^2 \cdot P_{-C \leftarrow -x}(\mathcal{M}_1) - x^2 \nonumber \\ & = C^2 (1 - P_{0 \leftarrow -x}(\mathcal{M}_1)) - x^2.\label{eq:var-m-3}
\end{align}
Hence,
\begin{align}
   & \textup{Var}_{\mathcal{M}_2}[Y | X = x] = \textup{Var}_{\mathcal{M}_2}[Y | X = -x] \nonumber\\&= \frac{\textup{Var}_{\mathcal{M}_2}[Y | X = x] + \textup{Var}_{\mathcal{M}_2}[Y | X = -x]}{2}\\&=\frac{\textup{Eq}.~(\ref{eq:var-m-1}) + \textup{Eq}.~(\ref{eq:var-m-3})}{2} \nonumber \\& \leq \textup {Max}\{\textup{Var}_{\mathcal{M}_1}[Y | X = -x], \textup{Var}_{\mathcal{M}_1}[Y | X = x] \}.~\nonumber
\end{align}
\qeda

\subsection{\textbf{Proof of Lemma ~\ref{lem:lemma-2}}}\label{app:proof-of-lem-2}
In essence, with Eq.~(\ref{eq:lem-2-1})~(\ref{eq:lem-2-1}), we have
\begin{align}
   & \frac{P_{C \leftarrow 1}(\mathcal{M}_3)}{P_{C \leftarrow -1}(\mathcal{M}_3)} \nonumber \\& = \frac{P_{C \leftarrow 1}(\mathcal{M}_2) - \frac{e^{\epsilon} P_{C \leftarrow -1}(\mathcal{M}_2)-P_{C \leftarrow 1}(\mathcal{M}_2)}{e^{\epsilon}-1}}{P_{C \leftarrow -1}(\mathcal{M}_2) - \frac{e^{\epsilon} P_{C \leftarrow -1}(\mathcal{M}_2)-P_{C \leftarrow 1}(\mathcal{M}_2)}{e^{\epsilon}-1}} \nonumber \\& = e^\epsilon. \nonumber
\end{align}
Similarly, we can prove that $\frac{P_{C \leftarrow 1}(\mathcal{M}_3)}{P_{C \leftarrow -1}(\mathcal{M}_3)}, \frac{P_{C \leftarrow 1}(\mathcal{M}_3)}{P_{C \leftarrow -1}(\mathcal{M}_3)} = e^\epsilon$.
Besides, 
\begin{align}
    P_{C \leftarrow x} + P_{-C \leftarrow x} + P_{0 \leftarrow x} = 1.~\nonumber
\end{align}
In addition,
\begin{align}
   &C \cdot  P_{C \leftarrow x} + (-C) \cdot P_{-C \leftarrow x} + 0 \cdot P_{0 \leftarrow x}~\nonumber \\
   & = C \cdot \left( P_{C \leftarrow x}(\mathcal{M}_2) - P_{-C \leftarrow x}(\mathcal{M}_2) \right) = x.~\nonumber
\end{align}
Hence, mechanism $\mathcal{M}_3$ satisfies requirements in (\ref{LDP})~(\ref{unbiased})~(\ref{proper}).

Because the symmetric mechanism $\mathcal{M}_3$ satisfies requirements in (\ref{LDP})~(\ref{unbiased})~(\ref{proper}), the variance of $\mathcal{M}_3$ is
\begin{align}
    &\textup{Var}_{\mathcal{M}_3}[Y|X = x] = E[Y^2 | X = x] - (E[Y | X= x])^2 \nonumber \\
    &  = C^2\cdot P_{C \leftarrow x}(\mathcal{M}_3) + 0 \cdot P_{C \leftarrow x}(\mathcal{M}_3) \nonumber\\& \quad + (-C)^2 \cdot P_{-C \leftarrow x}(\mathcal{M}_3) - x^2 \nonumber \\
    &  = C^2(P_{C \leftarrow x}(\mathcal{M}_3) +  P_{-C \leftarrow x}(\mathcal{M}_3)) - x^2 \nonumber\\
    & = C^2 (1 -P_{0 \leftarrow x}(\mathcal{M}_3) ) - x^2.~\nonumber
\end{align}
By comparing with variance of $\mathcal{M}_2$ in Eq.~(\ref{eq:var-m-2-0}), we obtain
\begin{align}
    &\textup{Var}_{\mathcal{M}_3}[Y|X = x] - \textup{Var}_{\mathcal{M}_2}[Y|X = x] \nonumber  \\
    & = C^2 (1 -P_{0 \leftarrow x}(\mathcal{M}_3) ) - x^2 - (C^2 (1 -P_{0 \leftarrow x}(\mathcal{M}_2) ) - x^2) \nonumber  \\
    & = C^2 (P_{0 \leftarrow x}(\mathcal{M}_2) -P_{0 \leftarrow x}(\mathcal{M}_3) ).~\nonumber
\end{align}
Because of $P_{0 \leftarrow x}(\mathcal{M}_2) > P_{0 \leftarrow x}(\mathcal{M}_3)$,  based on Eq.~(\ref{eq:lem-2-3}) and Inequality~(\ref{lem:lem-2-assumption}), we get $P_{0 \leftarrow x}(\mathcal{M}_2) > P_{0 \leftarrow x}(\mathcal{M}_3)$ which means that
\begin{align}
    \textup{Var}_{\mathcal{M}_2}[Y|X = x] > \textup{Var}_{\mathcal{M}_3}[Y|X = x].~\nonumber
\end{align}
Thus, the variance of $\mathcal{M}_3$ is smaller than the variance of $\mathcal{M}_2$ when $x \in [-1,1]$, so we obtain that the worst-case noise variance of $\mathcal{M}_3$ is smaller that of $\mathcal{M}_2$ using
\begin{align}
    \max_{x\in [-1,1]} \textup{Var}_{\mathcal{M}_2}[Y|X = x] > \max_{x\in [-1,1]} \textup{Var}_{\mathcal{M}_3}[Y|X = x].~\nonumber
\end{align}
\qeda

\subsection{\textbf{Proof of Lemma ~\ref{lem:lem-3}}}\label{app:proof-of-lem-3}
Since   $P_{-C \leftarrow 0} + P_{C \leftarrow 0} + P_{0 \leftarrow 0} = 1$ and unbiased estimation $-C \cdot P_{-C \leftarrow 0} + C \cdot P_{C \leftarrow 0} + 0 \cdot P_{0 \leftarrow 0} = 0$, we have
\begin{align}
P_{C \leftarrow 0} = P_{-C \leftarrow 0} = \frac{1-P_{0 \leftarrow 0}}{2}.\label{eq:probability-c-0}
\end{align}
Then, based on (\ref{LDP})~(\ref{unbiased})~(\ref{proper}) and Lemma~\ref{lem:lemma-2}, we can derive $C$ with the following steps:
\begin{align}
    &P_{-C \leftarrow 1} + P_{C \leftarrow 1} + P_{0 \leftarrow 1} = 1, \nonumber\\
    & -C \cdot P_{-C \leftarrow 1} + C \cdot P_{C \leftarrow 1} + 0 \cdot P_{0 \leftarrow 1} = 1.\nonumber
\end{align}
Therefore, we have
\begin{align}
    &P_{C \leftarrow 1} = \frac{1-P_{0 \leftarrow 1}+\frac{1}{C}}{2}, \nonumber\\
    &P_{-C \leftarrow 1} = \frac{1-P_{0 \leftarrow 1} - \frac{1}{C}}{2}.\nonumber
\end{align}
From Lemma~\ref{lem:lemma-2}, we obtain
\begin{align}
    \frac{1-P_{0 \leftarrow 1}+\frac{1}{C}}{2} = e^\epsilon \cdot \left(\frac{1-P_{0 \leftarrow 1} - \frac{1}{C}}{2}\right), \nonumber
\end{align}
which is equivalent to
\begin{align}
    C = \frac{e^\epsilon + 1}{(e^\epsilon -1)(1-P_{0 \leftarrow 1})}.\label{eq:value-c}
\end{align}
Hence,
\begin{align}
P_{C \leftarrow 1} &=P_{-C \leftarrow -1} = \frac{(1-P_{0 \leftarrow 1})e^{\epsilon}}{e^{\epsilon}+1},\nonumber \\ P_{-C \leftarrow 1} & =P_{C \leftarrow -1}= \frac{(1-P_{0 \leftarrow 1})}{e^{\epsilon}+1}.\label{eq:probability-c-negative-1}
\end{align}
Then, we compute the variance as follows:
\begin{itemize}
    \item[I.] For $x \in [0, 1]$, we have
        \begin{align}
            &\textup{Var}[Y | X =x] = E[Y^2 | X = x] - (E[Y | X= x])^2 \nonumber \\
            & = C^2 \cdot P_{C \leftarrow x} + 0 \cdot P_{0 \leftarrow x} + (-C)^2 \cdot P_{-C \leftarrow x} - x^2 \nonumber\\
            & = C^2 \left(P_{C \leftarrow x}  + P_{-C \leftarrow x} \right) -x^2.\label{eq:var-mechanism-m3}
        \end{align}
        Substituting Eq.~(\ref{eq:lem-1-eq-1}) and Eq.~(\ref{eq:lem-1-eq-2}) into Eq.~(\ref{eq:var-mechanism-m3}) yields
        \begin{align}
            & = C^2 \left(P_{C \leftarrow 0} + (P_{C \leftarrow 1}-P_{C \leftarrow 0})x \right) ~\nonumber \\& \quad + C^2 \left( P_{-C \leftarrow 0} - (P_{-C \leftarrow 0}-P_{-C \leftarrow 1})x \right) -x^2 ~\nonumber \\
            & = C^2 \left(P_{C \leftarrow 0} + P_{-C \leftarrow 0} \right) + C^2  (P_{C \leftarrow 1}+ P_{C \leftarrow -1})x  ~\nonumber \\& \quad - C^2   (P_{C \leftarrow 0}+ P_{-C \leftarrow 0})x  -x^2 ~\nonumber \\
            & = C^2 \left(P_{C \leftarrow 0} + P_{-C \leftarrow 0} \right) + C^2  (1 - P_{0 \leftarrow 1})x  ~\nonumber \\& \quad - C^2(1 - P_{0 \leftarrow 0})x  -x^2 ~\nonumber \\
            & = C^2 \left ( 1 - P_{0 \leftarrow 0}\right) + C^2(P_{0 \leftarrow 0} - P_{0 \leftarrow 1})x - x^2.\label{eq:var-final-1}
        \end{align}
    \item[II.] For $x \in [-1, 0]$, we have
        \begin{align}
            &\textup{Var}[Y | X =x] = E[Y^2 | X = x] - (E[Y | X= x])^2 \nonumber\\
            & = C^2 \left ( 1 - P_{0 \leftarrow 0}\right) + C^2 (P_{0 \leftarrow 0} - P_{0 \leftarrow 1}) (-x) - x^2.\label{eq:var-final-2}
        \end{align}
\end{itemize}
Hence, by summarizing Eq.~(\ref{eq:value-c}), Eq.~(\ref{eq:var-final-1}), and Eq.~(\ref{eq:var-final-2}), we get the variance as follows:
\begin{align}\label{eq:variance-three-outputs-1}
    &\textup{Var}[Y | X =x] \nonumber\\& = C^2 \left ( 1 - P_{0 \leftarrow 0} \right) + C^2 (P_{0 \leftarrow 0} - P_{0 \leftarrow 1})|x| - x^2 \nonumber \\&= \left( \frac{e^\epsilon + 1}{(e^\epsilon -1)(1-P_{0 \leftarrow 1})} \right)^2 \left( 1 - P_{0 \leftarrow 0} + (P_{0 \leftarrow 0} - P_{0 \leftarrow 1})|x|\right)\nonumber \\& \quad - x^2.
\end{align}
Derive the partial derivative of $\textup{Var}[Y | X =x]$ to $P_{0 \leftarrow 1}$, and we get
\begin{align}
   & \frac{d(\textup{Var}[Y | X =x])}{dP_{0 \leftarrow 1}} \nonumber\\& = \frac{(e^\epsilon+1)^2\left( 2P_{0 \leftarrow 0} + |x|(1-2P_{0 \leftarrow 0} + P_{0 \leftarrow 1}) -2  \right)}{(P_{0 \leftarrow 1} +1)^2(e^\epsilon-1)^2}.\label{eq:first-derivative-1}
\end{align}
Then, we have the following cases:
\begin{itemize}
    \item [I.] If $|x| = 0$, Eq.~(\ref{eq:first-derivative-1})~$= \frac{(e^\epsilon+1)^2\left( 2P_{0 \leftarrow 0} -2  \right)}{(P_{0 \leftarrow 1} +1)^2(e^\epsilon-1)^2} < 0$,
    \item [II.] If $|x| = 1$, Eq.~(\ref{eq:first-derivative-1})~$= \frac{(e^\epsilon+1)^2\left( P_{0 \leftarrow 1} -1  \right)}{(P_{0 \leftarrow 1} +1)^2(e^\epsilon-1)^2} < 0$.
\end{itemize}
Therefore, if given $P_{0 \leftarrow 0}$, the variance of the output given input $x$ is a strictly decreasing function of $P_{0 \leftarrow 1}$. Hence, we get the minimized variance when $P_{0 \leftarrow 1}  = \frac{P_{0 \leftarrow 0}}{e^{\epsilon}}$.
\qeda

\subsection{\textbf{Solve Eq.~(\ref{eq:solve-quartic-equation})}}\label{appendix:solve_quartic}
To find the optimal $t$ for $\min_{t} \max_{x \in [-1, 1]} \textup{Var}[Y|x]$, we calculate first-order derivative of the $\max_{x \in [-1, 1]} \textup{Var}[Y|x]$ as follows:
    \begin{align}
    \nonumber& \frac{2t}{3 (e^\epsilon -1)^2}  +\frac{4  }{3 (e^\epsilon -1)}  +\frac{4  }{3 (e^\epsilon -1)^2}  -\frac{4 t^{-2}}{3(e^\epsilon -1)^2} \\&  -\frac{4 t^{-2}}{3(e^\epsilon -1)} -\frac{2t^{-3}}{3 (e^\epsilon -1)^2} -\frac{4 t^{-3}}{3 (e^\epsilon -1)}  -\frac{2t^{-3}}{3} \nonumber \\&= \frac{2}{3 (e^\epsilon -1)^2} [t + 2 e^\epsilon   - 2e^\epsilon t^{-2} - e^{2\epsilon} t^{-3}].
    \end{align}
  \noindent  Next, we calculate the second-order derivative of $\max_{x \in [-1, 1]} \textup{Var}[Y|x]$ as follows:
  \begin{align} 
        \frac{2}{3 (e^\epsilon -1)^2} [1 + 4 e^\epsilon t^{-3}  + 3 e^\epsilon t^{-4}] > 0.
    \end{align} Since the second-order derivative of $\max_{x \in [-1, 1]} \textup{Var}[Y|x] >0$, we can conclude that $\max_{x \in [-1, 1]} \textup{Var}[Y|x]$ has minimum point in its domain. 
    
To find $t$ which minimizes $\max_{x \in [-1, 1]} \textup{Var}[Y|x]$, we set $t^{4} + 2 e^\epsilon t^{3}  - 2e^\epsilon t  - e^{2\epsilon}  = 0$. By solving 
\begin{align}
    t^4 + 2 e^{\epsilon} t^3 - 2 e^{\epsilon} t - e^{2\epsilon} =0,\label{eq:solve-quartic-equation}
\end{align} we obtain Eq.~(\ref{eq:t-value}).
Define Eq.~(\ref{eq:solve-quartic-equation})'s coefficients as $c_4:=1,~c_3:=2 e^{\epsilon},~c_2:=0,~c_1:=- 2 e^{\epsilon},~c_0 := - e^{2\epsilon}$, and we obtain
\begin{align}
       c_4 \cdot t^4 + c_3 \cdot t^3 + c_1 \cdot t + c_0 =0.\label{eq:quatic_function}
\end{align}
To change Eq.~(\ref{eq:quatic_function}) into a depressed quartic form, we substitute $f: = e^\epsilon, t := y-\frac{c_3}{4c_4} = y - \frac{f}{2}$ into Eq.~(\ref{eq:quatic_function}) and obtain
\begin{align}
    y^4 + p \cdot y^2 + q \cdot y + r = 0, \label{eq:change_to_y}
\end{align}
where 
\begin{align}
    p & = \frac{8c_2c_4 - 3c_3^2}{8c_4^2} =  -\frac{3f^2}{2},\label{eq:p-definition}\\
    q &= \frac{c_3^3 - 4c_2c_3c_4 + 8c_1c_4^2}{8c_4^3} =f^3-2f,\label{eq:q-definition}\\
    r &= \frac{-3c_3^4 + 256c_0c_4^3 - 64c_1c_3c_4^2 + 16c_2c_3^2c_4}{256c_4^4} \nonumber \\& = -\frac{3}{16}f^4.\label{eq:r-definition}
\end{align}
Rewrite Eq.~(\ref{eq:change_to_y}) to the following:
\begin{align}
    \bigg( y^2 + \frac{p}{2} \bigg) = -qy - r + \frac{p^2}{4}.\label{eq:depressed-quartic-form}
\end{align}
Then, we introduce a variable $m$ into the factor on the left-hand side of Eq.~(\ref{eq:depressed-quartic-form}) by adding $2y^2m + pm + m^2$ to both sides. Thus, we can change the equation to the following:
\begin{align}
    \bigg (  y^2 + \frac{p}{2} + m \bigg ) = 2my^2 -qy + m^2 + mp + \frac{p^2}{4} -r.
\end{align}
Since $m$ is arbitrarily chosen, we choose the value of $m$ to get a perfect square in the right-hand side. Hence,
\begin{align}
    8m^3 + 8pm^2 + (2p^2 - 8r)m - q^2 = 0. \label{cubic_equation}
\end{align}
To solve Eq.~(\ref{cubic_equation}), we substitute Eq.~(\ref{eq:p-definition}), Eq.~(\ref{eq:q-definition}) and Eq.~(\ref{eq:r-definition}) into the following equations:
\begin{align}
    c_3' &:=8, \nonumber \\ c_2' &:= 8p, \nonumber \\ c_1' &:=2p^2 - 8r, \nonumber\\ c_0' &:= - q^2, \nonumber\\
 \Delta_0 &= (c_2')^2 - 3c_3' c_1' =  (8p)^2-3 \cdot 8 \cdot (2p^2 - 8r) = 0,\nonumber \\
  \nonumber  \Delta_1 &=  2 (c_2')^2 - 9 c_3' \cdot c_2' \cdot c_1' + 27 (c_3')^2 \cdot c_0' \nonumber\\ &= 2(8p)^3-9 \cdot 8 \cdot 8p \cdot (2p^2 - 8r) +27\cdot 8^2 \cdot (- q^2) \nonumber \\& = 6912(f^4-f^2), \nonumber \\
   \nonumber C &= \sqrt[3]{\frac{\Delta_1 \pm \sqrt{\Delta_1^2 - 4 \Delta_0^3}}{2}}  = \sqrt[3]{\Delta_1} \\&= \sqrt[3]{6912(f^4-f^2)}.\label{eq:C-expression}
\end{align}
By solving the cubic function Eq.~(\ref{eq:C-expression}), we have roots as follows :
\begin{align}
    m_k &= - \frac{1}{3c_3'}(c_2' + \xi^kC + \frac{\Delta_0}{\xi^kC}) \nonumber \\&= \frac{f^2}{2}+\sqrt[3]{\frac{f^2-f^4}{2}}\xi^k, k \in {0, 1, 2}.
\end{align}
We only use the real-value root; thus, we get
\begin{align}
    m =  \frac{f^2}{2}+\sqrt[3]{\frac{f^2-f^4}{2}}.
\end{align}
Thus,
\begin{align}
       y = \frac{ \pm_1 \sqrt{2m} \pm_2 \sqrt{-(2p + 2m \pm_1 \frac{\sqrt{2}q}{\sqrt{m}})}}{2}.
\end{align}
Then, the solutions of the original quartic equation are
\begin{align}
   t = - \frac{c_3}{4c_4} + \frac{ \pm_1 \sqrt{2m} \pm_2 \sqrt{-(2p + 2m \pm_1 \frac{\sqrt{2}q}{\sqrt{m}})}}{2}.\label{eq:quartic_equation_results}
\end{align}
Since $t$ is a real number and $t > 0$, we obtain Eq.~(\ref{eq:t-value}) after substituting $c_3, c_4, m, p, q, f$ into Eq.~(\ref{eq:quartic_equation_results}).
\qeda

\subsection{\textbf{Proof of Lemma ~\ref{lem:optimal_a}}}\label{appendix:proof_optimal_a}
By summarizing Lemma~\ref{lem-symmetrization}~\ref{lem:lemma-2}~\ref{lem:lem-3}, our designed mechanism achieves minimum variance when it satisfies $P_{0 \leftarrow 1}  = \frac{P_{0 \leftarrow 0}}{e^{\epsilon}}$. Hence, the variance is
\begin{align}
    &\textup{Var}[Y | X =x] = C^2 \left ( 1 - P_{0 \leftarrow 0} \right) + C^2 P_{0 \leftarrow 0} ( 1- \frac{1}{e^\epsilon})|x|-x^2,\label{eq:minimal-variance}
\end{align}
where \begin{align}
    C = \frac{e^\epsilon + 1}{(e^\epsilon -1)(1-\frac{P_{0 \leftarrow 0}}{e^\epsilon})}.\label{eq:optimal-c}
\end{align}
For simplicity, we set 
\begin{align}
    a &= P_{0 \leftarrow 0},\label{eq:optimal-a} \\
    b &= P_{0 \leftarrow 0} ( 1- \frac{1}{e^\epsilon}) = a ( 1- \frac{1}{e^\epsilon}).\label{eq:optimal-b}
\end{align}
Since $x\in [-1,1]$, the worst-case noise variance is
\begin{align}\label{worst-val-2}
  \max_{x \in [-1,1]} \textup{Var}[Y|x]=\begin{cases}
    (1-a)C^2+\frac{C^4b^2}{4},&\text{~~ if }\frac{C^2b}{2}<1, \\
    (1-a+b)C^2-1,&\text{~~ if }\frac{C^2b}{2}\geq 1.
    \end{cases}
\end{align}
Substituting Eq.~(\ref{eq:optimal-a}), Eq.~(\ref{eq:optimal-b}) and Eq.~(\ref{eq:optimal-c}) into Eq.~(\ref{worst-val-2}) yields
\hspace{-10pt}\begin{align}\label{worst-val-3}
  &\max_{x \in [-1,1]}\textup{Var}[Y|x]\nonumber\\
 \hspace{-10pt}   &=\begin{cases}
    \frac{(e^\epsilon +1)^2\cdot e^{2\epsilon}}{(e^\epsilon -1)^2}\bigg(\frac{1-a}{(e^\epsilon -a)^2}+ \frac{(e^\epsilon +1)^2\cdot a^2}{4(e^\epsilon -a)^4} \bigg),\text{ if }\frac{C^2b}{2}<1, \\[6pt]
    \frac{(e^\epsilon +1)^2\cdot e^{2\epsilon}}{(e^\epsilon -1)^2}\bigg(\frac{1-a}{(e^\epsilon -a)^2}+ \frac{(e^\epsilon -1)\cdot a}{e^\epsilon(e^\epsilon-a)^2} \bigg)\hspace{-2pt}-\hspace{-2pt}1,\hspace{-5pt}\text{ if }\frac{C^2b}{2}\geq 1.
    \end{cases}\nonumber\\
    &=\begin{cases}
    \frac{(e^\epsilon +1)^2\cdot e^{2\epsilon}}{(e^\epsilon -1)^2}\bigg(\frac{1-a}{(e^\epsilon -a)^2}+ \frac{(e^\epsilon +1)^2\cdot a^2}{4(e^\epsilon -a)^4} \bigg),\text{ if }\frac{C^2b}{2}<1, \\[6pt]
    \frac{(e^\epsilon +1)^2\cdot e^{\epsilon}}{(e^\epsilon -1)^2 \cdot (e^\epsilon-a)^2 } - 1,\text{ if }\frac{C^2b}{2}\geq 1.
    \end{cases}
\end{align}
Substituting Eq.~(\ref{eq:optimal-b}) and Eq.~(\ref{eq:optimal-c}) yields
\begin{align}
    \frac{C^2b}{2}
    &=\frac{(e^\epsilon +1)^2\cdot e^{2\epsilon}}{2(e^\epsilon -1)^2(e^\epsilon-a)^2}\cdot\frac{a(e^\epsilon -1)}{e^\epsilon} \nonumber\\
    &=\frac{(e^\epsilon +1)^2\cdot e^{\epsilon}\cdot a}{2(e^\epsilon -1)(e^\epsilon-a)^2} \nonumber\\
    &<1,
\end{align}
and
\begin{align}\label{eqn>0}
    &2(e^\epsilon-1)a^2-\left[4(e^\epsilon-1)e^\epsilon +(e^\epsilon+1)^2\cdot e^\epsilon \right]a \nonumber\\ &\quad\quad\quad + 2(e^\epsilon-1)e^{2\epsilon}>0.
\end{align}
To solve Eq.~(\ref{eqn>0}), we denote the smaller solution of the quadratic function as
\begin{align}
    a^*=\frac{e^\epsilon(e^{2\epsilon}+6e^\epsilon-3)-(e^\epsilon+1)e^\epsilon\sqrt{(e^\epsilon+1)^2+8(e^\epsilon -1)}}{4(e^\epsilon -1)}.\label{eq:a-star}
\end{align}
From Eq.~(\ref{eq:lem-1-eq-2}), we get
\begin{align}
    P_{-C \leftarrow 0} & \geq P_{-C \leftarrow 1}.\label{eq:compare-coefficients}
\end{align}
Then, substituting $P_{-C \leftarrow 0}$ and $P_{-C \leftarrow 1}$ with \textup{Eq}.~(\ref{eq:probability-c-0}) and  \textup{Eq}.~(\ref{eq:probability-c-negative-1}) in Eq.~(\ref{eq:compare-coefficients}) yields
\begin{align}
    \frac{1-P_{0 \leftarrow 0}}{2}  &\geq  \frac{(1-P_{0 \leftarrow 1})}{e^{\epsilon}+1}.
\end{align}
Hence,
\begin{align}
    a = P_{0 \leftarrow 0} \leq \frac{e^\epsilon}{e^\epsilon + 2}.\label{eq:bound-for-a}
\end{align}
From Eq.~(\ref{eq:bound-for-a}), we know that the Eq.~(\ref{eqn>0}) will be ensured (i) when $0\leq a<a^*$ if $a^* < \frac{e^\epsilon}{e^\epsilon +2}$, or (ii) when $0\leq a \leq \frac{e^\epsilon}{e^\epsilon +2}$ if $a^*\geq \frac{e^\epsilon}{e^\epsilon +2}$.
Hence, by combining with Eq.~(\ref{worst-val-3}), we obtain
\begin{align}\label{worst-val-4}
 \hspace{-10pt} &\max_{x \in [-1,1]}\textup{Var}[Y|x]= \nonumber \\
    &\begin{cases}
    \begin{cases}
    \frac{(e^\epsilon +1)^2\cdot e^{2\epsilon}}{(e^\epsilon -1)^2}\bigg(\frac{1-a}{(e^\epsilon -a)^2}+ \frac{(e^\epsilon +1)^2\cdot a^2}{4(e^\epsilon -a)^4} \bigg), \\\text{ for }0\leq a < a^*, \\[6pt]
    \frac{(e^\epsilon +1)^2\cdot e^{\epsilon}}{(e^\epsilon -1)^2 \cdot (e^\epsilon-a)^2 } - 1, \\ \text{ for }a^*\leq a \leq \frac{e^\epsilon}{e^\epsilon +2},
    \end{cases}\hspace{-30pt} , \text{ if }a^*<\frac{e^\epsilon}{e^\epsilon +2}\\[6pt]
    \frac{(e^\epsilon +1)^2\cdot e^{2\epsilon}}{(e^\epsilon -1)^2}\bigg(\frac{1-a}{(e^\epsilon -a)^2}+ \frac{(e^\epsilon +1)^2\cdot a^2}{4(e^\epsilon -a)^4} \bigg), \\ \text{ for }0\leq a\leq \frac{e^\epsilon}{e^\epsilon +2},~~~~~~~~~~~~~~~~~~\text{ if }a^*\geq \frac{e^\epsilon}{e^\epsilon +2}.
    \end{cases}
\end{align}
Substituting Eq.~(\ref{eq:a-star}) into $a^* = \frac{e^\epsilon}{e^\epsilon +2}$ yields
\begin{align}
    &\frac{e^\epsilon(e^{2\epsilon}+6e^\epsilon-3)-(e^\epsilon+1)e^\epsilon\sqrt{(e^\epsilon+1)^2+8(e^\epsilon -1)}}{4(e^\epsilon -1)} \nonumber \\& \quad = \frac{e^\epsilon}{e^\epsilon +2}.\label{eq:solv-eq-a}
\end{align}
After solving Eq.~(\ref{eq:solv-eq-a}), we get $\epsilon = \ln 4$.

\begin{figure}[h]
    \centering
    \includegraphics[scale=0.4]{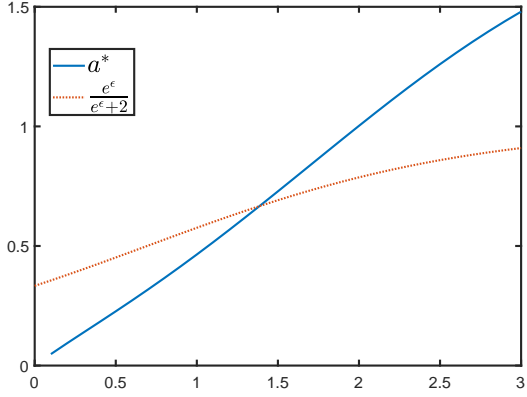}
    \caption{Compare $a^*$ with $\frac{e^\epsilon}{e^\epsilon+2}$.}
    \label{fig:solv_a_opt}
\end{figure}
According to Fig.~\ref{fig:solv_a_opt}, we obtain that $a^* \geq \frac{e^\epsilon}{e^\epsilon+2}$ if $0 < \epsilon \leq \ln4$. Since $\epsilon = \ln 4$ is the only solution if $\epsilon > 0$, we conclude that $a^* > \frac{e^\epsilon}{e^\epsilon+2}$ if $\epsilon > \ln4$. Therefore, we can replace the condition $a^*< \frac{e^\epsilon}{e^\epsilon +2}$ and write the variance as follows:

$\max_{x \in [-1,1]}\textup{Var}[Y|x] =$
\begin{subequations}
\begin{empheq}[left= {\empheqlbrace}]{align}
    &\begin{cases}
    \frac{(e^\epsilon +1)^2\cdot e^{2\epsilon}}{(e^\epsilon -1)^2}\bigg(\frac{1-a}{(e^\epsilon -a)^2}+ \frac{(e^\epsilon +1)^2\cdot a^2}{4(e^\epsilon -a)^4} \bigg), \\\text{ for }0\leq a < a^*, \\[6pt]
    \frac{(e^\epsilon +1)^2\cdot e^{\epsilon}}{(e^\epsilon -1)^2 \cdot (e^\epsilon-a)^2 } - 1, \\ \text{ for }a^*\leq a \leq \frac{e^\epsilon}{e^\epsilon +2},\label{worst-val-4-a}
    \end{cases}\hspace{-30pt} , \text{ if }\epsilon < \ln4,\\[6pt]
    &\frac{(e^\epsilon +1)^2\cdot e^{2\epsilon}}{(e^\epsilon -1)^2}\bigg(\frac{1-a}{(e^\epsilon -a)^2}+ \frac{(e^\epsilon +1)^2\cdot a^2}{4(e^\epsilon -a)^4} \bigg), \label{worst-val-4-b}\\ &\text{ for }0\leq a\leq \frac{e^\epsilon}{e^\epsilon +2},~~~~~~~~~~~~~~~~~~\text{ if }\epsilon \geq \ln4.\nonumber
\end{empheq}\label{worst-val-4}
\end{subequations}
To simplify the calculation of the minimum $\max_{x \in [-1,1]}\textup{Var}[Y|x]$ in Eq.~(\ref{worst-val-4}), we define
\begin{align}
   \hspace{-12pt} f_1(a) := \frac{(e^\epsilon +1)^2\cdot e^{2\epsilon}}{(e^\epsilon -1)^2}\bigg(\frac{1-a}{(e^\epsilon -a)^2}+ \frac{(e^\epsilon +1)^2\cdot a^2}{4(e^\epsilon -a)^4} \bigg),\label{eq:f_1}
\end{align}
and
\begin{align}
    f_2(a) := \frac{(e^\epsilon +1)^2\cdot e^{\epsilon}}{(e^\epsilon -1)^2 \cdot (e^\epsilon-a)^2 } - 1.\label{eq:f_2}
\end{align}
First order derivative of $f_2(a)$~in Eq.~(\ref{eq:f_2}) is
\begin{align}
    f_2'(a) = \frac{2e^\epsilon(e^\epsilon+1)^2}{(e^\epsilon-1)^2(e^\epsilon-a)^3} >0.
\end{align}
Since $f_2'(a)>0$, the worst-case noise variance monotonously increases if $a \in [a^*, \frac{e^\epsilon}{e^\epsilon + 2}]$, we can get optimal $a$ by analyzing $f_1(a)$ in Eq.~(\ref{eq:f_1}) when $a \in [0, a^*)$ if  $\epsilon < \ln 4$.
First order derivative of Eq.~(\ref{eq:f_1}) is
\begin{align}
  \hspace{-10pt}  f_1'(a) = \frac{2(1-a)}{(e^\epsilon-a)^3} - \frac{1}{(e^\epsilon-a)^2} + \frac{a(e^\epsilon+1)^2}{2(e^\epsilon-a)^4} + \frac{a^2(e^\epsilon+1)^2}{(e^\epsilon-a)^5}\label{eq:derivative_1} .
\end{align}
After simplifying $f_1'(a)$, we have
\begin{align}
   &f_1'(a) = \frac{-2a^3 - a^2(-e^{2\epsilon} - 5 -4e^{\epsilon})}{2(e^\epsilon -a)^5}\nonumber\\& \quad+\frac{-a(7e^{\epsilon} - 4e^{2\epsilon}-e^{3\epsilon}) - (2e^{3\epsilon}-4e^{2\epsilon}))}{2(e^\epsilon -a)^5}.\label{eq:derivative_1_simplify} 
\end{align}
Since $2(e^\epsilon -a)^5 > 0$, solving $f_1'(a) = 0$ is equivalent to solve the following equation
\begin{align}\label{eq:f_first_order}
   &2a^3 + a^2(-e^{2\epsilon} - 5 -4e^{\epsilon}) +a(7e^{\epsilon} - 4e^{2\epsilon}-e^{3\epsilon}) \nonumber\\& \quad   + (2e^{3\epsilon}-4e^{2\epsilon})=0.
\end{align}
We define  coefficients of Eq.~(\ref{eq:f_first_order}) as follows:
\begin{align}
    &c_3 := 2, \label{eq:c_3} \\& c_2 := -e^{2\epsilon} - 5 -4e^{\epsilon}\label{eq:c_2}, \\& c_1:= 7e^{\epsilon} - 4e^{2\epsilon}-e^{3\epsilon}, \label{eq:c_1}\\& c_0 := 2e^{3\epsilon}-4e^{2\epsilon}.\label{eq:c_0}
\end{align}
The general solution of the cubic equation involves calculation of
\begin{align}
    \Delta_0 &= \nonumber {c_2}^2 - 3c_3c_1 \\& \nonumber =(-e^{2\epsilon} - 5 -4e^{\epsilon})^2-3 \times 2(7e^{\epsilon} - 4e^{2\epsilon}-e^{3\epsilon})  \\& \nonumber  = e^{4\epsilon}+14e^{3\epsilon}+50e^{2\epsilon}-2e^{\epsilon}+25 > 0,\\
    \Delta_1 &= \nonumber 2{c_2}^3-9c_3c_2c_1+27c_3^2c_0 \\& \nonumber = 2(-e^{2\epsilon} - 5 -4e^{\epsilon})^3\\& \nonumber \quad -9 \times 2(-e^{2\epsilon}- 5 -4e^{\epsilon})(7e^{\epsilon} - 4e^{2\epsilon}-e^{3\epsilon}) \\& \nonumber \quad + 27 \times 2^2(2e^{3\epsilon}-4e^{2\epsilon})  \\&= \nonumber -2e^{6\epsilon} -42e^{5\epsilon}-270e^{4\epsilon}-404e^{3\epsilon}-918e^{2\epsilon}\\& \quad +30e^{\epsilon} -250 < 0, \nonumber\\
    C &= \sqrt[3]{\frac{\Delta_1 \pm \sqrt{\Delta_1^2-4\Delta_0^3}}{2}}. \nonumber 
\end{align}
Substituting $\Delta_0$ and $\Delta_1$ into $C$ yields
\begin{align}
   & \Delta_1^2-4\Delta_0^3 \nonumber \\& \quad = (-2c^6 -42c^5-270c^4-404c^3-918c^2+30c -250)^2 \nonumber \\& \quad -4(c^4+14c^3+50c^2-2c+25)^3 < 0,\nonumber
\end{align}
then
\begin{align}
 \sqrt{\Delta_1^2-4\Delta_0^3} =i \sqrt{4\Delta_0^3-\Delta_1^2},
\end{align} finally,
\begin{align}
    C &= \sqrt[3]{\frac{\Delta_1 - i\sqrt{4\Delta_0^3-\Delta_1^2}}{2} }. \nonumber\label{eq:C}
\end{align}
To eliminate the imaginary number, we change Eq.~(\ref{eq:C}) using Euler's formula. Define mold of $C$ as follows:
\begin{align}
  \hspace{-10pt}  |C| = (|C|^3)^{1/3}= \bigg (\sqrt{\frac{\Delta_1^2}{4} + \Delta_0^3  - \frac{\Delta_1^2}{4}} \bigg)^{1/3} = \sqrt{\Delta_0},
\end{align}
\begin{align}
    C &= |C|e^{i \theta}, \\
    C^3 &= |C|^3 e^{3i\theta}  = \sqrt{\Delta_0^3}e^{3i\theta}.
\end{align}
Therefore, we obtain $C = \sqrt{\Delta_0}e^{i \theta}.$
According to Euler's Formula, we have
\begin{align}
    e^{i3\theta} &= \cos 3\theta + i\sin 3\theta, \label{eq:eular}\\
    \cos 3\theta &= \frac{\Delta_1}{2 \Delta_0^\frac{3}{2}} < 0,\\
    \sin 3\theta &= -\frac{\sqrt{4 \Delta_0^3 - \Delta_1^2}}{2 \Delta_0^\frac{3}{2}} < 0.
\end{align}
Hereby,
\begin{align}
    3 \theta &= -\pi + \arccos(-\frac{\Delta_1}{2 \Delta_0^\frac{3}{2}}), \\\theta &= -\frac{\pi}{3} + \frac{1}{3}\arccos(-\frac{\Delta_1}{2 \Delta_0^\frac{3}{2}}).\label{eq:theta}
\end{align}
The solution of the cubic function is
\begin{align}
     a_k &= - \frac{1}{3c_3}(c_2 + \xi^kC +\frac{\Delta_0}{\xi^kC}),~ k \in \{0, 1, 2 \}.\label{eq:a_k}
\end{align}
To solve Eq.~(\ref{eq:a_k}), we have the following cases:
\begin{itemize}
    \item If $k=0$, we have
        \begin{align}
            a_0 &= - \frac{1}{3c_3}(c_2 + C +\frac{\Delta_0}{C}).\label{eq:a_k_0}
        \end{align}
         Substituting $\theta$ (\ref{eq:theta}) and $c_2$ (\ref{eq:c_2}) into Eq.~(\ref{eq:a_k_0}) yields
         \begin{align}
          & a_0\hspace{-2pt} =\hspace{-2pt} -\frac{1}{6}(-e^{2\epsilon}\hspace{-2pt}\hspace{-2pt} -\hspace{-2pt} 4e^{\epsilon}\hspace{-2pt} -5\hspace{-2pt}\nonumber \\ &\quad+\hspace{-2pt} 2 \sqrt{\Delta_0} \cos(-\frac{\pi}{3} + \frac{1}{3}\arccos(-\frac{\Delta_1}{2 \Delta_0^\frac{3}{2}}))).
         \end{align}
    \item If $k=1$, we have
            \begin{align}
                 a_1 &= - \frac{1}{3c_3}(c_2 + \xi C +\frac{\Delta_0}{\xi C}) \nonumber \\
                 & = - \frac{1}{3c_3}(c_2 + (-\frac{1}{2} + \frac{\sqrt{3}i}{2})C + \frac{\Delta_0}{(-\frac{1}{2} + \frac{\sqrt{3}i}{2})C}) \nonumber \\& = - \frac{1}{3c_3}(c_2 +  C e^{i \frac{2 \pi}{3}} + \frac{\Delta_0}{C} e^{i \frac{4 \pi}{3}}) \nonumber \\
                 & =- \frac{1}{3c_3}(c_2 + \sqrt{\Delta_0} e^{i(\theta + \frac{2}{3} \pi)} + \sqrt{\Delta_0} e^{i(\frac{4 \pi}{3}- \theta)}).\label{eq:a_k_1}
            \end{align}
            Simplify Eq.~(\ref{eq:a_k_1}) using Eq.~(\ref{eq:eular}), we have
            \begin{align}
                a_1 = - \frac{1}{3c_3}(c_2 +2 \sqrt{\Delta_0} \cos(\theta + \frac{2 \pi}{3})).\label{eq:a_k_1_simplify}
            \end{align}
           Substituting $\theta$ (\ref{eq:theta}) and $c_2$ (\ref{eq:c_2}) into Eq.~(\ref{eq:a_k_1_simplify}) yields
            \begin{align}
                a_1 &= -\frac{1}{6}(-e^{2\epsilon} - 4e^{\epsilon} -5 \nonumber\\& \quad+ 2 \sqrt{\Delta_0} \cos(\frac{\pi}{3} + \frac{1}{3}\arccos(-\frac{\Delta_1}{2 \Delta_0^\frac{3}{2}}))).\label{eq:a_k_1_final}
            \end{align}
        \item If $k=2$, we get
            \begin{align}
               & a_2 = - \frac{1}{3c_3}(c_2 + \xi^2 C +\frac{\Delta_0}{\xi^2 C}) \nonumber \\
                 & =  - \frac{1}{3c_3}(c_2 +(-\frac{1}{2} + \frac{\sqrt{3}i}{2})^2C + \frac{\Delta_0}{(-\frac{1}{2} + \frac{\sqrt{3}i}{2})^2C}) \nonumber \\
                 & = - \frac{1}{3c_3}(c_2 +(-\frac{1}{2} - \frac{\sqrt{3}i}{2})C + \frac{\Delta_0}{(-\frac{1}{2} - \frac{\sqrt{3}i}{2})C}) \nonumber \\
                 & =- \frac{1}{3c_3}(c_2 + \sqrt{\Delta_0} e^{i(\theta + \frac{4\pi}{3})} + \sqrt{\Delta_0} e^{i(\frac{2 \pi}{3}- \theta)}).\label{eq:a_k_2}
            \end{align}
            Simplify Eq.~(\ref{eq:a_k_2}) using Eq.~(\ref{eq:eular}), we obtain
            \begin{align}
                a_2 = - \frac{1}{3c_3}(c_2 +2 \sqrt{\Delta_0} \cos(\theta - \frac{ 2\pi}{3})).\label{eq:a_k_2_simplify}
            \end{align}
            Substituting $\theta$ (\ref{eq:theta}) and $c_2$ (\ref{eq:c_2}) into Eq.~(\ref{eq:a_k_2_simplify}) yields
            \begin{align}
               a_2 & = -\frac{1}{6}(-e^{2\epsilon} - 4e^{\epsilon} -5 \nonumber \\& \quad+ 2 \sqrt{\Delta_0} \cos(-\pi+\frac{1}{3}\arccos(-\frac{\Delta_1}{2 \Delta_0^\frac{3}{2}}))).
            \end{align}
\end{itemize}
The number of real and complex roots are determined by the discriminant of the cubic equation as follows:
\begin{align}
    \Delta = 18c_3c_2c_1c_0 - 4c_2^3c_0 + c_2^2c_1^2 - 4c_3c_1^3 - 27c_3^2c_0^2.\label{eq:Delta}
\end{align}
Substituting $c_3$~(\ref{eq:c_3}), $c_2$~(\ref{eq:c_2}), $c_1$~(\ref{eq:c_1}) and $c_0$~(\ref{eq:c_0}) into  $\Delta$~(\ref{eq:Delta}) yields
\begin{align}
   &\Delta = e^{2\epsilon}(e^\epsilon + 1)^2(e^{6\epsilon} + 30e^{5\epsilon} \nonumber \\& \quad+ 279e^{4\epsilon} + 580e^{3\epsilon} - 2385e^{2e^\epsilon} + 606e^\epsilon - 775).
\end{align}

If $\Delta = 0$, we get $\epsilon = \ln(\text{root of}~ \epsilon^6 + 30 \epsilon^5 + 279 \epsilon^4 + 580 \epsilon^3 - 2385 \epsilon^2 + 606 \epsilon - 775 ~\text{near}~ \epsilon = 1.87686) \approx 0.629598.$
\begin{itemize}
    \item If $0 < \epsilon < 0.629598,~ \Delta < 0$, the equation has one real root and two non-real complex conjugate roots.
    \item If $ \epsilon = 0.629598,~ \Delta = 0$, the equation has a multiple root all of its roots are real.
    \item If $\epsilon > 0.629598,~ \Delta > 0$, the equation has three distinct real roots.
\end{itemize}
From the simplified $f_1'(a)$ Eq.~(\ref{eq:derivative_1_simplify}), we know that the sign and roots of $f_1'(a)$ are same as its numerator, defined as follows:
\begin{align}
\nonumber  &g(a) := -2a^3 - a^2(-e^{2\epsilon} - 5 -4e^{\epsilon}) \\& \quad-a(7e^{\epsilon} - 4e^{2\epsilon}-e^{3\epsilon}) - (2e^{3\epsilon}-4e^{2\epsilon}).
\end{align}
Let $c = e^\epsilon$, we can change $g(a)$ to the following:
\begin{align}
  &g(a) = -2a^3 - a^2(-c^2 - 5 -4c) -a(7c - 4c^2-c^3) \nonumber \\& \quad - (2c^3-4c^2).
\end{align}

Case 1: If $0 < \epsilon < 0.629598$,
by observing $a_0, a_1, a_2$, it is obvious that $a_2$ is a real root. Fig.~\ref{fig:root_a_2} shows that $a_2 > 1$. Since $a_2$ is the only real value root, Eq.~(\ref{eq:derivative_1})~$f_1'(a) > 0$ if $0 < \epsilon \leq 0.629598$ and $f_1'(a) \leq 0$ if $\epsilon > 0.629598$, $f_1(a)$ monotonously increases  if $0 < \epsilon \leq 0.629598$. Therefore, we conclude that $a = 0$.

\begin{figure}[h]
\centering
  \includegraphics[width=0.45\textwidth]{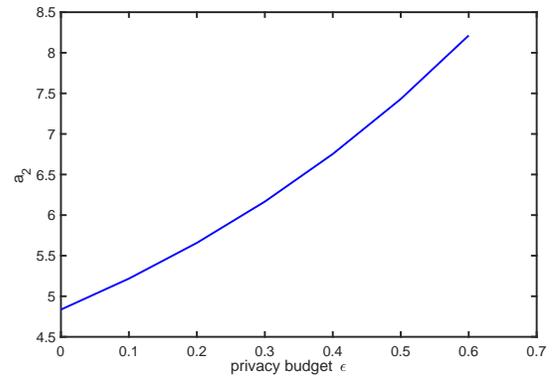} \vspace{-2pt}\caption{ $a_2 ~\text{if}~ \epsilon \in [ 0, 0.629598 ]$.}
 \label{fig:root_a_2}
\end{figure}

Case 2: If $0.629598 \leq \epsilon < \ln2$,
$\Delta \geq 0$, we get real roots. 
\begin{itemize}
    \item If $a = 0, ~g(0) =  - (2c^3-4c^2)$.
    \item If $a = 2, ~g(2) =  2(8c^2 + c +2) >0$.
    \item If $a = +\infty, ~\lim_{a \to \infty}g(a) = -\infty < 0$.
\end{itemize}
If there is a root is in $[0, \frac{e^\epsilon}{e^\epsilon + 2}]$, it means that $g(0) \leq 0$. By solving $g(0) = -(2c^3-4c^2) \leq 0$, we have $c \geq 2$ meaning $\epsilon \geq \ln2$. When $0.629598 < \epsilon < \ln2$, we have a $root \in (2, +\infty)$.

Based on the properties of cubic function, we have
\begin{align}
   & a_0a_1 + a_0a_2 + a_1a_2 = \frac{c_1}{c_3} = \frac{7e^\epsilon - 4e^{2\epsilon} - e^{3\epsilon}}{2},\label{eq:add_roots} \\& a_0 a_1 a_2 = -\frac{c_0}{c_3} = -\frac{2e^{3\epsilon}-4e^{2\epsilon}}{2}.\label{eq:multiply_roots}
\end{align}

\begin{figure}[h]
\centering
  \includegraphics[width=0.45\textwidth]{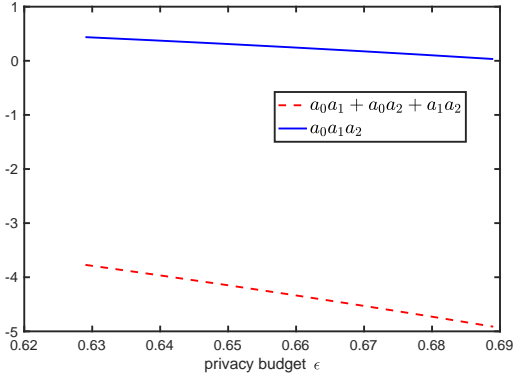} \vspace{-2pt}\caption{ $a_0a_1 + a_0a_2 + a_1a_2 ~\text{and}~  a_0 a_1 a_2~\text{if}~ \epsilon \in [0.629598, \ln2 ]$.}
 \label{fig:root_discriminant}
\end{figure}
 
Fig.~\ref{fig:root_discriminant} shows that $a_0a_1 + a_0a_2 + a_1a_2 < 0$ (Eq.~(\ref{eq:add_roots})) and $ a_0 a_1 a_2 >0$ (Eq.~(\ref{eq:multiply_roots})). Therefore, we can conclude that there are one positive real root and two negative real roots or a multiple root. Since two negative roots are out of the $a$'s domain, we only discuss  the positive root.

\begin{itemize}
    \item If $a \in [0, root), ~g(a) > 0$ meaning $f_1'(a) > 0$.
    \item If $a \in [root, +\infty), ~g(a) \leq 0$ meaning $f_1'(a) \leq 0$.
\end{itemize}
From above we know that $g(a) > 0$, so that $f_1(a)$ monotonously increases if $a \in [0, \frac{e^\epsilon}{e^\epsilon + 2}]$. Therefore, $a = 0$.

Case 3:  If $\ln2 \leq \epsilon \leq \ln5.53$, $\Delta >0$, there are three distinct real roots. Since $a_0a_1 + a_0a_2 + a_1a_2 < 0$ and $a_0 a_1 a_2 <0$, there are one negative root or three negative roots. If there are three negative roots, $a_0a_1 + a_0a_2 + a_1a_2 > 0$, there is only one negative root, $f_1'(a)$ have two positive  roots and one negative root.

\begin{itemize}
    \item If $a = 0, ~g(0) =  - (2c^3-4c^2) < 0$.
    \item If $a = 2, ~g(2) =  2(8c^2 + c +2) >0$.
    \item If $a = +\infty, ~\lim_{a \to \infty}g(a) = -\infty$.
\end{itemize}
From above results, we can deduce that there is one positive root in $(0, 2)$ defined as $root_1$, the other positive root is in $(2, +\infty)$ defined as $root_2$. Since $root_2 >1$, we only discuss $root_1$.
\begin{itemize}
    \item $a \in [0, root_1]$,~$g(a) \leq 0$.
    \item $a \in (root_1, root_2)$, ~$g(a) > 0$.
\end{itemize}
Therefore, if $g(\frac{c}{c+2}) \geq 0$, we can conclude that $root_1 \leq \frac{c}{c+2}$. The exact form of $g(\frac{c}{c+2})$ is
\begin{align}
    g \bigg(\frac{c}{c+2} \bigg) = \frac{c^2(c+1)^2(-c^2+3c+14)}{(c+2)^3}.\label{expression:g}
\end{align}
By solving $g(\frac{c}{c+2}) \geq 0$, we have $c \leq \ln \left(\frac{3 + \sqrt{65}}{2}\right) \approx 5.53$, i.e. $\epsilon \leq \ln5.53$.
From Fig.~\ref{fig:three_roots}, we can conclude that $a_1$ is the correct root, $a_0 < 0$ and $a_2 > 1$.
$a = a_1 = -\frac{1}{6}(-e^{2\epsilon} - 4e^{\epsilon} -5 + 2 \sqrt{\Delta_0} \cos(\frac{\pi}{3} + \frac{1}{3}\arccos(-\frac{\Delta_1}{2 \Delta_0^\frac{3}{2}})))$.

\begin{figure}[h]
\centering
  \includegraphics[width=0.45\textwidth]{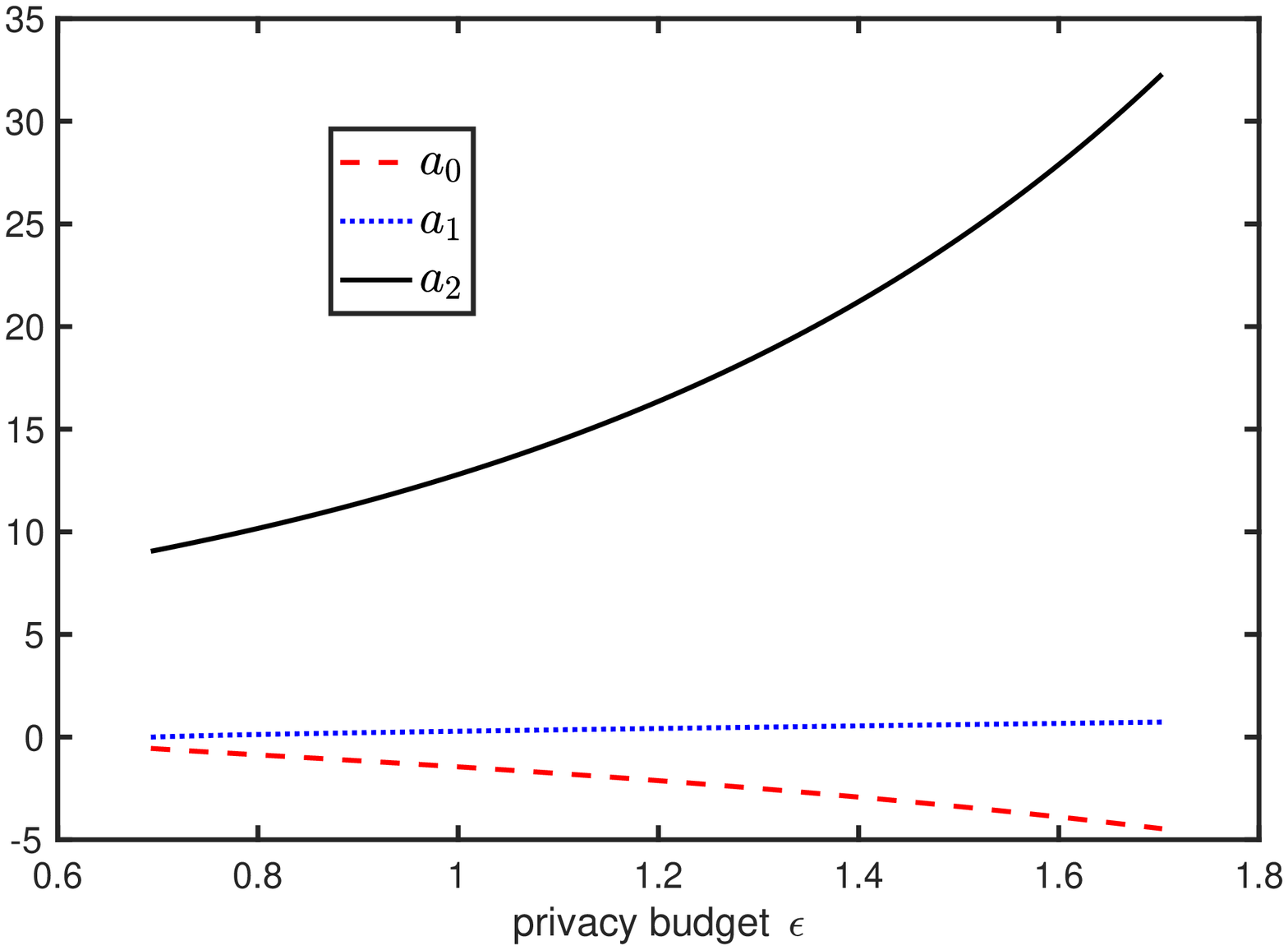} \vspace{-2pt}\caption{ $a_0, a_1 ~\text{and}~ a_2~\text{if}~ \epsilon \in [\ln2, \ln5.53]$.}
 \label{fig:three_roots}
\end{figure}

Case 4: If $\epsilon > \ln 5.53$, $\Delta >0$, there are three distinct real roots. From analysis in Case 3, we know that if $\epsilon > \ln 5.53$, $root_1 > \frac{c}{c+2}$ and $g(\frac{c}{c+2}) < 0$. We know that $g(a) \leq 0$ if $a \in [0, \frac{c}{c+2}]$, meaning $f_1'(a) < 0$, so that $f_1(a)$ monotonously decreases if $\epsilon > \ln 5.53$.  Since $a \in [0, \frac{e^\epsilon}{e^\epsilon + 2}]$, we have $a =  \frac{e^\epsilon}{e^\epsilon + 2} $.

Summarize above, we obtain the optimal $a$ which is named as $P_{0 \leftarrow 0}$ in the Eq.~(\ref{a_opt_1}).

\qeda

\subsection{\textbf{Proof of Lemma~\ref{lem:worst-case-variance-three-outputs}}}\label{appendix:lem-three-outputs-variance}

By substituting the optimal $P_{0 \leftarrow 0}$ of Eq.~(\ref{a_opt_1}) with $a$ in the $\max_{x \in [-1,1]}\textup{Var}[Y|x]$ of Eq.~(\ref{worst-val-4}), we obtain the worst-case noise variance of \texttt{Three-Outputs} as follows:
\begin{align}\label{worst-val-case1}
    &\hspace{-10pt}\min_{P_{0 \leftarrow 0}}\max_{x \in [-1,1]}\textup{Var}[Y|x]=\nonumber\\
    &\hspace{-10pt}\begin{cases}
        \hspace{-10pt}&\hspace{-10pt}\frac{(e^\epsilon + 1)^2}{(e^\epsilon - 1)^2},~\text{for}~\epsilon < \ln2, \\
        \hspace{-10pt}&\hspace{-10pt} \frac{(e^\epsilon +1)^2\cdot e^{2\epsilon}}{(e^\epsilon -1)^2}\bigg(\frac{1-P_{0 \leftarrow 0}}{(e^\epsilon - P_{0 \leftarrow 0})^2} + \frac{(e^\epsilon +1)^2\cdot P_{0 \leftarrow 0}^2}{4(e^\epsilon - P_{0 \leftarrow 0})^4} \bigg),~\text{for}~\ln2 \leq \epsilon \leq \ln5.53,\\& \text{where}~ P_{0 \leftarrow 0} =  -\frac{1}{6}(-e^{2\epsilon}-4e^\epsilon-5 \\& \quad +2\sqrt{\Delta_0} \cos(\frac{\pi}{3} + \frac{1}{3} \arccos(-\frac{\Delta_1}{2\Delta_0^\frac{3}{2}})))\\
        \hspace{-10pt}& \hspace{-10pt}\frac{(e^\epsilon + 2)(e^\epsilon+10)}{4(e^\epsilon-1)^2},~\text{ for}~\epsilon > \ln5.53.
    \end{cases}
\end{align}

\qeda

\subsection{\textbf{Proof of $\frac{2(e^\epsilon-a)^2(e^\epsilon-1)-ae^\epsilon(e^\epsilon+1)^2}{2(e^\epsilon-a)^2(e^\epsilon+t)-ae^\epsilon(e^\epsilon+1)^2} \leq \frac{e^\epsilon-1}{e^\epsilon+t}$~if~$t = \frac{\epsilon}{3}$}} \label{compare}
$\textbf{Proposition 1.} ~ \frac{2(e^\epsilon-a)^2(e^\epsilon-1)-ae^\epsilon(e^\epsilon+1)^2}{2(e^\epsilon-a)^2(e^\epsilon+t)-ae^\epsilon(e^\epsilon+1)^2} \leq \frac{e^\epsilon-1}{e^\epsilon+t}$~for~ $\epsilon > 0$.

Define
\begin{align}
 & f := \frac{2(e^\epsilon-a)^2(e^\epsilon-1)-ae^\epsilon(e^\epsilon+1)^2}{2(e^\epsilon-a)^2(e^\epsilon+\frac{\epsilon}{3})-ae^\epsilon(e^\epsilon+1)^2}  - \frac{e^\epsilon-1}{e^\epsilon+ \frac{\epsilon}{3}} \nonumber \\
 &= \frac{ae^\epsilon(e^\epsilon+1)^2(t+1)}{(ae^\epsilon(e^\epsilon+1)^2-2(e^\epsilon-a)^2(e^\epsilon+ \frac{\epsilon}{3}))(e^\epsilon + \frac{\epsilon}{3})}, \label{eq:f}  \nonumber \\
\end{align}
and 
\begin{align}
    h := 2(e^\epsilon-a)^2(e^\epsilon+ \frac{\epsilon}{3}) - ae^\epsilon(e^\epsilon+1)^2.\label{eq:h} 
\end{align}
When $\epsilon >0$, we  have $e^\epsilon + \frac{\epsilon}{3} >0$ and  $ae^\epsilon(e^\epsilon+1)^2(\frac{\epsilon}{3}+1) > 0$ (the numerator of Eq.~(\ref{eq:f})). 
\begin{itemize}
    \item If~ $0< \epsilon < \ln2$, we have ~$a = 0$ and $~h
    = -2e^{2\epsilon}(e^\epsilon + e^{\epsilon/3}) < 0 $. Therefore, we conclude that Eq.~(\ref{eq:f})~$ < 0.$ \nonumber
    \item   If $\ln2 \leq \epsilon \leq \ln5.53$, we have $a = -\frac{1}{6}(-e^{2\epsilon}-4e^\epsilon-5 +2\sqrt{\Delta_0} \cos(\frac{\pi}{3} + \frac{1}{3} \arccos(-\frac{\Delta_1}{2\Delta_0^\frac{3}{2}}))).$~Fig.~$\ref{condition-1}$ shows that $h := 2(e^\epsilon-a)^2(e^\epsilon+t) - ae^\epsilon(e^\epsilon+1)^2 < 0 $, so that we obtain Eq.~(\ref{eq:f})~$ < 0$.\\
    \item  If $\epsilon > \ln 5.53$, we have $a = \frac{e^\epsilon}{e^\epsilon + 2}
    $ and $h = \frac{e^{2\epsilon}(e^\epsilon+1)^2(-2+e^\epsilon + 2e^{\epsilon/3})}{(e^\epsilon + 2)^2} \nonumber$.
   Since~ $e^\epsilon$ and $e^{\epsilon/3} > 1$, we obtain $-2+e^\epsilon + 2e^{\epsilon/3}> 0$ and $h > 0 $. Hence, we conclude that Eq.~(\ref{eq:f})$~< 0.$ \nonumber
\end{itemize}
Based on above analysis, we have Eq.~(\ref{eq:f}))$~< 0$ when~ $\epsilon > 0$, meaning that $\frac{2(e^\epsilon-a)^2(e^\epsilon-1)-ae^\epsilon(e^\epsilon+1)^2}{2(e^\epsilon-a)^2(e^\epsilon+t)-ae^\epsilon(e^\epsilon+1)^2} \leq \frac{e^\epsilon-1}{e^\epsilon+t}$ when~ $\epsilon > 0$.

\qeda

\subsection{\textbf{Proof of ~$2(e^\epsilon-a)^2(e^\epsilon+t) - ae^\epsilon(e^\epsilon+1)^2 > 0$}.}\label{condition_1_proof}

From values of $\epsilon$, we have the following cases:
\begin{itemize}
    \item If $0<\epsilon \leq \ln 5.53$, we have $2(e^\epsilon-a)^2(e^\epsilon+t) - ae^\epsilon(e^\epsilon+1)^2 > 0$ referring to Fig.~\ref{condition-1}.
    \item If $\epsilon > \ln 5.53$, we have $a = \frac{e^\epsilon}{e^\epsilon+2}$ and $t = e^\frac{\epsilon}{3}$, so that
    \begin{align}
       & 2(e^\epsilon-a)^2(e^\epsilon+t)-ae^\epsilon(e^\epsilon+1)^2  \nonumber \\
       & \quad =\frac{e^{2\epsilon}(e^\epsilon+1)^2(e^\epsilon-2+2e^\frac{\epsilon}{3})}{(e^\epsilon+2)^2}.
    \end{align}\label{eq-1}
    Since $\frac{e^{2\epsilon}(e^\epsilon+1)^2(-e^\epsilon+2-2e^\frac{\epsilon}{3})}{(e^\epsilon+2)^2} > 0$ if $\epsilon > 0$, we have $2(e^\epsilon-a)^2(e^\epsilon+t) - ae^\epsilon(e^\epsilon+1)^2 > 0$~ if~$\epsilon > \ln 5.53.$
\end{itemize}

\begin{figure}[h]
\centering 
  \includegraphics[width=0.45\textwidth]{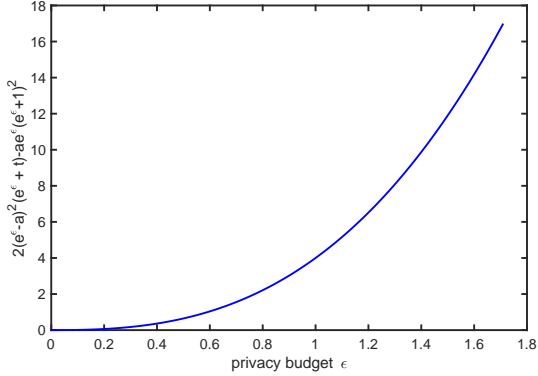} \vspace{-2pt}\caption{Value $2(e^\epsilon-a)^2(e^\epsilon+t)-ae^\epsilon(e^\epsilon+1)^2 > 0$, when $\epsilon \in ( 0, \ln 5.53]$.}
 \label{condition-1}
\end{figure}

Thus, we conclude that $2(e^\epsilon-a)^2(e^\epsilon+t) - ae^\epsilon(e^\epsilon+1)^2 > 0$ if $\epsilon > 0$. 
\qeda

\subsection{\textbf{Proving Lemma~\ref{PM-lemma-1}}}\label{appendix:lem-1}

From Eq.~(\ref{eq:pdf-ldp-a}) and Eq.~(\ref{eq:pdf-ldp-b}), for any $Y \in [-A, A]$ and any two input values $x_1, x_2 \in [-1, 1]$, we have $\frac{pdf(Y|x_1)}{pdf(Y|x_2)} \leq \frac{c}{d} = \exp(\epsilon)$.
Thus, Algorithm \ref{algo:PM-OPT} satisfies local differential privacy.
For notational simplicity, with a fixed $\epsilon$ below, we will write $L(\epsilon,x, t)$ and $R(\epsilon,x, t)$ as $L_x$ and $R_x$. Based on proper probability distribution, we have
\begin{align}
    \int_{-A}^A \bfu{Y=y|x} \, \text{d} y = c(R_x-L_x) + d[2A-(R_x-L_x)] =1.\label{eq:prob-distribution}
\end{align}
In addition,
\begin{align}
   & \mathbb{E}[Y=y|x] = \int_{-A}^A y \cdot pdf(Y=y|x) \nonumber \, \text{d} y \\& = c(R_x-L_x) + d[2A-(R_x-L_x)] \nonumber \\& = \frac{d}{2}\cdot ({L_x}^2-A^2) + \frac{c}{2} ({R_x}^2-{L_x}^2) + \frac{d}{2}\cdot (A^2-{R_x}^2) \nonumber \\&  = x. \label{eq:constraint-unbiased}
\end{align}
By solving above Eq.~(\ref{eq:prob-distribution}) and Eq.~(\ref{eq:constraint-unbiased}), we have
\begin{align} \label{eqnitem4plusx}
\begin{cases}
{L_x} & =   \frac{x }{1-2Ad }-\frac{1-2Ad}{2(c-d)}   ,   \\ R_x &  = \frac{x}{1-2Ad }+\frac{1-2Ad}{2(c-d)}  .
\end{cases}
\end{align}
With $-A \leq y \leq A$, the constraint $-A \leq {L_x}  < R_x \leq A$ for any $-1 \leq x \leq 1$ in Eq.~(\ref{eq:pdf-ldp}), Eq.~(\ref{eq:opt-pdf-d}) and Eq.~(\ref{eq:opt-pdf-A}) implies
\begin{align}
 Ad & < \frac{1}{2},
 \label{eqnitem7} \\
 A   &  \geq \frac{1}{1-2Ad }+\frac{1-2Ad}{2(c-d)}. \label{eqnitem6}
\end{align}
For notational simplicity, we define $\alpha$ and $\xi$ as
\begin{align}
 \alpha &:= Ad , \label{eqnitem8} \\ \xi  &:= \frac{c-d}{d},\label{eqnitem9}
\end{align}
where it is clear under privacy parameter $\epsilon$ that
\begin{align}
\xi = e^{\epsilon}-1. \label{eqnitem10}
\end{align}
Applying Eq.~(\ref{eqnitem8}) and Eq.~(\ref{eqnitem9}) to Inequality~(\ref{eqnitem7}) and  Eq.~(\ref{eqnitem6}), we obtain
\begin{align}   \frac{\alpha}{d}   &\geq \frac{1}{1-2\alpha}+\frac{1-2\alpha}{2 \xi d}  , \label{eqnitem12}  \\ \alpha  & < \frac{1}{2}. \label{eqnitem14}
\end{align}
The condition Eq.~(\ref{eqnitem12}) induces $d \leq \frac{(2\xi+4)\alpha-(4+4\xi)\alpha^2-1}{2\xi}  = \frac{[(2\xi+2)\alpha-1](1-2\alpha)}{2\xi}.$ In view of $\frac{1}{2(\xi+1)} = \frac{1}{2e^\epsilon}<\alpha<\frac{1}{2}$, we define $t$ satisfying $0<t<\infty$ such that
$\alpha =  \frac{t+1}{2(t+e^\epsilon)} $. Note that $\lim_{t \to 0}  \frac{t+1}{2(t+e^\epsilon)} = \frac{1}{2e^\epsilon}$, $\lim_{t \to \infty}  \frac{t+1}{2(t+e^\epsilon)} = \frac{1}{2}$ and $d =   \frac{[(2\xi+2)\alpha-1](1-2\alpha)}{2\xi}  = \frac{1}{2\xi}   \cdot \frac{\xi t}{t+1+\xi} \cdot \frac{\xi}{t+1+\xi} = \frac{t (e^{\epsilon}-1)}{2(t+e^\epsilon)^2}$.

Applying $\alpha$, $d$ and $\xi$ to Eq.~(\ref{eqnitem4plusx}), we have
\begin{align}
\hspace{-10pt}
\begin{cases}
L_x \hspace{-10pt} & = \frac{x }{1-2\alpha }-  \frac{1-2\alpha }{2\xi d} = x \cdot \frac{e^{\epsilon}+t}{e^{\epsilon}-1} -\frac{e^{\epsilon}+t}{t (e^{\epsilon}-1)} = \frac{(e^{\epsilon}+t)(xt-1)}{t (e^{\epsilon}-1)},   \\ R_x \hspace{-10pt} & = \frac{x }{1-2\alpha } + \frac{1-2\alpha }{2\xi d} = x \cdot \frac{e^{\epsilon}+t}{e^{\epsilon}-1} + \frac{e^{\epsilon}+t}{t (e^{\epsilon}-1)} = \frac{(e^{\epsilon}+t)(xt+1)}{t (e^{\epsilon}-1)}.
\end{cases}
\end{align}
Furthermore, the variance of $Y$ is 
\begin{align} \label{eqnitem3}
\hspace{-20pt}&\textup{Var}[Y|x]  = \mathbb{E}[Y^2|~x] - (\mathbb{E}[Y|x])^2 \nonumber \\& \nonumber  =\int_{-A}^A y^2 \bfu{Y=y|x} \, \text{d} y  - x^2~ \\&  \nonumber  =  \int_{-A}^{L_x} d y^2 \, \text{d} y + \int_{L_x}^{R_x} c y^2 \, \text{d} y + \int_{R_x}^{A} d y^2 \, \text{d} y   - x^2   \nonumber \\ & = \nonumber \max_{x \in [-1, 1]}   \frac{d}{3} [{L_x}^3-(-A)^3] + \frac{c}{3} ({R_x}^3-{L_x}^3) + \\& \quad  \frac{d}{3} (A^3-{R_x}^3)  - x^2   \nonumber \\ & =  \frac{2d}{3} A^3 + \frac{(c-d)}{3} ({R_x}^3-{L_x}^3)    - x^2  .  
\end{align}
Substituting Eq.~(\ref{eqnitem4plusx}) into Eq.~(\ref{eqnitem3}) yields
\begin{align}\label{PM-OPT_var}
  &\textup{Var}[Y|x] = \frac{2d}{3} A^3  + \frac{(c-d)}{3} \cdot \nonumber \\ & \nonumber \hspace{-5pt} \left[\left(\frac{x }{1-2Ad }+\frac{1-2Ad}{2(c-d)}\right)^3 -\left(\frac{x  }{1-2Ad } \hspace{-2pt}-\hspace{-2pt}\frac{1-2Ad}{2(c-d)}\right)^3\right] \\ & \quad - x^2  \nonumber \\  & = \frac{2d}{3} A^3 + \frac{(c-d)}{3} \cdot \nonumber \\ & \quad \left\{6\left(\frac{x  }{1-2Ad }\right)^2 \times \frac{1-2Ad}{2(c-d)}+2\left[\frac{1-2Ad}{2(c-d)}\right]^3\right\}   - x^2      \nonumber \\ &   = \left( \frac{1}{1-2Ad}  - 1 \right)x^2     +  \frac{2d}{3} A^3  + \frac{(1-2Ad)^3}{12(c-d)^2}.
\end{align}
Substituting $1-2\alpha = \frac{e^{\epsilon}-1}{e^{\epsilon}+t}$,
$d = \frac{t (e^{\epsilon}-1)}{2(t+e^\epsilon)^2}$,
$\xi  = \frac{c-d}{d}, \xi = e^{\epsilon}-1$, and $\alpha =  \frac{t+1}{2(t+e^\epsilon)}$ into Eq.~(\ref{PM-OPT_var}) yields
\begin{align}
   \hspace{-5pt} \textup{Var}[Y|x] =  \frac{t+1}{e^\epsilon-1}x^2 + \frac{(t+e^\epsilon)\big ((t+1)^3 + e^\epsilon -1 \big)}{3t^2(e^\epsilon-1)^2}.\label{eq:var-ldp-1}
\end{align}
\qeda

\subsection{\textbf{Calculate the probability of a variable $Y$ falling in the interval $[L(\epsilon,x,e^{\epsilon/3}), R(\epsilon,x,e^{\epsilon/3})]$.}}\label{appendix:y-pdf}

By replacing $t$ in Eq.(\ref{eq:pdf-ldp}) of \texttt{PM-OPT} with $e^{\epsilon/3}$, we obtain the probability as follows:

\begin{align}
    &\bp{L(\epsilon,x,e^{\epsilon/3}) \leq Y \leq R(\epsilon,x,e^{\epsilon/3})} \nonumber
    \\&= \int_{L(\epsilon,x,e^{\epsilon/3})}^{R(\epsilon,x,e^{\epsilon/3})} c~\textup{d}Y \nonumber
    \\&= \int_{\frac{(e^{\epsilon}+e^{\epsilon/3})(xe^{\epsilon/3}-1)}{e^{\epsilon/3} (e^{\epsilon}-1)}}^{\frac{(e^{\epsilon}+e^{\epsilon/3})(xe^{\epsilon/3}+1)}{e^{\epsilon/3} (e^{\epsilon}-1)}} \frac{e^{\epsilon} t (e^{\epsilon}-1)}{2(t+e^\epsilon)^2}~\textup{d}Y \nonumber
    \\& = \frac{e^\epsilon}{e^{\epsilon/3}+e^\epsilon}. \nonumber
\end{align}

\subsection{\textbf{Proof of Lemma~\ref{lem:proof-discretize-probability}}}\label{appendix:proof-of-discreization-probability}

The expression of these two probabilities in Eq.~(\ref{eq:perturbation-probabilities}) can be solved from the following:
\begin{subequations}
    \begin{empheq}[left=\hspace{-15pt} {\empheqlbrace}]{align}
        &\text{proper distribution so that } \nonumber \\ &\bp{Z=\frac{kC}{m}~|~Y=y} \nonumber \\ & \quad+ \bp{Z=\frac{(k+1)C}{m}~|~Y=y}=1,\label{eq:prob-dis}\\[8pt]
    &\bE{Z~|~Y=y} = y \text{ so that } \nonumber\\  &\left(\begin{array}{l}
        \frac{kC}{m} \times \bp{Z=\frac{kC}{m}~|~Y=y}  \\ + \frac{(k+1)C}{m} \times \bp{Z=\frac{(k+1)C}{m}~|~Y=y} \end{array}\right)=y.\label{eq:expectation-dis}
    \end{empheq}
\end{subequations}
Summarizing \ding{172} and \ding{173},    with $k:= \lfloor \frac{ym}{C} \rfloor$, we have
\begin{align}
 \bp{Z=z~|~Y=y}  = \begin{cases}
k+1- \frac{ym}{C},  &\text{if} ~z=\frac{kC}{m}, \\[8pt]   \frac{ym}{C}-k, &\text{if} ~z=\frac{(k+1)C}{m}.
    \end{cases} \label{eq-Z-given-Y}
\end{align}

In the perturbation step, the distribution of $Y$ given the input $x$ is given by
\begin{align}
  \hspace{-8pt}  \bfu{Y=y~|~x} = \begin{cases}
    p_1, ~\text{if}~ y \in [L(x), R(x)], \\
    p_2, ~\text{if}~y \in [-C, L(x)) \cup (R(x), C].
    \end{cases}\label{eq:perturbation-probabilities}
\end{align} 

Hence,
\begin{align}
 &   \bp{Z=z~|~x} \nonumber \\&  = \int_{y} \bp{Z=z~|~x \text{ and } Y=y } \bfu{Y=y~|~x} \,\text{d}y \nonumber \\& =  \int_{y} \bp{Z=z~|~Y=y} \bfu{Y=y~|~x} \,\text{d}y.
\end{align}
In addition, our mechanisms are unbiased, such that
\begin{align}
    \mathbb{E}[Y~|~x] =  \int_{y} y \times \bfu{Y=y~|~x} \text{d}y = x.\label{eq:expection-y}
\end{align}
Therefore, we obtain
\begin{align}
    & \mathbb{E}[Z~|~x] = \sum_{z} z \times \bp{Z=z~|~x} \nonumber \\& = \sum_{z} z \times \int_{y} \bp{Z=z~|~Y=y} \bfu{Y=y~|~x} \text{d}y  \nonumber \\& =  \int_{y} \left( \sum_{z} z \times  \bp{Z=z~|~Y=y} \right ) \bfu{Y=y~|~x} \text{d}y  \nonumber \\& = \int_{y} y \times \bfu{Y=y~|~x} \text{d}y = x.\label{eq:expection-y-to-x}
\end{align}
\qeda

\subsection{\textbf{Proof of Lemma~\ref{lem:discretization-worst-case-var}}}\label{app:proof-of-lem-8}

To prove $\textup{Var}[Z | X = x] \geq \textup{Var}[Y | X = x]$, it is equivalent to prove 
\begin{align}
    &\bE{Z^2 | X = x} - (\bE{Z | X = x})^2 \nonumber\\& \quad \geq \bE {Y^2 | X = x} - (\bE {Y | X = x})^2. \nonumber
\end{align}
Since $\mathcal{M}_1$ and $\mathcal{M}_2$ are unbiased, we have $\bE{Z|X = x} = \bE{Y | X = x } = x$. Hence, it is sufficient to prove 
\begin{align}
    \bE{Z^2 | X = x} \geq \bE {Y^2|X = x}.\label{ineq:expectation-z-square-to-y-square}
\end{align} We can derive that
\begin{align}
    &\bE{Z^2 | X = x} \nonumber\\ &= \sum_z z^2 \cdot  \bp{Z = z | X=x} \nonumber \\
    & = \sum_z z^2 \int_y \bp{Z= z|Y = y} \cdot  \bfu{
    Y = y | X = x} \textup{d}y \nonumber \\
    & = \int_y \sum_z z^2 \bp{Z= z|Y = y} \cdot \bfu{Y = y | X = x} \textup{d}y \nonumber \\
    & = \int_y \bE{Z^2|Y = y} \cdot \bfu{Y = y | X = x} \textup{d}y,\label{eq:expectation-z-square}
\end{align}
and 
\begin{align}
        \bE{Y^2 | X = x} = \int_y  y^2 \cdot \bfu{Y = y | X = x} \textup{d}y.\label{eq:expectation-y-square}
\end{align}

To prove Inequality~(\ref{ineq:expectation-z-square-to-y-square}), because of Eq.~(\ref{eq:expectation-z-square}) and Eq.~(\ref{eq:expectation-y-square}), it is sufficient to prove 
\begin{align}
    \bE{Z^2|Y = y} \geq y^2, ~\forall y \in \text{Range}(Y). \label{Ineq:z-square-greater-y-square}
\end{align}
After getting the intermediate output $y$ from $\mathcal{M}_1$, we may discretize the intermediate output $y$ into $z_1$ with probability $p_1$ and $z_2$ with probability $p_2$. Hereby,
\begin{align}
    p_1 + p_2 = 1. \nonumber
\end{align}
Mechanism $\mathcal{M}_2$ is unbiased, so that $\bE{Z | Y = y} = y$, and we have
\begin{align}
        p_1 \cdot z_1 + p_2 \cdot z_2 = y. \nonumber
\end{align}

According to Cauchy–Schwarz inequality, we have 
\begin{align}
     \bE{Z^2|Y = y} &= p_1 \cdot z_1^2 + p_2 \cdot z_2^2 \nonumber\\&= [(\sqrt{p_1})^2 + (\sqrt{p_2})^2][(\sqrt{p_1} z_1)^2 + (\sqrt{p_2} z_2)^2]~\nonumber \\ & \geq (p_1 \cdot z_1 + p_2 \cdot z_2)^2 = y^2.~\nonumber
\end{align}
Thus, we get Inequality~(\ref{ineq:expectation-z-square-to-y-square}) and Inequality~(\ref{Ineq:z-square-greater-y-square}).
\qeda

\subsection{\textbf{Proof of Lemma~\ref{lem:beta}}}\label{proof_lemma_varH}
The~$\max_{x \in [-1,1]} \textup{Var}_{\mathcal{H}}[Y|x]$~is minimized when
    \begin{align}
        &\beta =  \label{eq:beta-value}\\
        &\begin{cases}
            0 ,~~\text{If}~~ 0 < \epsilon < \epsilon^*, \nonumber \\
            \beta_1 = \frac{2(e^\epsilon-a)^2(e^\epsilon-1)-ae^\epsilon(e^\epsilon+1)^2}{2(e^\epsilon-a)^2(e^\epsilon+t)-ae^\epsilon(e^\epsilon+1)^2}, ~~\text{If}~~ \epsilon^* \leq \epsilon < \ln2, \\
            \beta_3 = \frac{-\sqrt{\frac{B}{A}} + e^\epsilon -1}{e^\epsilon +e^{\epsilon/3}},~\text{If}~  \epsilon \geq \ln2,
        \end{cases} 
    \end{align}
Where \begin{align}
      \epsilon^*  &:= 3 \ln(\text{root of}~ 3 x^5 - 2 x^3 + 3 x^2 - 5 x - 3~\text{near}~x = 1.22588) \nonumber \\&\approx 0.610986.\label{eq:eps_star}
  \end{align}

\textbf{Proof.} If $x = x^* = \frac{(\beta-1)ae^\epsilon(e^\epsilon+1)^2}{2(e^\epsilon-a)^2(\beta (e^\epsilon+ t)-e^\epsilon+1)}$, we have variance of $Y$ as follows:
\begin{align}
 &\textup{Var}_{\mathcal{H}}[Y|x^*]\nonumber\\& = \nonumber (\beta \frac{t+1}{e^\epsilon-1} + \beta - 1) \bigg(\frac{(\beta-1)ae^\epsilon(e^\epsilon+1)^2}{2(e^\epsilon-a)^2(\beta (e^\epsilon+ t)-e^\epsilon+1)}\bigg)^2\\&\nonumber + (1-\beta)  \frac{ae^\epsilon(e^\epsilon+1)^2}{(e^\epsilon-1)(e^\epsilon-a)^2} \bigg (\frac{(\beta-1)ae^\epsilon(e^\epsilon+1)^2}{2(e^\epsilon-a)^2(\beta (e^\epsilon+ t)-e^\epsilon+1)}\bigg ) \\& + \bigg ( \frac{(t+e^\epsilon)((t+1)^3+e^\epsilon-1)}{3t^2(e^\epsilon-1)^2} \beta \nonumber \\& +(1- \beta)(1-a) \frac{e^{2\epsilon}(e^\epsilon+1)^2}{(e^\epsilon-1)^2(e^\epsilon-a)^2} \bigg ). \label{varH_*}
\end{align}

Let $\gamma := \beta(e^\epsilon+t)-e^\epsilon + 1$ and $c: = e^\epsilon$,~we can transform $\textup{Var}_{\mathcal{H}}[Y|x^*]$ ~Eq.~(\ref{varH_*})~to the following: \\
\begin{align}
   \nonumber & \gamma \bigg (\frac{a^2c^2(c+1)^4}{4(c+t)^2(c-a)^4(c-1)}- \frac{a^2c^2(c+1)^4}{2(c+t)^2(c-1)(c-a)^4} \\& \nonumber + \frac{(t+1)^3+c-1}{3t^2(c-1)^2}-\frac{(1-a)c^2(c+1)^2}{(c+t)(c-1)^2(c-a)^2} \bigg ) \\& \nonumber + \frac{1}{\gamma} \bigg (\frac{(1+t)^2a^2c^2(c+1)^4}{4(c+t)^2(c-a)^4(c-1)} - \frac{(1+t)^2a^2c^2(c+1)^4}{2(c+t)^2(c-1)(c-a)^4} \bigg ) \\& \nonumber - \frac{(1+t)a^2c^2(c+1)^4}{2(c+t)^2(c-a)^4(c-1)} + \frac{(1+t)a^2c^2(c+1)^4}{(c+t)^2(c-1)(c-a)^4} \\&  + \frac{(t+1)^3+c-1}{3t^2(c-1)} + \frac{(1+t)(1-a)c^2(c+1)^2}{(c+t)(c-1)^2(c-a)^2}.\label{eq:var_gamma}
\end{align}
$\text{Set coefficient of} ~\gamma~\text{as}~ A: $
\begin{align}
  \nonumber  A &:= \frac{a^2c^2(c+1)^4}{4(c+t)^2(c-a)^4(c-1)}- \frac{a^2c^2(c+1)^4}{2(c+t)^2(c-1)(c-a)^4} \\&  + \frac{(t+1)^3+c-1}{3t^2(c-1)^2}-\frac{(1-a)c^2(c+1)^2}{(c+t)(c-1)^2(c-a)^2}.\label{eq:A}
\end{align}
$\text{Set coefficient of} ~\frac{1}{\gamma}~\text{ as}~ B: $
\begin{align}
    B &:= -\frac{(1+t)^2a^2c^2(c+1)^4}{4(c+t)^2(c-a)^4(c-1)}.\label{eq:B}
\end{align}
Set C as :
\begin{align}
    \nonumber C &:= - \frac{(1+t)a^2c^2(c+1)^4}{2(c+t)^2(c-a)^4(c-1)} + \frac{(1+t)a^2c^2(c+1)^4}{(c+t)^2(c-1)(c-a)^4} \\& + \frac{(t+1)^3+c-1}{3t^2(c-1)} + \frac{(1+t)(1-a)c^2(c+1)^2}{(c+t)(c-1)^2(c-a)^2}.\label{eq:C}
\end{align}
Since $\gamma$ monotonically increases with $\beta$ in the domain $\beta \in (0, \beta_1)$, the minimum $\gamma$ is $\gamma_1 := -e^\epsilon + 1$ at $\beta = 0$, maximum $\gamma$ is $\gamma_2 := 0$ at $\beta = \beta_1$.
\begin{itemize}
    \item If $0 < \epsilon < \ln2$, $a = 0$, we have
    \begin{align}
      \nonumber &A  = \frac{(t+1)^3+c-1}{3t^2(c-1)^2}-\frac{(c+1)^2}{(c+t)(c-1)^2}, \\
      \nonumber &B  = 0, \\
      \nonumber &\textup{Var}_{\mathcal{H}} [Y|x^*]  = A \gamma + C~\text{is a linear function.}
    \end{align}

    -If $\beta \in (0, \beta_1)$, Appendix~\ref{value:A} proves:
    \begin{itemize}
        \item[a)] $A>0$, if $0< \epsilon < 0.610986$.
        \item[b)] $A=0$, if $\epsilon = 0.610986$.
        \item[c)] $A<0$, if $0.610986 < \epsilon < \beta_1$.
    \end{itemize}
    Therefore, $\min_\beta \max_{x \in [-1,1]} \textup{Var}_{\mathcal{H}}[Y|x]$ is at:
    \begin{itemize}
        \item[a)] $\beta = 0$, if $0< \epsilon < 0.610986$.
        \item[b)] $\beta = \beta_1$, if $0.610986 \leq \epsilon <\ln2$.
    \end{itemize}

    -If $\beta \in [\beta_1, \beta_2]$,  $\textup{slope}_1 = \frac{t+1}{e^\epsilon-1} + 1 + \frac{(t+e^\epsilon)((t+1)^3+e^\epsilon-1)}{3t^2(e^\epsilon-1)^2}-\frac{(e^\epsilon+1)^2}{(e^\epsilon-1)^2}$,~$\textup{slope}_2 =\frac{(t+e^\epsilon)((t+1)^3+e^\epsilon-1)}{3t^2(e^\epsilon-1)^2} - \frac{(e^\epsilon+1)^2}{(e^\epsilon-1)^2}$. \\
     Fig.~\ref{slope_1} proves that $\textup{slope}_1 > 0$, when $\epsilon \in [0, \ln2]$.
     
    \begin{figure}[h]
    \centering
      \includegraphics[width=0.45\textwidth]{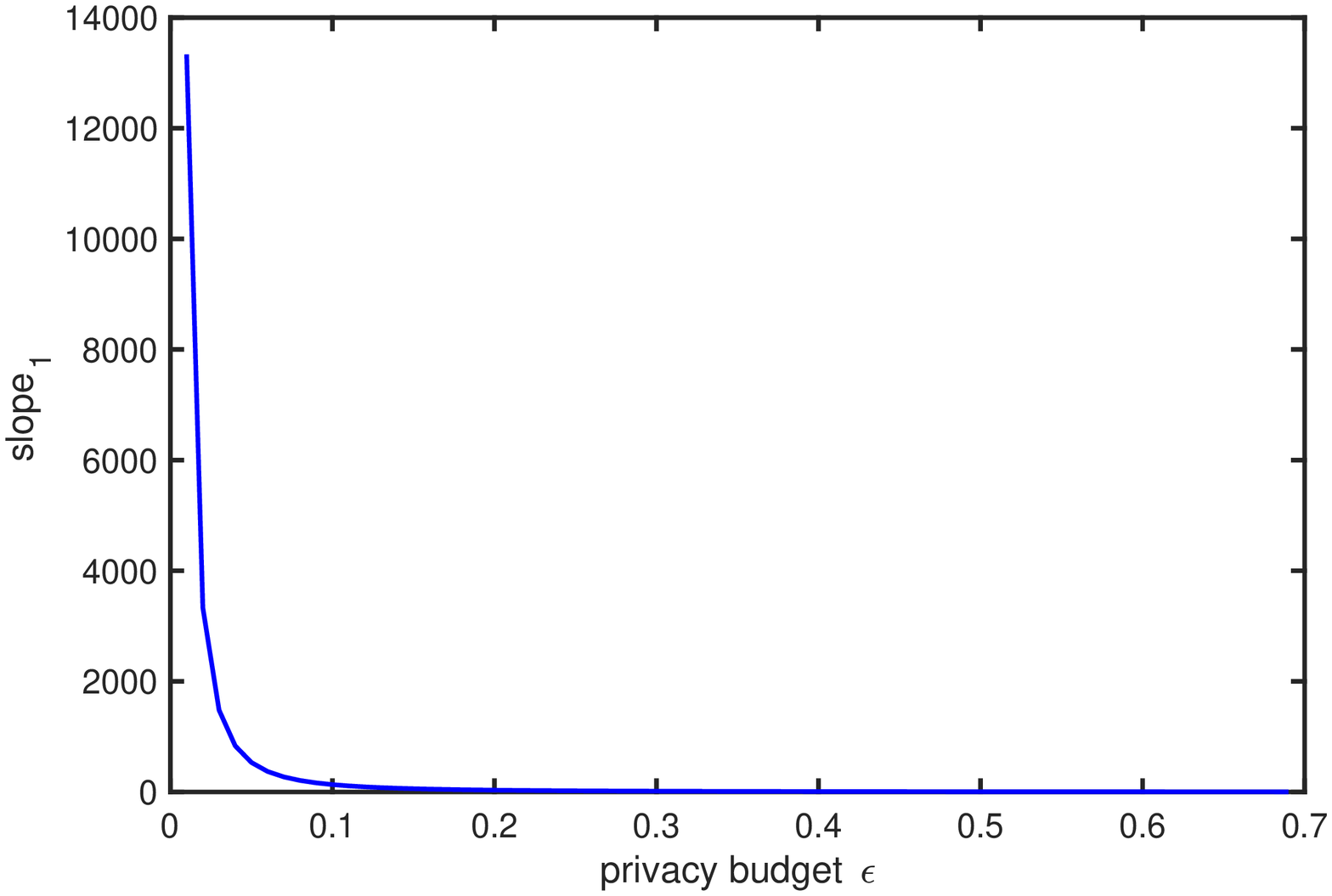} \vspace{-2pt}\caption{$\textup{slope}_1 ~\text{when}~ \epsilon \in [0, \ln2]$.}
     \label{slope_1}
    \end{figure}

     When $a = 0$, 
     $\beta_1$~($\ref{de:beta_1}$) $= \beta_{intersection}$~($\ref{de:intersection}$)  $= \frac{e^\epsilon -1}{t + e^\epsilon}$, the intersection of $\textup{Var}_{\mathcal{H}} [Y|0]$ and $\textup{Var}_{\mathcal{H}} [Y|1]$~is at $\beta_1$.

     When $\textup{slope}_2 = 0$, we have root at $\epsilon = 3 \ln(\text{root of}~ 3 x^5 - 2 x^3 + 3 x^2 - 5 x - 3~\text{near}~x = 1.22588) \approx 0.610986$.
     
    From Fig.~\ref{slope_2}, we have:
     \begin{itemize}
         \item[(1)] If $0 < \epsilon < 0.610986,~\textup{slope}_2 > 0$.
         \item[(2)] If $\epsilon = 0.610986,~\textup{slope}_2 = 0$.
         \item[(3)] If $\epsilon > 0.610986,~\textup{slope}_2 < 0$.
     \end{itemize}
     Based on previous analysis, we have:
     \begin{itemize}
         \item[(1)] If $0 < \epsilon < 0.610986$, $\min_\beta \max_{x \in [-1,1]} \textup{Var}_{\mathcal{H}}[Y|x] = \max_{x \in [-1,1]}\textup{Var}_{\mathcal{H}} [Y|x, \beta = \beta_1]$.
         \item[(2)] If $\epsilon = 0.610986$, $\min_\beta \max_{x \in [-1,1]} \textup{Var}_{\mathcal{H}}[Y|x] = \textup{Var}_{\mathcal{H}} [Y|1, \beta = \beta_1]$.
         \item[(3)] If $\epsilon > 0.610986$, $\min_\beta \max_{x \in [-1,1]} \textup{Var}_{\mathcal{H}}[Y|x] = \textup{Var}_{\mathcal{H}} [Y|1, \beta = \beta_1]$.
     \end{itemize}

    \begin{figure}[h]
    \centering
      \includegraphics[width=0.45\textwidth]{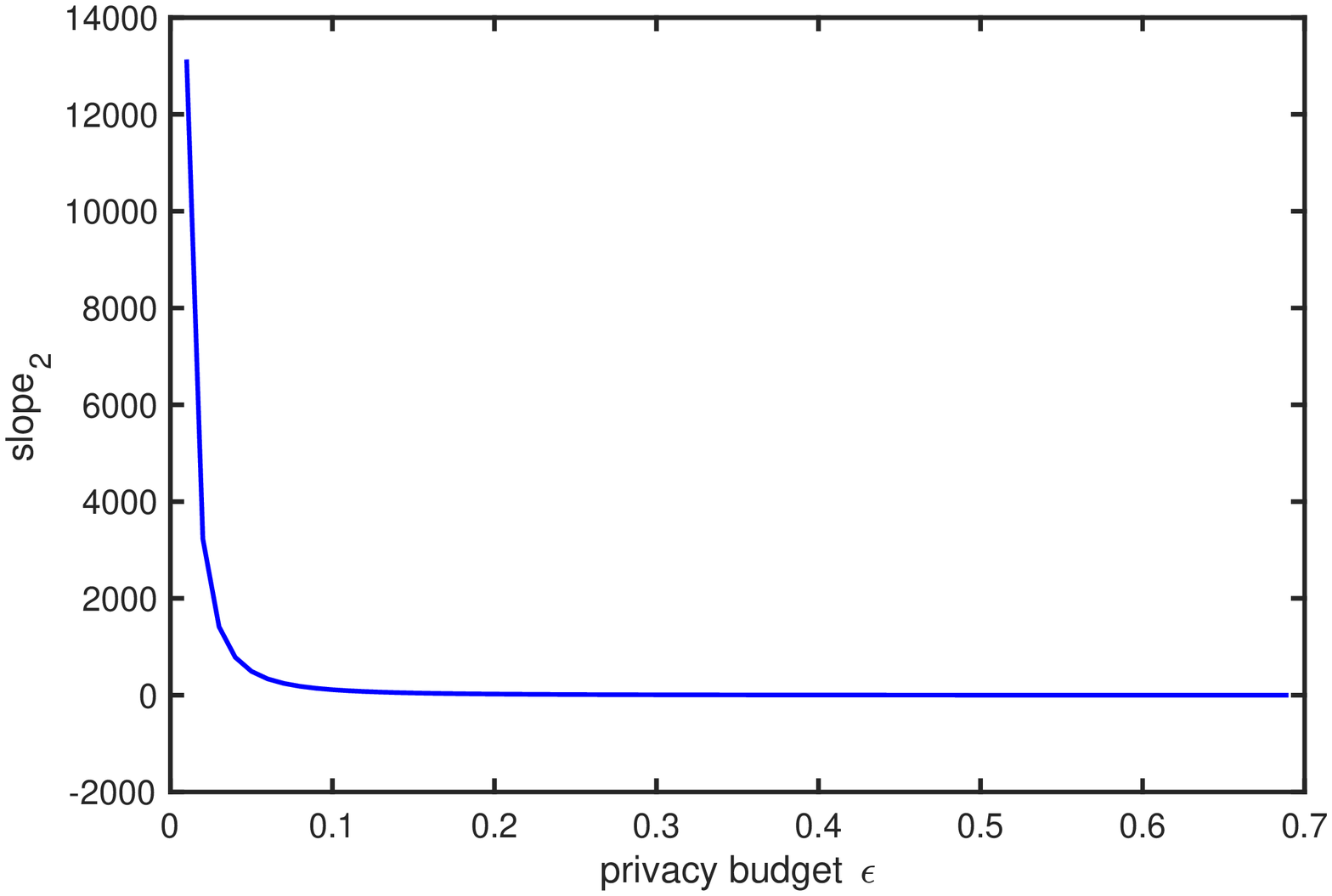} \vspace{-2pt}\caption{$\textup{slope}_2 ~\text{when}~ \epsilon \in [0, \ln2]$.}
     \label{slope_2}
    \end{figure}
     
    Therefore, $\min_\beta \max_{x \in [-1,1]} \textup{Var}_{\mathcal{H}}[Y|x]$ is at $\beta = \beta_1$ if $\beta \in [\beta_1, \beta_2]$. Summarize above analysis, we can conclude that $\min_\beta \max_{x \in [-1,1]} \textup{Var}_{\mathcal{H}}[Y|x] =$
     \begin{itemize}
         \item $\textup{Var}_{\mathcal{H}} [Y|x^*, \beta=0] $, if $0 < \epsilon < 0.610986$.
         \item $\max_{x \in [-1,1]} \textup{Var}_{\mathcal{H}} [Y|x, \beta= \beta_1] $, if  $0.610986 \leq \epsilon < \ln2$.
     \end{itemize} 

    \item If $\ln2 \leq \epsilon \leq \ln5.53$,\\
    When $\beta \in (0, \beta_1) $,~Appendix~\ref{value:A} proves~$A < 0,~B < 0$,~Appendix~\ref{monotonicity_of_gamma} proves when $\gamma = -\sqrt{\frac{B}{A}}$, we have  $\min_\beta \max_{x \in [-1,1]} \textup{Var}_{\mathcal{H}}[Y|x] = \max_{x \in [-1,1]} \textup{Var}_{\mathcal{H}} [Y|x, \beta = \beta_3]$. Since $\gamma := \beta(c+t)-c + 1$, we have~$\beta_3 := \frac{-\sqrt{\frac{B}{A}} + c -1}{c+t}$. $\min_\beta \max_{x \in [-1,1]} \textup{Var}_{\mathcal{H}}[Y|x] = \max_{x \in [-1,1]} \textup{Var}_{\mathcal{H}} [Y|x^*, \beta = \beta_3]$. \\ 
    When $\beta \in [\beta_1, \beta_2] $,
    Fig.~\ref{slope_1_2} shows that $\textup{slope}_1 > 0$ if $\epsilon \in [\ln2, \ln 5.53]$.
    Fig.~\ref{slope_2_2} shows that $\textup{slope}_2$:
    \begin{itemize}
         \item If $0 < \epsilon < 1.4338,~\textup{slope}_2 < 0,~\beta_{intersection} < \beta_1,$~see Fig.~\ref{diff_intersection_1}, $~\min_\beta \max_{x \in [-1,1]} \textup{Var}_{\mathcal{H}}[Y|x] =  \max_{x \in [-1,1]} \textup{Var}_{\mathcal{H}} [Y|x ,\beta = \beta_1]$.
         \item If $\epsilon \approx 1.4338,~\textup{slope}_2 = 0,~\beta_{intersection} < \beta_1,$~see Fig.~\ref{diff_intersection_1},~$\min_\beta \max_{x \in [-1,1]} \textup{Var}_{\mathcal{H}}[Y|x] = \max_{x \in [-1,1]} \textup{Var}_{\mathcal{H}} [Y|x ,\beta = \beta_1]$.
         \item If $1.4338 < \epsilon \leq \ln5.53$, $\textup{slope}_2 > 0$, $\min_\beta \max_{x \in [-1,1]} \textup{Var}_{\mathcal{H}}[Y|x] = \max_{x \in [-1,1]} \textup{Var}_{\mathcal{H}} [Y|x ,\beta = \beta_1]$.
    \end{itemize}
    Since $\textup{Var}_{\mathcal{H}} [Y|x, \beta= \beta_3] < \textup{Var}_{\mathcal{H}} [Y|x, \beta=\beta_1]$, $\min_\beta \max_{x \in [-1,1]} \textup{Var}_{\mathcal{H}}[Y|x]$ is at $\beta = \beta_3$.
    
    \begin{figure}[h]
    \centering
      \includegraphics[width=0.45\textwidth]{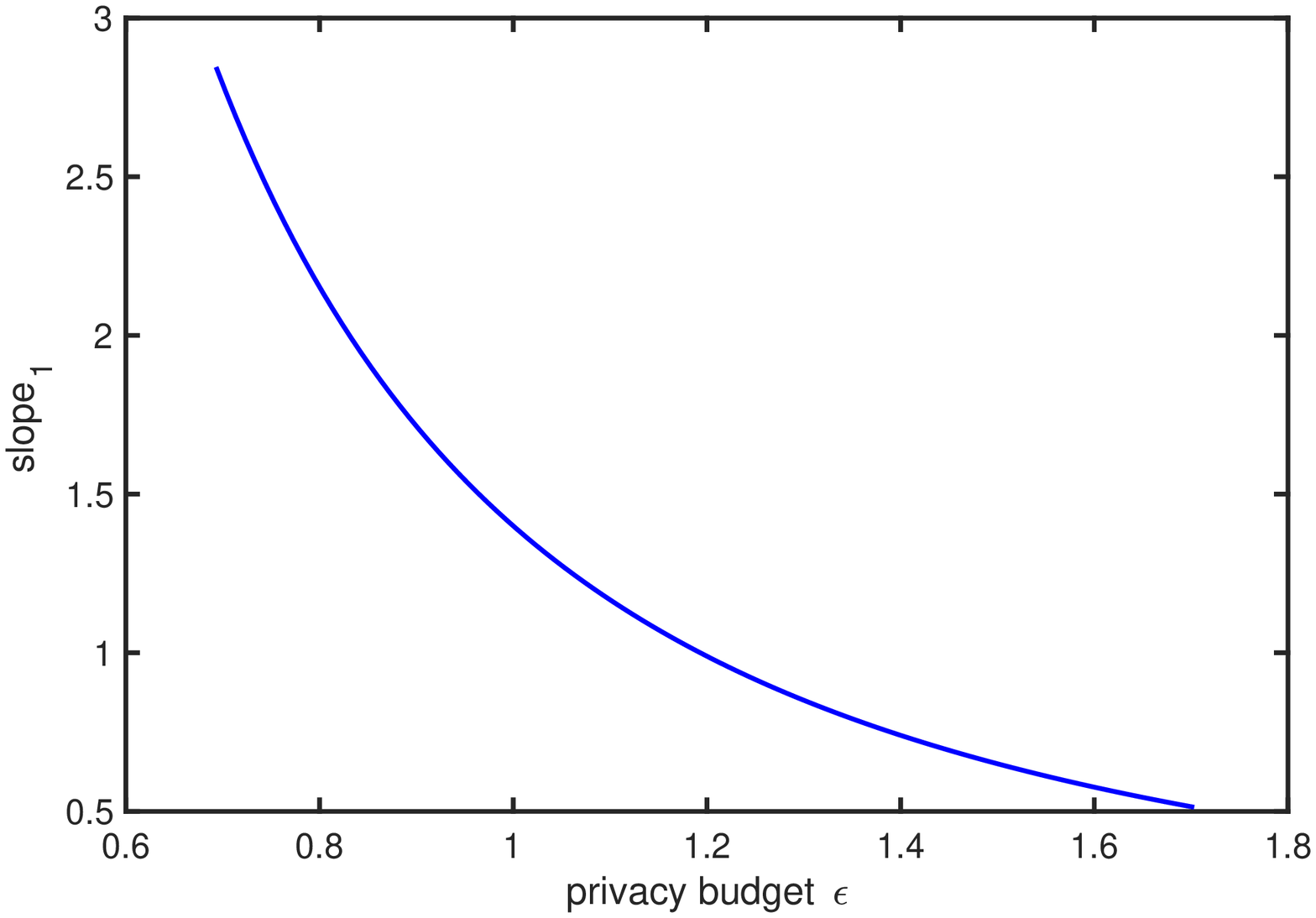} \vspace{-2pt}\caption{slope$_1 ~\text{when}~ \epsilon \in [\ln2, \ln5.53]$.}
     \label{slope_1_2}
    \end{figure}
    
    \begin{figure}[h]
    \centering
      \includegraphics[width=0.45\textwidth]{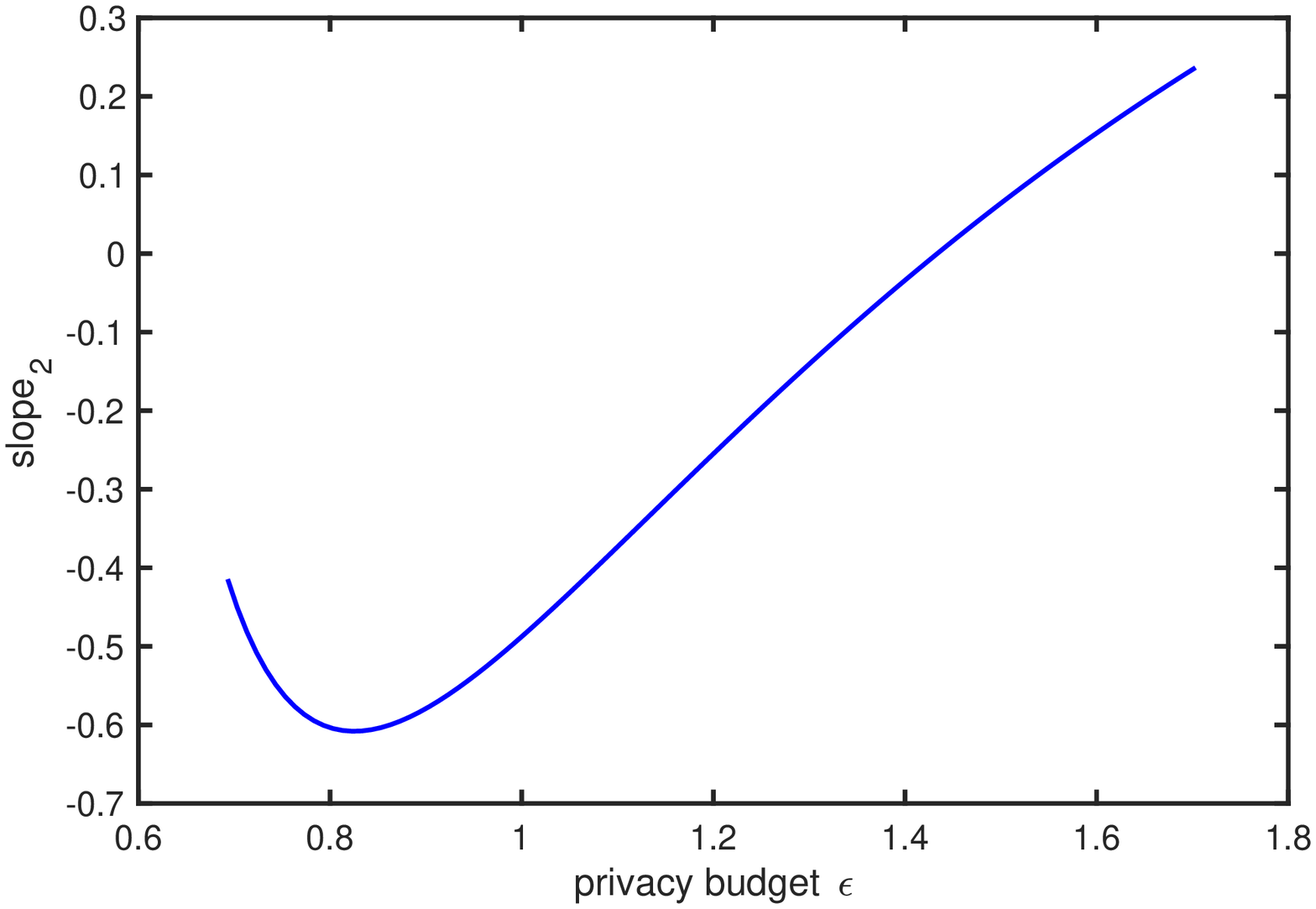} \vspace{-2pt}\caption{slope$_2 ~\text{when}~ \epsilon \in [\ln2, \ln5.53]$.}
     \label{slope_2_2}
    \end{figure}
    
    \begin{figure}[h]
    \centering
      \includegraphics[width=0.45\textwidth]{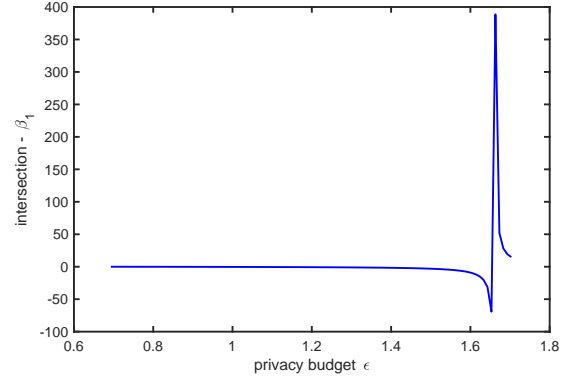} \vspace{-2pt}\caption{$\beta_{intersection} - \beta_1 ~\text{when}~ \epsilon \in [\ln2, \ln5.53]$.}
     \label{diff_intersection_1}
    \end{figure}
    
    \item If $\epsilon > \ln5.53$,
    \begin{itemize}
        \item  If $\beta \in (0, \beta_1) $,~Appendix~\ref{value:A} proves~$A<0$ and $B<0$,~Appendix~\ref{monotonicity_of_gamma} proves when $\gamma = -\sqrt{\frac{B}{A}}$, $\beta_4 := \frac{-\sqrt{\frac{B}{A}} + c -1}{c+t}$, we get $\min_\beta \max_{x \in [-1,1]} \textup{Var}_{\mathcal{H}}[Y|x] = \max_{x \in [-1,1]} \textup{Var}_{\mathcal{H}} [Y|x, \beta = \beta_4]$.
        \item  If $\beta \in [\beta_1, \beta_2] $, Appendix~\ref{proof:value_slope_1} proves $\textup{slope}_1 > 0$ and Appendix~\ref{proof:value_slope_2} proves $\textup{slope}_2 >0$. Therefore, we obtain $\min_\beta \max_{x \in [-1,1]} \textup{Var}_{\mathcal{H}}[Y|x] = \max_{x \in [-1,1]} \textup{Var}_{\mathcal{H}} [Y|x, \beta = \beta_1]$.
    \end{itemize}
\end{itemize}
\qeda

\subsection{\textbf{Proof of the monotonicity of $\textup{Var}_{\mathcal{H}} [Y|x^*]$}}\label{monotonicity_of_gamma}

  Substituting $A$ (Eq.~\ref{eq:A}), $B$ (Eq.~\ref{eq:B}) and $C$ (Eq.~\ref{eq:C}) into $\textup{Var}_{\mathcal{H}} [Y|x^*]$ (Eq.~\ref{eq:var_gamma}) yields
    \begin{align}
       \textup{Var}_{\mathcal{H}} [Y|x^*]=  A \gamma + \frac{B}{\gamma} + C.\label{eq:simplify_var_gamma}
    \end{align}
    The first order derivative of (\ref{eq:simplify_var_gamma}) is
    \begin{align}
      \textup{Var}_{\mathcal{H}} [Y|x^*]' =  A - \frac{ B}{\gamma^2}.
    \end{align}
    If $A - \frac{ B}{\gamma^2} = 0$, we get two roots: $$\gamma_1= -\sqrt{\frac{B}{A}} , \gamma_2 =  \sqrt{\frac{B}{A}}.$$ Since $\gamma < 0$, $\gamma_1$ is eligible. Hereby,
    \begin{itemize}
        \item     If $\gamma \in (-\infty, -\sqrt{\frac{B}{A}}]$, $\textup{Var}_{\mathcal{H}} [Y|x^*]' < 0$, $\textup{Var}_{\mathcal{H}} [Y|x^*]$ monotonically decreases.
        \item     If $\gamma \in (-\sqrt{\frac{B}{A}}, 0)$, $\textup{Var}_{\mathcal{H}} [Y|x^*]' > 0$, $ \textup{Var}_{\mathcal{H}} [Y|x^*]$ monotonically increases.
    \end{itemize}

\qeda

\subsection{\textbf{The sign of $A$ to $\epsilon$}}\label{value:A}

If $A = 0$, we have $\epsilon = 3 \ln(\text{root of}~ 3 \epsilon^5 - 2 \epsilon^3 + 3 \epsilon^2 - 5 \epsilon - 3~\text{near}~x = 1.22588) \approx 0.610986$.

First order derivative of $A$ is
\begin{align}
    \nonumber  &A' =  -(25 e^\epsilon - 27e^{2\epsilon} - 9e^{3\epsilon} - 12e^{\epsilon/3} + 19e^{2\epsilon/3} - e^{4\epsilon/3} \\ & \quad + 41e^{5\epsilon/3} + 7e^{7\epsilon/3} + 5)/(9e^{2\epsilon/3}(e^{2\epsilon/3} + 1)^2(e^\epsilon - 1)^3).
\end{align}
\begin{itemize}

\item If $0 <\epsilon < \ln2$, Fig.~\ref{A_less_than_zero} shows that $A'<0$ and $A$ monotonically decreases if $\epsilon \in (0, \ln2)$. Therefore, we have
\begin{itemize}
    \item $A>0$, if $ 0 < \epsilon < 0.610986$.
    \item $A = 0$, if $ \epsilon = 0.610986$.
    \item $A<0$, if $0.610986<\epsilon < \ln2$.
\end{itemize}

\begin{figure}[h]
\centering
  \includegraphics[width=0.45\textwidth]{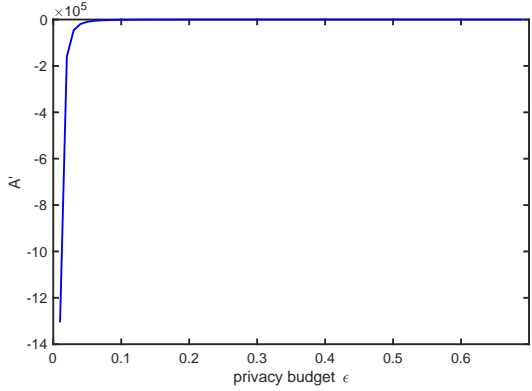} \vspace{-2pt}\caption{First order derivative of $A$ is less than 0, if ~$0 < \epsilon \leq \ln2$.}
 \label{A_less_than_zero}
\end{figure}

\item If $\ln2 \leq \epsilon \leq \ln5.53$, Fig.~\ref{A_less_than_zero_2} shows $A <0$.
\begin{figure}[h]
\centering
  \includegraphics[width=0.45\textwidth]{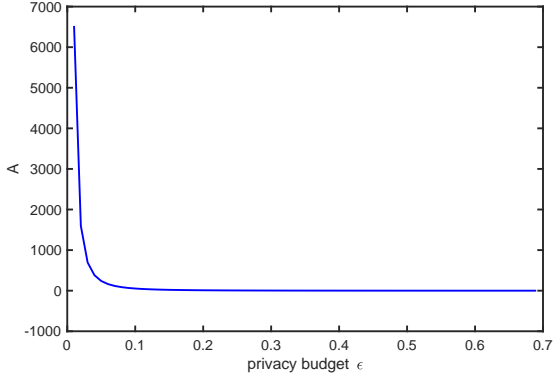} \vspace{-2pt}\caption{Value $A$ is less than 0, if ~$\ln2 <\epsilon \leq \ln5.53$.}
 \label{A_less_than_zero_2}
\end{figure}

\item If $\epsilon > \ln5.53$, we obtain
\begin{align}
    &A =  -(16e^\epsilon + 21e^{2\epsilon} + 3e^{3\epsilon} + 36e^{\frac{\epsilon}{3}} - 12e^\frac{2\epsilon}{3} \nonumber \\& - 28e^\frac{5\epsilon}{3} - 8e^\frac{7\epsilon}{3} - 12))/(12e^\frac{2\epsilon}{3}(e^\epsilon - e^\frac{2\epsilon}{3} + e^\frac{5\epsilon}{3} - 1)^2.
\end{align}
When $A = 0$, we have three roots:
$$r_1 \approx -16.9563, r_2 \approx -1.2284, r_3 \approx 0.0463914.$$ If the denominator $$= 12e^\frac{2\epsilon}{3}(e^\epsilon - e^\frac{2\epsilon}{3} + e^\frac{5\epsilon}{3} - 1)^2 >0$$ ~and~
\begin{align}
    &\lim_{\epsilon \to \infty} -(16e^\epsilon + 21e^{2\epsilon} + 3e^{3\epsilon} \nonumber \\&+ 36e^{\frac{\epsilon}{3}} - 12e^\frac{2\epsilon}{3} - 28e^\frac{5\epsilon}{3} - 8e^\frac{7\epsilon}{3} - 12) = -\infty, \nonumber
\end{align}$r_3$ is the largest real value root, the sign of  $$-(16e^\epsilon + 21e^{2\epsilon} + 3e^{3\epsilon} + 36e^{\frac{\epsilon}{3}} - 12e^\frac{2\epsilon}{3} - 28e^\frac{5\epsilon}{3} - 8e^\frac{7\epsilon}{3} - 12) $$ doesn't change, so that $A <0$ when $\epsilon > \ln 5.53$.
\end{itemize}

\qeda

\subsection{\textbf{The sign of $\textup{slope}_1$ when $\epsilon > \ln5.53$}}\label{proof:value_slope_1}

 \begin{align}
    &\textup{If~slope}_1 = \frac{t+1}{e^\epsilon-1} + 1 - \frac{ae^\epsilon(e^\epsilon+1)^2}{(e^\epsilon-1)(e^\epsilon-a)^2} \nonumber \\&+  \frac{(t+e^\epsilon)((t+1)^3+e^\epsilon-1)}{3t^2(e^\epsilon-1)^2} -  \frac{(1-a)e^{2\epsilon}(e^\epsilon+1)^2}{(e^\epsilon-1)^2(e^\epsilon-a)^2}  = 0, \nonumber
\end{align} we have two roots:
$$e^\epsilon_1 \approx -0.141506, e^\epsilon_2 \approx 2.21598.$$
The first order derivative of $\textup{slope}_1$ is
$$\textup{slope}_1' = -{{e^{{{2\epsilon}\over{3}}}\,\left(10\,e^\epsilon+20\right)+20\,e^\epsilon+\left(-27\,e^\epsilon-45 \right)\,e^{{{\epsilon}\over{3}}}+10}\over{9e^{{{\epsilon}\over{3}}}\,(e^\epsilon - 1)^3}}.$$
The denominator of the $\textup{slope}_1'$ is $>0$. Thus,
\begin{align}
    &\lim_{\epsilon \to \infty} -{e^{{{2\epsilon}\over{3}}}\,\left(10\,e^\epsilon+20\right)+20\,e^\epsilon+\left(-27\,e^\epsilon-45 \right)\,e^{{{\epsilon}\over{3}}}+10} \nonumber \\& \quad = -\infty.\nonumber
\end{align}
If $\textup{slope}_1' = 0$, we have three roots:
\begin{align}
    e^\epsilon_1 \approx -1.35696, e^\epsilon_2 \approx 0.0169067, e^\epsilon_3 \approx 4.22192. \nonumber
\end{align}
Since $e^\epsilon_3 \approx 4.22192$ is the largest real value root, the sign of $\textup{slope}_1'$ doesn't change when $e^\epsilon > 4.22192$. 
Therefore, when $\epsilon > \ln5.53$ and $\textup{slope}_1' <0$, $\textup{slope}_1$ monotonically decreases.
By simplifying $\textup{slope}_1$, we get 
\begin{align}
    \textup{slope}_1 = \frac{-(9e^\epsilon - 5 e^{2\epsilon/3} -5e^{4\epsilon/3} +3)}{3(e^\epsilon - 1)^2}. \nonumber
\end{align}  Then, we obtain $\lim_{\epsilon \to \infty} \textup{slope}_1$ = 0. Thus, we have $\textup{slope}_1 > 0$ if $\epsilon > \ln5.53$.

\qeda

\subsection{\textbf{The sign of $\textup{slope}_2$ when $\epsilon > \ln5.53$}}\label{proof:value_slope_2}

When
\begin{align}
   &\textup{slope}_2 = \frac{(t+e^\epsilon)((t+1)^3+e^\epsilon-1)}{3t^2(e^\epsilon-1)^2} \nonumber \\& \quad-  \frac{(1-a)e^{2\epsilon}(e^\epsilon+1)^2}{(e^\epsilon-1)^2(e^\epsilon-a)^2} = 0, \nonumber 
\end{align} we have
\begin{align}
    \epsilon_1 \approx \ln -1.24835, \epsilon_2 \approx \ln 1.52144. \nonumber
\end{align} 
The first order derivative of $\textup{slope}_2$ is :
\begin{align}
    &\textup{slope}_2' =  \nonumber  \\& {{-4\,e^{2\epsilon}+e^{{{2\epsilon}\over{3}}}\,\left(9\,e^\epsilon+63\right)-23\,e^\epsilon+\left(-20\, e^\epsilon-10\right)\,e^{{{\epsilon}\over{3}}}-3}\over{9e^{{{2\epsilon}\over{3}}}\,\left(e^\epsilon - 1\right)^3}}.\label{eq:slope-2}
\end{align}
The denominator of the Eq.~$(\ref{eq:slope-2})$ is $>0$. Besides, from nominator of Eq.~$(\ref{eq:slope-2})$, we obtain
\begin{align}
    &\lim_{\epsilon \to \infty} (-4\,e^{2\epsilon}+e^{{{2\epsilon}\over{3}}}\,\left(9\,e^\epsilon+63\right)-23\,e^\epsilon+\left(-20\, e^\epsilon-10\right)\,e^{{{\epsilon}\over{3}}}-3) \nonumber\\& \quad = -\infty. \nonumber
\end{align}
If Eq.~$(\ref{eq:slope-2}) = 0$, we have two roots:
\begin{align}
    &\epsilon_1 = 3 \ln(\text{root of}~ 4 x^6 - 9 x^5 + 20 x^4 + 23 x^3 - 63 x^2 + 10 x + 3) \nonumber \\& \quad \approx -3.13865, \nonumber  \\
    &\epsilon_2 = 3 \ln(\text{root of}~ 4 x^6 - 9 x^5 + 20 x^4 + 23 x^3 - 63 x^2 + 10 x + 3) \nonumber \\& \quad \approx 0.709472.\nonumber 
\end{align} Since $\epsilon \approx 0.709472$ is the largest real value root, the sign of $\textup{slope}_2'$ doesn't change if $\epsilon > 0.709472$. Therefore, $\textup{slope}_2' < 0 $ if $\epsilon > \ln 5.53$. Simplify $$\textup{slope}_2 = \frac{3e^\frac{\epsilon}{3} - 3e^\epsilon + 5e^\frac{2\epsilon}{3} + 2e^\frac{4\epsilon}{3} - 9}{3(e^\epsilon - 1)^2},$$ and then we get $$\lim_{\epsilon \to \infty} \textup{slope}_2  = 0.$$ Thus, we have $\textup{slope}_2 > 0$ if $\epsilon > \ln5.53$.
\qeda

\subsection{\textbf{Proof of Lemma~\ref{lem-d-2}}}\label{proof:lem-d-2}

For any $i \in [1, n]$,  the random variable $Y[t_j]- x[t_j]$ has zero mean based on Lemma~\ref{lem-d-1}. In both \texttt{PM-SUB} and $\texttt{HM}_{\text{PM-SUB}, \text{Three-Outputs}}$, $|Y[t_j]- x[t_j]| \leq \frac{d}{k} \cdot \frac{(e^{\frac{\epsilon}{k}} + e^\frac{\epsilon}{3k})(e^\frac{\epsilon}{3k} + 1)}{e^\frac{\epsilon}{3k}(e^\frac{\epsilon}{k} - 1)}$. By Bernstein's inequality, we have
\begin{align}
   & \bp{|Z[t_j] - X[t_j]| \geq \lambda}\nonumber \\ & = Pr \bigg [ | \sum_{i=1}^n \{ Y[t_j] - x[t_j]\} | \geq n \lambda \bigg ] \nonumber \\& \leq 2 \cdot \exp\bigg (\frac{-(n \lambda)^2}{2 \sum_{i=1}^n \textup{Var}[Y[t_j]] + \frac{2}{3} \cdot n\lambda \cdot \frac{d}{k} \cdot \frac{(e^{\frac{\epsilon}{k}} + e^\frac{\epsilon}{3k})(e^\frac{\epsilon}{3k} + 1)}{e^\frac{\epsilon}{3k}(e^\frac{\epsilon}{k} - 1)} } \bigg).~\nonumber 
\end{align}
In Algorithm~\ref{PM-d-dimension-Algo}, $Y[t_j]$ equals $\frac{d}{k} y_{j}$ with probability $\frac{k}{d}$ and 0 with probability $1 - \frac{k}{d}$. Moreover, we obtain $\mathbb{E}[Y[t_j]] = x[t_j]$ from Lemma \ref{lem-d-1}, and then we get 
\begin{align}
    &\textup{Var}[Y[t_j]] = \mathbb{E}[(Y[t_j])^2] - \mathbb{E}[Y[t_j]]^2 \nonumber \\&= \frac{k}{d}E[(\frac{d}{k}y_{j})^2] - (x[t_j])^2 = \frac{d}{k}\mathbb{E}[(y_{j})^2] - (x[t_j])^2. \label{var-PM-d}
\end{align}
In Algorithm~\ref{PM-d-dimension-Algo}, if Line 5 uses \texttt{PM-SUB}, we the variance in Eq.~(\ref{eq:var-ldp-1}) to compute $\mathbb{E}[(y_{j})^2] $, the asymptotic expression involving $\epsilon$ are in the sense of $\epsilon \rightarrow 0$.
\begin{align}
   & E[(y_{j})^2] = \textup{Var}[y_{j}] + (\mathbb{E}[y_{j}])^2 \nonumber \\& = \frac{t(\frac{\epsilon}{k})+1}{e^\frac{\epsilon}{k}-1}(x[t_j])^2 + \frac{(t(\frac{\epsilon}{k})+e^\frac{\epsilon}{k})\big ((t(\frac{\epsilon}{k})+1)^3 + e^\frac{\epsilon}{k} -1 \big)}{3(t(\frac{\epsilon}{k}))^2(e^\frac{\epsilon}{k}-1)^2} \nonumber \\& \quad + (x[t_j])^2  = O\bigg (\frac{k^2}{\epsilon^2} \bigg ). \label{eq:expectation-x-square}
\end{align}

In Algorithm~\ref{PM-d-dimension-Algo}, if Line 5 uses \texttt{Three-Outputs}, and then we the variance in Eq.~(\ref{Three-OUTPUTS-VAR}) to compute $\mathbb{E}[(y_{j})^2] $, the asymptotic expression involving $\epsilon$ are in the sense of $\epsilon \rightarrow 0$.
    \begin{align}
       & \mathbb{E}[(y_{j})^2] = \textup{Var}[y_{j}] + (\mathbb{E}[y_{j}])^2 \nonumber \\& = \frac{(1-a)e^{\frac{2\epsilon}{k}}(e^\frac{\epsilon}{k}+1)^2}{(e^\frac{\epsilon}{k} -1)^2(e^\frac{\epsilon}{k} -a)^2}
        +\frac{b|x[t_j]|e^{\frac{2\epsilon}{k}}(e^\frac{\epsilon}{k}+1)^2}{(e^\frac{\epsilon}{k} -1)^2(e^\frac{\epsilon}{k} -a)^2}-(x[t_j])^2 \nonumber \\& \quad + (x[t_j])^2 \nonumber \\& = \frac{(1-a)e^{\frac{2\epsilon}{k}}(e^\frac{\epsilon}{k}+1)^2}{(e^\frac{\epsilon}{k} -1)^2(e^\frac{\epsilon}{k} -a)^2}
        +\frac{b|x[t_j]|e^{\frac{2\epsilon}{k}}(e^\frac{\epsilon}{k}+1)^2}{(e^\frac{\epsilon}{k} -1)^2(e^\frac{\epsilon}{k} -a)^2} = O(\frac{k^2}{\epsilon^2}),\nonumber
    \end{align}

In Algorithm~\ref{PM-d-dimension-Algo}, if Line 5 uses \texttt{HM-TP}, we have:

$E[(y_{j})^2] = \textup{Var}[y_{j}] + (E[y_{j}])^2$
\begin{align}
    &=
    \begin{cases}
     \frac{(e^\frac{\epsilon}{k} + 1)^2}{(e^\frac{\epsilon}{k} -1)^2} + (x[t_j])^2,~\text{If}~0 < \epsilon < \epsilon^*,\\
     \textup{Var}_{\mathcal{H}}[Y|1,\beta_1, \frac{\epsilon}{k}] + (x[t_j])^2,  ~\text{If}~\epsilon^* \leq \epsilon < \ln2,\\
     \textup{Var}_{\mathcal{H}}[Y|1,\beta_3, \frac{\epsilon}{k}] + (x[t_j])^2,~\text{If}~  \epsilon \geq \ln2
    \end{cases} \nonumber \\
   & = O\bigg (\frac{k^2}{\epsilon^2} \bigg ),
\end{align}
where $\epsilon^*$ is defined in the Eq.~(\ref{eq:eps_star}).
Then,
\begin{align}
    &\textup{Var}[Y[t_j]] = \frac{d}{k} \bigg (\frac{t(\frac{\epsilon}{k})+1}{e^\frac{\epsilon}{k}-1}(x[t_j])^2 \nonumber \\&+ \frac{(t(\frac{\epsilon}{k})+e^\frac{\epsilon}{k})\big ((t(\frac{\epsilon}{k})+1)^3 + e^\frac{\epsilon}{k} -1 \big)}{3(t(\frac{\epsilon}{k}))^2(e^\frac{\epsilon}{k}-1)^2} + (x[t_j])^2 \bigg ) \nonumber \\& - (x[t_j])^2. \nonumber
\end{align}
Substituting Eq.~(\ref{eq:expectation-x-square}) into Eq.~(\ref{var-PM-d}) yields
\begin{align}
    \textup{Var}[Y[t_j]] = \frac{d}{k} \cdot O\bigg (\frac{k^2}{\epsilon^2} \bigg ) - (x[t_j])^2 = O\bigg (\frac{dk}{\epsilon^2} \bigg ). 
\end{align}
Therefore, we obtain $$\bp{|Z[t_j] - X[t_j]| \geq \lambda} \leq 2 \cdot \exp \bigg ( - \frac{n \lambda^2}{O(dk/\epsilon^2) + \lambda \cdot O(d/ \epsilon)} \bigg ).$$
By the union bound, there exists $\lambda = O \bigg ( \frac{\sqrt{d \ln(d/\beta))}}{\epsilon \sqrt{n}} \bigg )$. Therefore, $\max_{j \in [1,d]} |Z[t_j] - X[t_j]| = \lambda = O \bigg ( \frac{\sqrt{d \ln(d/\beta))}}{\epsilon \sqrt{n}} \bigg ).$

\qeda

\subsection{\textbf{Calculate $k$ for $\texttt{PM-SUB}$ and $\texttt{Three-Outputs}$}}\label{proof:k}

We calculate the optimal $k$ for \texttt{PM-SUB} and \texttt{Three-Outputs}.

\begin{itemize}
    \item[(I)] We calculate the $k$ for $d$ dimension \texttt{PM-SUB}. When $x[t_j] = 1$, we get
        \begin{align}
            &\max \textup{Var}[Y[t_j]]  \nonumber \\& = \frac{d}{k} \bigg (\frac{t(\frac{\epsilon}{k})+1}{e^\frac{\epsilon}{k}-1}  \nonumber + \frac{(t(\frac{\epsilon}{k})+e^\frac{\epsilon}{k})\big ((t(\frac{\epsilon}{k})+1)^3 + e^\frac{\epsilon}{k} -1 \big)}{3(t(\frac{\epsilon}{k}))^2(e^\frac{\epsilon}{k}-1)^2} + 1 \bigg ) - 1. \nonumber
        \end{align}
For \texttt{PM-SUB}, we have
        \begin{align}
             &\max \textup{Var}[Y[t_j]]  \nonumber \\& = \frac{d}{k} \bigg (\frac{e^\frac{\epsilon}{3k}+1}{e^\frac{\epsilon}{k}-1} \nonumber + \frac{(e^\frac{\epsilon}{3k}+e^\frac{\epsilon}{k})\big ((e^\frac{\epsilon}{3k}+1)^3 + e^\frac{\epsilon}{k} -1 \big)}{3(e^\frac{\epsilon}{3k})^2(e^\frac{\epsilon}{k}-1)^2} + 1 \bigg )- 1.\nonumber
        \end{align}
Let $s = \frac{\epsilon}{k}$, and then
        \begin{align}
             &\max \textup{Var}[Y[t_j]]  \nonumber \\& = \frac{d}{\epsilon} \cdot s \bigg (\frac{e^\frac{s}{3}+1}{e^s-1} \nonumber + \frac{(e^\frac{s}{3}+e^s)\big ((e^\frac{s}{3}+1)^3 + e^s -1 \big)}{3(e^\frac{s}{3})^2(e^s-1)^2} + 1 \bigg )  - 1.\nonumber
        \end{align}
        
Let
        \begin{align}
             f(s) =  s \cdot \bigg (\frac{e^\frac{s}{3}+1}{e^s-1} \nonumber + \frac{(e^\frac{s}{3}+e^s)\big ((e^\frac{s}{3}+1)^3 + e^s -1 \big)}{3(e^\frac{s}{3})^2(e^s-1)^2} + 1 \bigg ), \nonumber
        \end{align}
and then we obtain
        \begin{align}
             \max \textup{Var}[Y[t_j]]  = \frac{d}{\epsilon} \cdot f(s) - 1.
        \end{align}
        
        \begin{figure}[h]
        \centering
          \includegraphics[width=0.45\textwidth]{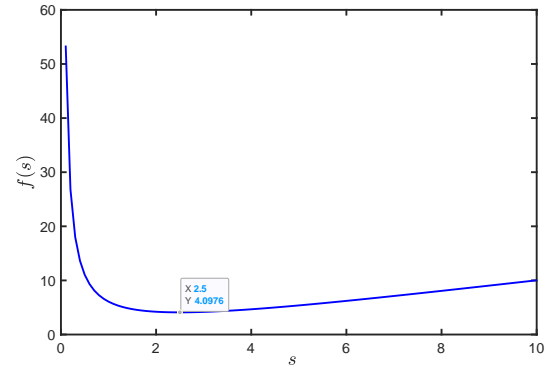} \vspace{-2pt}\caption{ Find $s$ for $\min f(s)$.}
         \label{fig:optimal_s}
        \end{figure}
        From numerical experiments shown in Fig.~\ref{fig:optimal_s}, we conclude that we can get $\min f(s)$ and $\min \max \textup{Var}[Y[t_j]]$ if $s = 2.5$, i.e. $k = \frac{\epsilon}{2.5}$.
    \item[(II)] Calculate the $k$ for $d$ dimension \texttt{Three-Outputs}.
        The variance of $Y[t_j]$ is
        \begin{align}
            &\textup{Var}[Y[t_j]] = \nonumber\\& \frac{d}{k} \bigg ( \frac{(1-a)e^{\frac{2\epsilon}{k}}(e^\frac{\epsilon}{k}+1)^2}{(e^\frac{\epsilon}{k} -1)^2(e^\frac{\epsilon}{k} -a)^2}
            +\frac{b|x[t_j]|e^{\frac{2\epsilon}{k}}(e^\frac{\epsilon}{k}+1)^2}{(e^\frac{\epsilon}{k} -1)^2(e^\frac{\epsilon}{k} -a)^2} \bigg ) - (x[t_j])^2 \nonumber \\& = \frac{d}{\epsilon} \cdot s \cdot \bigg ( \frac{(1-a)e^{\frac{2\epsilon}{k}}(e^\frac{\epsilon}{k}+1)^2}{(e^\frac{\epsilon}{k} -1)^2(e^\frac{\epsilon}{k} -a)^2}
            +\frac{b|x[t_j]|e^{\frac{2\epsilon}{k}}(e^\frac{\epsilon}{k}+1)^2}{(e^\frac{\epsilon}{k} -1)^2(e^\frac{\epsilon}{k} -a)^2} \bigg ) \nonumber \\& \quad - (x[t_j])^2, \nonumber
        \end{align}
        where $b$ is from Eq.~(\ref{eq:optimal-b}) and $a$ is from Eq.~(\ref{eq:optimal-a}).

        Let $x[t_j]' = \frac{dbe^\frac{2\epsilon}{k}(e^\frac{\epsilon}{k} + 1)^2}{2k(e^\frac{\epsilon}{k} - 1)^2(e^\frac{\epsilon}{k} - a)^2}$, if $0< x[t_j]' < 1$, the worst-case noise variance of $Y$ is
        \begin{align}
            &\max_{x \in [-1,1]}  \textup{Var}[Y[t_j]] = \nonumber\\
            & \begin{cases}
                \textup{Var}[x[t_j]'],\text{if}~ 0 < x[t_j]' < 1,   \\
               \max \{\textup{Var}[0], \textup{Var}[1]\}, ~~\text{otherwise.}
            \end{cases}
        \end{align}
        Let $s = \frac{\epsilon}{k}$, and then  
        \begin{align}
           & x[t_j]' = \frac{d}{\epsilon} \cdot s\frac{b(s)e^{2s}(e^s + 1)^2}{2(e^s - 1)^2(e^s - a(s))^2} \nonumber \\&= \frac{d}{\epsilon} \cdot s \frac{a(s) e^s (e^s+1)^2}{2(e^s - 1)(e^s-a(s))^2}. \nonumber
        \end{align}
        
        If $0 < \frac{d}{\epsilon} \cdot s \frac{a(s) e^s (e^s+1)^2}{2(e^s - 1)(e^s-a(s))^2} < 1$, and then
        \begin{align}
            &\max \textup{Var}[Y[t_j]] = \max \textup{Var}[x[t_j]'] = \nonumber \\& = \frac{d}{k} \bigg ( \frac{(1-a(\frac{\epsilon}{k}))e^{\frac{2\epsilon}{k}}(e^\frac{\epsilon}{k}+1)^2}{(e^\frac{\epsilon}{k} -1)^2(e^\frac{\epsilon}{k} -a)^2}
            +\frac{b(\frac{\epsilon}{k}) x[t_j]'e^{\frac{2\epsilon}{k}}(e^\frac{\epsilon}{k}+1)^2}{(e^\frac{\epsilon}{k} -1)^2(e^\frac{\epsilon}{k} -a)^2} \bigg ) \nonumber \\& \quad - (x[t_j]')^2 \nonumber
            \\&  = \frac{d^2}{\epsilon^2} \cdot s^2 \frac{(b(s))^2 e^{4s}(e^s+1)^4}{2(e^s-1)^4(e^s-a(s))^4} \nonumber \\& \quad -\frac{d^2}{\epsilon^2} \cdot s^2 \frac{(b(s))^2 e^{4s}(e^s+1)^4}{4(e^s-1)^4(e^s-a(s))^4} \nonumber \\& \quad + \frac{d}{\epsilon} \cdot s \frac{(1-a(s))e^{2s}(e^s + 1)^2}{(e^s-1)^2(e^s-a(s))^2} \nonumber
            \\&  = \frac{d^2}{\epsilon^2} \cdot s^2 \frac{(b(s))^2 e^{4s}(e^s+1)^4}{4(e^s-1)^4(e^s-a(s))^4} \nonumber \\& \quad+ \frac{d}{\epsilon} \cdot s \frac{(1-a(s))e^{2s}(e^s + 1)^2}{(e^s-1)^2(e^s-a(s))^2}.\label{eq:max_var_three_outputs}
        \end{align}
        
        Substituting ~$b = a \cdot \frac{e^s - 1}{e^s}$ ~into Eq.~(\ref{eq:max_var_three_outputs}) yields
        \begin{align}
            &= \frac{d^2}{\epsilon^2} \cdot s^2 \frac{(a(s))^2 e^{2s}(e^s+1)^4}{4(e^s-1)^2(e^s-a(s))^4} \nonumber \\ &\quad + \frac{d}{\epsilon} \cdot s \frac{(1-a(s))e^{2s}(e^s + 1)^2}{(e^s-1)^2(e^s-a(s))^2}. \nonumber      
        \end{align}

        \begin{itemize}
            \item If $\epsilon < \ln2$, $a = 0, b = 0$, first-order derivative of $\max \textup{Var}[Y[t_j]] $ is
                \begin{align}
                    \max \textup{Var}[Y[t_j]]' =  \frac{d}{\epsilon} \cdot \frac{(e^s + 1)(-4se^s + e^{2s} - 1)}{(e^s - 1)^3}.
                \end{align}
                When $ \max \textup{Var}[Y[t_j]]' = 0$, we have root $s \approx 2.18$.
            \item If $\ln2 < \epsilon < \ln5.5$, by numerical experiments, we have optimal $s \approx 2.5$.
            \item If $\epsilon \geq \ln 5.5$, by numerical experiments, we have optimal $s \approx 2.5$.
        \end{itemize}
        Therefore, we pick $s = 2.5$ and $k =  \frac{\epsilon}{2.5}$ to simplify the experimental evaluation.
    \end{itemize}
\qeda

\subsection{\textbf{Extending \texttt{Three-Outputs} for Multiple Numeric Attributes}}

\begin{lem}\label{eq:lem-extend-three}
    For a d-dimensional numeric tuple $x$ which is perturbed as $Y$ under $\epsilon$-LDP, and for each $t_j$ of the $d$ attribute, the variance of $Y[t_j]$ induced by \textit{Three-Outputs} is
    \begin{align}
        &\textup{Var}[Y[t_j]] = \frac{d}{k} \bigg ( \frac{(1-a)e^{\frac{2\epsilon}{k}}(e^\frac{\epsilon}{k}+1)^2}{(e^\frac{\epsilon}{k} -1)^2(e^\frac{\epsilon}{k} -a)^2}
        +\frac{b|x[t_j]|e^{\frac{2\epsilon}{k}}(e^\frac{\epsilon}{k}+1)^2}{(e^\frac{\epsilon}{k} -1)^2(e^\frac{\epsilon}{k} -a)^2} \bigg ) \nonumber\\&  - (x[t_j])^2 \nonumber. 
    \end{align}
\end{lem}

\noindent
\textbf{Proof of Lemma~\ref{eq:lem-extend-three}.}
The variance of $Y[t_j]$ is computed as
    \begin{align}
        &\textup{Var}[Y[t_j]] = \mathbb{E}[(Y[t_j])^2] - \mathbb{E}[Y[t_j]]^2 \nonumber \\&= \frac{k}{d}\mathbb{E}[(\frac{d}{k}y_{j})^2] - (x[t_j])^2 \nonumber \\& = \frac{d}{k}\mathbb{E}[(y_{j})^2] - (x[t_j])^2.
    \end{align}
    We use variance Eq.~(\ref{Three-OUTPUTS-VAR}) to compute
    \begin{align}
       & \mathbb{E}[(y_{j})^2] = \textup{Var}[y_{j}] + (\mathbb{E}[y_{j}])^2 \nonumber \\& = \frac{(1-a)e^{\frac{2\epsilon}{k}}(e^\frac{\epsilon}{k}+1)^2}{(e^\frac{\epsilon}{k} -1)^2(e^\frac{\epsilon}{k} -a)^2}
        +\frac{b|x[t_j]|e^{\frac{2\epsilon}{k}}(e^\frac{\epsilon}{k}+1)^2}{(e^\frac{\epsilon}{k} -1)^2(e^\frac{\epsilon}{k} -a)^2}-(x[t_j])^2 \nonumber \\& \quad + (x[t_j])^2 \nonumber \\& = \frac{(1-a)e^{\frac{2\epsilon}{k}}(e^\frac{\epsilon}{k}+1)^2}{(e^\frac{\epsilon}{k} -1)^2(e^\frac{\epsilon}{k} -a)^2}
        +\frac{b|x[t_j]|e^{\frac{2\epsilon}{k}}(e^\frac{\epsilon}{k}+1)^2}{(e^\frac{\epsilon}{k} -1)^2(e^\frac{\epsilon}{k} -a)^2}. \nonumber
    \end{align}
Then,
    \begin{align}
        &\textup{Var}[Y[t_j]] =  \frac{d}{k} \bigg ( \frac{(1-a)e^{\frac{2\epsilon}{k}}(e^\frac{\epsilon}{k}+1)^2}{(e^\frac{\epsilon}{k} -1)^2(e^\frac{\epsilon}{k} -a)^2}
        +\frac{b|x[t_j]|e^{\frac{2\epsilon}{k}}(e^\frac{\epsilon}{k}+1)^2}{(e^\frac{\epsilon}{k} -1)^2(e^\frac{\epsilon}{k} -a)^2} \bigg ) \nonumber \\& \quad - (x[t_j])^2. \nonumber
    \end{align}
\qeda

\subsection{\textbf{Proof of Lemma}~\ref{lem:hm-var}} \label{sec:proof-hm-var}

Given \texttt{PM-SUB}'s variance in Eq.~(\ref{PM-OPT_var}), we have
\begin{align}
    &\textup{Var}_{\mathcal{P}}[Y|x]  = \left( \frac{1}{1-2Ad}  - 1 \right)x^2     +  \frac{2d}{3} A^3  + \frac{(1-2Ad)^3}{12(c-d)^2}.  \nonumber
\end{align}
Substituting $\alpha = Ad$ yields
\begin{align}
    &\textup{Var}_{\mathcal{P}}[Y|x]  = \left( \frac{1}{1-2\alpha}  - 1 \right)x^2     +  \frac{2\alpha}{3d^2}  + \frac{(1-2\alpha)^3}{12(c-d)^2}.  \nonumber
\end{align}
In addition, substitute $1-2\alpha = \frac{e^{\epsilon}-1}{e^{\epsilon}+t}$,
$d = \frac{t (e^{\epsilon}-1)}{2(t+e^\epsilon)^2},
\xi  := \frac{c-d}{d}, \xi = e^{\epsilon}-1, \alpha =  \frac{t+1}{2(t+e^\epsilon)} $.
Then, we have
\begin{align}\label{PM-SUB-VAR}
 \textup{Var}_{\mathcal{P}}[Y|x] = \frac{t+1}{e^\epsilon-1}x^2 + \frac{(t+e^\epsilon)\big ((t+1)^3 + e^\epsilon -1 \big)}{3t^2(e^\epsilon-1)^2}.
\end{align}

Given \texttt{Three-Outputs}'s variance in Eq.~(\ref{eq:minimal-variance}), we can simplify it as
\begin{align}
    &\textup{Var}_{\mathcal{T}}[Y|x]=(1-a)C^2+C^2b|x|-x^2. \nonumber
\end{align}
According to Eq.~(\ref{eq:optimal-c}), we have $C=\frac{e^\epsilon(e^\epsilon+1)}{(e^\epsilon -1)(e^\epsilon -a)}$, so that
\begin{align}\label{Three-OUTPUTS-VAR}
    \textup{Var}_{\mathcal{T}}[Y|x]=\frac{(1-a)e^{2\epsilon}(e^\epsilon+1)^2}{(e^\epsilon -1)^2(e^\epsilon -a)^2}
    +\frac{b|x|e^{2\epsilon}(e^\epsilon+1)^2}{(e^\epsilon -1)^2(e^\epsilon -a)^2}-x^2.
\end{align}
Based on Eq.~(\ref{PM-SUB-VAR}) and Eq.~(\ref{Three-OUTPUTS-VAR}), we can construct variance of hybrid mechanism as follows
\begin{align}
    \hspace{-15pt} \nonumber &\textup{Var}_{\mathcal{H}}[Y|x]= \beta (\frac{t+1}{e^\epsilon-1}x^2 + \frac{(t+e^\epsilon)\big ((t+1)^3 + e^\epsilon -1 \big)}{3t^2(e^\epsilon-1)^2})+  \\& \nonumber (1-\beta) (\frac{(1-a)e^{2\epsilon}(e^\epsilon+1)^2}{(e^\epsilon -1)^2(e^\epsilon -a)^2}
    +\frac{b|x|e^{2\epsilon}(e^\epsilon+1)^2}{(e^\epsilon -1)^2(e^\epsilon -a)^2}-x^2),
\end{align}
where $ t = e^{\epsilon/3}$. 

From Eq.~(\ref{eq:optimal-c}), we set $b =  a\cdot \frac{e^\epsilon -1}{e^\epsilon}$ to get the worst-case noise variance. Then, we have variance of the hybrid mechanism as
\begin{align}
    \hspace{-10pt} \textup{Var}_{\mathcal{H}}[Y|x]&= \nonumber (\beta \frac{t+1}{e^\epsilon-1} + \beta - 1) x^2\\& \nonumber \quad + (1-\beta)  \frac{ae^\epsilon(e^\epsilon+1)^2}{(e^\epsilon-1)(e^\epsilon-a)^2} |x| \\& \nonumber \quad + \bigg ( \frac{(t+e^\epsilon)((t+1)^3+e^\epsilon-1)}{3t^2(e^\epsilon-1)^2} \beta \\& \quad +(1- \beta)(1-a) \frac{e^{2\epsilon}(e^\epsilon+1)^2}{(e^\epsilon-1)^2(e^\epsilon-a)^2} \bigg ),\label{Eq:var_1}
\end{align}
where $ t = e^{\epsilon/3}$. Based on Eq.~(\ref{Eq:var_1}), we get 
\begin{align}
    &\max_{x \in [-1,1]} \textup{Var}_{\mathcal{H}}[Y|x]= \nonumber\\
    & \begin{cases}
        \textup{Var}_{\mathcal{H}}[Y|x^*],\text{if}~\beta \frac{t+1}{e^\epsilon-1} + \beta - 1 < 0,  0 < x^* < 1,   \\
       \max \{\textup{Var}_{\mathcal{H}}[Y|0], \textup{Var}_{\mathcal{H}}[Y|1] \}, ~~\text{otherwise},
    \end{cases}
\end{align} where $x^* := \frac{(\beta-1)ae^\epsilon(e^\epsilon+1)^2}{2(e^\epsilon-a)^2(\beta (e^\epsilon+ t)-e^\epsilon+1)}$. 

Therefore, we have the following cases to compute $\max_{x \in [-1,1]} \textup{Var}_{\mathcal{H}}[Y|x]$:
\begin{itemize}
\item[(I)] 
    If $\beta \frac{t+1}{e^\epsilon-1} + \beta - 1 < 0,  0 < Y < 1$, we obtain:
    \begin{itemize}
        \item $\beta < \frac{e^\epsilon-1}{e^\epsilon+t}$,
        \item For $Y := \frac{(\beta-1)ae^\epsilon(e^\epsilon+1)^2}{2(e^\epsilon-a)^2(\beta (e^\epsilon+ t)-e^\epsilon+1)}$ , if $ 0 < Y < 1$, we have $0 < \frac{(\beta-1)ae^\epsilon(e^\epsilon+1)^2}{2(e^\epsilon-a)^2(\beta (e^\epsilon+ t)-e^\epsilon+1)} < 1$.
\end{itemize}

If $\frac{(\beta-1)ae^\epsilon(e^\epsilon+1)^2}{2(e^\epsilon-a)^2(\beta (e^\epsilon+ t)-e^\epsilon+1)} < 1$, we have :
    \begin{align}
       \nonumber &\beta (2(e^\epsilon-a)^2(e^\epsilon+t)-ae^\epsilon(e^\epsilon+1)^2) > \\& \quad 2(e^\epsilon-a)^2(e^\epsilon-1)-ae^\epsilon(e^\epsilon+1)^2.
    \end{align}

    \begin{itemize}
       \item If $2(e^\epsilon-a)^2(e^\epsilon+t)-ae^\epsilon(e^\epsilon+1)^2 > 0$, we have
        $$\beta > \frac{2(e^\epsilon-a)^2(e^\epsilon-1)-ae^\epsilon(e^\epsilon+1)^2}{2(e^\epsilon-a)^2(e^\epsilon+t)-ae^\epsilon(e^\epsilon+1)^2}.$$
       \item If $2(e^\epsilon-a)^2(e^\epsilon+t)-ae^\epsilon(e^\epsilon+1)^2 = 0$ and $2(e^\epsilon-a)^2(e^\epsilon-1)-ae^\epsilon(e^\epsilon+1)^2 > 0$, no $\beta$ satisfies the condition.
       \item If $2(e^\epsilon-a)^2(e^\epsilon+t)-ae^\epsilon(e^\epsilon+1)^2 = 0$ and $2(e^\epsilon-a)^2(e^\epsilon-1)-ae^\epsilon(e^\epsilon+1)^2 \leq 0$, any $\beta$ satisfies the condition.
       \item If $2(e^\epsilon-a)^2(e^\epsilon+t)-ae^\epsilon(e^\epsilon+1)^2 < 0$, we have \\ \quad
        $\beta < \frac{2(e^\epsilon-a)^2(e^\epsilon-1)-ae^\epsilon(e^\epsilon+1)^2}{2(e^\epsilon-a)^2(e^\epsilon+t)-ae^\epsilon(e^\epsilon+1)^2}$. Since $\beta < \frac{e^\epsilon-1}{e^\epsilon+t}$, to get the correct domain, we compare $\frac{2(e^\epsilon-a)^2(e^\epsilon-1)-ae^\epsilon(e^\epsilon+1)^2}{2(e^\epsilon-a)^2(e^\epsilon+t)-ae^\epsilon(e^\epsilon+1)^2}$ and $\frac{e^\epsilon-1}{e^\epsilon+t}$, see Appendix~\ref{compare}, we have $\frac{2(e^\epsilon-a)^2(e^\epsilon-1)-ae^\epsilon(e^\epsilon+1)^2}{2(e^\epsilon-a)^2(e^\epsilon+t)-ae^\epsilon(e^\epsilon+1)^2} \leq \frac{e^\epsilon-1}{e^\epsilon+t}$. Therefore, $\beta < \frac{2(e^\epsilon-a)^2(e^\epsilon-1)-ae^\epsilon(e^\epsilon+1)^2}{2(e^\epsilon-a)^2(e^\epsilon+t)-ae^\epsilon(e^\epsilon+1)^2}$.
    \end{itemize}
Summarize above analysis, we have the following cases to compute $\max_{x \in [-1,1]} \textup{Var}_{\mathcal{H}}[Y|x]$ :
    \begin{itemize}
        \item[1)] If $2(e^\epsilon-a)^2(e^\epsilon+t)-ae^\epsilon(e^\epsilon+1)^2 > 0$, we have:
            \begin{align}
                &\max_{x \in [-1,1]} \textup{Var}_{\mathcal{H}}[Y|x]= \nonumber\\
                & \begin{cases}
                   \nonumber \textup{Var}_{\mathcal{H}}[Y|x^*],~\text{if}~\frac{2(e^\epsilon-a)^2(e^\epsilon-1)-ae^\epsilon(e^\epsilon+1)^2}{2(e^\epsilon-a)^2(e^\epsilon+t)-ae^\epsilon(e^\epsilon+1)^2} <  \beta < \frac{e^\epsilon-1}{e^\epsilon+t},\\
                    \max \{\textup{Var}_{\mathcal{H}}[Y|0], \textup{Var}_{\mathcal{H}}[Y|1] \}, ~~\text{otherwise.} \nonumber
                \end{cases}
            \end{align}
        \item[2)] If $2(e^\epsilon-a)^2(e^\epsilon+t)-ae^\epsilon(e^\epsilon+1)^2 = 0$~and~$ae^\epsilon(e^\epsilon+1)^2 + 2(e^\epsilon-a)^2(1-e^\epsilon) > 0$, we have:
            \begin{align}
                &\max_{x \in [-1,1]} \textup{Var}_{\mathcal{H}}[Y|x]= \max \{\textup{Var}_{\mathcal{H}}[Y|0], \textup{Var}_{\mathcal{H}}[Y|1] \} \nonumber
            \end{align}
        \item[3)] If $2(e^\epsilon-a)^2(e^\epsilon+t)-ae^\epsilon(e^\epsilon+1)^2 =0$~and~$2(e^\epsilon-a)^2(e^\epsilon-1)-ae^\epsilon(e^\epsilon+1)^2 \leq 0$, we have:
            \begin{align}
                &\max_{x \in [-1,1]} \textup{Var}_{\mathcal{H}}[Y|x]= \nonumber \\
                &\begin{cases}
                    \textup{Var}_{\mathcal{H}}[Y|x^*],~\text{if}~ \beta < \frac{e^\epsilon-1}{e^\epsilon+t}, \\
                    \max \{\textup{Var}_{\mathcal{H}}[Y|0], \textup{Var}_{\mathcal{H}}[Y|1] \} , ~~\text{otherwise.} \nonumber
                \end{cases}
            \end{align}
       \item[4)] If $2(e^\epsilon-a)^2(e^\epsilon+t)- ae^\epsilon(e^\epsilon+1)^2<0$, we have:
            \begin{align}
                &\max_{x \in [-1,1]} \textup{Var}_{\mathcal{H}}[Y|x]= \nonumber\\
                & \begin{cases}
                    \textup{Var}_{\mathcal{H}}[Y|x^*],\text{if}~ \beta < \frac{2(e^\epsilon-a)^2(e^\epsilon-1)-ae^\epsilon(e^\epsilon+1)^2}{2(e^\epsilon-a)^2(e^\epsilon+t)-ae^\epsilon(e^\epsilon+1)^2},\\
                    \max \{\textup{Var}_{\mathcal{H}}[Y|0], \textup{Var}_{\mathcal{H}}[Y|1] \} , ~~\text{otherwise.} \nonumber
                \end{cases}
            \end{align}
    \end{itemize}
Appendix~\ref{condition_1_proof} proves that $$2(e^\epsilon-a)^2(e^\epsilon+t) - ae^\epsilon(e^\epsilon+1)^2 >0$$ and Appendix~\ref{compare}~ $$\frac{2(e^\epsilon-a)^2(e^\epsilon-1)-ae^\epsilon(e^\epsilon+1)^2}{2(e^\epsilon-a)^2(e^\epsilon+t)-ae^\epsilon(e^\epsilon+1)^2} \leq \frac{e^\epsilon-1}{e^\epsilon+t}.$$ Therefore, we have
\begin{align}
    &\max_{x \in [-1,1]} \textup{Var}_{\mathcal{H}}[Y|x]= \nonumber\\
    & \begin{cases}
        \textup{Var}_{\mathcal{H}}[Y|x^*],\text{if}~ 0 < \beta < \frac{2(e^\epsilon-a)^2(e^\epsilon-1)-ae^\epsilon(e^\epsilon+1)^2}{2(e^\epsilon-a)^2(e^\epsilon+t)-ae^\epsilon(e^\epsilon+1)^2},\\
        \max \{\textup{Var}_{\mathcal{H}}[Y|0], \textup{Var}_{\mathcal{H}}[Y|1] \}, ~~\text{otherwise.} \nonumber
    \end{cases}
\end{align}

\item[(II)] Based on above analysis, to make $\max_{x \in [-1,1]} \textup{Var}_{\mathcal{H}}[Y=y|x]=\max \{\textup{Var}_{\mathcal{H}}[Y|0], \textup{Var}_{\mathcal{H}}[Y|1] \}$, $\beta$ should satisfy constraint $  \frac{2(e^\epsilon-a)^2(e^\epsilon-1)-ae^\epsilon(e^\epsilon+1)^2}{2(e^\epsilon-a)^2(e^\epsilon+t)-ae^\epsilon(e^\epsilon+1)^2} \leq \beta \leq 1$.

To get the exact value of $\max_{x \in [-1,1]} \textup{Var}_{\mathcal{H}}[Y|x]$, we compare $\textup{Var}_{\mathcal{H}}[Y|1]$ and $\textup{Var}_{\mathcal{H}}[Y|0]$, values of $\textup{Var}_{\mathcal{H}}[Y|1]$ and $\textup{Var}_{\mathcal{H}}[Y|0]$ are: \\
\begin{itemize}
     \item $\textup{Var}_{\mathcal{H}}[Y|1] = (\beta \frac{t+1}{e^\epsilon-1} + \beta - 1) + (1-\beta)  \frac{ae^\epsilon(e^\epsilon+1)^2}{(e^\epsilon-1)(e^\epsilon-a)^2} + \big ( \frac{(t+e^\epsilon)((t+1)^3+e^\epsilon-1)}{3t^2(e^\epsilon-1)^2} \beta +(1- \beta)(1-a) \frac{e^{2\epsilon}(e^\epsilon+1)^2}{(e^\epsilon-1)^2(e^\epsilon-a)^2} \big ) = \beta \big( \frac{t+1}{e^\epsilon-1} + 1 - \frac{ae^\epsilon(e^\epsilon+1)^2}{(e^\epsilon-1)(e^\epsilon-a)^2} +  \frac{(t+e^\epsilon)((t+1)^3+e^\epsilon-1)}{3t^2(e^\epsilon-1)^2} -  \frac{(1-a)e^{2\epsilon}(e^\epsilon+1)^2}{(e^\epsilon-1)^2(e^\epsilon-a)^2} \big ) -1 + \frac{ae^\epsilon(e^\epsilon+1)^2}{(e^\epsilon-1)(e^\epsilon-a)^2} +  \frac{(1-a)e^{2\epsilon}(e^\epsilon+1)^2}{(e^\epsilon-1)^2(e^\epsilon-a)^2} $,
     \item $\textup{Var}_{\mathcal{H}}[Y|0] =  \frac{(t+e^\epsilon)((t+1)^3+e^\epsilon-1)}{3t^2(e^\epsilon-1)^2} \beta +(1- \beta)(1-a) \frac{e^{2\epsilon}(e^\epsilon+1)^2}{(e^\epsilon-1)^2(e^\epsilon-a)^2} = \beta \big (  \frac{(t+e^\epsilon)((t+1)^3+e^\epsilon-1)}{3t^2(e^\epsilon-1)^2} -  \frac{(1-a)e^{2\epsilon}(e^\epsilon+1)^2}{(e^\epsilon-1)^2(e^\epsilon-a)^2} \big )  +  \frac{(1-a)e^{2\epsilon}(e^\epsilon+1)^2}{(e^\epsilon-1)^2(e^\epsilon-a)^2}$.
\end{itemize}
Since $\textup{Var}_{\mathcal{H}}[Y|1]$ and $\textup{Var}_{\mathcal{H}}[Y|0]$ are linear equations respect to $\beta$, we compare slopes of $\beta$ in $\textup{Var}_{\mathcal{H}}[Y|1]$ and $\textup{Var}_{\mathcal{H}}[Y|0]$. We define the slope of $\beta$ in $\textup{Var}_{\mathcal{H}}[Y|1]$ as
\begin{align}
 \nonumber  &\textup{slope}_1 := \frac{t+1}{e^\epsilon-1} + 1 - \frac{ae^\epsilon(e^\epsilon+1)^2}{(e^\epsilon-1)(e^\epsilon-a)^2} \\&+  \frac{(t+e^\epsilon)((t+1)^3+e^\epsilon-1)}{3t^2(e^\epsilon-1)^2} -  \frac{(1-a)e^{2\epsilon}(e^\epsilon+1)^2}{(e^\epsilon-1)^2(e^\epsilon-a)^2},\label{de:slope_1}
\end{align}
and the slope of $\beta$ in $\textup{Var}_{\mathcal{H}}[Y|0]$ as
\begin{align}
 \textup{slope}_2 := \frac{(t+e^\epsilon)((t+1)^3+e^\epsilon-1)}{3t^2(e^\epsilon-1)^2} -  \frac{(1-a)e^{2\epsilon}(e^\epsilon+1)^2}{(e^\epsilon-1)^2(e^\epsilon-a)^2}.\label{de:slope_2}   
\end{align}
Then, we represent left boundary of $\beta$ as
\begin{align}
    \beta_1 := \frac{2(e^\epsilon-a)^2(e^\epsilon-1)-ae^\epsilon(e^\epsilon+1)^2}{2(e^\epsilon-a)^2(e^\epsilon+t)-ae^\epsilon(e^\epsilon+1)^2},\label{de:beta_1}
\end{align}
and the right boundary of $\beta$ as
\begin{align}
    \beta_2 := 1, \label{de:beta_2}
\end{align}
and the value of $\beta$ at the intersection of $\textup{\textup{slope}}_1$ and $\textup{\textup{slope}}_2$ is 
\begin{align}
    \beta_{intersection} := \frac{(c-1)(c-a)^2 - ac(c+1)^2}{(t+c)(c-a)^2-ac(c+1)^2}.\label{de:intersection}
\end{align}
\end{itemize}

Then, $\textup{slope}_1$ and $\textup{slope}_2$ have the following possible combinations:
\begin{itemize}
    \item[1)] If $\textup{slope}_1 >0$, $\textup{slope}_2 > 0$, $\beta = \beta_1$.
    \item[2)] If $\textup{slope}_1 <0$, $\textup{slope}_2 < 0$, $\beta = \beta_2$.
    \item[3)] If $\textup{slope}_1 \cdot \textup{slope}_2 < 0$, If $\beta_1<\beta_{intersection} < \beta_2$, $\beta = \beta_{intersection}$.
    \item[4)] If $\textup{slope}_1 \cdot \textup{slope}_2 < 0$, If $\beta_{\beta_{intersection}} < \beta_1$ or $\beta_{intersection} > \beta_2$,~find $\beta$ for $\min \big \{ \max \{ \textup{Var}[Y|1, \beta = \beta_1], \textup{Var}[Y|0, \beta = \beta_1] \}\\, \max \{ \textup{Var}[Y|1,\beta =  \beta_2], \textup{Var}[Y|0,\beta =  \beta_2] \}\big \}$.
    \item[5)]If $\textup{slope}_1 \cdot \textup{slope}_2 = 0$, \\
    \textbf{Case 1:} $\textup{slope}_1 = 0$, $\textup{slope}_2 \neq 0$,
    \begin{itemize}
        \item[a)] If $\textup{slope}_1 = 0$, $\textup{slope}_2 > 0$, $\beta_{intersection} \in [\beta_1, \beta_2]$, $\beta = [\beta_1, \beta_{intersection}]$.
        \item[b)] If $\textup{slope}_1 = 0$, $\textup{slope}_2 <0$, $\beta_{intersection} \in [\beta_1, \beta_2]$, $\beta = [\beta_{intersection}, \beta_2]$.
        \item[c)] If $\textup{slope}_1 = 0$, $\textup{slope}_2 > 0$, $\beta_{intersection} < \beta_1$ or $\beta_{intersection} > \beta_2$, $\textup{Var}_{\mathcal{H}}[Y|1, \beta = [\beta_1, \beta_2]] > \textup{Var}_{\mathcal{H}}[Y|0, \beta = [\beta_1, \beta_2]]$, $\beta = [\beta_1,\beta_2]$.
        \item[d)] If $\textup{slope}_1 = 0$, $\textup{slope}_2 > 0$, $\beta_{intersection} < \beta_1$ or $\beta_{intersection} > \beta_2$, $\textup{Var}_{\mathcal{H}}[Y|0, \beta = [\beta_1, \beta_2]] > \textup{Var}_{\mathcal{H}}[Y|1, \beta = [\beta_1, \beta_2]]$, $\beta = \beta_1$.
        \item[e)] If $\textup{slope}_1 = 0$, $\textup{slope}_2 < 0$, $\beta_{intersection} < \beta_1$ or $\beta_{intersection} > \beta_2$, $\textup{Var}_{\mathcal{H}}[Y|1, \beta = [\beta_1, \beta_2]] > \textup{Var}_{\mathcal{H}}[Y|0, \beta = [\beta_1, \beta_2]]$, $\beta = [\beta_1,\beta_2]$.
        \item[f)] If $\textup{slope}_1 = 0$, $\textup{slope}_2 < 0$, $\beta_{intersection} < \beta_1$ or $\beta_{intersection} > \beta_2$, $\textup{Var}_{\mathcal{H}}[Y|1, \beta = [\beta_1, \beta_2]] < \textup{Var}_{\mathcal{H}}[Y|0, \beta = [\beta_1, \beta_2]]$, $\beta = \beta_2$.
    \end{itemize}
    
    \textbf{Case 2:} $\textup{slope}_2 = 0$, $\textup{slope}_1 \neq 0$,
    
    \begin{itemize}
        \item[a)] If $\textup{slope}_1 >0$, $\textup{slope}_2 = 0$ , $\beta_{intersection} \in [\beta_1, \beta_2]$, $\beta = [\beta_1, \beta_{intersection}]$.
        \item[b)] If $\textup{slope}_1 <0$, $\textup{slope}_2 = 0$ , $\beta_{intersection} \in [\beta_1, \beta_2]$, $\beta = [\beta_{intersection}, \beta_2]$.
        \item[c)] If $\textup{slope}_2 = 0$, $\textup{slope}_1 > 0$, $\beta_{intersection} < \beta_1$ or $\beta_{intersection} > \beta_2$, $\textup{Var}_{\mathcal{H}}[Y|1, \beta = [\beta_1, \beta_2]] > \textup{Var}_{\mathcal{H}}[Y|0, \beta = [\beta_1, \beta_2]]$, $\beta = \beta_1$.
        \item[d)] If $\textup{slope}_2 = 0$, $\textup{slope}_1 > 0$, $\beta_{intersection} < \beta_1$ or $\beta_{intersection} > \beta_2$, $\textup{Var}_{\mathcal{H}}[Y|0, \beta = [\beta_1, \beta_2]] > \textup{Var}_{\mathcal{H}}[Y|1, \beta = [\beta_1, \beta_2]]$, $\beta = [\beta_1, \beta_2]$.
        \item[e)] If $\textup{slope}_2 = 0$, $\textup{slope}_1 < 0$, $\beta_{intersection} < \beta_1$ or $\beta_{intersection} > \beta_2$, $\textup{Var}_{\mathcal{H}}[Y|1, \beta = [\beta_1, \beta_2]] > \textup{Var}_{\mathcal{H}}[Y|0, \beta = [\beta_1, \beta_2]]$, $\beta = \beta_2$.
        \item[f)] If $\textup{slope}_2 = 0$, $\textup{slope}_1 < 0$, $\beta_{intersection} < \beta_1$ or $\beta_{intersection} > \beta_2$, $\textup{Var}_{\mathcal{H}}[Y|1, \beta = [\beta_1, \beta_2]] < \textup{Var}_{\mathcal{H}}[Y|0, \beta = [\beta_1, \beta_2]]$, $\beta = [\beta_1, \beta_2]$.
    \end{itemize}
    
    \textbf{Case 3:} $\textup{slope}_1 = 0$ and $\textup{slope}_2 = 0$,
    \begin{itemize}
        \item[a)] If $\textup{Var}_{\mathcal{H}}[Y|1, \beta = [\beta_1, \beta_2]] < \textup{Var}_{\mathcal{H}}[Y|0, \beta = [\beta_1, \beta_2]]$, $\beta = [\beta_1,\beta_2]$.
        \item[b)] If $\textup{Var}_{\mathcal{H}}[Y|1, \beta = [\beta_1, \beta_2]] > \textup{Var}_{\mathcal{H}}[Y|0, \beta = [\beta_1, \beta_2]]$, $\beta = [\beta_1,\beta_2]$.
        \item[c)] If $\textup{Var}_{\mathcal{H}}[Y|1, \beta = [\beta_1, \beta_2]] = \textup{Var}_{\mathcal{H}}[Y|0, \beta = [\beta_1, \beta_2]]$, $\beta = [\beta_1,\beta_2]$.
    \end{itemize}

\end{itemize}

\begin{proof}
1) If $\textup{slope}_1 >0$, $\textup{slope}_2 > 0$, $\textup{Var}_{\mathcal{H}}[Y|1]$ and $\textup{Var}_{\mathcal{H}}[Y|0]$~monotonically increase $\beta \in [\beta_1, \beta_2]$, $\min_\beta \max_{x \in [-1,1]} \textup{Var}_{\mathcal{H}}[Y|x]$ is at $\beta = \beta_1$.

2) Similar to 1), we have $\min_\beta \max_{x \in [-1,1]} \textup{Var}_{\mathcal{H}}[Y|x]$ is at $\beta =1$.

3) If $\beta_1<\beta_{intersection} < \beta_2$, we have:
    \begin{itemize}
        \item If~$\textup{slope}_1 > 0,~\textup{slope}_2 < 0$, $\textup{Var}[Y|1]$ monotonically increases and $\textup{Var}[Y|0]$ monotonically decreases, so~ when $\beta \in [\beta_1, \beta_{intersection}]$, $\max_{x \in [-1,1]} \textup{Var}_{\mathcal{H}}[Y|x]= \textup{Var}_{\mathcal{H}}[Y|0]$. When $\beta \in [\beta_{intersection}, \beta_2]$, $\max_{x \in [-1,1]} \textup{Var}_{\mathcal{H}}[Y|x]= \textup{Var}_{\mathcal{H}}[Y|1]$. Therefore, $\min_\beta \max_{x \in [-1,1]} \textup{Var}_{\mathcal{H}}[Y|x]= \textup{Var}_{\mathcal{H}}[Y|1] = \textup{Var}_{\mathcal{H}}[Y|0]$ at $\beta = \beta_{intersection}$.
        \item If~$\textup{slope}_1 < 0,~\textup{slope}_2 > 0$, $\textup{Var}[Y|1]$ monotonically decreases and $\textup{Var}[Y|0]$ monotonically increases, so~ when $\beta \in [\beta_1, \beta_{intersection}]$, $\max_{x \in [-1,1]} \textup{Var}_{\mathcal{H}}[Y|x]= \textup{Var}_{\mathcal{H}}[Y|1]$. When $\beta \in [\beta_{intersection}, \beta_2]$, $\max_{x \in [-1,1]} \textup{Var}_{\mathcal{H}}[Y|x]= \textup{Var}_{\mathcal{H}}[Y|0]$. Therefore, $\min_\beta \max_{x \in [-1,1]} \textup{Var}_{\mathcal{H}}[Y|x]= \textup{Var}_{\mathcal{H}}[Y|1] = \textup{Var}_{\mathcal{H}}[Y|0]$ at $\beta = \beta_{intersection}$.
    \end{itemize}
4) If $ \beta_{intersection} < \beta_1$ and $\beta_2 > \beta_{intersection}$, we have:
    \begin{itemize}
        \item If $\textup{slope}_1 > 0$, $\textup{slope}_2 < 0$ and $ \textup{Var}_{\mathcal{H}}[Y|1] > \textup{Var}_{\mathcal{H}}[Y|0]$, since $\textup{Var}_{\mathcal{H}}[Y|1]$ monotonically increases in the domain, $\min_\beta \textup{Var}_{\mathcal{H}}[Y|1]$ is at $\beta = \beta_1$. $\min_\beta \max_{x \in [-1,1]} \textup{Var}_{\mathcal{H}}[Y|x]$ is at $\beta = \beta_1$.
        \item If $\textup{slope}_1 > 0$, $\textup{slope}_2 < 0$ and $ \textup{Var}_{\mathcal{H}}[Y|1] < \textup{Var}_{\mathcal{H}}[Y|0]$. Since $\textup{Var}_{\mathcal{H}}[Y|0]$ monotonically decreases in the domain, $\min_\beta \textup{Var}_{\mathcal{H}}[Y|0]$ is at $\beta = \beta_2$. $\min_\beta \max_{x \in [-1,1]} \textup{Var}_{\mathcal{H}}[Y|x]$ is at $\beta = \beta_2$.
        \item If $\textup{slope}_1 < 0$, $\textup{slope}_2 > 0$ and $\textup{Var}_{\mathcal{H}}[Y|1] > \textup{Var}_{\mathcal{H}}[Y|0], \max_{x \in [-1,1]}\textup{Var}_{\mathcal{H}}[Y|x] = \textup{Var}_{\mathcal{H}}[Y|1]$, since $\textup{Var}_{\mathcal{H}}[Y|1]$ monotonically decreases in the domain, $\min_\beta \textup{Var}_{\mathcal{H}}[Y|1]$ is at $\beta = \beta_2$. $\min_\beta \max_{x \in [-1,1]} \textup{Var}_{\mathcal{H}}[Y|x]$ is at $\beta = \beta_2$.
        \item If $\textup{slope}_1 < 0$ , $\textup{slope}_2 > 0$ and $ \textup{Var}_{\mathcal{H}}[Y|1] < \textup{Var}_{\mathcal{H}}[Y|0]$, $\max_{x \in [-1,1]} \textup{Var}_{\mathcal{H}}[Y|x] = \textup{Var}_{\mathcal{H}}[Y|0]$, since $\textup{Var}_{\mathcal{H}}[Y|0]$ monotonically increases in the domain, $\min_\beta \textup{Var}_{\mathcal{H}}[Y|0]$ is at $\beta = \beta_1$. $\min_\beta \max_{x \in [-1,1]} \textup{Var}_{\mathcal{H}}[Y|x]$ is at $\beta = \beta_1$.
    \end{itemize}
5)\begin{itemize}
   \item \textbf{Case 1:}
       \begin{itemize}
            \item[a)] If $\textup{slope}_1 = 0$, $\textup{slope}_2 >0$, $\beta_{intersection} \in [\beta_1, \beta_2]$, we can conclude that $\beta \in [\beta_1, \beta_{intersection}], \textup{Var}_{\mathcal{H}}[Y|1] > \textup{Var}_{\mathcal{H}}[Y|0] $. When $\beta \in (\beta_{intersection}, \beta_2], \textup{Var}_{\mathcal{H}}[Y|0] > \textup{Var}_{\mathcal{H}}[Y|1]$. Since $\textup{Var}_{\mathcal{H}}[Y|0]$ monotonically increases if $\beta \in (\beta_{intersection}, \beta_2]$, $\min_\beta \max_{x \in [-1,1]} \textup{Var}_{\mathcal{H}}[Y|x] = \textup{Var}_{\mathcal{H}}[Y|1, \beta = [\beta_1, \beta_{intersection}]]$. 
            \item[b)] If $\textup{slope}_1 = 0, \textup{slope}_2 <0, \beta_{intersection} \in [\beta_1, \beta_2]$, we can conclude that $\beta \in [\beta_1, \beta_{intersection}], \textup{Var}_{\mathcal{H}}[Y|0] > \textup{Var}_{\mathcal{H}}[Y|1] $. When $\beta \in (\beta_{intersection}, \beta_2], \textup{Var}_{\mathcal{H}}[Y|1] > \textup{Var}_{\mathcal{H}}[Y|0] $. Since $\textup{Var}_{\mathcal{H}}[Y|1]$ does not change and $\textup{Var}_{\mathcal{H}}[Y|0]$ monotonically decreases, $\min_\beta \max_{x \in [-1,1]} \textup{Var}_{\mathcal{H}}[Y|x] = \textup{Var}_{\mathcal{H}}[Y|0, \beta = [ \beta_{intersection}, \beta_2]]$. 
            \item[c)] If $\textup{slope}_1 = 0, \textup{slope}_2 >0, \textup{Var}_{\mathcal{H}}[Y|1, \beta = [\beta_1, \beta_2]] > \textup{Var}_{\mathcal{H}}[Y|0, \beta = [\beta_1, \beta_2]]$, the $\beta_{intersection} > \beta_2$. Since $\textup{Var}_{\mathcal{H}}[Y|1]$ does not change if $\beta \in [\beta_1, \beta_2]$, $\min_\beta \max_{x \in [-1,1]} \textup{Var}_{\mathcal{H}}[Y|x] = \textup{Var}_{\mathcal{H}}[Y|1, \beta =[\beta_1, \beta_2]]$.
            \item[d)] If $\textup{slope}_1 = 0, \textup{slope}_2 >0, \textup{Var}_{\mathcal{H}}[Y|0, \beta = [\beta_1, \beta_2]] > \textup{Var}_{\mathcal{H}}[Y|1, \beta = [\beta_1, \beta_2]]$, the $\beta_{intersection} > \beta_2$. Since $\textup{Var}_{\mathcal{H}}[Y|0]$ monotonically decreases if $\beta \in [\beta_1, \beta_2]$, $\min_\beta \max_{x \in [-1,1]} \textup{Var}_{\mathcal{H}}[Y|x] = \textup{Var}_{\mathcal{H}}[Y|0, \beta = \beta_2]$.
            \item[e)] If $\textup{slope}_1 = 0, \textup{slope}_2 <0, \textup{Var}_{\mathcal{H}}[Y|1, \beta = [\beta_1, \beta_2]] > \textup{Var}_{\mathcal{H}}[Y|0, \beta = [\beta_1, \beta_2]]$, the $\beta_{intersection} < \beta_1$. Since $\textup{Var}_{\mathcal{H}}[Y|1]$ does not change if $\beta \in [\beta_1, \beta——2]$, $\min_\beta \max_{x \in [-1,1]} \textup{Var}_{\mathcal{H}}[Y|x] = \textup{Var}_{\mathcal{H}}[Y|1, \beta = [\beta_1, \beta_2]]$.
            \item[f)] If $\textup{slope}_1 = 0, \textup{slope}_2 <0, \textup{Var}_{\mathcal{H}}[Y|0, \beta = [\beta_1, \beta_2]] > \textup{Var}_{\mathcal{H}}[Y|1, \beta = [\beta_1, \beta_2]]$, the $\beta_{intersection} > \beta_2$. Since $\textup{Var}_{\mathcal{H}}[Y|0]$ monotonically decreases if $\beta \in [\beta_1, \beta_2]$, $\min_\beta \max_{x \in [-1,1]} \textup{Var}_{\mathcal{H}}[Y|x] = \textup{Var}_{\mathcal{H}}[Y|0, \beta = \beta_2]$.
       \end{itemize}
    \item \textbf{Case 2:} The proof is similar to Case 1.
    \item \textbf{Case 3:} $\textup{Var}_{\mathcal{H}}[Y|0]$ and $\textup{Var}_{\mathcal{H}}[Y|1]$ are unchanged when $\beta \in [\beta_1, \beta_2]$. Hence, $\min_\beta \max_{x \in [-1,1]} \textup{Var}_{\mathcal{H}}[Y|x] = \max_{x \in [-1,1]} \textup{Var}_{\mathcal{H}}[Y|x, \beta = [\beta_1, \beta_2]]$.
\end{itemize}

\end{proof}
\end{document}